\documentclass[manuscript, screen]{acmart}

\usepackage{subfigure}
\usepackage{enumitem}

\usepackage[ruled, vlined, linesnumbered]{algorithm2e}

\usepackage{multirow}

\usepackage{textcomp}
\usepackage{pifont}

\setcopyright{none}

\begin{document}

\title{Efficient Sampling Algorithms for Approximate Motif Counting in Temporal Graph Streams}

\author{Jingjing Wang}
\authornote{This research was done when Dr.~Jingjing Wang worked as a Ph.D.~student at Hunan University.}
\affiliation{%
  \department{Department of Mathematics and Computer Science}
  \institution{Changsha University}
  \city{Changsha}
  \country{China}
}
\email{wangjingjing@ccsu.edu.cn}

\author{Yanhao Wang}
\affiliation{%
  \department{School of Data Science and Engineering}
  \institution{East China Normal University}
  \city{Shanghai}
  \country{China}
}
\email{yhwang@dase.ecnu.edu.cn}

\author{Wenjun Jiang}
\authornote{Corresponding author}
\affiliation{%
  \department{College of Computer Science and Electronic Engineering}
  \institution{Hunan University}
  \city{Changsha}
  \country{China}
}
\email{jiangwenjun@hnu.edu.cn}

\author{Yuchen Li}
\affiliation{%
  \department{School of Computing and Information Systems}
  \institution{Singapore Management University}
  \city{Singapore}
  \country{Singapore}
}
\email{yuchenli@smu.edu.sg}

\author{Kian-Lee Tan}
\affiliation{%
  \department{School of Computing}
  \institution{National University of Singapore}
  \city{Singapore}
  \country{Singapore}
}
\email{tankl@comp.nus.edu.sg}

\renewcommand{\shortauthors}{Wang et al.}

\begin{abstract}
A great variety of complex systems ranging from user interactions in communication networks to transactions in financial markets can be modeled as \emph{temporal graphs}, which consist of a set of vertices and a series of timestamped and directed edges. \emph{Temporal motifs} in temporal graphs are generalized from subgraph patterns in static graphs which take into account edge orderings and durations in addition to topologies. Counting the number of occurrences of temporal motifs is a fundamental problem for temporal network analysis. However, existing methods either cannot support temporal motifs or suffer from performance issues. Moreover, they cannot work in the streaming model where edges in a temporal graph are observed incrementally over time. In this paper, we focus on approximate temporal motif counting via random sampling. We first propose two sampling algorithms for approximate temporal motif counting in the offline setting, where the whole temporal graph is available in advance and can be kept in main memory. The first is a generic edge sampling (ES) algorithm for estimating the number of instances of any temporal motif and the second is an improved edge-wedge sampling (EWS) algorithm that hybridizes edge sampling with wedge sampling for counting temporal motifs with $3$ vertices and $3$ edges. Furthermore, we propose two algorithms to count temporal motifs incrementally in temporal graph streams by extending the ES and EWS algorithms, which are referred to as the SES and SEWS algorithms, respectively. We provide comprehensive analyses of the theoretical bounds and complexities of our proposed algorithms. Finally, we perform extensive experimental evaluation of our proposed algorithms on several real-world temporal graphs. The results show that ES and EWS have higher efficiency, better accuracy, and greater scalability than state-of-the-art sampling methods for temporal motif counting in the offline setting. What is more, SES and SEWS further achieve up to three orders of magnitude speedups over ES and EWS while having comparable estimation errors for temporal motif counting in the streaming setting.
\end{abstract}

\ccsdesc[500]{Theory of computation~Sketching and sampling}
\ccsdesc[500]{Theory of computation~Graph algorithms analysis}

\keywords{temporal network; motif counting; random sampling}

\maketitle

\section{Introduction}\label{sec:intro}

Graphs are one of the most fundamental data structures that are widely used for modeling complex systems across diverse domains from bioinformatics~\cite{DBLP:journals/bioinformatics/Przulj07}, to neuroscience~\cite{DBLP:journals/ploscb/VarshneyCPHC11}, to social sciences~\cite{FAUST2010221}. Modern graph datasets increasingly incorporate temporal information to describe the dynamics of relations over time. Such graphs are referred to as \emph{temporal graphs}~\cite{HOLME201297} and typically represented by a set of vertices and a sequence of timestamped and directed edges between vertices called \emph{temporal edges}. For example, a communication network~\cite{DBLP:conf/cikm/ZhaoTHOJL10, DBLP:conf/sigmod/GurukarRR15, DBLP:journals/pvldb/WangFLT17, DBLP:journals/tois/WangLFT18, DBLP:conf/edbt/WangLT19} is often denoted by a temporal graph, where each person is a vertex and each message sent from one person to another is a temporal edge. Similarly, computer networks and financial transactions can also be modeled as temporal graphs. Due to the ubiquitousness of temporal graphs, they have attracted much attention~\cite{DBLP:conf/sigmod/GurukarRR15, DBLP:conf/wsdm/ParanjapeBL17, DBLP:conf/cikm/ZhaoTHOJL10, DBLP:conf/icde/LiSQYD18, DBLP:conf/cikm/NamakiWSLG17, DBLP:conf/cikm/GalimbertiBBCG18, DBLP:journals/pvldb/ShaLHT17, DBLP:journals/pvldb/GuoLST17} recently.

One fundamental problem in temporal graphs with wide real-world applications such as network characterization~\cite{DBLP:conf/wsdm/ParanjapeBL17}, structure prediction~\cite{DBLP:conf/wsdm/LiuBC19}, and fraud detection~\cite{DBLP:journals/pvldb/KumarC18}, is to count the number of occurrences of small (connected) subgraph patterns (i.e., \emph{motifs}~\cite{Milo824}). To capture the temporal dynamics in network analysis, the notion of \emph{motif}~\cite{DBLP:conf/wsdm/ParanjapeBL17, DBLP:conf/edbt/KosyfakiMPT19, Kovanen_2011, DBLP:conf/wsdm/LiuBC19} in temporal graphs is more general than its counterpart in static graphs. It takes into account not only the subgraph structure (i.e., \emph{subgraph isomorphism}~\cite{DBLP:journals/jacm/Ullmann76, DBLP:conf/sigmod/GuoL0HXT20}) but also the temporal information including edge ordering and motif duration. As an illustrative example, two temporal motifs $M$ and $M'$ in Fig.~\ref{fig:example0} are different temporal motifs. Though they are exactly the same in structure, they are distinguished from each other by the ordering of edges. Consequently, although there has been a considerable amount of work on subgraph counting~\cite{DBLP:journals/tkdd/WangLRTZG14, DBLP:conf/www/JhaSP15, DBLP:conf/icde/WangLTZ16, DBLP:conf/wsdm/0002C00P17, DBLP:conf/www/PinarSV17, DBLP:journals/tkde/WangZZLCLTTG18, DBLP:conf/sdm/KoldaPS13, DBLP:conf/icdm/TurkogluT17, DBLP:conf/www/TurkT19, DBLP:conf/cikm/EtemadiLT16} in static graphs, they cannot be directly used for counting temporal motifs.

Generally, it is a challenging task to count temporal motifs. Firstly, the problem is at least as hard as subgraph counting in static graphs, whose time complexity increases exponentially with the number of edges in the query subgraph. Secondly, it becomes even more computationally difficult because the temporal information is taken into consideration. For example, counting the number of instances of $k$-stars is simple in static graphs. However, counting temporal $k$-stars is proven to be NP-hard~\cite{DBLP:conf/wsdm/LiuBC19} due to the combinatorial nature of edge ordering. Thirdly, temporal graphs are a kind of \emph{multi-graph} that is permitted to have multiple edges between the same two vertices at different timestamps. As a result, there may exist many different instances of a temporal motif within the same set of vertices, which leads to more challenges for counting problems. There have been a few methods for exact temporal motif counting~\cite{DBLP:conf/wsdm/ParanjapeBL17} or enumeration~\cite{DBLP:conf/bigdataconf/MackeyPFCC18, DBLP:journals/pvldb/KumarC18}. However, they suffer from efficiency issues and often cannot scale well in massive temporal graphs with hundreds of millions of edges~\cite{DBLP:conf/wsdm/LiuBC19}.

In many scenarios, it is not necessary to count motifs exactly, and finding an approximate number is sufficient for practical use. A recent work~\cite{DBLP:conf/wsdm/LiuBC19} has proposed a sampling method for approximate temporal motif counting. It partitions a temporal graph into equal-time intervals, utilizes an exact algorithm~\cite{DBLP:conf/bigdataconf/MackeyPFCC18} to count the number of motif instances in a subset of intervals, and computes an estimate from the per-interval counts. However, this method still cannot achieve satisfactory performance in massive datasets. On the one hand, it fails to provide an accurate estimate when the sampling rate and length of intervals are small. On the other hand, its efficiency is not significantly improved upon that of exact methods when the sampling rate and length of intervals are too large.

\begin{figure}
  \centering
  \includegraphics[width=0.3\textwidth]{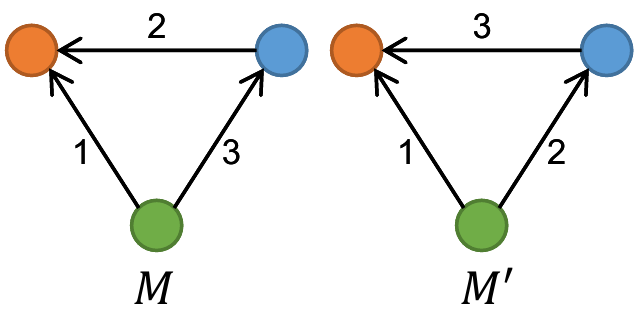}
  \caption{Examples for temporal motifs}
  \Description{example0}
  \label{fig:example0}
\end{figure}

Moreover, the vertices and edges in a temporal graph are typically observed incrementally over time in a streaming manner. For example, in communication networks, newly registered users are observed as new vertices and new messages between any two existing or new users from time to time are generated continuously as new temporal edges. In such scenarios, it is almost impossible to obtain the whole temporal graph all at once. Even if the whole dataset is already available, it may still be infeasible to keep it entirely in main memory for motif counting due to high space consumption. In addition, when new edges arrive over time, an offline counting method has to be rerun from scratch for maintaining the count in the updated dataset. However, to the best of knowledge, all the existing methods~\cite{DBLP:conf/wsdm/ParanjapeBL17, DBLP:conf/wsdm/LiuBC19, DBLP:conf/bigdataconf/MackeyPFCC18, DBLP:journals/pvldb/KumarC18} for temporal motif counting and enumeration are designed for the offline setting and become very inefficient for temporal graph streams.

To address the above problems, we propose more efficient and accurate sampling algorithms for approximate temporal motif counting in this paper. The basic idea of our algorithms is to first uniformly draw a set of random edges from a temporal graph (stream), then exactly count or estimate the number of \emph{local} motif instances that contain each sampled edge, and finally compute the global motif count from local counts. We first propose two offline algorithms for temporal motif counting. For any $k$-vertex $l$-edge temporal motif, we propose a generic Edge Sampling (ES) algorithm, which exactly counts the number of local motif instances by enumerating them. Next, temporal motifs with $3$ vertices and $3$ edges (i.e., triadic patterns) are one of the most important classes of motifs, whose distribution is an indicator to characterize temporal networks~\cite{DBLP:conf/sdm/KoldaPS13, DBLP:conf/wsdm/ParanjapeBL17, DBLP:journals/snam/UzupyteW20, FAUST2010221}. We propose an improved Edge-Wedge Sampling (EWS) algorithm for counting any $3$-vertex $3$-edge temporal motif, which estimates the local counts by \emph{wedge sampling}~\cite{DBLP:conf/sdm/KoldaPS13, DBLP:conf/icdm/TurkogluT17}. Furthermore, based on the above two offline algorithms, we propose a reservoir sampling-based framework to extend them for counting temporal motifs in temporal graph streams. We analyze the theoretical bounds and complexities of our proposed algorithms, and perform extensive experiments to show their accuracy and efficiency. Our main contributions in this paper are summarized as follows.
\begin{itemize}
  \item We propose a generic Edge Sampling (ES) algorithm to estimate the number of instances of any temporal motif in a temporal graph. It exploits the \textsc{BackTracking} (BT) algorithm~\cite{DBLP:journals/jacm/Ullmann76, DBLP:conf/bigdataconf/MackeyPFCC18} for subgraph isomorphism to enumerate local motif instances. We devise simple heuristics to determine the matching order of a temporal motif to reduce the search space.
  \item We propose an improved Edge-Wedge Sampling (EWS) algorithm that combines \emph{edge sampling} with \emph{wedge sampling}~\cite{DBLP:conf/sdm/KoldaPS13, DBLP:conf/icdm/TurkogluT17} specialized for counting any $3$-vertex $3$-edge temporal motif. Instead of enumerating all instances containing a sampled edge, EWS estimates the number of local instances via \emph{temporal wedge sampling}. In this way, EWS avoids the computationally intensive enumeration and greatly improves the efficiency upon ES.
  \item We further propose two algorithms on top of ES and EWS, namely SES and SEWS, to estimate the number of instances of a temporal motif over a temporal graph stream. SES and SEWS utilize the same methods as ES and EWS, respectively, to compute the local count for each sampled edge. Moreover, they adopt a reservoir sampling-based framework to maintain a fixed-size set of sampled edges over time and thus always keep the up-to-date global count dynamically w.r.t.~the set of sampled edges.
  \item Finally, we evaluate the performance of our proposed algorithms on several real-world temporal graphs. The experimental results confirm the accuracy, efficiency, and scalability of our proposed algorithms. In the offline setting, ES and EWS run up to $10.3$ and $48.5$ times faster than the state-of-the-art sampling method while having lower estimation errors. In the streaming setting, SES and SEWS further achieve up to three orders of magnitude speedups over ES and EWS while the estimation errors are still comparable.
\end{itemize}

\textbf{Differences from Prior Conference Paper~\cite{DBLP:conf/cikm/00040JLT20}:}
A preliminary version~\cite{DBLP:conf/cikm/00040JLT20} of this paper has been published on CIKM '20. The new contributions of this extended version are listed as follows. Firstly, all the proofs for the unbiasedness and variances of ES and EWS, which are omitted in the preliminary version due to space limitations, are included in the extended version. Secondly, we propose two novel algorithms (i.e., SES and SEWS) for temporal motif counting in temporal graph streams. We analyze the theoretical bounds and complexities of SES and SEWS. Thirdly, we conduct new experiments to evaluate the performance of SES and SEWS in temporal graph streams. And the experimental results confirm the efficiency and effectiveness of SES and SEWS over ES and EWS in the streaming setting.

\textbf{Paper Organization:}
The remainder of this paper is organized as follows. Section~\ref{sec:related:work} reviews the related work. Section~\ref{sec:def} introduces the background and formulation of \emph{temporal motif counting}. Section~\ref{sec:alg} presents the ES and EWS algorithms for temporal motif counting and analyzes them theoretically. Section~\ref{sec:s-alg} proposes the SES and SEWS algorithms for streaming temporal motif counting and provides the theoretical analyses accordingly. Section~\ref{sec:exp} describes the setup and results of the experiments. Finally, Section~\ref{sec:conclusion} provides some concluding remarks.

\section{Related Work}\label{sec:related:work}

\textbf{Subgraph (Motif) Counting in Static Graphs:}
The problem of counting the number of instances of a query subgraph in a large data graph has been extensively studied across several decades. Since counting the exact number of instances by enumeration is computationally intensive due to the NP-hardness of \emph{subgraph isomorphism}~\cite{DBLP:journals/jacm/Ullmann76}, more efforts have been made to estimate the counts within bounded errors using random sampling (see~\cite{DBLP:journals/csur/RibeiroPSAS21} for a survey). First of all, as \emph{triangles} are the simplest yet most fundamental subgraph with wide applications in many network analysis tasks, a large number of sampling methods were proposed for triangle counting in massive graphs~\cite{DBLP:conf/kdd/TsourakakisKMF09, DBLP:journals/ipl/PaghT12, DBLP:journals/pvldb/PavanTTW13, DBLP:conf/kdd/JhaSP13, DBLP:conf/sdm/KoldaPS13, DBLP:conf/cikm/ParkC13, DBLP:conf/kdd/AhmedDNK14, DBLP:conf/kdd/LimK15, DBLP:conf/kdd/StefaniERU16, DBLP:journals/algorithmica/BulteauFKP16, DBLP:conf/cikm/EtemadiLT16, DBLP:journals/pvldb/WangQSZTG17, DBLP:journals/tkdd/StefaniERU17, DBLP:conf/icdm/Shin17, DBLP:journals/pvldb/AhmedDWR17, DBLP:conf/icdm/TurkogluT17, DBLP:journals/siamcomp/EdenLRS17, 2018PES, DBLP:conf/pkdd/ShinKHF18, DBLP:conf/www/TurkT19, DBLP:journals/tkdd/ShinOKHF20, DBLP:conf/sigmod/Gou021, DBLP:journals/tkdd/ShinLOHF21}. The above methods considered the triangle counting problem in many different settings, including offline graphs~\cite{DBLP:conf/kdd/TsourakakisKMF09, DBLP:conf/sdm/KoldaPS13, DBLP:conf/cikm/EtemadiLT16, DBLP:journals/siamcomp/EdenLRS17, DBLP:conf/icdm/TurkogluT17, DBLP:conf/www/TurkT19}, insertion-only~\cite{DBLP:conf/kdd/TsourakakisKMF09, DBLP:conf/kdd/JhaSP13, DBLP:journals/pvldb/PavanTTW13, DBLP:conf/kdd/AhmedDNK14, DBLP:conf/kdd/LimK15, DBLP:conf/kdd/StefaniERU16, DBLP:journals/pvldb/WangQSZTG17, DBLP:journals/pvldb/AhmedDWR17, DBLP:journals/tkdd/StefaniERU17, DBLP:conf/icdm/Shin17, 2018PES} and fully-dynamic~\cite{DBLP:journals/algorithmica/BulteauFKP16, DBLP:conf/pkdd/ShinKHF18, DBLP:journals/tkdd/ShinOKHF20} graph streams, sliding windows~\cite{DBLP:conf/sigmod/Gou021}, and distributed graphs~\cite{DBLP:journals/ipl/PaghT12, DBLP:conf/cikm/ParkC13, DBLP:journals/tkdd/ShinLOHF21}. Moreover, sampling methods were also proposed for estimating more complex motifs than triangles, e.g., 4-vertex motifs~\cite{DBLP:conf/www/JhaSP15, DBLP:conf/kdd/Sanei-MehriST18}, 5-vertex motifs~\cite{DBLP:conf/www/PinarSV17, DBLP:journals/tkdd/WangLRTZG14, DBLP:conf/icde/WangLTZ16, DBLP:journals/tkde/WangZZLCLTTG18}, motifs with 6 or more vertices~\cite{DBLP:conf/wsdm/0002C00P17}, $k$-cliques~\cite{DBLP:conf/www/JainS17, DBLP:conf/stoc/EdenRS18}, sparse motifs with low counts~\cite{DBLP:journals/tkdd/StefaniTU21}, and butterflies in bipartite graphs~\cite{DBLP:conf/kdd/Sanei-MehriST18}. However, all the above methods were not designed for temporal graphs. They considered neither the temporal information nor the ordering of edges.
Therefore, they could not be used for temporal motif counting directly.

\textbf{Motifs in Temporal Graphs:}
Prior studies have considered different types of \emph{temporal network motifs}. Viard et al.~\cite{DBLP:conf/asunam/ViardLM15, DBLP:journals/tcs/ViardLM16} and Himmel et al.~\cite{DBLP:conf/asunam/HimmelMNS16} extended the notion of \emph{maximal clique} to temporal networks and proposed efficient algorithms for maximal clique enumeration. Li et al.~\cite{DBLP:conf/icde/LiSQYD18} proposed the notion of \emph{$(\theta,\tau)$-persistent $k$-core} to capture the persistence of a community in temporal networks. However, these notions of temporal motifs were different from ours since they did not take \emph{edge ordering} into account. Zhao et al.~\cite{DBLP:conf/cikm/ZhaoTHOJL10} and Gurukar et al.~\cite{DBLP:conf/sigmod/GurukarRR15} studied the \emph{communication motifs}, which are frequent subgraphs to characterize the patterns of information propagation in social networks. Kovanen et al.~\cite{Kovanen_2011} and Kosyfaki et al.~\cite{DBLP:conf/edbt/KosyfakiMPT19} defined the \emph{flow motifs} to model flow transfer among a set of vertices within a time window in temporal networks. Although both definitions accounted for edge ordering, they were more restrictive than ours because the former assumed any two adjacent edges must occur within a fixed time span while the latter assumed edges in a motif must be consecutive events for a vertex~\cite{DBLP:conf/wsdm/ParanjapeBL17}.

\textbf{Temporal Motif Counting \& Enumeration:}
There have been several existing studies on counting and enumerating temporal motifs. Paranjape et al.~\cite{DBLP:conf/wsdm/ParanjapeBL17} first formally defined the notion of \emph{temporal motifs} we use in this paper. They proposed exact algorithms for counting temporal motifs based on subgraph enumeration and timestamp-based pruning. Kumar and Calders~\cite{DBLP:journals/pvldb/KumarC18} proposed an efficient algorithm called 2SCENT to enumerate all simple temporal cycles in a directed interaction network. Although 2SCENT was shown to be effective for cycles, it could not be used for enumerating temporal motifs of any other type. Mackey et al.~\cite{DBLP:conf/bigdataconf/MackeyPFCC18} proposed an efficient \textsc{BackTracking} algorithm for temporal subgraph isomorphism. The algorithm could count temporal motifs exactly by enumerating all of them. Very recently, Micale et al.~\cite{DBLP:journals/ans/MicaleLPF21} proposed a subgraph isomorphism algorithm specialized for flow motifs in temporal graphs. Liu et al.~\cite{DBLP:conf/wsdm/LiuBC19} proposed an interval-based sampling framework for counting temporal motifs. To the best of our knowledge, this is the only existing work on approximate temporal motif counting via sampling. In this paper, we present several improved sampling algorithms for temporal motif counting in offline and streaming temporal graphs and compare them with the algorithms in~\cite{DBLP:conf/wsdm/ParanjapeBL17, DBLP:conf/bigdataconf/MackeyPFCC18, DBLP:journals/pvldb/KumarC18, DBLP:conf/wsdm/LiuBC19}.

\section{Preliminaries}\label{sec:def}

In this section, we formally define \emph{temporal graph (stream)}, \emph{temporal motif}, and the \emph{temporal motif counting} problem. Here, we follow the definition of \emph{temporal motifs} in~\cite{DBLP:conf/wsdm/ParanjapeBL17, DBLP:conf/wsdm/LiuBC19, DBLP:conf/bigdataconf/MackeyPFCC18} for its simplicity and generality. Other types of temporal motifs have been discussed in Section~\ref{sec:related:work}.

\textbf{Temporal Graph:}
A \emph{temporal graph} $\varGamma = (V_{\varGamma},E_{\varGamma})$ is defined by a set $V_{\varGamma}$ of $n$ vertices and a sequence $E_{\varGamma}$ of $m$ temporal edges among vertices in $V_{\varGamma}$. Each temporal edge $e=(u, v, t)$, where $u, v \in V_{\varGamma}$ and $t \in \mathbb{R}^{+}$, is a timestamped directed edge from $u$ to $v$ at time $t$. There may be more than one temporal edge from $u$ to $v$ at different timestamps (e.g., a user can call another user many times in a communication network). For ease of presentation, we assume the timestamp $t$ of each temporal edge $e$ is unique so that the temporal edges in $E_{\varGamma}$ are strictly ordered. Note that our algorithms can also handle non-unique timestamps by using any consistent rule to break ties. We also consider the case when a temporal graph is generated and observed incrementally in a streaming manner. In this case, a \emph{temporal graph stream} $\varGamma_t = (V_{\varGamma, t}, E_{\varGamma, t})$ at time $t$ is composed by the set $V_{\varGamma, t}$ of $n_{t}$ vertices observed until time $t$ and a sequence $E_{\varGamma, t}$ of $m_{t}$ temporal edges whose timestamps are less than or equal to $t$. New vertices and edges are inserted into the temporal graph incrementally over time, i.e., for some $t^{\prime} > t$, new vertices and edges from time $t$ to $t^{\prime}$ will be added to $\varGamma_t$ as $\varGamma_{t^{\prime}}$.

\textbf{Temporal Motif:}
We formalize the notion of \emph{temporal motifs}~\cite{DBLP:conf/wsdm/ParanjapeBL17, DBLP:conf/wsdm/LiuBC19} in the following definition.
\begin{definition}[Temporal Motif]
  A temporal motif $M=(V_M,$ $E_M,\sigma)$ consists of a connected graph with a set of $k$ vertices $V_M$ and a set of $l$ edges $E_M$, and an ordering $\sigma$ on the edges in $E_M$.
\end{definition}

Intuitively, a temporal motif $M$ can be represented as an ordered sequence of edges $\langle e^{\prime}_1 = (u^{\prime}_1, v^{\prime}_1),\ldots,e^{\prime}_l=(u^{\prime}_l, v^{\prime}_l) \rangle$. Given a temporal motif $M$ as a template pattern, we aim to count how many times this pattern appears in a temporal graph. Furthermore, we only consider the instances where the pattern is formed within a short time span. For example, an instance formed in an hour is more interesting than one formed accidentally in one year on a communication network~\cite{DBLP:conf/cikm/ZhaoTHOJL10, DBLP:conf/sigmod/GurukarRR15, DBLP:conf/wsdm/ParanjapeBL17}. Therefore, given a temporal graph $\varGamma$ and a temporal motif $M$, our goal is to find a sequence of edges $S \subseteq E_{\varGamma}$ such that (1) $S$ exactly matches (i.e., \emph{is isomorphic to}) $M$, (2) $S$ is in the same order as specified by $\sigma$, and (3) all edges in $S$ occur within a time span of at most $\delta$. We call such an edge sequence $S$ as a \emph{$\delta$-instance}~\cite{DBLP:conf/wsdm/ParanjapeBL17, DBLP:conf/wsdm/LiuBC19} of $M$ and the difference between $t_l$ and $t_1$ as the \emph{duration} $\Delta(S)$ of instance $S$. The formal definition is given in the following.
\begin{definition}[Motif $\delta$-instance]
  A sequence of $l$ edges $S=\langle (w_1, x_1, t_1), \ldots, (w_l, x_l, t_l) \rangle$ ($t_1 < \ldots < t_l$) from a temporal graph $\varGamma$ is a $\delta$-instance of a temporal motif $M=\langle (u^{\prime}_{1}, v^{\prime}_{1}), \ldots, (u^{\prime}_{l}, v^{\prime}_{l}) \rangle$ if (1) there exists a bijection $f$ between the vertex sets of $S$ and $M$ such that $f(w_i) = u^{\prime}_{i}$ and $f(x_i) = v^{\prime}_{i}$ for $i = 1, \ldots, l$ and (2) the duration $\Delta(S)$ is at most $\delta$, i.e., $t_{l} - t_{1} \leq \delta$.
\end{definition}

\begin{figure}
  \centering
  \includegraphics[width=0.6\textwidth]{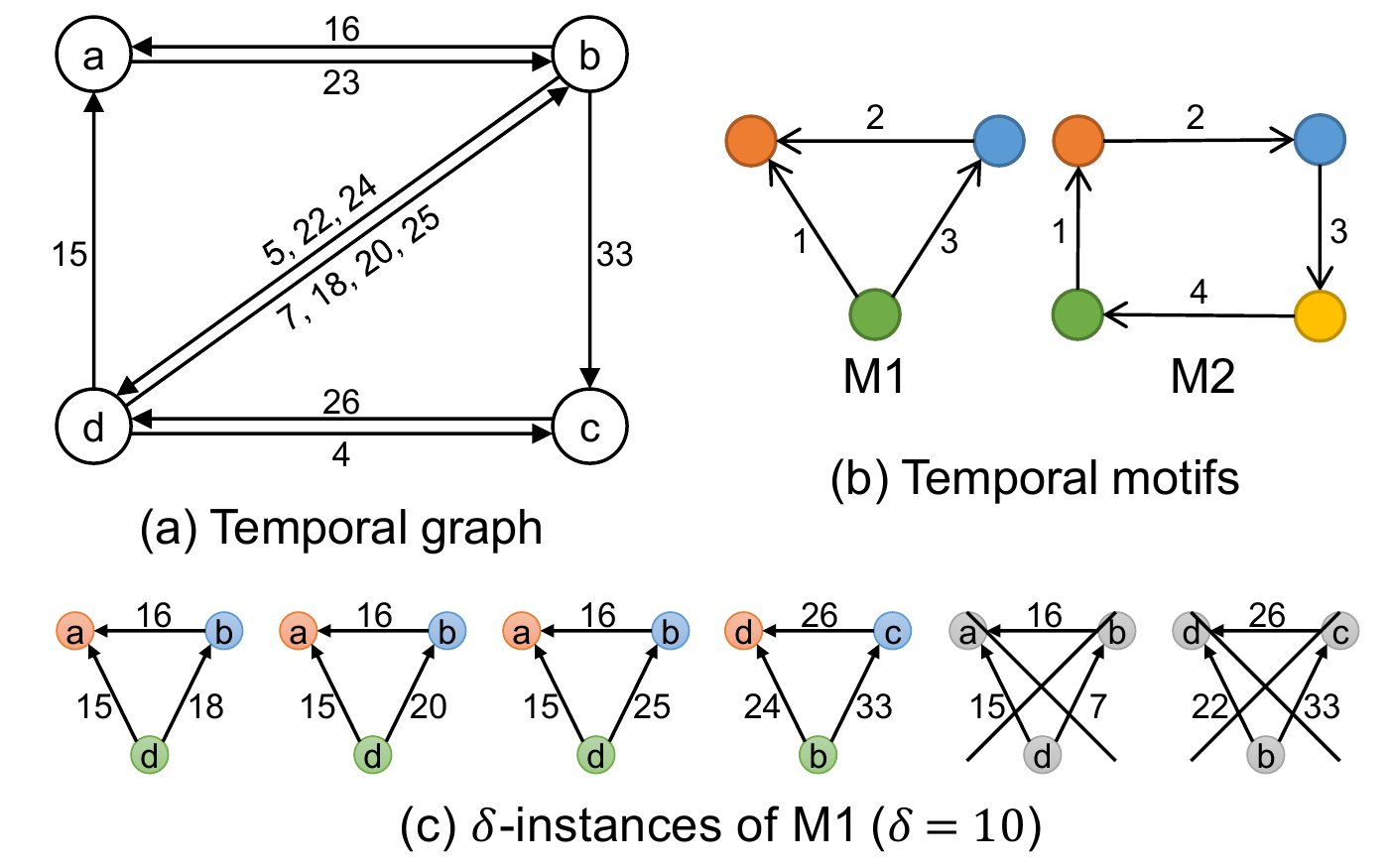}
  \caption{Examples for temporal graph and motif}
  \Description{example1}
  \label{fig:example}
\end{figure}

\begin{example}
  In Fig.~\ref{fig:example}(a), we illustrate a temporal graph with $4$ vertices and $13$ temporal edges. Let us consider the problem of finding all $\delta$-instances ($\delta=10$) of temporal motif $M1$ in Fig.~\ref{fig:example}(b). As shown in Fig.~\ref{fig:example}(c), there are $4$ valid $10$-instances of $M1$ found. These instances can match $M1$ in terms of both structure and edge ordering and their durations are within $10$. In addition, we also give $2$ invalid instances of $M1$, which are isomorphic to $M1$ but violate either the edge ordering or duration constraint.
\end{example}

\textbf{Temporal Motif Counting:}
According to the above notions, we present the offline and streaming \emph{temporal motif counting} problems studied in this paper.
\begin{definition}[Temporal Motif Counting]\label{def:count}
  For a temporal graph $\varGamma$, a temporal motif $M$, and a time span $\delta$, the temporal motif counting problem returns the number $C_M$ of $\delta$-instances of motif $M$ appeared in temporal graph $\varGamma$.
\end{definition}
\begin{definition}[Streaming Temporal Motif Counting]\label{def:count1}
  For a temporal graph stream $\varGamma$, a temporal motif $M$, a time span $\delta$, and a timestamp $t$, the streaming temporal motif counting problem returns the number $C_{M, t}$ of $\delta$-instances of motif $M$ appeared in the partially observed temporal graph stream $\varGamma_{t}$ at time $t$.
\end{definition}

Both temporal motif counting problems are proven to be NP-hard even for very simple motifs, e.g.~$k$-stars~\cite{DBLP:conf/wsdm/LiuBC19}, because the edge ordering is taken into account. According to previous results~\cite{DBLP:conf/wsdm/LiuBC19}, although there is a simple polynomial algorithm to count the number of $k$-stars on a graph, it is NP-hard to exactly count the number of temporal $k$-stars. Typically, counting temporal motifs exactly on massive graphs with millions or even billions of edges is a computationally intensive task~\cite{DBLP:conf/wsdm/ParanjapeBL17, DBLP:conf/wsdm/LiuBC19}. Therefore, we focus on designing efficient and scalable sampling algorithms for estimating the number of temporal motifs approximately in Sections~\ref{sec:alg} and \ref{sec:s-alg}. The frequently used notations are summarized in Table~\ref{tab:freq}.

\begin{table}[t]
  \small
  \centering
  \caption{Frequently used notations}\label{tab:freq}
  \begin{tabular}{c l}
    \hline
    \textbf{Symbol} & \textbf{Description} \\
    \hline
    $\varGamma, \varGamma_{t}$ & Temporal graph and (partially observed) temporal graph stream at time $t$ \\
    $V_{\varGamma}, E_{\varGamma}$ & Set of vertices and edges in $\varGamma$ \\
    $V_{\varGamma,t}, E_{\varGamma,t}$ & Set of vertices and edges in $\varGamma_{t}$ at time $t$ \\
    $n, m$ & Number of vertices and edges in $\varGamma$ \\
    $n_{t}, m_{t}$ & Number of vertices and edges in $\varGamma_{t}$ at time $t$ \\
    $M$ & Temporal motif \\
    $V_M, E_M$ & Set of vertices and edges in $M$ \\
    $k, l$ & Number of vertices and edges in $M$ \\
    $\delta$ & Maximum time span of a motif instance \\
    $S$ & Motif $\delta$-instance \\
    $C_M, C_{M, t}$ & Numbers of $\delta$-instances of $M$ in $\varGamma$ and $\varGamma_{t}$, respectively \\
    $\widehat{C}_M, \widehat{C}_{M, t}$ & Unbiased estimators of $C_M$ and $C_{M, t}$, respectively \\
    $p$ & Probability of edge sampling \\
    $r$ & Sample size in reservoir sampling\\
    $\widehat{E}_{\varGamma}, \widehat{E}_{\varGamma, t}$ & Sets of sampled edges from $E_{\varGamma}$ and $E_{\varGamma, t}$ at time $t$, respectively\\
    $\eta(e)$ & Number of $\delta$-instances of $M$ containing edge $e$ \\
    $\eta_j(e)$ & Number of $\delta$-instances of $M$ when $e$ is mapped to $e_j^\prime$ \\
    $q$ & Probability of wedge sampling \\
    $W$ & Temporal wedge \\
    $\eta(W)$ & Number of $\delta$-instances of $M$ containing $W$ \\
    $W_j^\prime$ & Temporal wedge pattern for $M$ when $e$ is mapped to $e_j^\prime$ \\
    $\widehat{\mathcal{W}}_j(e)$ & Set of sampled $\delta$-instances of $W_j^\prime$  \\
    $\widehat{\eta}_j(e)$ & Unbiased estimator of $\eta_j(e)$ \\
    \hline
  \end{tabular}
\end{table}

\section{Algorithms for Offline Temporal Graphs}\label{sec:alg}

In this section, we present our proposed algorithms for approximate temporal motif counting in the offline setting. We first describe our generic Edge Sampling (ES) algorithm in Section~\ref{subsec:alg:1}. Then, we introduce our improved EWS algorithm specific for counting $3$-vertex $3$-edge temporal motifs in Section~\ref{subsec:alg:2}. In addition, we theoretically analyze the expected values and variances of the estimates returned by both algorithms.

\subsection{Generic Edge Sampling Algorithm}\label{subsec:alg:1}

\textbf{Algorithmic Description:}
The Edge Sampling (ES) algorithm is motivated by an exact subgraph counting algorithm called \emph{edge iterator}~\cite{DBLP:journals/tkde/WuYL16}. Given a temporal graph $\varGamma$, a temporal motif $M$, and a time span $\delta$, we use $\eta(e)$ to denote the number of local $\delta$-instances of $M$ containing an edge $e$. To count all $\delta$-instances of $M$ in $\varGamma$ exactly, we can simply count $\eta(e)$ for each $e \in E_{\varGamma}$ and then sum them up. In this way, each instance is counted $l$ times and the total number of instances is equal to the sum divided by $l$, i.e.,~$C_M=\frac{1}{l}\sum_{e \in E_{\varGamma}}\eta(e)$.

Based on the above idea, we propose the ES algorithm for estimating $C_M$: For each edge $e \in E_{\varGamma}$, we randomly sample it and compute $\eta(e)$ with fixed probability $p$. Then, we obtain an unbiased estimator $\widehat{C}_M$ of $C_M$ by adding up $\eta(e)$ for each sampled edge $e$ and scaling the sum by a factor of $\frac{1}{p l}$, i.e., $\widehat{C}_M=\frac{1}{p l}\sum_{e \in \widehat{E}_{\varGamma}}\eta(e)$ where $\widehat{E}_{\varGamma}$ is the set of sampled edges.

Now the remaining problem becomes how to compute $\eta(e)$ for an edge $e$. The ES algorithm adopts the well-known \textsc{BackTracking} algorithm~\cite{DBLP:journals/jacm/Ullmann76, DBLP:conf/bigdataconf/MackeyPFCC18} to enumerate all $\delta$-instances that contain an edge $e$ for computing $\eta(e)$. Specifically, the \textsc{BackTracking} algorithm runs $l$ times for each edge $e$; in the $j$\textsuperscript{th} run, it first maps edge $e$ to the $j$\textsuperscript{th} edge $e_j^\prime$ of $M$ and then uses a tree search to find all different combinations of the remaining $l-1$ edges that can form $\delta$-instances of $M$ with edge $e$. Let $\eta_j(e)$ be the number of $\delta$-instances of $M$ where $e$ is mapped to $e_j^\prime$. It is obvious that $\eta(e)$ is equal to the sum of $\eta_j(e)$ for $j=1,\ldots,l$, i.e., $\eta(e)=\sum_{j=1}^{l}\eta_j(e)$.

\begin{algorithm}
  \caption{Edge Sampling}\label{alg:es}
  \KwIn{Temporal graph $\varGamma$, temporal motif $M$, time span $\delta$, edge sampling probability $p$.}
  \KwOut{Estimator $\widehat{C}_M$ of the number of $\delta$-instances of $M$ in $\varGamma$}
  Initialize $\widehat{E}_{\varGamma} \gets \varnothing$\;\label{ln:es:sample:s}
  \ForEach{$e\in E_{\varGamma}$}
  {
    Toss a biased coin with success probability $p$\;
    \If{success}
    {
      $\widehat{E}_{\varGamma} \gets \widehat{E}_{\varGamma} \cup \{e\}$\;\label{ln:es:sample:t}
    }
  }
  \ForEach{$e=(u, v, t) \in \widehat{E}_{\varGamma}$\label{ln:es:count:s}}
  {
    Set $\eta(e) \gets 0$\; 
    \For{$j \in 1,\ldots,l$}
    {
      Generate an initial instance $S_j^{(1)}$ by mapping $e$ to $e_j^\prime$\;\label{ln:es:init}
      Run \textsc{BackTracking} on $E_{\varGamma}[t-\delta,t+\delta]$ starting from $S_j^{(1)}$ to find the set $\mathcal{S}_j(e) = \{S_j(e) \,:\, S_j(e)$ is a $\delta$-instance of $M$ where $e$ is mapped to $e_j^\prime\}$\;\label{ln:es:bt}
      Set $\eta_j(e) \gets |\mathcal{S}_j(e)|$ and $\eta(e) \gets \eta(e) + \eta_j(e)$\;\label{ln:es:count:t}
    }
  }
  \Return{$\widehat{C}_M \gets \frac{1}{pl}\sum_{e \in \widehat{E}_{\varGamma}}\eta(e)$}\;\label{ln:es:estimate}
\end{algorithm}

We depict the procedure of our ES algorithm in Algorithm~\ref{alg:es}. The first step of ES is to generate a random sample $\widehat{E}_{\varGamma}$ of edges from the edge set $E_{\varGamma}$ where the probability of adding any edge is $p$ (Lines~\ref{ln:es:sample:s}--\ref{ln:es:sample:t}). Then, in the second step (Lines~\ref{ln:es:count:s}--\ref{ln:es:count:t}), it counts the number $\eta(e)$ of local $\delta$-instances of $M$ for each sampled edge $e$ by running the \textsc{BackTracking} algorithm to enumerate each instance $S_j(e)$ that is a $\delta$-instance of $M$ and maps $e$ to $e_j^\prime$ for $j=1,\ldots,l$. Note that \textsc{BackTracking} (BT) runs on a subset $E_{\varGamma}[t-\delta,t+\delta]$ of $E_{\varGamma}$ which consists of all edges with timestamps from $t-\delta$ to $t+\delta$ for edge $e=(u, v, t)$ since it is safe to ignore any other edge due to the duration constraint. Here, we omit the detailed procedure of the BT algorithm because it generally follows an existing algorithm for subgraph isomorphism in temporal graphs~\cite{DBLP:conf/bigdataconf/MackeyPFCC18}. The main difference between our algorithm and the one in~\cite{DBLP:conf/bigdataconf/MackeyPFCC18} lies in the matching order, which will be discussed later. After counting $\eta(e)$ for each sampled edge $e$, it finally returns an estimate $\widehat{C}_M$ of $C_M$(Line~\ref{ln:es:estimate}).

\textbf{Matching Order for \textsc{BackTracking}:}
Now we discuss how to determine the matching order of a temporal motif. The BT algorithm in~\cite{DBLP:conf/bigdataconf/MackeyPFCC18} adopts a time-first matching order: it always matches the edges of $M$ in order of $\langle e_1^\prime,\ldots,e_l^\prime \rangle$. The advantage of this matching order is that it best exploits the temporal information for search space pruning. For a partial instance $S^{(j)}=\langle (w_1,x_1,t_1),\ldots,(w_j,$ $x_j, t_j) \rangle$ after $e_j^\prime$ is mapped, the search space for mapping $e_{j+1}^\prime$ is restricted to $E_{\varGamma}[t_j, t_1+\delta]$. However, the time-first matching order may not work well in ES. First, it does not consider the \emph{connectivity} of the matching order: If $e_{j+1}^\prime$ is not connected with any prior edge, it has to be mapped to all edges in $E_{\varGamma}[t_j, t_1+\delta]$, which may lead to a large number of redundant partial matchings. Second, the time-first order is violated by Line~\ref{ln:es:init} of Algorithm~\ref{alg:es} when $j>1$ since it first maps $e$ to $e_{j}^\prime$.

In order to overcome the above two drawbacks, we propose two heuristics to determine the matching order of a given motif $M$ for reducing the search space, and generate $l$ matching orders for $M$, in each of which $e_{j}^\prime$ ($j=1,\ldots,l$) is placed first: (1) \emph{enforcing connectivity}: For each $i=2,\ldots,l $, the $i$\textsuperscript{th} edge in the matching order must be adjacent to at least one prior edge that has been matched; (2) \emph{boundary edge first}: If there are multiple unmatched edges that satisfy the \emph{connectivity} constraint, the boundary edge (i.e., the first or last unmatched edge in the ordering $\sigma$ of $M$) will be matched first. The first rule can avoid redundant partial matchings and the second rule can restrict the temporal range of tree search, both of which are effective for search space pruning.

\textbf{Theoretical Analysis:}
Next, we analyze the estimate $\widehat{C}_M$ returned by Algorithm~\ref{alg:es} theoretically. We first prove that $\widehat{C}_M$ is an unbiased estimator of $C_M$ in Theorem~\ref{thm:es:exp}. The variance of $\widehat{C}_M$ is given in Theorem~\ref{thm:es:var}.

\begin{theorem}\label{thm:es:exp}
  The expected value $\mathbb{E}[\widehat{C}_M]$ of $\widehat{C}_M$ returned by Algorithm~\ref{alg:es} is $C_M$.
\end{theorem}
\begin{proof}
  Here, we consider the edges in $E_{\varGamma}$ are indexed by $[1,m]$ and use an indicator $\omega_i$ to denote whether the $i$\textsuperscript{th} edge $e_i$ is sampled, i.e.,
  \begin{equation*}
    \omega_i = \begin{cases}
      1, & e_i \in \widehat{E}_{\varGamma} \\
      0, & e_i \notin \widehat{E}_{\varGamma}
    \end{cases}
  \end{equation*}
  Then, we have
  \begin{equation}\label{Eq:n1}
    \widehat{C}_M = \frac{1}{pl}\sum_{e \in \widehat{E}_{\varGamma}}\eta(e) = \frac{1}{pl}\sum_{i=1}^{m}\omega_{i}\cdot\eta(e_i)
  \end{equation}
  Next, based on Equation~\ref{Eq:n1} and the fact that $\mathbb{E}[\omega_{i}]=p$, we have
  \begin{equation*}
    \mathbb{E}[\widehat{C}_M] = \frac{1}{pl}\sum_{i=1}^{m}\mathbb{E}[\omega_{i}]\cdot\eta(e_i) = \frac{1}{l}\sum_{i=1}^{m}\eta(e_i)=C_M
  \end{equation*}
  and conclude the proof.
\end{proof}

\begin{theorem}\label{thm:es:var}
  The variance $\textnormal{Val}[\widehat{C}_M]$ of $\widehat{C}_M$ returned by Algorithm~\ref{alg:es} is at most $\frac{1-p}{p} \cdot C_M^2$.
\end{theorem}
\begin{proof}
  According to Equation~\ref{Eq:n1}, we have
  \begin{equation*}
    \textnormal{Val}[\widehat{C}_M] =\textnormal{Val}\Big[\sum_{i=1}^m \frac{\eta(e_i)}{pl} \cdot \omega_i \Big] =\sum_{i, j = 1}^{m} \frac{\eta(e_i)}{pl} \cdot \frac{\eta(e_j)}{pl} \cdot \textnormal{Cov}(\omega_i,\omega_j)
  \end{equation*}
  Because the indicators $\omega_{i}$ and $\omega_{j}$ are independent if $i \neq j$, we have $\textnormal{Cov}(\omega_i,\omega_j)=0$ for any $i \neq j$. In addition, $\textnormal{Cov}(\omega_i,\omega_i)=\textnormal{Val}[\omega_i]=p-p^2$. Based on the above results, we have
  \begin{align*}
    \textnormal{Val}[\widehat{C}_M]
    & =\sum_{i=1}^{m} \frac{\eta^2(e_i)}{p^2 l^2}(p-p^2)
    =\frac{1-p}{pl^2}\sum_{i=1}^{m}\eta^2(e_i) \\
    & \leq \frac{1-p}{pl^2}\Big(\sum_{i=1}^{m}\eta(e_i)\Big)^2
    =\frac{1-p}{p} \cdot C_M^2
  \end{align*}
  and conclude the proof.
\end{proof}

Finally, we can derive Theorem~\ref{thm:es:est} by applying Chebyshev's inequality to Theorem~\ref{thm:es:var}.

\begin{theorem}\label{thm:es:est}
  It holds that $\textnormal{Pr}[|\widehat{C}_M-C_M|\geq\varepsilon\cdot C_M]\leq\frac{1-p}{p\varepsilon^2}$.
\end{theorem}
\begin{proof}
  By applying Chebyshev's inequality, we have $\textnormal{Pr}[|\widehat{C}_M-C_M|\geq\varepsilon\cdot C_M] \leq\frac{\textnormal{Val}[\widehat{C}_M]}{\varepsilon^2 C_M^2}$ and thus prove the theorem by substituting $\textnormal{Val}[\widehat{C}_M]$ with $\frac{1-p}{p} \cdot C_M^2$ according to Theorem~\ref{thm:es:var}.
\end{proof}
According to Theorem~\ref{thm:es:est}, we can say $\widehat{C}_M$ is an $(\varepsilon,\gamma)$-estimator of $C_M$ for parameters $\varepsilon,\gamma \in (0,1)$, i.e., $\textnormal{Pr}[|\widehat{C}_M-C_M|<\varepsilon\cdot C_M]> 1-\gamma$, when $p=\frac{1}{1+\gamma\varepsilon^{2}}$.

\textbf{Time Complexity:}
We first analyze the time complexity of computing $\eta(e)$ for an edge $e$. For \textsc{BackTracking}, the search space of each matching step is at most the number of (in-/out-)edges within range $[t-\delta,t]$ or $[t, t+\delta]$ connected with a vertex $v$. Here, we use $d_{\delta}$ to denote the maximum number of (in-/out-)edges connected with one vertex within any $\delta$-length time interval. The time complexity of \textsc{BackTracking} is $O(d_{\delta}^{l-1})$ and thus the time complexity of computing $\eta(e)$ is $O(l d_{\delta}^{l-1})$. Therefore, ES provides an $(\varepsilon,\gamma)$-estimator of $C_M$ in $O(\frac{ml d_{\delta}^{l-1}}{1+\gamma\varepsilon^2})$ time.

\subsection{Improved Edge-Wedge Sampling Algorithm}\label{subsec:alg:2}

\begin{figure}
  \centering
  \includegraphics[width=0.6\textwidth]{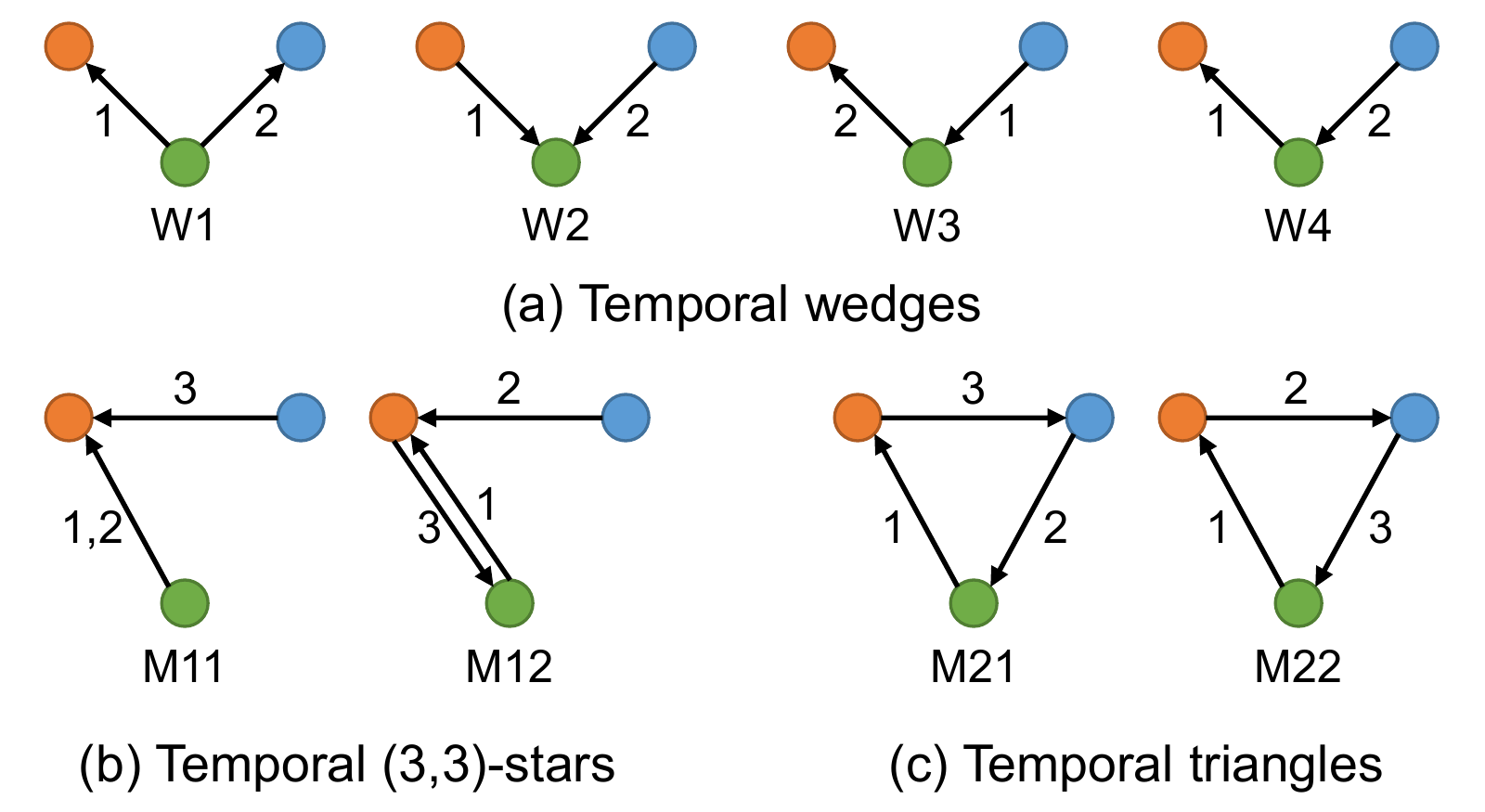}
  \caption{Examples of temporal wedges, $(3,3)$-stars, and triangles}
  \Description{wedge}
  \label{fig:wedge}
\end{figure}

\textbf{Algorithmic Description:}
The ES algorithm in Section~\ref{subsec:alg:1} is generic and able to count any connected temporal motif. Nevertheless, there are still opportunities to further reduce the computational overhead of ES when the query motif is limited to $3$-vertex $3$-edge temporal motifs (i.e., triadic patterns), which are one of the most important classes of motifs to characterize temporal networks~\cite{DBLP:conf/sdm/KoldaPS13, DBLP:conf/wsdm/ParanjapeBL17, DBLP:journals/snam/UzupyteW20, FAUST2010221}.

In this section, we propose an improved Edge-Wedge Sampling (EWS) algorithm that combines \emph{edge sampling} with \emph{wedge sampling} for counting $3$-vertex $3$-edge temporal motifs. Wedge sampling~\cite{DBLP:conf/sdm/KoldaPS13, DBLP:conf/icdm/TurkogluT17, DBLP:conf/www/TurkT19, DBLP:journals/tkde/WuYL16} is a widely used method for triangle counting. Its basic idea is to draw a sample of wedges (i.e., $3$-vertex $2$-edge subgraph patterns) uniformly from a graph and check the ratio of ``closed wedges'' (i.e., form a triangle in the graph) to estimate the number of triangles. However, traditional wedge-sampling methods are proposed for undirected static graphs and cannot be directly used on temporal graphs. First, they consider that all wedges are isomorphic and treat them equally. But there are four temporal wedge patterns with different edge directions and orderings as illustrated in Fig.~\ref{fig:wedge}(a). Second, they are designed for simple graphs where one wedge can form at most one triangle. However, since temporal graphs are multi-graphs and there may exist multiple edges between the same two vertices, one temporal wedge can participate in more than one instance of a temporal motif. Therefore, in the EWS algorithm, we extend \emph{wedge sampling} for temporal motif counting by addressing both issues.

\begin{algorithm} 
  \caption{Edge-Wedge Sampling}\label{alg:ews}
  \KwIn{Temporal graph $\varGamma$, temporal motif $M$, time span $\delta$, edge sampling probability $p$, wedge sampling probability $q$.}
  \KwOut{Estimator $\widehat{C}_M$ of the number of $\delta$-instances of $M$ in $\varGamma$}
  Generate $\widehat{E}_{\varGamma}$ using Line~\ref{ln:es:sample:s}--\ref{ln:es:sample:t} of Algorithm~\ref{alg:es}\;\label{ln:ews:sample}
  \ForEach{$e=(u, v, t) \in \widehat{E}_{\varGamma}$\label{ln:ews:wedge:s}}
  {
    \For{$j \gets 1,2,3$}
    {
      Map edge $e$ to $e_j^\prime$\;\label{ln:ews:mapping}
      Initialize $\widehat{\eta}_j(e) \gets 0$ and $\widehat{\mathcal{W}}_j(e) \gets \varnothing$\;\label{ln:ews:wedge:begin}
      \uIf{$M$ is a temporal $(3,3)$-star\label{ln:ews:wedge:gb}}
      {
        Select $W_j^\prime$ including $e_j^\prime$ centered at the center of $M$\;
      }
      \ElseIf{$M$ is a temporal triangle}
      {
        Select $W_j^\prime$ including $e_j^\prime$ centered at the vertex mapped to the one with a lower degree in $u$ and $v$\;\label{ln:ews:wedge:ge}
      }
      $E_j(e) \gets$ all edges that form $\delta$-instances of $W_j^\prime$ with $e$\;\label{ln:ews:edges}
      \ForEach{$g \in E_j(e)$ \label{ln:ews:wedge:p}}
      {
        Add a $\delta$-instance $W$ of $W_j^\prime$ comprising $e$ and $g$ to $\widehat{\mathcal{W}}_j(e)$ with probability $q$\;\label{ln:ews:wedge:end}
      }
      \ForEach{$W \in \widehat{\mathcal{W}}_j(e)$ \label{ln:ews:est:begin}}
      {
        Let $\eta(W)$ be the number of edges that form $\delta$-instances of $M$ together with $W$\;\label{ln:ews:est:W}
        $\widehat{\eta}_j(e) \gets \widehat{\eta}_j(e) + \frac{\eta(W)}{q}$\;\label{ln:ews:est:end}
      }
    }
    \label{ln:ews:wedge:t}
  }
  \Return{$\widehat{C}_M \gets \frac{1}{3p}\sum_{e \in \widehat{E}_{\varGamma}} \sum_{j=1}^{3}\widehat{\eta}_j(e)$}\;\label{ln:ews:estimate}
\end{algorithm}

The detailed procedure of EWS is presented in Algorithm~\ref{alg:ews}. First of all, it uses the same method as ES to sample a set $\widehat{E}_{\varGamma}$ of edges (Line~\ref{ln:ews:sample}). For each sampled edge $e\in \widehat{E}_{\varGamma}$ and $j = 1,2,3$, it also maps $e$ to $e_j^\prime$  for computing $\eta_j(e)$ (Line~\ref{ln:ews:mapping}), i.e., the number of $\delta$-instances of $M$ where $e$ is mapped to $e_j^\prime$. But, instead of running \textsc{BackTracking} to compute $\eta_j(e)$ exactly, it utilizes \emph{temporal wedge sampling} to estimate $\eta_j(e)$ approximately without full enumeration (Lines~\ref{ln:ews:wedge:begin}--\ref{ln:ews:est:end}), which is divided into two subroutines as discussed later. At last, it obtains an estimate $\widehat{C}_M$ of $C_M$ from each estimate $\widehat{\eta}_j(e)$ of $\eta_j(e)$ using a similar method to ES (Line~\ref{ln:ews:estimate}).

\textbf{Sample Temporal Wedges (Lines~\ref{ln:ews:wedge:begin}--\ref{ln:ews:wedge:end}):}
The first step of \emph{temporal wedge sampling} is to determine which temporal wedge pattern is to be matched according to the query motif $M$ and the mapping from $e$ to $e_j^\prime$. Specifically, we categorize 3-vertex 3-edge temporal motifs into two types, i.e., \emph{temporal $(3,3)$-stars} and \emph{temporal triangles} as shown in Fig.~\ref{fig:wedge}, based on whether they are closed. Interested readers may refer to~\cite{DBLP:conf/wsdm/ParanjapeBL17} for a full list of all $3$-vertex $3$-edge temporal motifs. For a star or wedge pattern, the vertex connected with all edges is its \emph{center}. Given that $e=(u, v, t)$ has been mapped to $e_j^\prime$, EWS should find a temporal wedge pattern $W_j^\prime$ containing $e_j^\prime$ from $M$ for sampling. Here, different strategies are adopted to determine $W_j^\prime$ for star and triangle motifs (Lines~\ref{ln:ews:wedge:gb}--\ref{ln:ews:wedge:ge}): If $M$ is a temporal $(3,3)$-star, it must select  $W_j^\prime$ that contains $e_j^\prime$ and has the same center as $M$; If $M$ is a temporal triangle, it may use the vertex mapped to either $u$ or $v$ as the center to generate a wedge pattern. In this case, the center of $W_j^\prime$ will be mapped to the vertex with a lower degree between $u$ and $v$ for search space reduction. After deciding $W_j^\prime$, it enumerates all edges that form a $\delta$-instance of $W_j^\prime$ together with $e$ as $E_j(e)$ from the adjacency list of the central vertex (Line~\ref{ln:ews:edges}). By selecting each edge $g \in E_j(e)$ with probability $q$, it generates a sample $\widehat{\mathcal{W}}_j(e)$ of $\delta$-instances of $W_j^\prime$ (Lines~\ref{ln:ews:wedge:p}--\ref{ln:ews:wedge:end}).

\textbf{Estimate $\eta_j(e)$ (Lines~\ref{ln:ews:est:begin}--\ref{ln:ews:est:end}):}
Now, it estimates $\eta_j(e)$ from the set $\widehat{\mathcal{W}}_j(e)$ of sampled temporal wedges. For each $W \in \widehat{\mathcal{W}}_j(e)$, it counts the number $\eta(W)$ of $\delta$-instances of $M$ that contain $W$ (Line~\ref{ln:ews:est:W}). Specifically, after matching $W$ with $W_j^\prime$, it can determine the starting and ending vertices as well as the temporal range for the mapping of the third edge of $M$. For the fast computation of $\eta(W)$, EWS maintains a hash table that uses an ordered combination $\langle u, v \rangle$ ($u, v \in V_{\varGamma}$) as the key and a sorted list of the timestamps of all edges from $u$ to $v$ as the value on the edge set $E_{\varGamma}$ of $\varGamma$. In this way, $\eta(W)$ can be computed by a hash search followed by at most two binary searches on the sorted list. Finally, $\eta_j(e)$ can be estimated by summing up $\eta(W)$ for each $W \in \widehat{\mathcal{W}}_j(e)$ (Line~\ref{ln:ews:est:end}), i.e., $\widehat{\eta}_j(e)=\frac{1}{q}\sum_{W\in\widehat{\mathcal{W}}_j(e)}\eta(W)$.

\textbf{Theoretical Analysis:}
Next, we analyze the estimate $\widehat{C}_M$ returned by Algorithm~\ref{alg:ews} theoretically. We prove the unbiasedness and variances of $\widehat{C}_M$ in Theorem~\ref{thm:ews:exp} and Theorem~\ref{thm:ews:var}, respectively.

\begin{theorem}\label{thm:ews:exp}
  The expected value $\mathbb{E}[\widehat{C}_M]$ of $\widehat{C}_M$ returned by Algorithm~\ref{alg:ews} is $C_M$.
\end{theorem}
\begin{proof}
  By applying the result of Theorem~\ref{thm:es:exp}, we only need to show $\mathbb{E}[\widehat{\eta}_j(e)]=\eta_j(e)$ to prove Theorem~\ref{thm:ews:exp}. Here, we index the edges in $E_j(e)$ by $[1,\ldots,|E_j(e)|]$ and use an indicator $\omega_r$ to denote whether the wedge $W_r$ w.r.t.~the $r$\textsuperscript{th} edge in $E_j(e)$ is sampled. We have the following equality:
  \begin{displaymath}
    \mathbb{E}[\widehat{\eta}_j(e)] = \frac{1}{q}\sum_{r=1}^{|E_j(e)|}\mathbb{E}[\omega_r]\cdot\eta(W_r) = \sum_{r=1}^{|E_j(e)|}\eta(W_r) = \eta_j(e)
  \end{displaymath}
  and conclude the proof.
\end{proof}

\begin{theorem}\label{thm:ews:var}
  The variance $\textnormal{Val}[\widehat{C}_M]$ of $\widehat{C}_M$ returned by Algorithm~\ref{alg:ews} is at most $\frac{1 - p q}{p q} \cdot C_M^2$.
\end{theorem}
\begin{proof}
  Let us index the edges in $E_{\varGamma}$ by $[1,m]$ and the edges in $E_j(e_i)$ when $e_i$ is mapped to $e_j^\prime$ by $[1,m_{i j}]$ where $m_{i j}=|E_j(e_i)|$. Similar to the proof of Theorem~\ref{thm:es:var}, we have
  \begin{align*}
    \textnormal{Val}[\widehat{C}_M]
    & = \textnormal{Val}\Big[\frac{1}{3pq} \sum_{i=1}^{m}\sum_{j=1}^{3}\sum_{r=1}^{m_{ij}} \omega_i \cdot \omega_{i j r} \cdot \eta(W_{i j r}) \Big] \\
    & = \sum_{i=1}^{m}\sum_{j=1}^{3}\sum_{r=1}^{m_{ij}} \frac{\eta^2(W_{i j r})}{9 p^2 q^2} \cdot \textnormal{Var}[\omega_i \cdot \omega_{i j r}] \\
    & = \frac{1-pq}{9pq} \cdot \sum_{i=1}^{m}\sum_{j=1}^{3}\sum_{r=1}^{m_{ij}} \eta^2(W_{i j r}) \leq \frac{1-pq}{pq} \cdot C_M^2
  \end{align*}
  where $\eta(W_{i j r})$ is the number of $\delta$-instances of $M$ containing a temporal wedge $W_{i j r}$ and $\omega_{i j r}$ is its indicator, the second equality holds because $\omega_i$ and $\omega_{i j r}$ are mutually independent, the third equality holds because $\textnormal{Var}[\omega_i \cdot \omega_{i j r}]=p q - p^2 q^2$, and the last inequality holds for $C_M=\frac{1}{3}\sum_{i=1}^{m}\sum_{j=1}^{3}\sum_{r=1}^{m_{i j}}\eta(W_{i j r})$.
\end{proof}

According to the result of Theorem~\ref{thm:ews:var} and Chebyshev's inequality, we have $\textnormal{Pr}[|\widehat{C}_M-C_M|\geq\varepsilon\cdot C_M]\leq\frac{1 - p q}{p q \varepsilon^2}$ and $\widehat{C}_M$ is an $(\varepsilon,\gamma)$-estimator of $C_M$ for parameters $\varepsilon,\gamma \in (0,1)$ when $p q = \frac{1}{1+\gamma\varepsilon^{2}}$.

\textbf{Time Complexity:}
We first analyze the time to compute $\widehat{\eta}_j(e)$. First, $|E_j(e)|$ is bounded by the maximum number of (in-/out-)edges connected with one vertex within any $\delta$-length time interval, i.e., $d_{\delta}$. Second, the time to compute $\eta(W)$ using a hash table is $O(\log{h})$ where $h$ is the maximum number of edges between any two vertices. Therefore, the time complexity per edge in EWS is $O(d_{\delta}\log{h})$. This is lower than $O(d_{\delta}^2)$ time per edge in ES (when $k, l = 3$). Finally, EWS provides an $(\varepsilon,\gamma)$-estimator of $C_M$ in $O(\frac{m d_{\delta}\log{h}}{1+\gamma\varepsilon^2})$ time.

\section{Algorithms For Temporal Graph Streams}\label{sec:s-alg}

\textbf{Algorithmic Description:}
In this section, we extend our proposed algorithms, i.e., ES and EWS, to work in a streaming setting so that they can deal with temporal graphs that are generated incrementally over time or too large to fit in main memory. Our basic idea is to sample a fixed-size subset of edges from a temporal graph stream via \emph{reservoir sampling}~\cite{DBLP:journals/toms/Vitter85} and then utilize a similar method as used in ES or EWS on the sampled edges for estimation. The advantages of reservoir sampling are three-fold. First, it uses a fixed memory that can be set in advance, even for a graph stream of unknown size. Thus, we can determine an appropriate memory size for the algorithm before running it. Second, it only requires a single pass over the stream without pre/post-processing. Third, it maintains a uniform sample of the observed data stream at any time, based on which an up-to-date estimator can always be acquired in real time. For the above advantages of reservoir sampling, it has been adopted in many different counting methods in graph streams~\cite{DBLP:journals/pvldb/AhmedDWR17, DBLP:journals/tkdd/StefaniERU17, DBLP:conf/icdm/Shin17, 2018PES}. Nevertheless, these methods still do not consider the temporal information of the graph and thus cannot be used for temporal motif counting. Next, we will present a general reservoir sampling-based framework to extend ES and EWS for temporal graph streams. We note that the extended streaming algorithms based on ES and EWS are referred to as the Streaming ES (SES) and Streaming EWS (SEWS) algorithms, respectively.

\begin{algorithm}
  \caption{Streaming ES/EWS}\label{alg:sf}
  \KwIn{Temporal graph stream $\varGamma$, temporal motif $M$, time span $\delta$, sample size $r$.}
  \KwOut{Estimator $\widehat{C}_{M, t}$ of the number of $\delta$-instances of $M$ in $\varGamma_t$ at time $t$}
  Initialize $\widehat{E}_{\varGamma} \gets \varnothing$, $C_{M}(\widehat{E}_{\varGamma}) \gets 0$\;\label{ln:sf:sample:s}
  \ForEach{$e=(u, v, t)\in \varGamma$}
  {
    Add $e$ to the set $E_{\varGamma}[t-\delta,t]$ of active edges\;\label{ln:sf:update:1}
    Remove the edges with timestamps smaller than $t-\delta$ from $E_{\varGamma}[t-\delta,t]$\;\label{ln:sf:update:2}
    \uIf{$m_t < r$\label{ln:sf:sample:s1}}
    {
      $\widehat{E}_{\varGamma} \gets \widehat{E}_{\varGamma} \cup \{e\}$\;\label{ln:sf:sample:t1}
      Count $\eta(e)$ by running \textsc{BackTracking} with $j=l$ using Line~\ref{ln:es:init}--\ref{ln:es:count:t} of Algorithm~\ref{alg:es} for SES (or estimate $\widehat{\eta}(e)$ when $j=3$ using Line~\ref{ln:ews:mapping}--\ref{ln:ews:wedge:t} of Algorithm~\ref{alg:ews} for SEWS)\;\label{ln:sf:local1}
      $C_{M}(\widehat{E}_{\varGamma}) \gets C_{M}(\widehat{E}_{\varGamma}) + \eta(e)$ for SES (or $C_{M}(\widehat{E}_{\varGamma}) + \widehat{\eta}(e)$ for SEWS)\;\label{ln:sf:counter:1}
      \Return{$\widehat{C}_{M,t} \gets C_{M}(\widehat{E}_{\varGamma})$}\;\label{ln:sf:estimate1}
    }
    \Else
    {
      Toss a biased coin with success probability $\frac{r}{m_t}$\;\label{ln:sf:sample:s2}
      \If{success}
      {
        Select an edge $g$ randomly from $\widehat{E}_{\varGamma}$\;\label{ln:sf:sample:s3}
        $\widehat{E}_{\varGamma} \gets \widehat{E}_{\varGamma} \setminus \{g\} \cup \{e\}$\;\label{ln:sf:sample:t2}
        Count $\eta(e)$ by running \textsc{BackTracking} when $j=l$ using Line~\ref{ln:es:init}--\ref{ln:es:count:t} of Algorithm~\ref{alg:es} for SES (or estimate $\widehat{\eta}(e)$ when $j=3$ using Line~\ref{ln:ews:mapping}--\ref{ln:ews:wedge:t} of Algorithm~\ref{alg:ews} for SEWS)\;\label{ln:sf:local2}
        $C_{M}(\widehat{E}_{\varGamma}) \gets C_{M}(\widehat{E}_{\varGamma}) + \eta(e) - \eta(g)$ for SES (or $C_{M}(\widehat{E}_{\varGamma}) + \widehat{\eta}(e) - \widehat{\eta}(g)$ for SEWS)\;\label{ln:sf:counter:2}
      }
      \Return{$\widehat{C}_{M,t} \gets \frac{m_t}{r} C_{M}(\widehat{E}_{\varGamma})$}\;\label{ln:sf:estimate2}
    }
  }  
\end{algorithm}

The detailed procedures of Streaming ES/EWS are presented in Algorithm~\ref{alg:sf}. We are given a temporal graph stream $\varGamma$, of which all the temporal edges are generated and observed in chronological order, a temporal motif $M$, a time span $\delta$, and a sample size $r$ for reservoir sampling. First of all, it initializes the set $\widehat{E}_{\varGamma}$ of sampled edges to $\varnothing$ and the counter $C_{M}(\widehat{E}_{\varGamma})$ for the number of instances of $M$ containing each sampled edge in $\widehat{E}_{\varGamma}$ to $0$ (Line~\ref{ln:sf:sample:s}). Then, when an edge $e=(u, v, t)$ arrives, it first updates the set $E_{\varGamma}[t-\delta,t]$ of active edges, which only contains the edges with timestamps between $t-\delta$ and $t$, by adding the new edge $e$ while deleting old edges whose timestamps are smaller than $t-\delta$ (Lines~\ref{ln:sf:update:1}--\ref{ln:sf:update:2}). Next, the reservoir sampling is performed to determine whether $e$ is added to $\widehat{E}_{\varGamma}$. Specifically, there are two cases for $e$. First, if the number $m_t$ of temporal edges in $\varGamma_{t}$ is smaller than the sample size $r$, i.e., the reservoir has not been filled yet, then $e$ will be definitely added into $\widehat{E}_{\varGamma}$ (Lines~\ref{ln:sf:sample:s1}); Otherwise, when $\widehat{E}_{\varGamma}$ has contained $r$ edges, $e$ will be added to $\widehat{E}_{\varGamma}$ with probability $\frac{r}{m_t}$ (Line~\ref{ln:sf:sample:s2}). If $e$ is successfully added to $\widehat{E}_{\varGamma}$, it will pick another edge $g$ from $\widehat{E}_\varGamma$ randomly to be replaced with $e$ so that the number of edges in $\widehat{E}_{\varGamma}$ is still equal to $r$ (Lines~\ref{ln:sf:sample:s3}--\ref{ln:sf:sample:t2}). If $e$ is sampled into $\widehat{E}_{\varGamma}$, it uses the same method as ES or EWS to count exactly or estimate approximately the number ${\eta}(e)$ of local $\delta$-instances of $M$ containing $e$ (see Lines~\ref{ln:sf:local1} and~\ref{ln:sf:local2}). After $\eta(e)$ or $\widehat{\eta}(e)$ is acquired, the counter $C_{M}(\widehat{E}_{\varGamma})$ will be updated by adding $\eta(e)$ or $\widehat{\eta}(e)$ (and subtracting $\eta(g)$ or $\widehat{\eta}(g)$ when $g$ is deleted from the sample) (Lines~\ref{ln:sf:counter:1} and~\ref{ln:sf:counter:2}).

Note that the procedure of counting or estimating ${\eta}(e)$ is slightly different from that of ES or EWS: Instead of mapping $e$ to all the $l$ edges of $M$, it only maps $e$ to the last edge $e^{\prime}_{l}$ of $M$ for computing $\eta(e)$ or $\widehat{\eta}(e)$. This is because only the active edges in $E_{\varGamma}[t-\delta,t]$ are available for matching the remaining edges. When $e$ is mapped to $e^{\prime}_{l}$, the remaining matched edges must fall in $E_{\varGamma}[t-\delta,t]$. But if $e$ is mapped to prior edges in $M$, since the edges after $e$ have not been observed yet, the mapping cannot be completed based on $E_{\varGamma}[t-\delta,t]$. Therefore, each local instance is enumerated at most once, and the final estimate $C_{M, t}$ is not scaled by $l$ in this case. To be specific, the estimate $C_{M, t}$ of the motif count in $\varGamma_t$ will be exactly $C_{M}(\widehat{E}_{\varGamma})$ when $m_t<r$ and the sampling probability is $1$ (Line~\ref{ln:sf:estimate1}) or $\frac{m_t}{r} C_{M}(\widehat{E}_{\varGamma})$ when $m_t \geq r$ and the sampling probability is $\frac{r}{m_t}$ (Line~\ref{ln:sf:estimate2}).

\textbf{Theoretical Analysis:}
Next, we will analyze the estimates $\widehat{C}_{M, t}$ returned by SES and SEWS theoretically. We first prove that the estimates $\widehat{C}_{M, t}$ of SES and SEWS are both unbiased in Theorem~\ref{thm:ses:exp}. Then, we provide the upper bounds of the variances of $\widehat{C}_{M, t}$ returned by SES and SEWS in Theorem~\ref{thm:ses:var} and~\ref{thm:sews:var}, respectively.

\begin{theorem}\label{thm:ses:exp}
  The expected values $\mathbb{E}[\widehat{C}_{M, t}]$ of $\widehat{C}_{M, t}$ returned by SES and SEWS are both $C_{M, t}$.
\end{theorem}
\begin{proof}
  We first consider the SES algorithm. When $m_t < r$, $\widehat{E}_{\varGamma}$ contains all the edges in $E_{\varGamma,t}$ at time $t$. Since each edge $e \in E_{\varGamma,t}$ is sampled and each $\delta$-instance of $M$ is counted once when $e$ is mapped to $e^{\prime}_{l}$, we have $\widehat{C}_{M, t}$ is exactly equal to $C_{M, t}$. Then, when $m_t \geq r$, we also index the edges in $E_{\varGamma,t}$ by $[1,m_t]$ and use an indicator $\omega_i$ to denote whether the $i$\textsuperscript{th} edge $e_i$ is sampled, i.e.,
  \begin{equation*}
    \omega_i =
    \begin{cases}
      1, & e_i \in \widehat{E}_{\varGamma} \\
      0, & e_i \notin \widehat{E}_{\varGamma}
    \end{cases}
  \end{equation*}
  Then, we have
  \begin{equation}\label{Eq:sn1}
    \widehat{C}_{M, t} = \frac{m_t}{r} \sum_{e \in \widehat{E}_{\varGamma}} \eta(e) = \frac{m_t}{r} \sum_{i=1}^{m_t} \omega_{i} \cdot \eta(e_i)
  \end{equation}
  Next, based on the result of Equation~\ref{Eq:sn1} and the fact that $\mathbb{E}[\omega_{i}]=\frac{r}{m_t}$ according to the property of \emph{reservoir sampling}~\cite{DBLP:journals/toms/Vitter85}, we have
  \begin{equation*}
    \mathbb{E}[\widehat{C}_{M, t}] = \frac{m_t}{r} \sum_{i=1}^{m_t} \mathbb{E}[\omega_{i}] \cdot \eta(e_i) = \sum_{i=1}^{m_t} \eta(e_i) = C_{M, t}
  \end{equation*}
  and conclude that $\widehat{C}_{M, t}$ returned by SES is unbiased.

  By combining the above results for SES together with the fact that $\mathbb{E}[\widehat{\eta}(e)]=\eta(e)$ as proven in Theorem~\ref{thm:ews:exp}, we also conclude that $\mathbb{E}[\widehat{C}_{M, t}] = C_{M, t}$ for the estimate $\widehat{C}_{M, t}$ returned by SEWS.
\end{proof}

\begin{theorem}\label{thm:ses:var}
  The variance $\textnormal{Val}[\widehat{C}_{M, t}]$ of $\widehat{C}_{M, t}$ returned by SES is at most $\frac{m_t-r}{r} \cdot C^2_{M, t}$.
\end{theorem}
\begin{proof}
  According to Equation~\ref{Eq:sn1}, we have
  \begin{equation*}
    \textnormal{Val}[\widehat{C}_{M, t}] = \textnormal{Val}\Big[\sum_{i=1}^{m_t} \frac{m_t \eta(e_i)}{r} \cdot \omega_i \Big] = \sum_{i, j = 1}^{m_t} \frac{m_t\eta(e_i)}{r} \cdot \frac{m_t\eta(e_j)}{r} \cdot \textnormal{Cov}(\omega_i,\omega_j)
  \end{equation*}
  Because the indicator $\omega_{i}$ and $\omega_{j}$ are independent if $i \neq j$, we have $\textnormal{Cov}(\omega_i,\omega_j) = 0$ for any $i \neq j$. In addition, $\textnormal{Cov}(\omega_i,\omega_i) = \textnormal{Val}[\omega_i]=\frac{r}{m_t}-(\frac{r}{m_t})^2$. Based on the above results, we have
  \begin{align*}
    \textnormal{Val}[\widehat{C}_{M, t}]
    & = \sum_{i=1}^{m_t} \frac{{m_t}^2\eta^2(e_i)}{{r}^2}(\frac{r}{m_t}-\frac{r^2}{m^2_t}) = \frac{m_t-r}{r}\sum_{i=1}^{m_t}\eta^2(e_i) \\
    & \leq \frac{m_t-r}{r} \Big(\sum_{i=1}^{m_t} \eta(e_i)\Big)^2 = \frac{m_t-r}{r} \cdot C^2_{M, t}
  \end{align*}
  and conclude the proof.
\end{proof}

\begin{theorem}\label{thm:sews:var}
  The variance $\textnormal{Val}[\widehat{C}_{M, t}]$ of $\widehat{C}_{M, t}$ returned by SEWS is at most $\frac{m_t - r q}{r q} \cdot C^2_{M, t}$.
\end{theorem}
\begin{proof}
  Let us index the edges in $E_{\varGamma,t}$ at time $t$ by $[1,m_t]$ and the edges in $E(e_i)$ when $e_i$ is mapped to $e^\prime_3$ by $[1,m_{ti}]$, where $m_{ti}=|E(e_i)|$. Similar to the proof of Theorem~\ref{thm:ews:var}, we have
  \begin{align*}
    \textnormal{Val}[\widehat{C}_{M, t}]
    & = \textnormal{Val}\Big[\frac{m_t}{r q} \sum_{i=1}^{m_t}\sum_{j=1}^{m_{ti}} \omega_i \cdot \omega_{ij} \cdot \eta(W_{ij}) \Big] \\
    & = \sum_{i=1}^{m_t}\sum_{j=1}^{m_{ti}} \frac{m_t^2\eta^2(W_{ij})}{r^2q^2} \cdot \textnormal{Var}[\omega_i \cdot \omega_{ij}] \\
    & = \frac{m_t - r q}{r q} \cdot \sum_{i=1}^{m_t}\sum_{j=1}^{m_{ti}} \eta^2(W_{ij}) \leq \frac{m_t - r q}{r q} \cdot C^2_{M, t}
  \end{align*}
  where $\eta(W_{i j})$ is the number of $\delta$-instances of $M$ containing a temporal wedge $W_{i j}$ and $\omega_{i j}$ is its indicator. All the inequalities are adapted from the ones in Theorem~\ref{thm:ews:var}.
\end{proof}

\textbf{Space and Time Complexity:}
First of all, both SES and SEWS only keep the set $E_{\varGamma}[t-\delta,t]$ of active edges and the set $\widehat{E}_{\varGamma}$ of sampled edges as well as their local counts in memory. Thus, the space complexities of both algorithms are $O(m_{\delta} + r)$, where $m_{\delta} = \max_{t} |E_{\varGamma}[t-\delta,t]|$. Then, for both algorithms, the procedures of maintaining active edges, reservoir sampling, updating and returning the estimate take only constant time per edge. The SES algorithm takes $O(d_{\delta}^{l-1})$ time to compute $\eta(e)$ for each edge $e$, where $d_{\delta}$ is the maximum number of (in-/out-)edges connected with one vertex in $E_{\varGamma}[t-\delta,t]$, since SES only runs \textsc{BackTracking} once, which is lower than $O(l d_{\delta}^{l-1})$ time per edge in ES. Similar to EWS, the time to compute $\widehat{\eta}(e)$ in SEWS is also $O(d_{\delta}\log{h})$, where $O(\log{h})$ is the time to compute $\eta(W)$ using a hash table. Therefore, the time complexities of SES and SEWS to update $\widehat{C}_{M, t}$ for each edge are $O(d_{\delta}^{l-1})$ and $O(d_{\delta}\log{h})$, respectively.

\section{Experimental Evaluation}\label{sec:exp}

In this section, we evaluate the empirical performance of our proposed algorithms on real-world datasets. We first introduce the experimental setup in Section~\ref{subsec:exp:setup}. The experimental results for offline and streaming temporal graphs are presented in Sections~\ref{subsec:exp:results} and~\ref{subsec:exp:results1}, respectively.

\subsection{Experimental Setup}\label{subsec:exp:setup}

\noindent\textbf{Experimental Environment:}
All experiments were conducted on a server running Ubuntu 18.04.1 LTS with an Intel\textsuperscript{\textregistered} Xeon\textsuperscript{\textregistered} Gold 6140 2.30GHz processor and 250GB RAM. We downloaded the code\footnote{\url{http://snap.stanford.edu/temporal-motifs/}}$^,$\footnote{\url{https://github.com/rohit13k/CycleDetection}}$^,$\footnote{\url{https://gitlab.com/paul.liu.ubc/sampling-temporal-motifs}} of the baseline algorithms published by the authors and followed the compilation and usage instructions. Our proposed algorithms were implemented in \texttt{C++11} compiled by \texttt{GCC~v7.4} with \texttt{-O3} optimizations, and ran on a single thread. Our code are publicly available on GitHub\footnote{\url{https://github.com/jingjing-hnu/Temporal-Motif-Counting}}.

\textbf{Datasets:}
We used five different real-world datasets in our experiments including AskUbuntu (AU), SuperUser (SU), StackOverflow (SO), BitCoin (BC), and RedditComments (RC). All the datasets were downloaded from publicly available sources like the SNAP repository~\cite{DBLP:journals/tist/LeskovecS16}. Each dataset contains a sequence of temporal edges in chronological order. We report several statistics of these datasets in Table~\ref{tab:Datasets}, including the number of vertices, the number of static edges, the number of temporal edges, and the overall time span, i.e., the time difference between the first and last edges.

\textbf{Algorithms:}
The offline algorithms we compare are listed as follows.
\begin{itemize}
  \item \textbf{EX}: An exact algorithm for temporal motif counting in~\cite{DBLP:conf/wsdm/ParanjapeBL17}. The implementation published by the authors is applicable only to $3$-edge motifs and cannot support motifs with $4$ or more edges (e.g.,~Q5 in Fig.~\ref{fig:queries}).
  \item \textbf{2SCENT}: An algorithm for simple temporal cycle (e.g.,~Q4 and Q5 in Fig.~\ref{fig:queries}) enumeration in~\cite{DBLP:journals/pvldb/KumarC18}.
  \item \textbf{BT}: A \textsc{BackTracking} algorithm for temporal subgraph isomorphism in~\cite{DBLP:conf/bigdataconf/MackeyPFCC18}. It can provide the exact count of any temporal motif by enumerating all instances of them.
  \item \textbf{IS-BT}: An interval-based sampling algorithm for temporal motif counting in~\cite{DBLP:conf/wsdm/LiuBC19}. BT~\cite{DBLP:conf/bigdataconf/MackeyPFCC18} is used as a subroutine for any motif with more than $2$ vertices.
  \item \textbf{ES}: Our generic edge sampling algorithm for temporal motif counting in Section~\ref{subsec:alg:1}.
  \item \textbf{EWS}: Our improved edge-wedge sampling algorithm for counting temporal motifs with $3$ vertices and $3$ edges (e.g.~Q1--Q4 in Figure~\ref{fig:queries}) in Section~\ref{subsec:alg:2}.
\end{itemize}
The two streaming algorithms we compare are listed as follows.
\begin{itemize}
  \item \textbf{SES}: Our streaming extension of ES for counting temporal motifs of any kind in Section~\ref{sec:s-alg}.
  \item \textbf{SEWS}: Our streaming extension of EWS for counting temporal motifs with $3$ vertices and $3$ edges in Section~\ref{sec:s-alg}.
\end{itemize}

\begin{table}[t]
  \small
  \centering
  \caption{Statistics of datasets}
  \label{tab:Datasets}
  \begin{tabular}{|c|c|c|c|c|}
    \hline
    \textbf{Dataset} & \#\textbf{Vertices} & \#\textbf{Static Edges} & \#\textbf{Temporal Edges} & \textbf{Time Span} \\
    \hline
    \textbf{AU} & $157,222$  &$544,621$ & $726,639$  & $7.16$ years \\
    \hline
    \textbf{SU} & $192,409$  & $854,377$& $1,108,716$ & $7.60$ years \\
    \hline
    \textbf{SO} & $2,584,164$ & $34,875,684$& $47,902,865$ & $7.60$ years \\
    \hline
    \textbf{BC} & $48,098,591$ &$86,798,226$ & $113,100,979$  & $7.08$ years \\
    \hline
    \textbf{RC} & $5,688,164$ &$329,485,956$ & $399,523,749$  & $7.44$ years \\ 
    \hline
  \end{tabular}
\end{table}

\begin{figure}[t]
  \centering
  \includegraphics[width=0.6\textwidth]{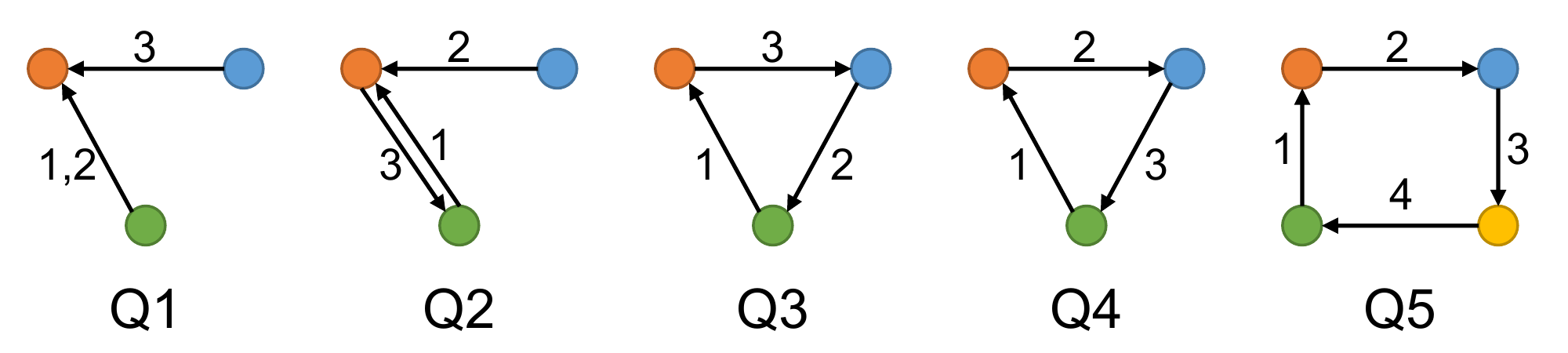}
  \caption{Query motifs}
  \Description{queries}
  \label{fig:queries}
\end{figure}

\textbf{Queries:}
The five query motifs we use in the experiments are listed in Fig.~\ref{fig:queries}. Since different algorithms specialize in different types of motifs, we select a few motifs that can best represent the specializations of all the algorithms. As discussed above, an algorithm may not be applicable to some of the query motifs. In this case, the algorithm is ignored in the experiments on these motifs.

\textbf{Performance Measures:}
In the offline setting, the efficiency is measured by the CPU time (in seconds) of an algorithm to count a query motif in a temporal graph. In the streaming setting, when a new edge arrives, an algorithm updates and returns an estimate.
The efficiency is measured by the CPU time (in seconds) of an algorithm to update the motif count for each 1,000 new edges. The accuracy of a sampling algorithm is measured by the relative error $\frac{|\widehat{x}-x|}{x}$ where $x$ is the exact number of instances of a query motif in a temporal graph and $\widehat{x}$ is an estimate of $x$ returned by the algorithm. In each experiment, we run all the algorithms $10$ times and use the average CPU time and relative errors for comparison.

\subsection{Experimental Results in Offline Temporal Graphs}\label{subsec:exp:results}

\textbf{Overall Performance.}
The experimental results of all offline algorithms are reported in Table~\ref{tab:results}. Here, the time span $\delta$ is set to $86400$ seconds (i.e., one day) on AU and SU, and $3600$ seconds (i.e., one hour) on SO, BC, and RC (Note that we use the same values of $\delta$ across all experiments, unless specified). For IS-BT, we report the results in the default setting as indicated in~\cite{DBLP:conf/wsdm/LiuBC19}, i.e., we fix the interval length to $30\delta$ and present the result for the smallest interval sampling probability that can guarantee the relative error is at most $5\%$. For ES and EWS, we report the results when $p=0.01$ by default; in a few cases when the numbers of motif instances are too small or their distribution is highly skewed among edges, we report the results when $p=0.1$ (marked with ``*'' in Table~\ref{tab:results}) because ES and EWS cannot provide accurate estimates when $p=0.01$. In addition, we set $q=1$ on AU and SU and $q=0.1$ on SO, BC, and RC for EWS.

First of all, the efficiencies of EX and 2SCENT are significantly lower than the other algorithms. This is because they use an algorithm for subgraph isomorphism or cycle detection in static graphs as a subroutine for candidate generation without considering the temporal information of edges. As a result, a large number of redundant candidates that violates the duration constraint are generated, which leads to the degradation in performance. Second, on medium-sized datasets (i.e., AU and SU), ES runs faster than IS-BT in most cases. Meanwhile, their relative errors are close to each other. On large datasets (i.e, SO, BC, and RC), ES demonstrates both much higher efficiency (up to $10.3$x speedup) and lower estimation errors ($2.42\%$ vs. $4.61\%$) than IS-BT. Third, EWS runs $1.7$x--$19.6$x faster than ES due to its lower computational cost per edge. The relative errors of ES and EWS are the same on AU and SU since we use $q=1$ for EWS. When $q=0.1$, EWS achieves further speedups at the expense of higher relative errors. A more detailed analysis of the effect of $q$ will be provided in the following paragraph.

\begin{table}[t]
  \small
  \centering
  \caption{Running time (in seconds) and average errors ($\%$) of all offline algorithms on each dataset. We use ``---'' and ``\ding{53}'' to denote ``motif not supported'' and ``running out of memory'', respectively. For IS-BT, ES, and EWS, we show their speedup ratios over BT for comparison. We use ``*'' to mark the results of ES and EWS for $p=0.1$ instead of $p=0.01$.}
  \label{tab:results}
  \begin{tabular}{|c|c|c|c|c|c|c|c|c|c|c|}
    \hline
    \multirow{2}{*}{\textbf{Dataset}} & \multirow{2}{*}{\textbf{Motif}} & \textbf{EX} & \textbf{2SCENT} & \textbf{BT} & \multicolumn{2}{c|}{\textbf{IS-BT}} & \multicolumn{2}{c|}{\textbf{ES}} & \multicolumn{2}{c|}{\textbf{EWS}} \\ \cline{3-11}
    & & time (s) & time (s) & time (s) & error & time (s) & error & time (s) & error & time (s) \\
    \hline
    \multirow{5}{*}{\textbf{AU}}
    & Q1 & \multirow{2}{*}{1.8} & \multirow{3}{*}{\textbf{---}} & 0.758 & 4.84\% & 0.402/1.9x &\textbf{4.32\%} & 0.059/12.8x & 4.32\% & \textbf{0.027/28.1x} \\ \cline{2-2} \cline{5-11} 
    & Q2 &  &  & 1.104 & \textbf{4.16\%} & 0.434/2.5x & 4.57\% & 0.048/23.0x & 4.57\% & \textbf{0.029/38.1x} \\ \cline{2-3} \cline{5-11} 
    & Q3 & \multirow{2}{*}{2.3} &  & 0.884 & 3.97\% & 0.50/1.8x & *\textbf{3.73\%} & *0.605/1.5x & *3.73\% & *\textbf{0.183/4.8x} \\ \cline{2-2} \cline{4-11}
    & Q4 &  & \multirow{2}{*}{23.68} & 1.038 & 4.67\% & 0.492/2.1x & *\textbf{4.63\%} & *0.628/1.7x &*4.63\% & *\textbf{0.173/6x} \\ \cline{2-3} \cline{5-11} 
    & Q5 & \textbf{---} &  & 1.262 & \textbf{3.98\%} & 0.536/2.4x & *4.62\% & *\textbf{0.322/3.9x} & \multicolumn{2}{c|}{\textbf{---}}\\ \hline
    \multirow{5}{*}{\textbf{SU}}
    & Q1 & \multirow{2}{*}{3.26} & \multirow{3}{*}{\textbf{---}} & 1.499 & 3.99\% & 0.620/2.4x & \textbf{3.06\%} & 0.102/14.7x & 3.06\% & \textbf{0.052/28.8x} \\ \cline{2-2} \cline{5-11} 
    & Q2 &  &  & 1.650 & 3.23\% & 0.671/2.5x & \textbf{2.47\%} & 0.083/19.9x & 2.47\% & \textbf{0.046/35.9x} \\ \cline{2-3} \cline{5-11} 
    & Q3 & \multirow{2}{*}{4.6} &  & 1.506 & 4.85\% & 0.723/2.1x & \textbf{4.66\%} & 0.113/13.3x & 4.66\% & \textbf{0.030/50.2x} \\ \cline{2-2} \cline{4-11} 
    & Q4 &  & \multirow{2}{*}{46.0} & 1.434 & \textbf{3.79\%} & 0.725/2.0x & 4.63\% & 0.128/11.2x & 4.63\% & \textbf{0.042/34.1x} \\ \cline{2-3} \cline{5-11} 
    & Q5 & \textbf{---} &  & 1.521 & 4.55\% & 0.759/2.0x & *\textbf{4.52\%} & *\textbf{0.453/3.4x} & \multicolumn{2}{c|}{\textbf{---}}\\ \hline
    \multirow{5}{*}{\textbf{SO}}
    & Q1 & \multirow{2}{*}{169} & \multirow{3}{*}{\textbf{---}} & 105.8 & 4.82\% & 8.626/12.3x &\textbf{0.97\%} & 4.419/23.9x & 1.22\% & \textbf{1.528/69.2x} \\ \cline{2-2} \cline{5-11} 
    & Q2 &  &  & 110.7 & 4.82\% & 27.48/4.0x & \textbf{0.20\%} & 3.985/27.8x & 0.89\% & \textbf{1.514/73.1x} \\ \cline{2-3} \cline{5-11} 
    & Q3 & \multirow{2}{*}{466} &  & 107.4 & 4.30\% & 25.70/4.2x & \textbf{1.36\%} & 4.031/26.6x & 3.6\% & \textbf{1.235/87x} \\ \cline{2-2} \cline{4-11} 
    & Q4 &  & \multirow{2}{*}{243.7} & 105.5 & 4.90\% & 6.775/15.6x & \textbf{1.78\%} & 3.936/26.8x & 3.31\% & \textbf{1.153/91.5x}\\ \cline{2-3} \cline{5-11} 
    & Q5 & \textbf{---} &  & 91.83 & 4.91\% & 9.451/9.7x & \textbf{3.48\%} & \textbf{1.505/61.0x} & \multicolumn{2}{c|}{\textbf{---}}\\ \hline
    \multirow{5}{*}{\textbf{BC}}
    & Q1 & \multirow{2}{*}{8143} & \multirow{3}{*}{\textbf{---}} & 220.0 & 4.75\% & 50.02/4.4x & \textbf{0.64\%} & 59.12/3.7x & 0.67\% & \textbf{9.463/23.2x} \\ \cline{2-2} \cline{5-11} 
    & Q2 &  &  & 399.8 & 4.90\% & 125.1/3.2x & \textbf{1.11\%} & 34.74/11.5x & 1.16\% & \textbf{8.126/49.2x} \\ \cline{2-3} \cline{5-11} 
    & Q3 & \multirow{2}{*}{8116} &  & 396.8 & 3.89\% & 90.19/4.4x & \textbf{1.49\%} & 41.49/9.6x & 3.02\% & \textbf{2.121/187x} \\ \cline{2-2} \cline{4-11} 
    & Q4 &  & \multirow{2}{*}{473.7} & 473.4 & 4.93\% & 95.47/5.0x & \textbf{0.83\%} & 37.43/12.6x & 1.91\% & \textbf{2.262/209x} \\ \cline{2-3} \cline{5-11} 
    & Q5 & \textbf{---} &  & 596.4 & 4.83\% & 319.7/1.9x & \textbf{2.92\%} & \textbf{20.47/29.1x} & \multicolumn{2}{c|}{\textbf{---}}\\ \hline
    \multirow{5}{*}{\textbf{RC}}
    & Q1 & \multirow{2}{*}{2799} & \multirow{3}{*}{\textbf{---}} & 1966 & 4.76\% & 840.5/2.3x & \textbf{3.27\%} & 257.4/7.6x & 3.36\% & \textbf{31.49/62.4x} \\ \cline{2-2} \cline{5-11} 
    & Q2 &  &  & 2113 & 4.67\% & 428/4.9x & 0.63\% & 120.6/17.5x & \textbf{0.6\%} & \textbf{30.57/69.1x} \\ \cline{2-3} \cline{5-11} 
    & Q3 &  \multirow{2}{*}{\ding{53}} &  & 2069 & 4.61\% & 784.4/2.6x & 2.42\% & 76.09/27.2x & \textbf{2.27\%} & \textbf{16.17/128x} \\ \cline{2-2} \cline{4-11} 
    & Q4 &  & \multirow{2}{*}{2245} & 1897 & 4.86\% & 683/2.8x & \textbf{3.47\%} & 68.60/27.7x & 4.57\% & \textbf{15.91/119x} \\ \cline{2-3} \cline{5-11} 
    & Q5 & \textbf{---} &  & 1613 & 4.41\% & 706.6/2.3x & *\textbf{4.32\%} & *\textbf{120.3/13.4x} & \multicolumn{2}{c|}{\textbf{---}}\\ \hline
  \end{tabular}
\end{table}

\begin{figure}[t]
  \subfigure[Q2 on SU]{
    \label{subfig:q:a1}
    \includegraphics[height=0.9in]{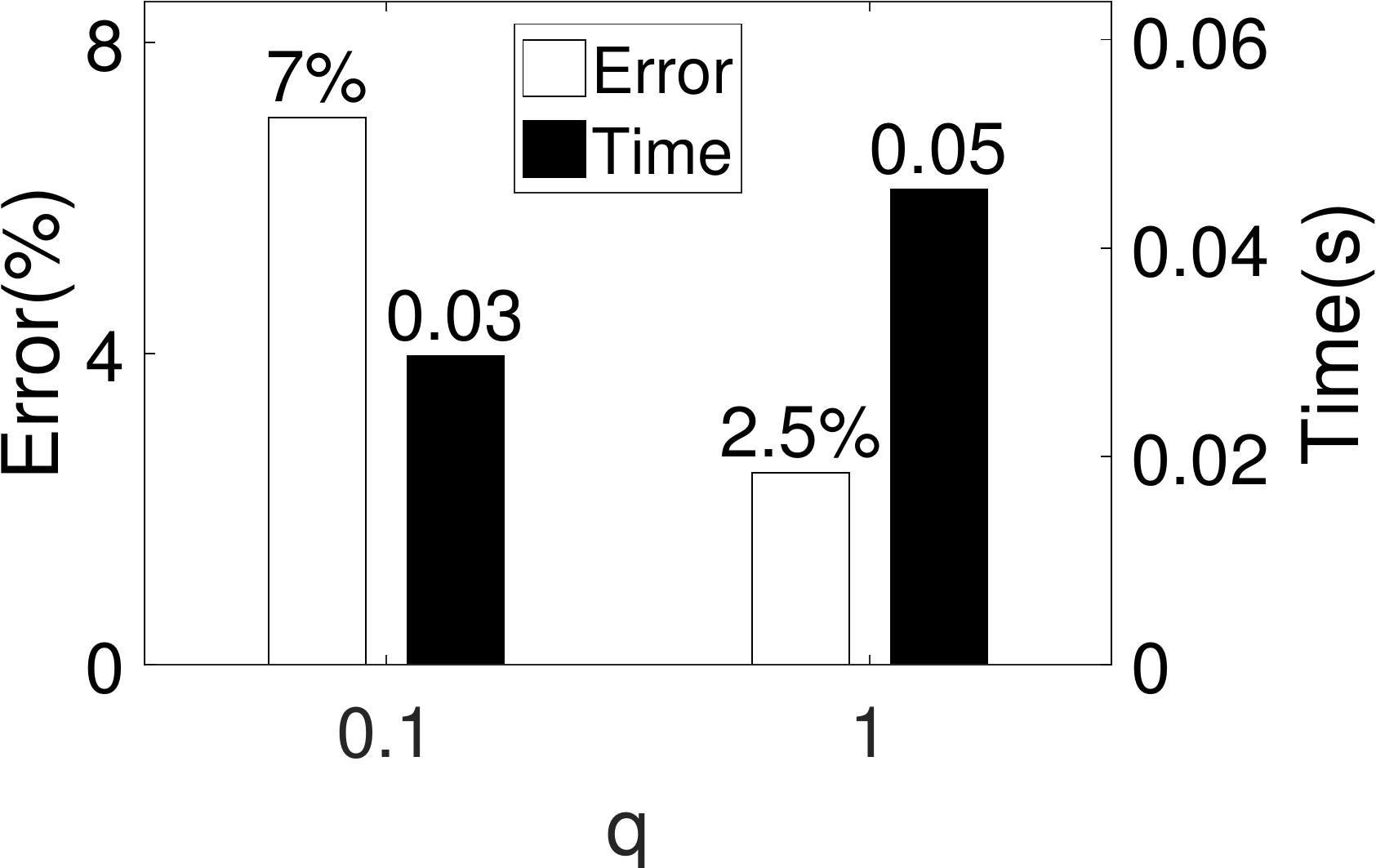}
  }
  \hfill
  \subfigure[Q3 on SU]{
    \label{subfig:q:a2}
    \includegraphics[height=0.9in]{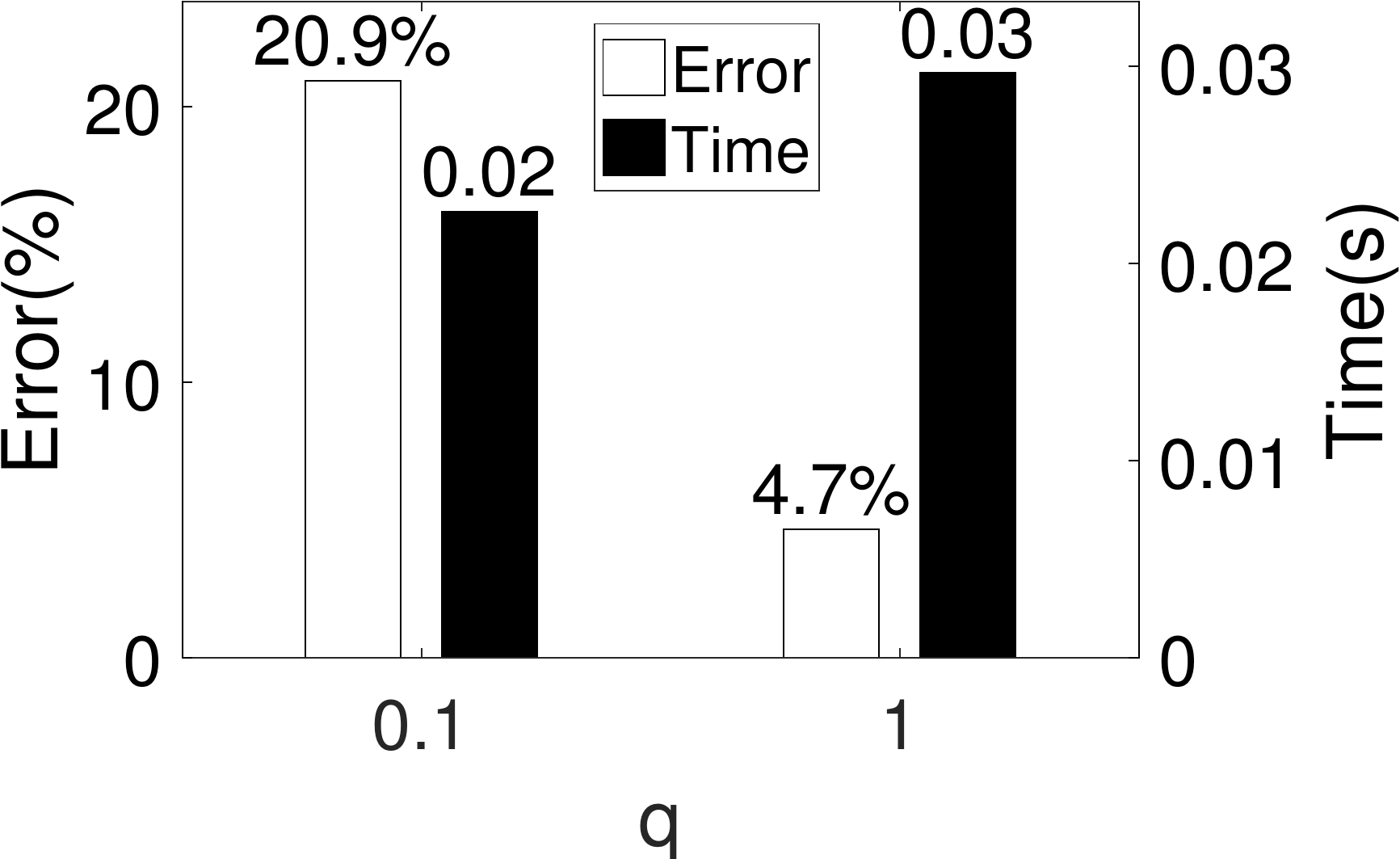}
  }
  \hfill
  \subfigure[Q2 on BC]{
    \label{subfig:q:b1}
    \includegraphics[height=0.9in]{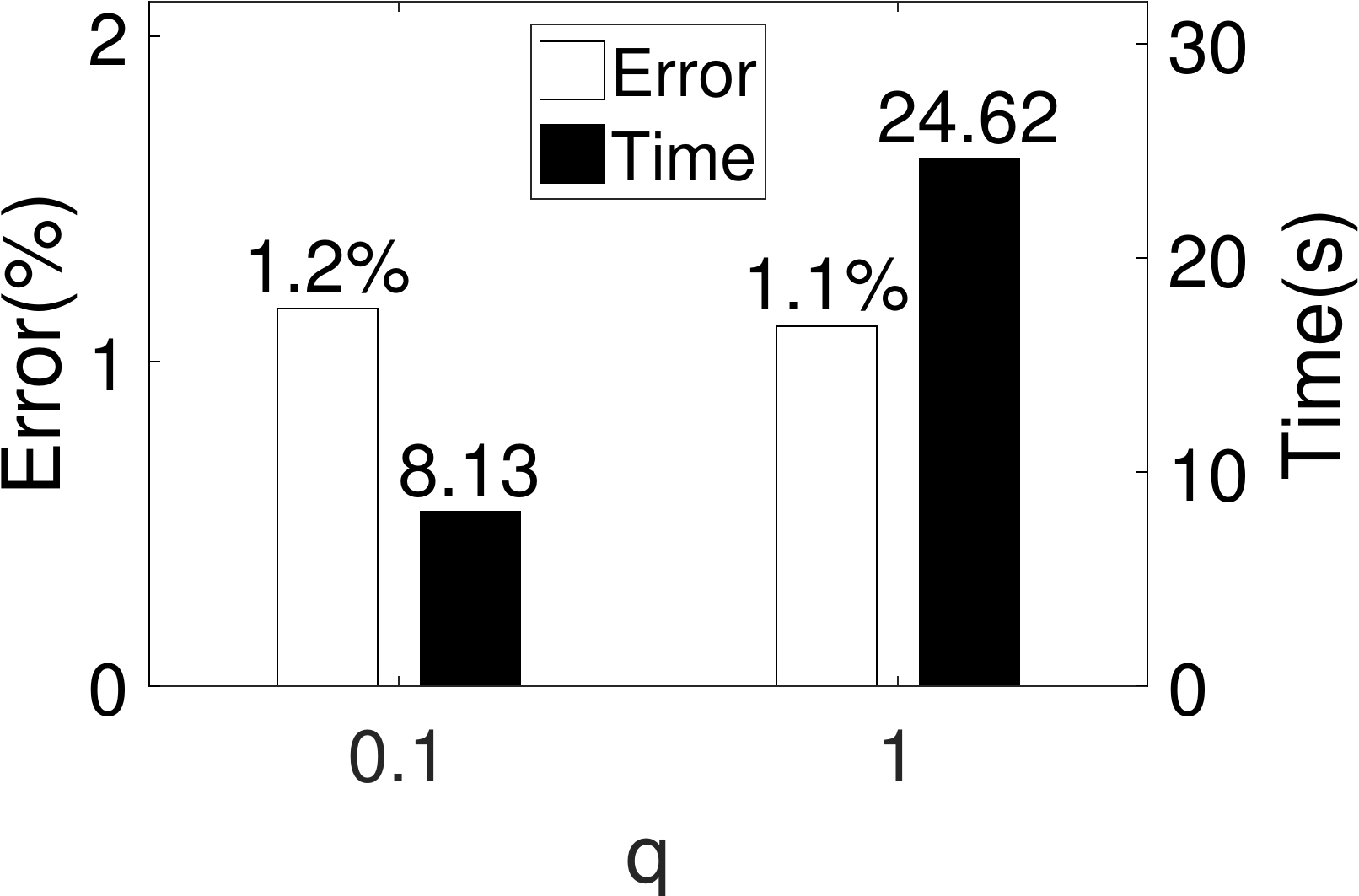}
  }
  \hfill
  \subfigure[Q3 on BC]{
    \label{subfig:q:b2}
    \includegraphics[height=0.9in]{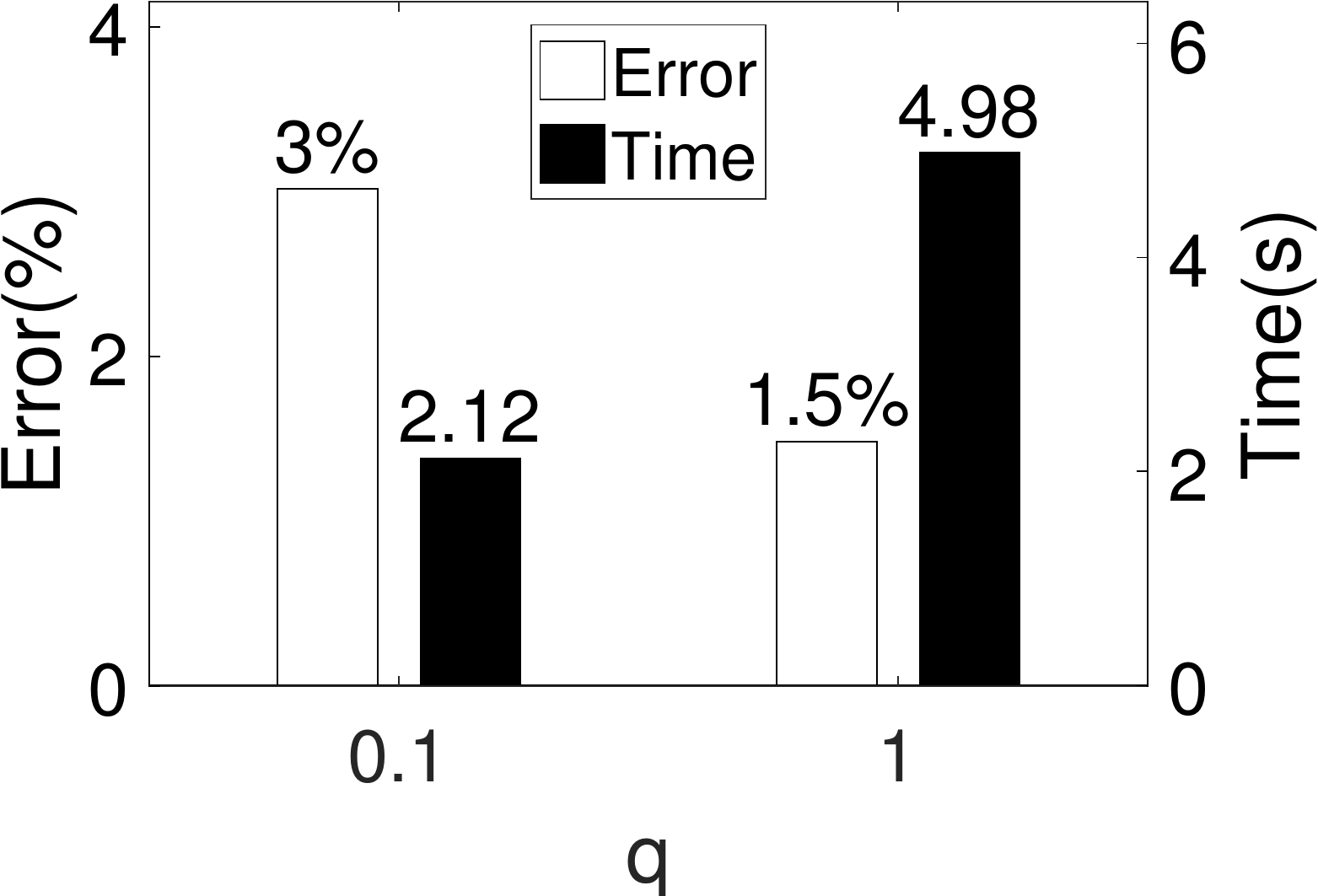}
  }
  \caption{Comparison of the performance of EWS when $q=0.1$ and $1$.}
  \Description{Q}
  \label{fig:q}
\end{figure}

\textbf{Effect of $q$ for EWS:}
In Fig.~\ref{fig:q}, we compare the relative errors and running time of EWS for $q=1$ and $0.1$ when $p$ is fixed to $0.01$. We observe different effects of $q$ on medium-sized (e.g., SU) and large (e.g., BC) datasets. On the SU dataset, the benefit of smaller $q$ is marginal: the running time decreases slightly but the errors become obviously higher. But on the BC dataset, by setting $q=0.1$, EWS achieves $2$x--$3$x speedups without affecting the accuracy seriously. These results imply that \emph{temporal wedge sampling} is more effective on larger datasets than smaller ones, since the number of instances of a query motif is much greater on larger datasets and thus acceptable estimation errors can be achieved at lower sampling rates ($0.001$ vs.~$0.01$). Therefore, we set $q=1$ on AU and SU, and $q=0.1$ on SO, BC, and RC for EWS in the remaining experiments.

\begin{figure}
  \centering
  \includegraphics[height=0.12in]{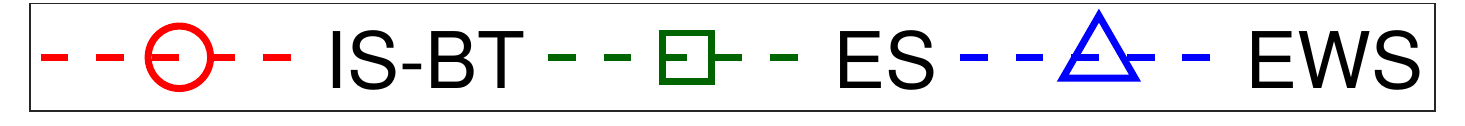}
  \\
  \subfigure[Q1 on AU]{
    \label{fig:p:m1:au}
    \includegraphics[height=0.78in]{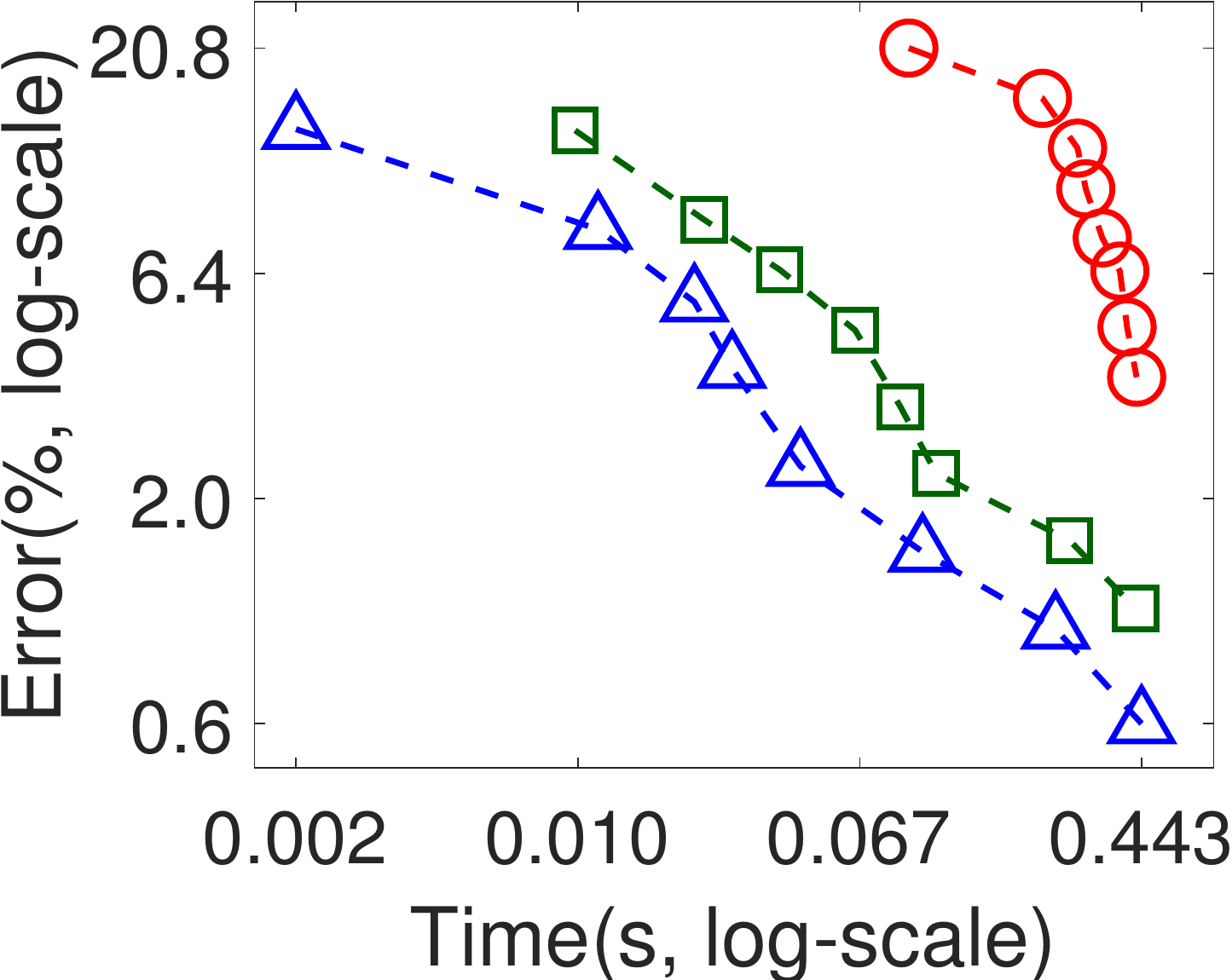}
  }
  \hfill
  \subfigure[Q2 on AU]{
    \label{fig:p:m2:au}
    \includegraphics[height=0.78in]{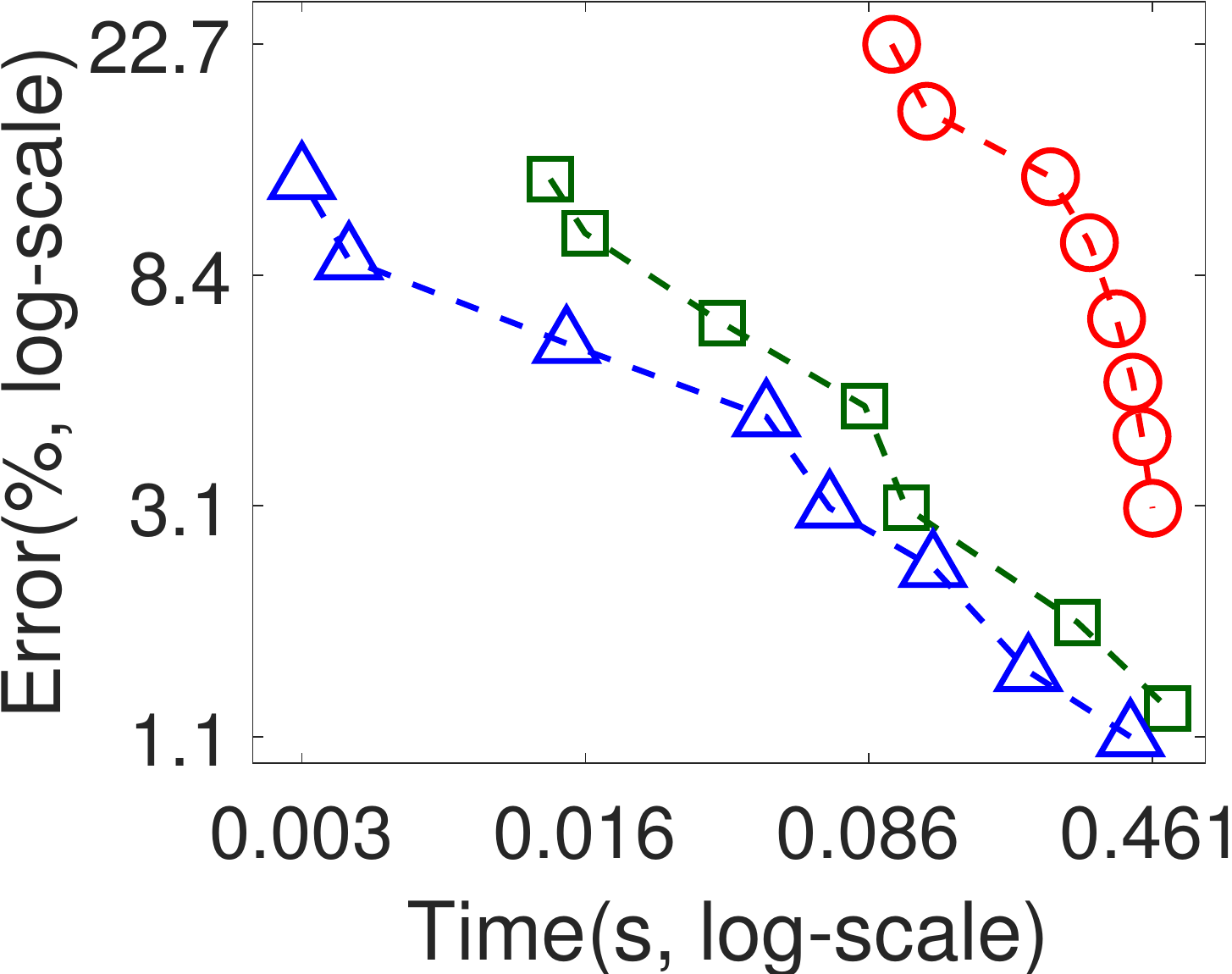}
  }
  \hfill
  \subfigure[Q3 on AU]{
    \label{fig:p:m3:au}
    \includegraphics[height=0.78in]{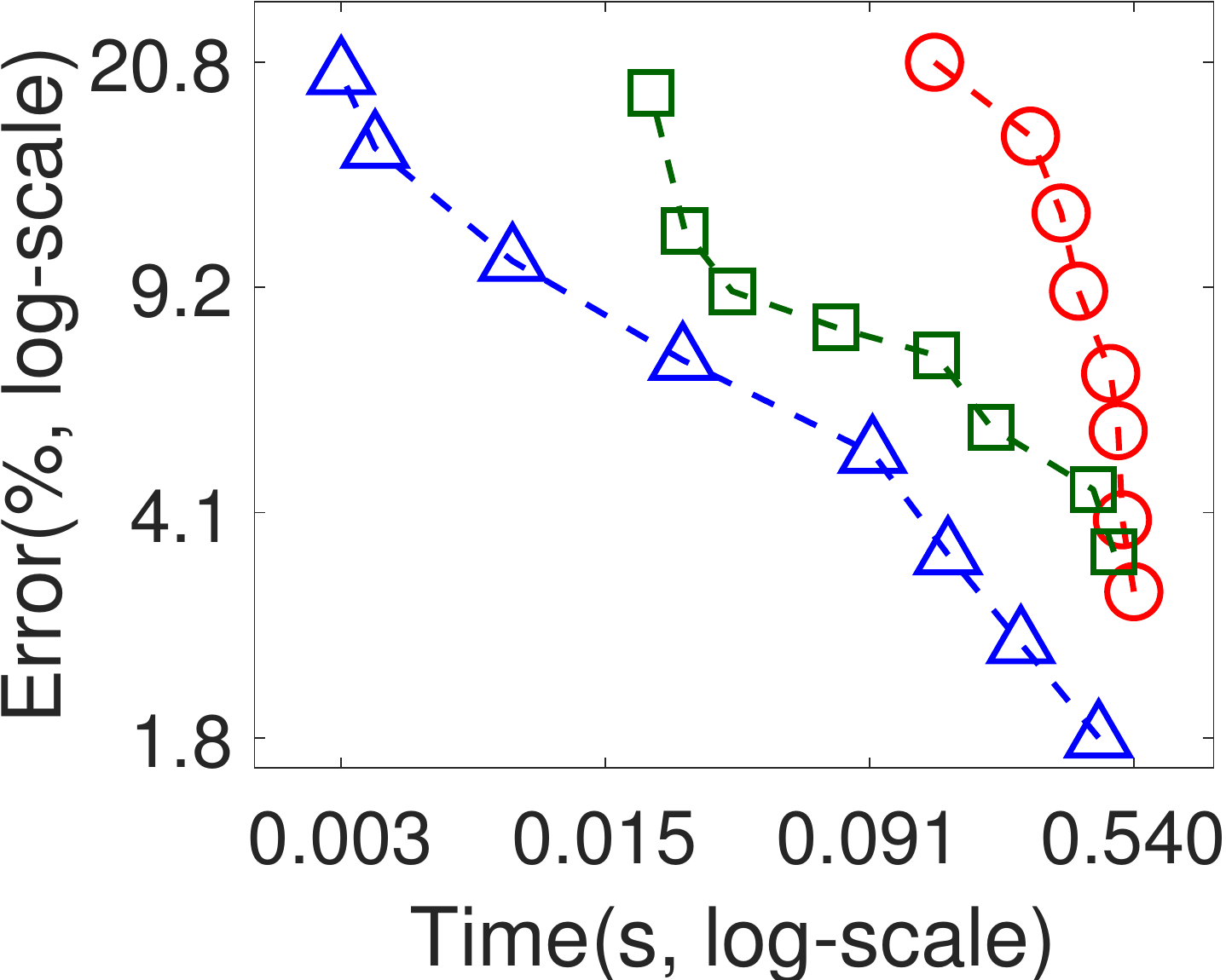}
  }
  \hfill
  \subfigure[Q4 on AU]{
    \label{fig:p:m4:au}
    \includegraphics[height=0.78in]{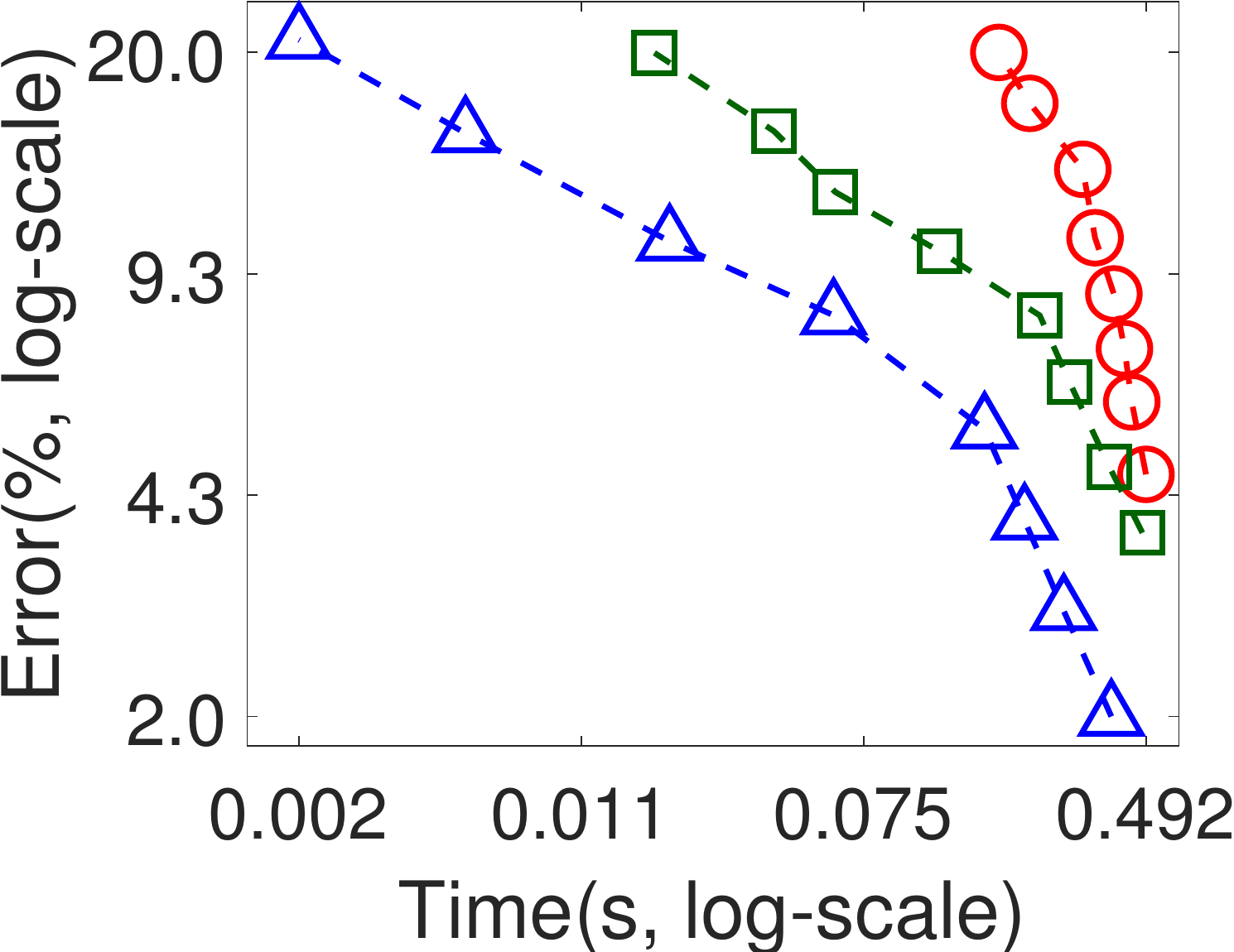}
  }
  \hfill
  \subfigure[Q5 on AU]{
    \label{fig:p:m5:au}
    \includegraphics[height=0.78in]{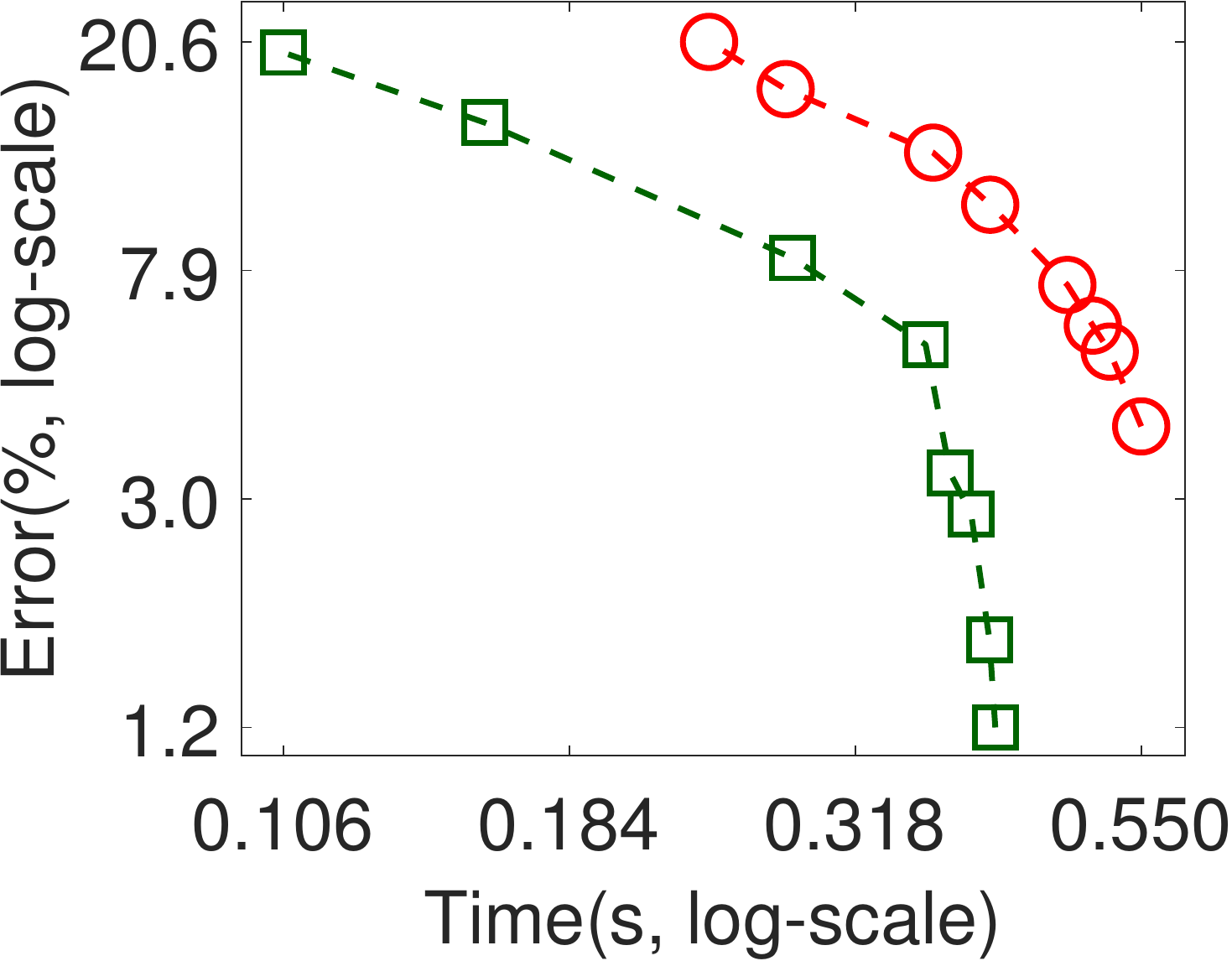}
  }
  \subfigure[Q1 on SU]{
    \label{fig:p:m1:su}
    \includegraphics[height=0.77in]{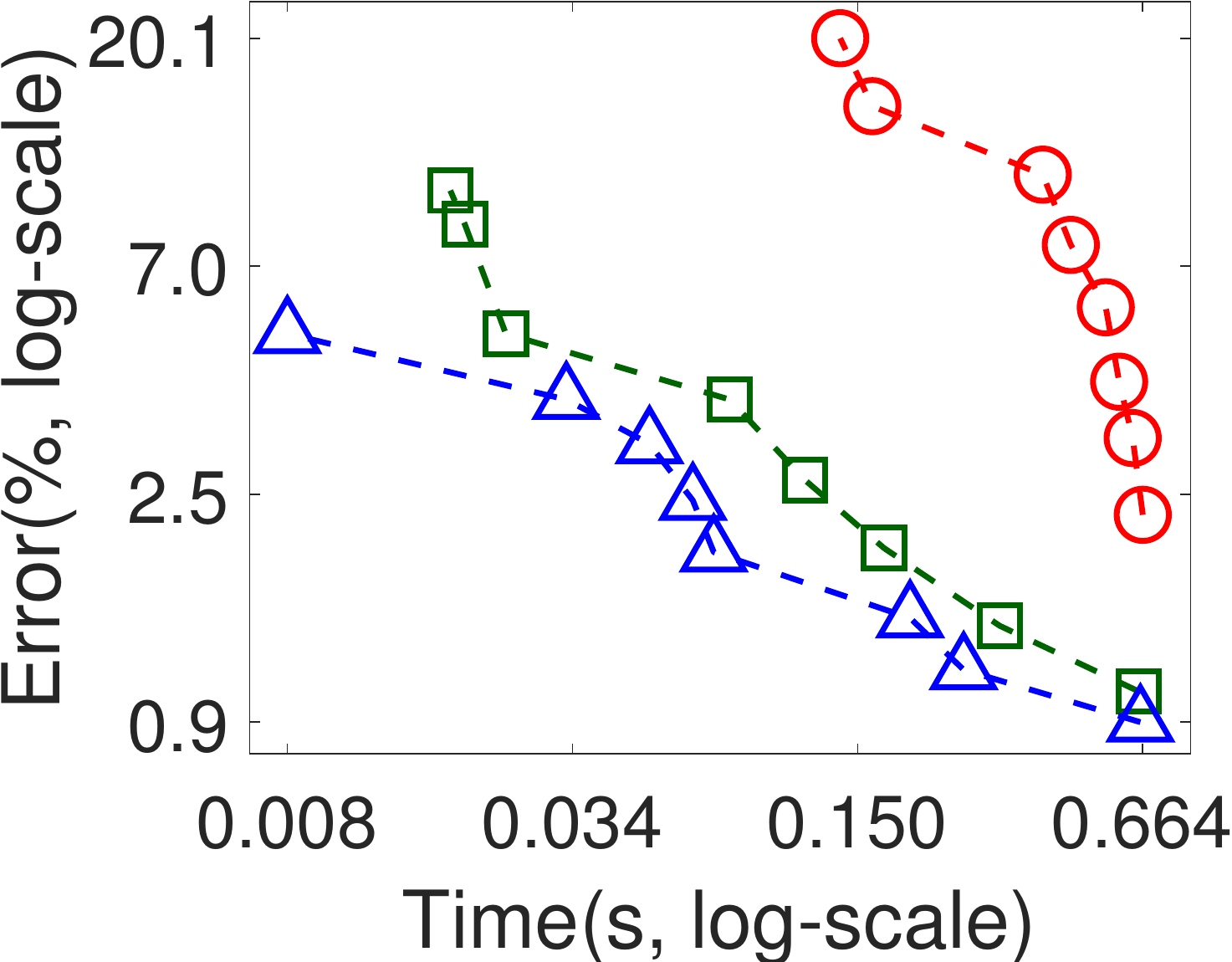}
  }
  \hfill
  \subfigure[Q2 on SU]{
    \label{fig:p:m2:su}
    \includegraphics[height=0.77in]{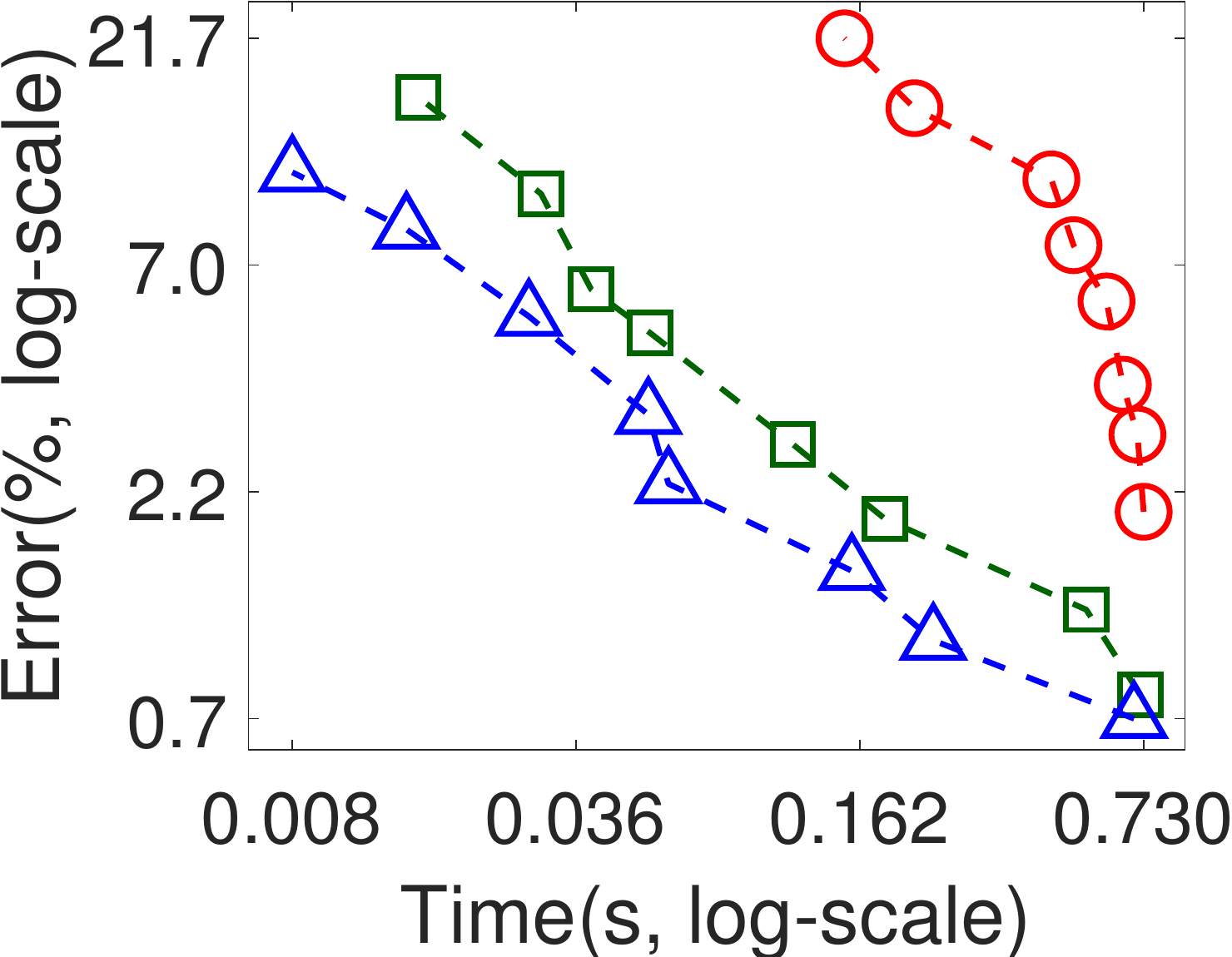}
  }
  \hfill
  \subfigure[Q3 on SU]{
    \label{fig:p:m3:su}
    \includegraphics[height=0.77in]{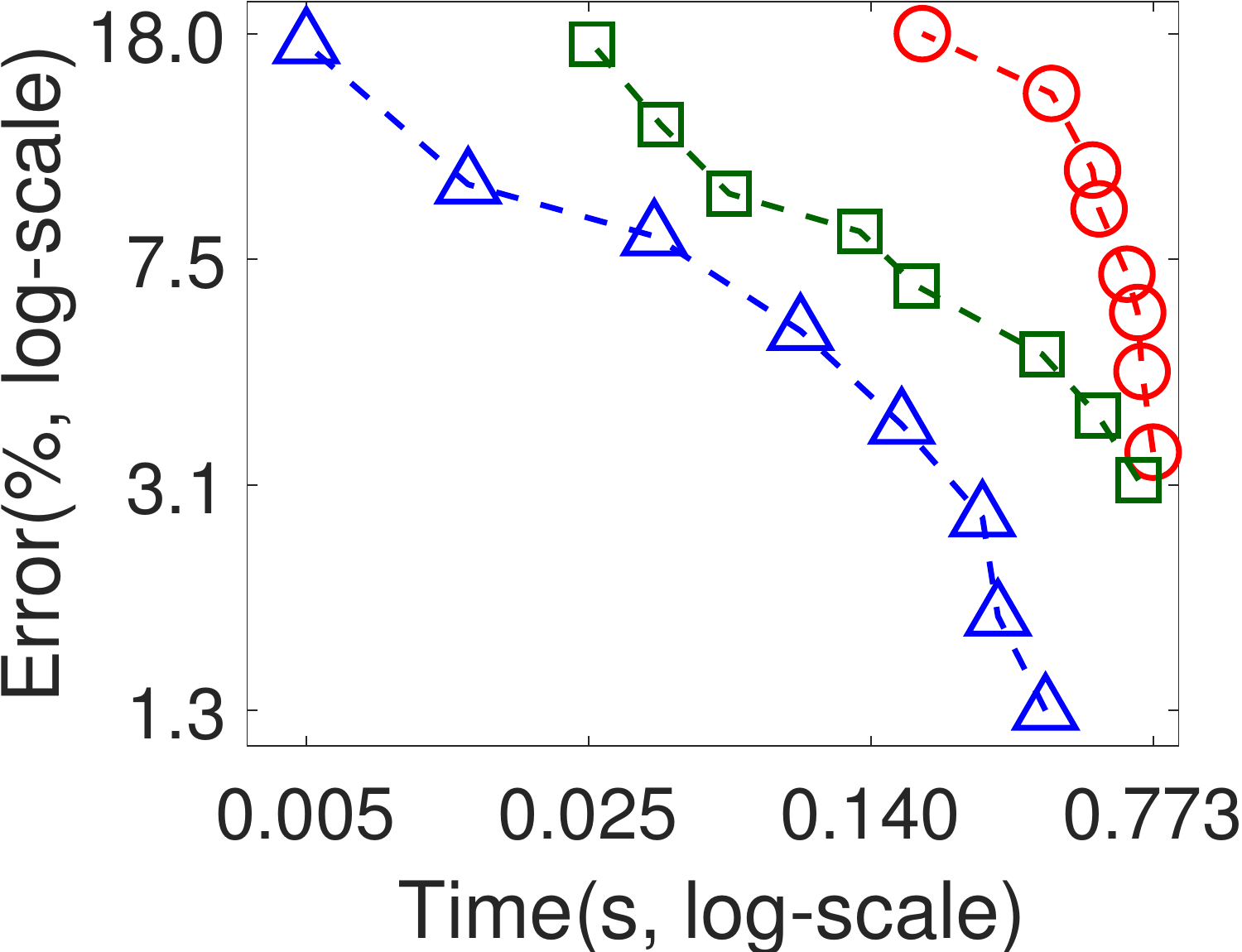}
  }
  \hfill
  \subfigure[Q4 on SU]{
    \label{fig:p:m4:su}
    \includegraphics[height=0.77in]{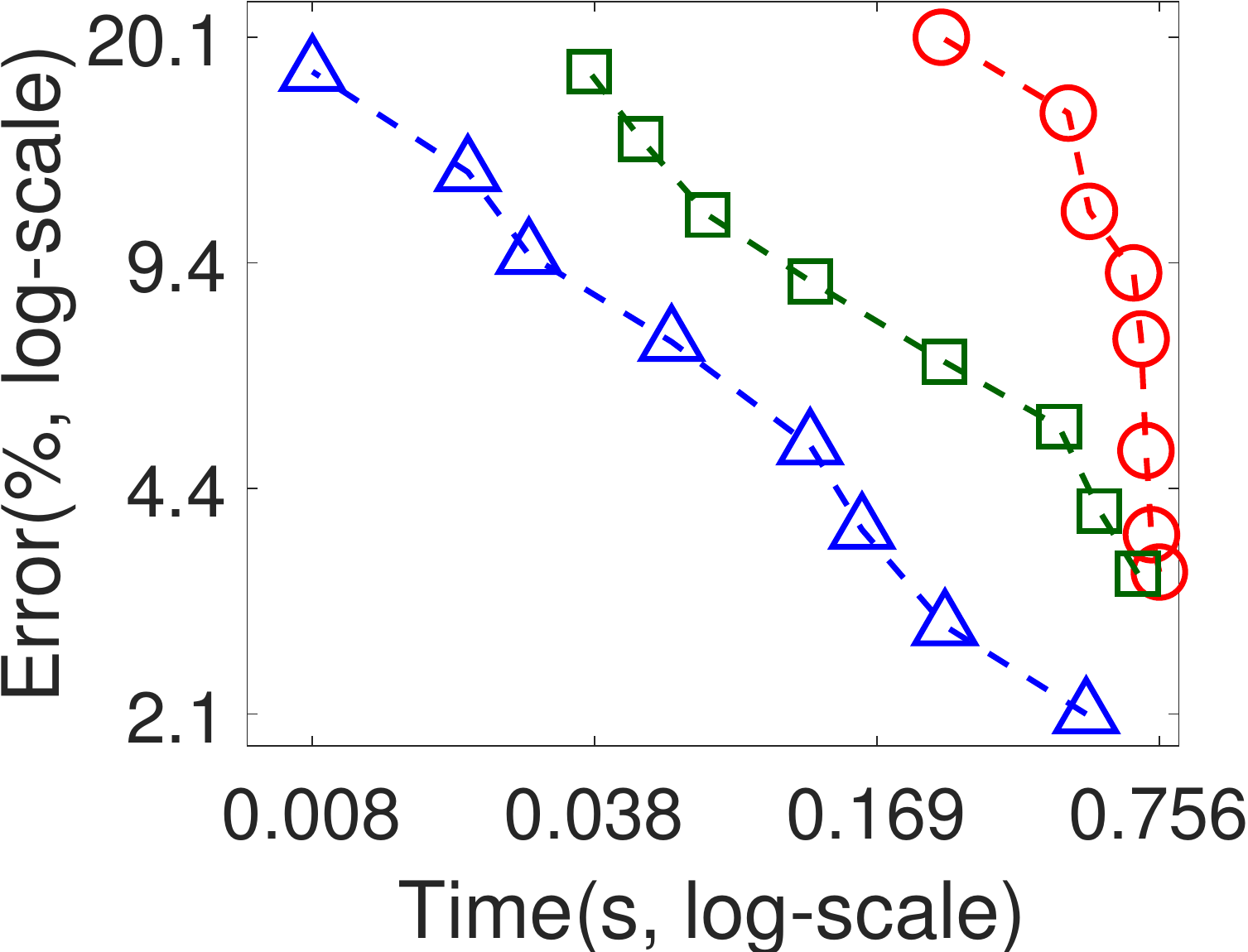}
  }
  \hfill
  \subfigure[Q5 on SU]{
    \label{fig:p:m5:su}
    \includegraphics[height=0.77in]{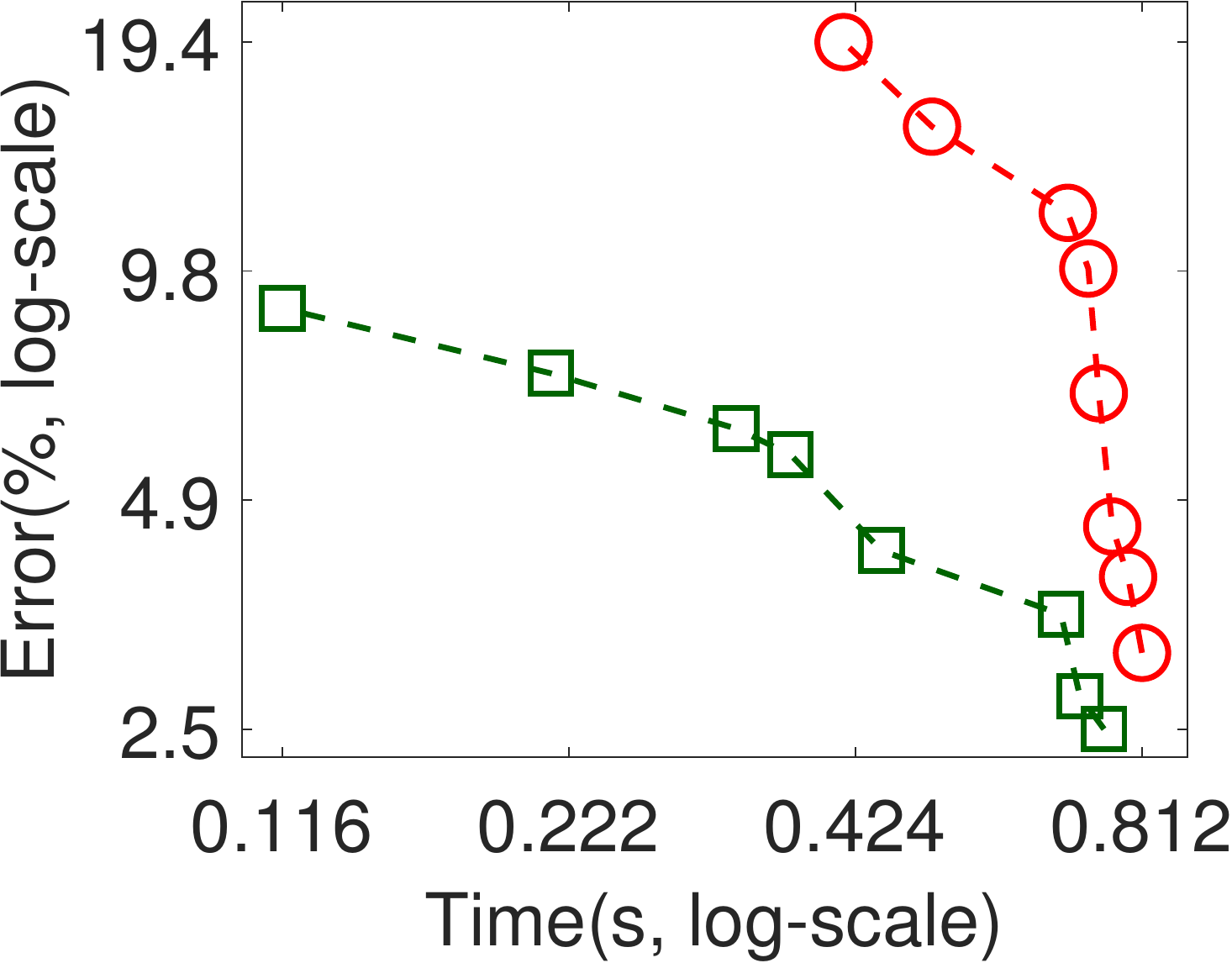}
  }
  \subfigure[Q1 on SO]{
    \label{fig:p:m1:so}
    \includegraphics[height=0.78in]{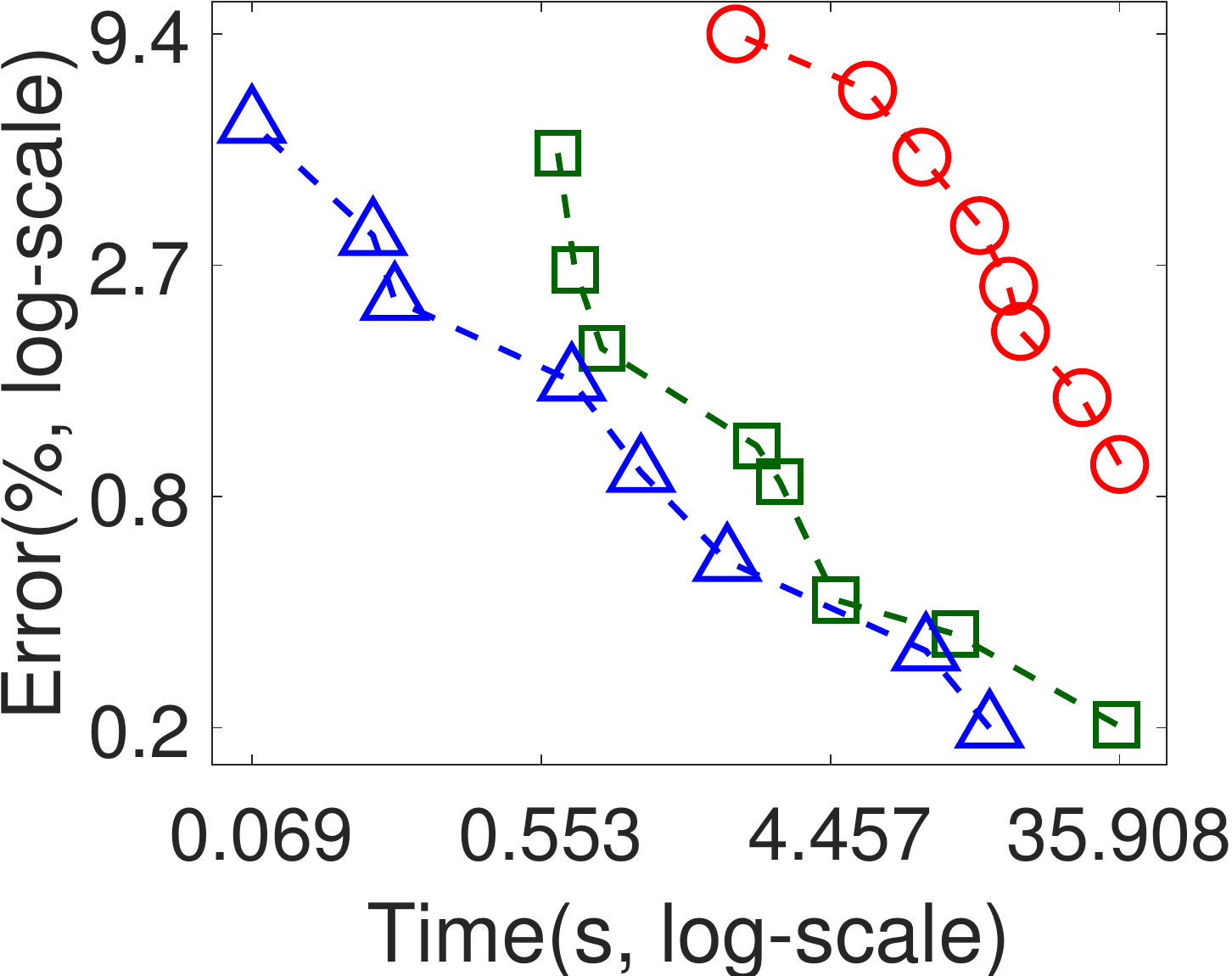}
  }
  \hfill
  \subfigure[Q2 on SO]{
    \label{fig:p:m2:so}
    \includegraphics[height=0.78in]{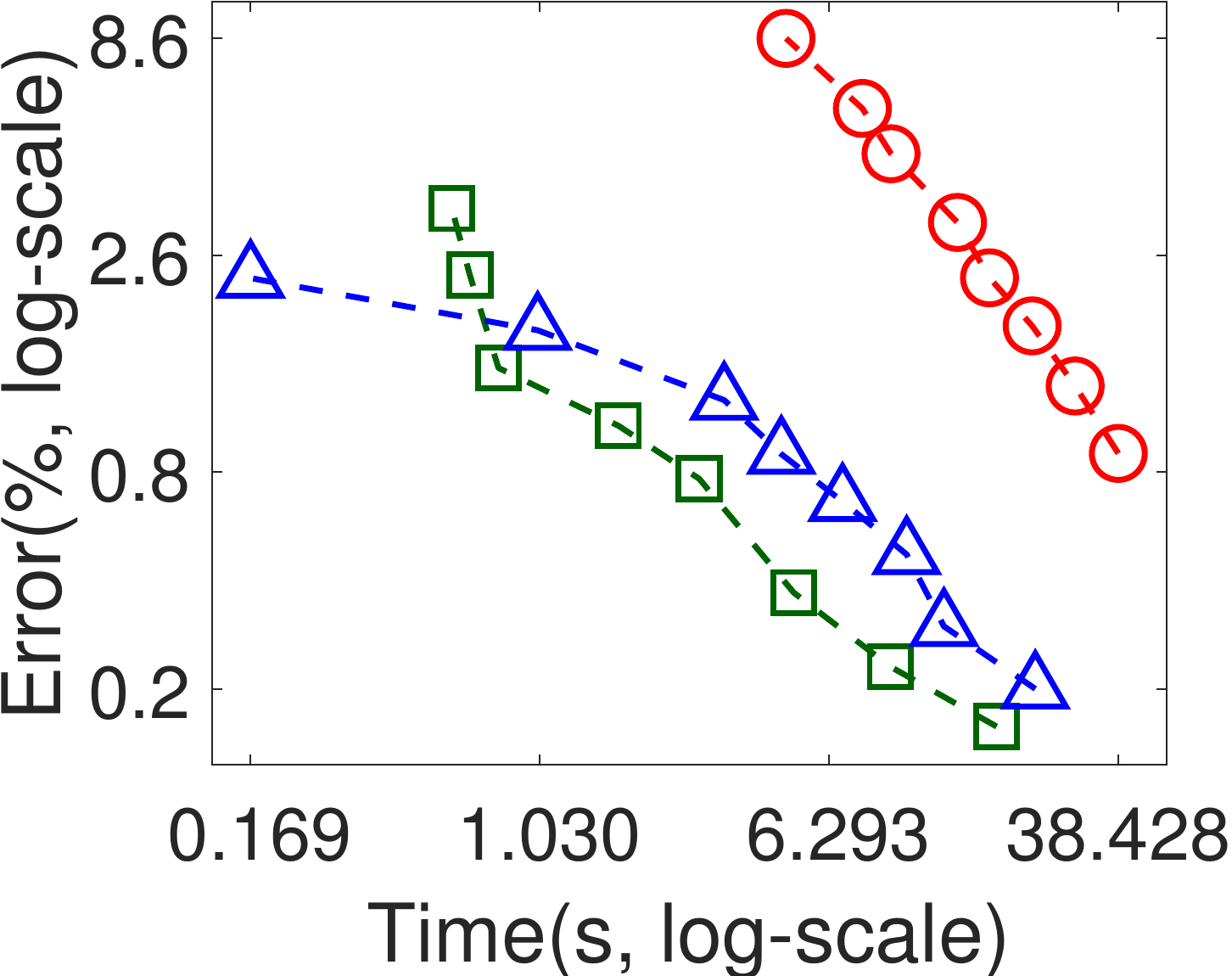}
  }
 \hfill
  \subfigure[Q3 on SO]{
    \label{fig:p:m3:so}
    \includegraphics[height=0.78in]{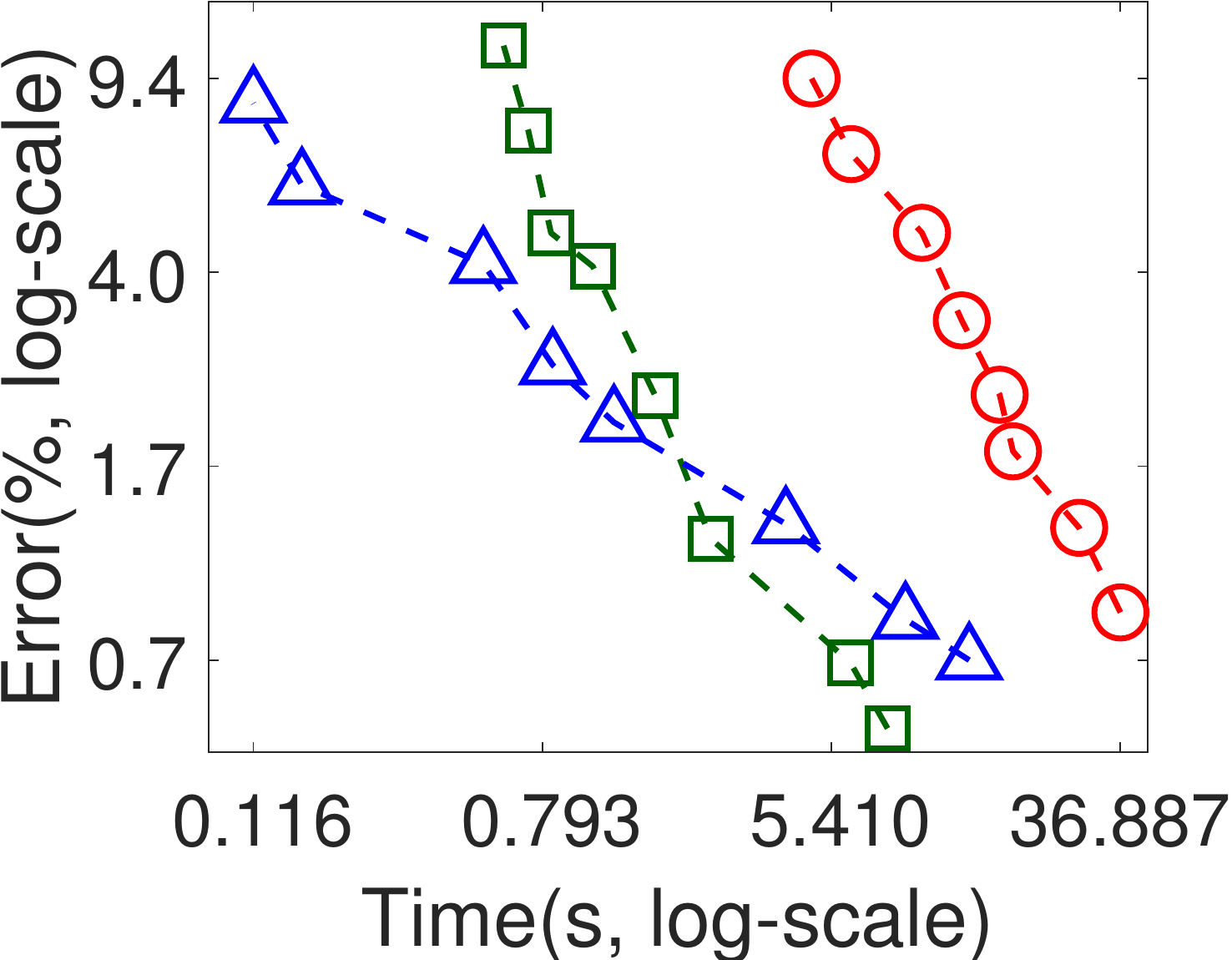}
  }
  \hfill
  \subfigure[Q4 on SO]{
    \label{fig:p:m4:so}
    \includegraphics[height=0.78in]{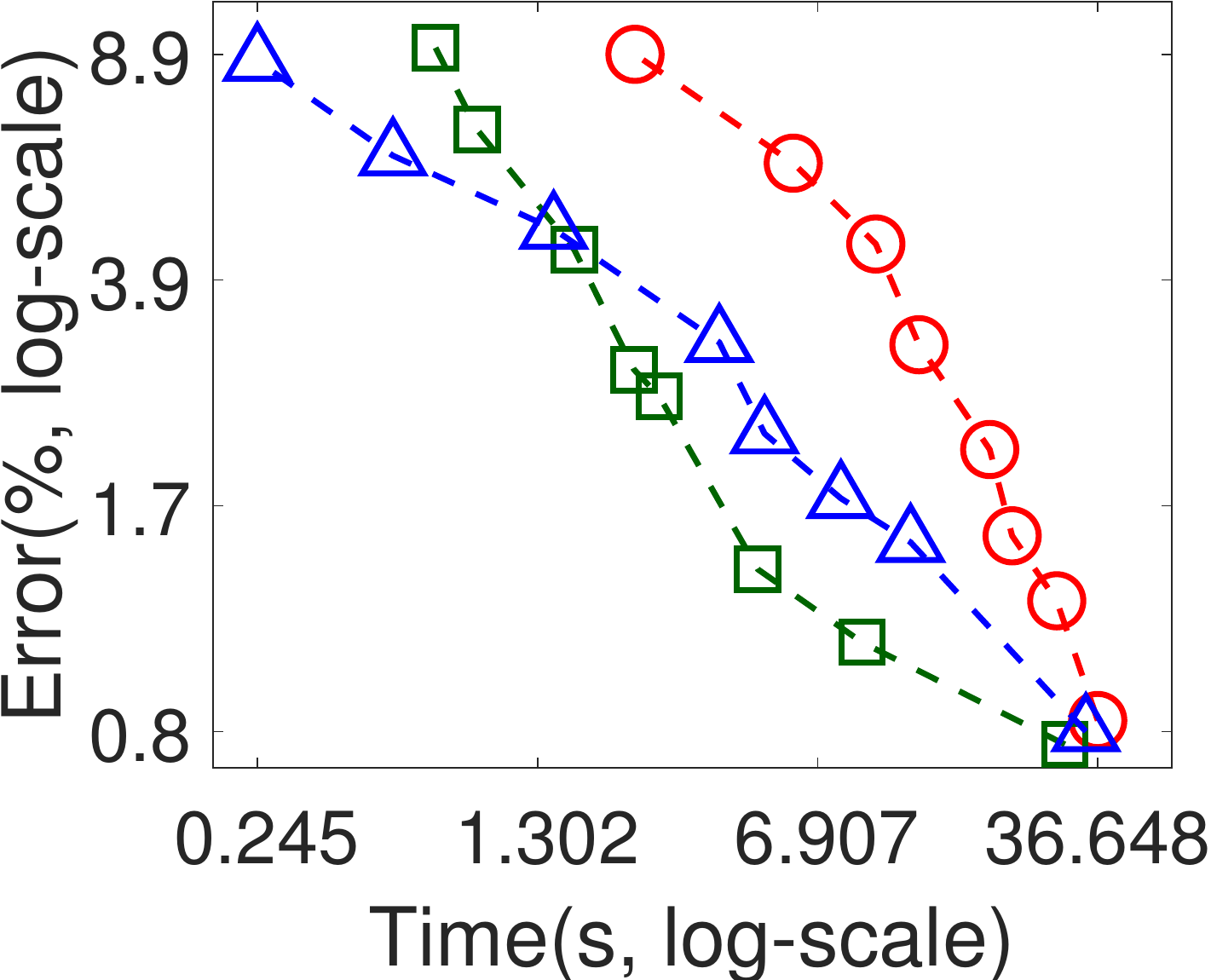}
  }
  \hfill
  \subfigure[Q5 on SO]{
    \label{fig:p:m5:so}
    \includegraphics[height=0.78in]{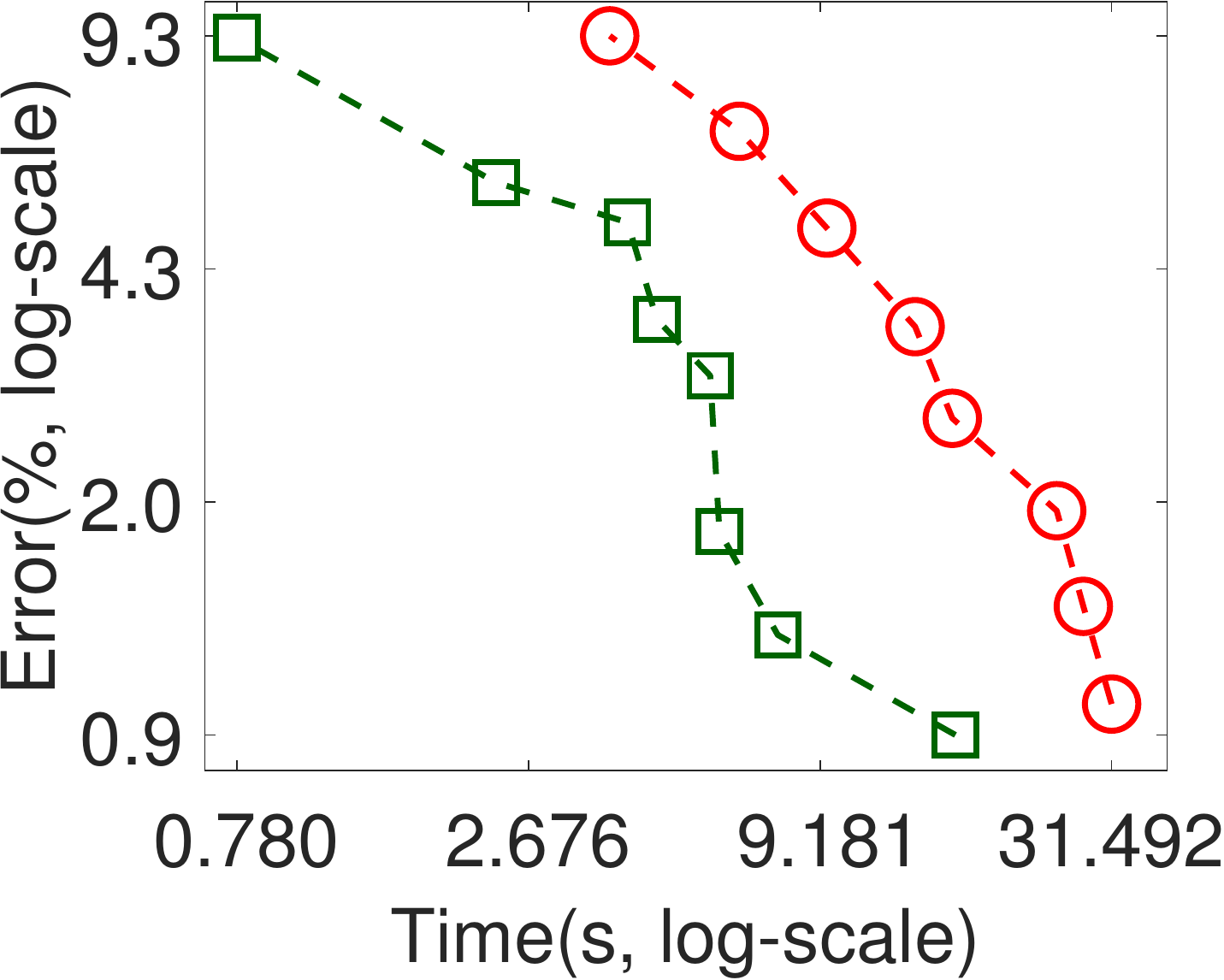}
  }
  \subfigure[Q1 on BC]{
    \label{fig:p:m1:bt}
    \includegraphics[height=0.77in]{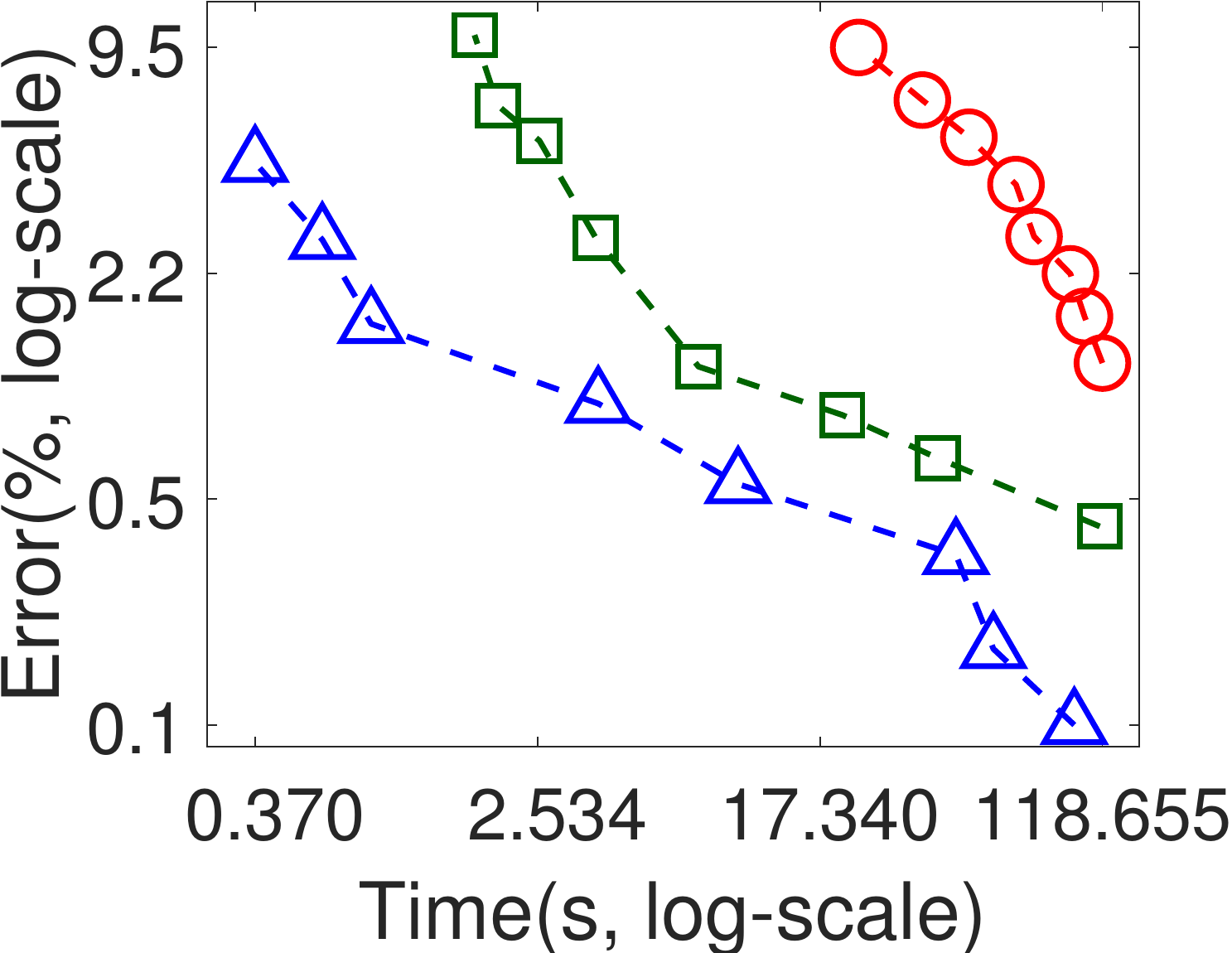}
  }
  \hfill
  \subfigure[Q2 on BC]{
    \label{fig:p:m2:bt}
    \includegraphics[height=0.77in]{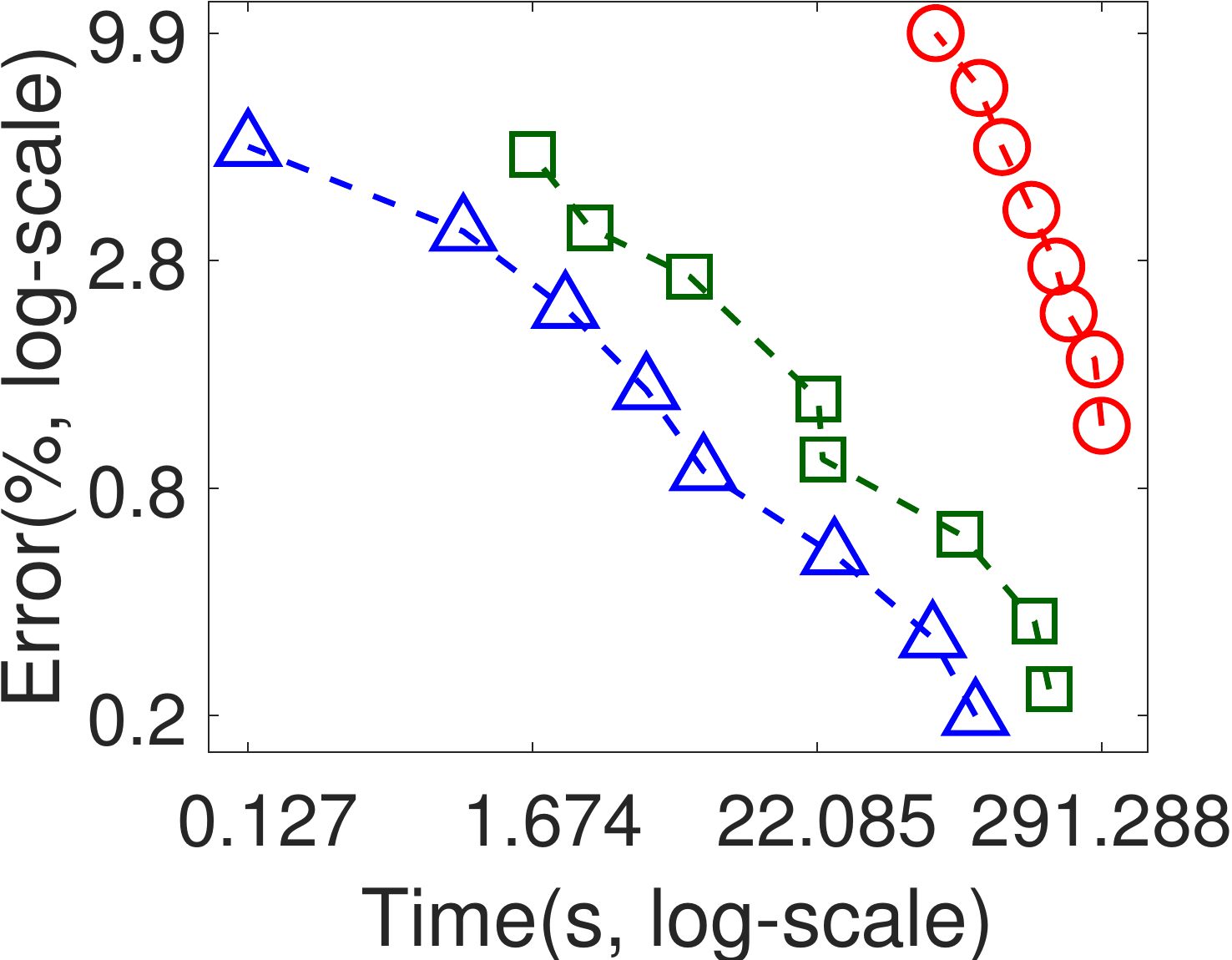}
  }
  \hfill
  \subfigure[Q3 on BC]{
    \label{fig:p:m3:bt}
    \includegraphics[height=0.77in]{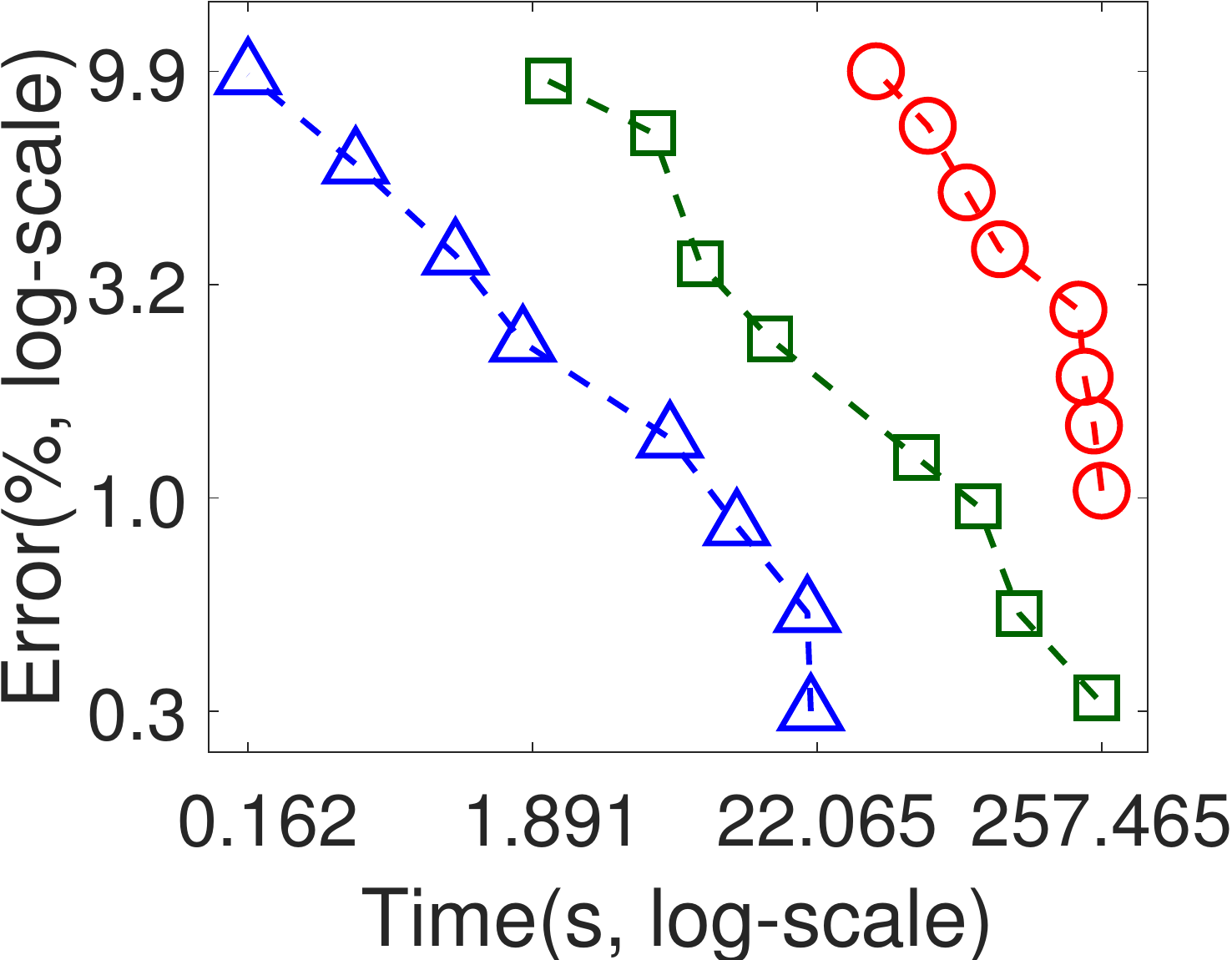}
  }
  \hfill
  \subfigure[Q4 on BC]{
    \label{fig:p:m4:bt}
    \includegraphics[height=0.77in]{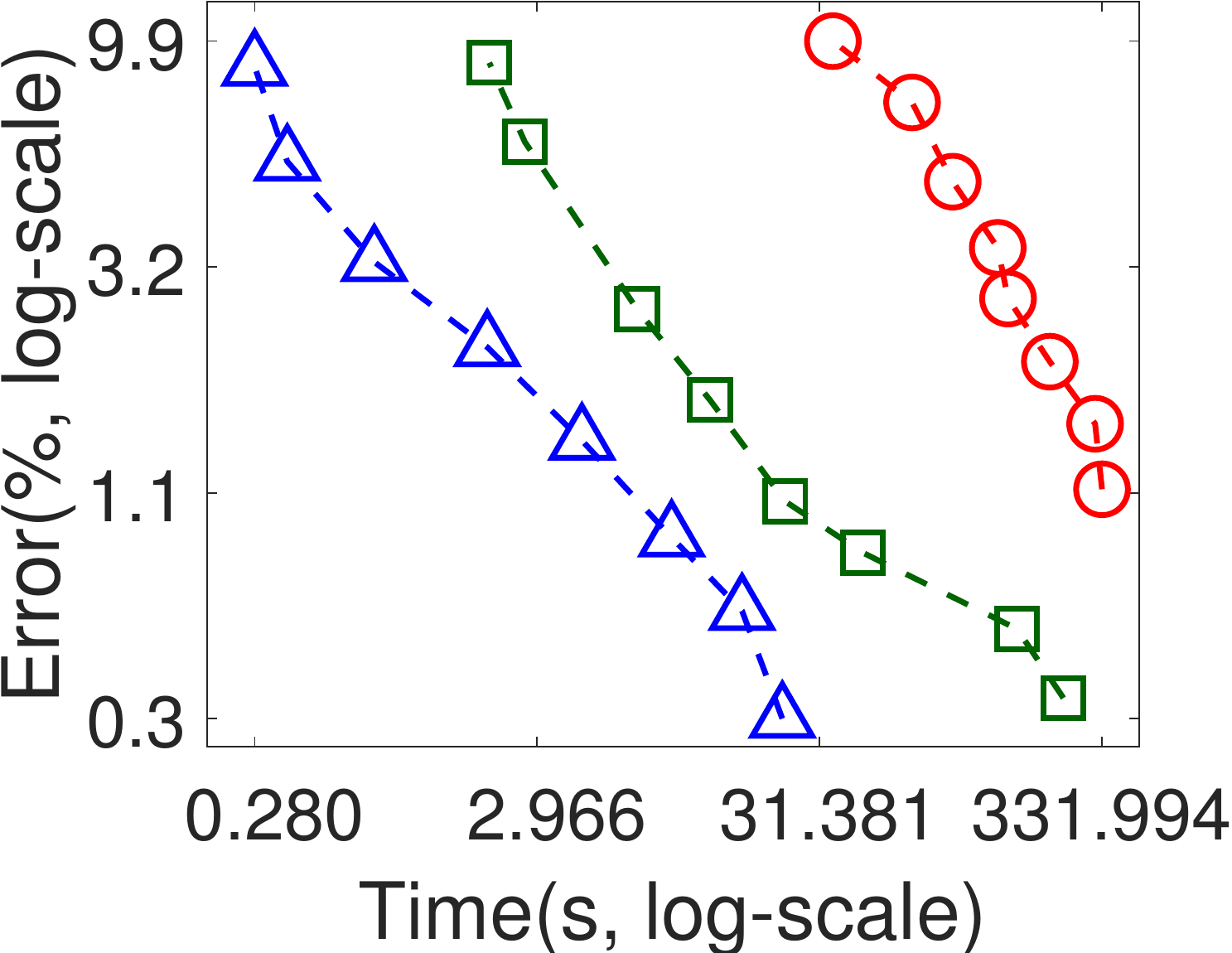}
  }
  \hfill
  \subfigure[Q5 on BC]{
    \label{fig:p:m5:bt}
    \includegraphics[height=0.77in]{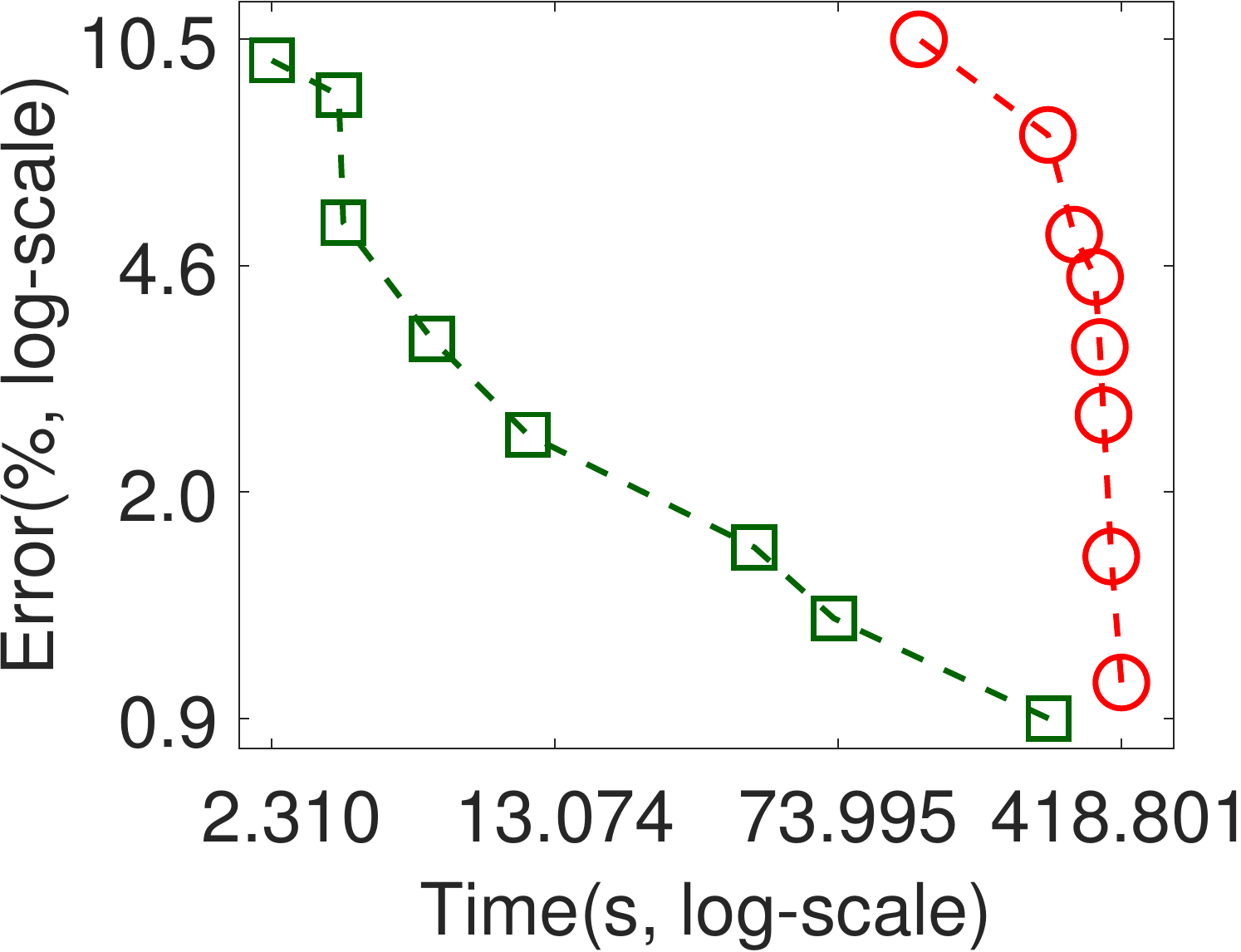}
  }
  \subfigure[Q1 on RC]{
    \label{fig:p:m1:rc}
    \includegraphics[height=0.76in]{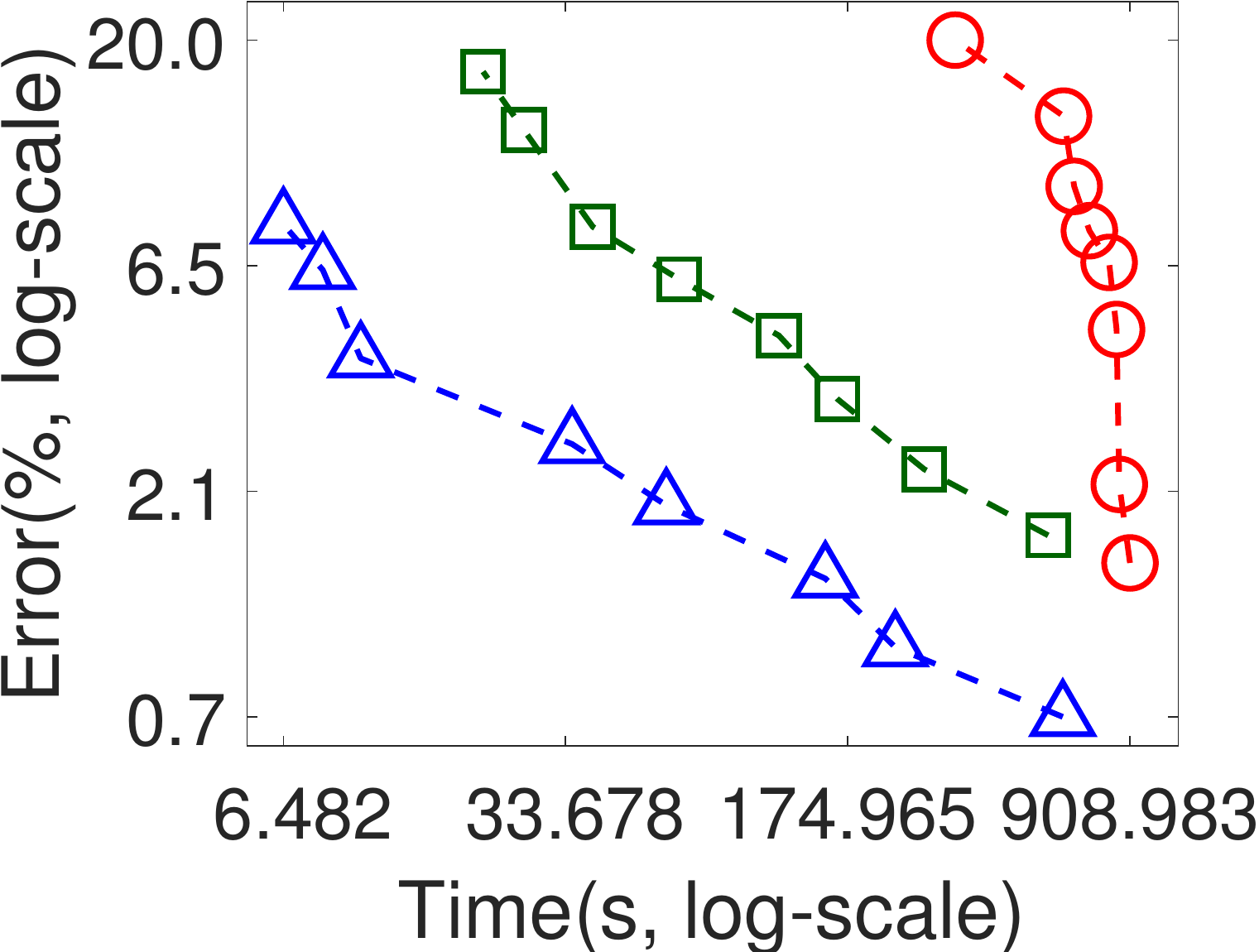}
  }
  \hfill
  \subfigure[Q2 on RC]{
    \label{fig:p:m2:rc}
    \includegraphics[height=0.76in]{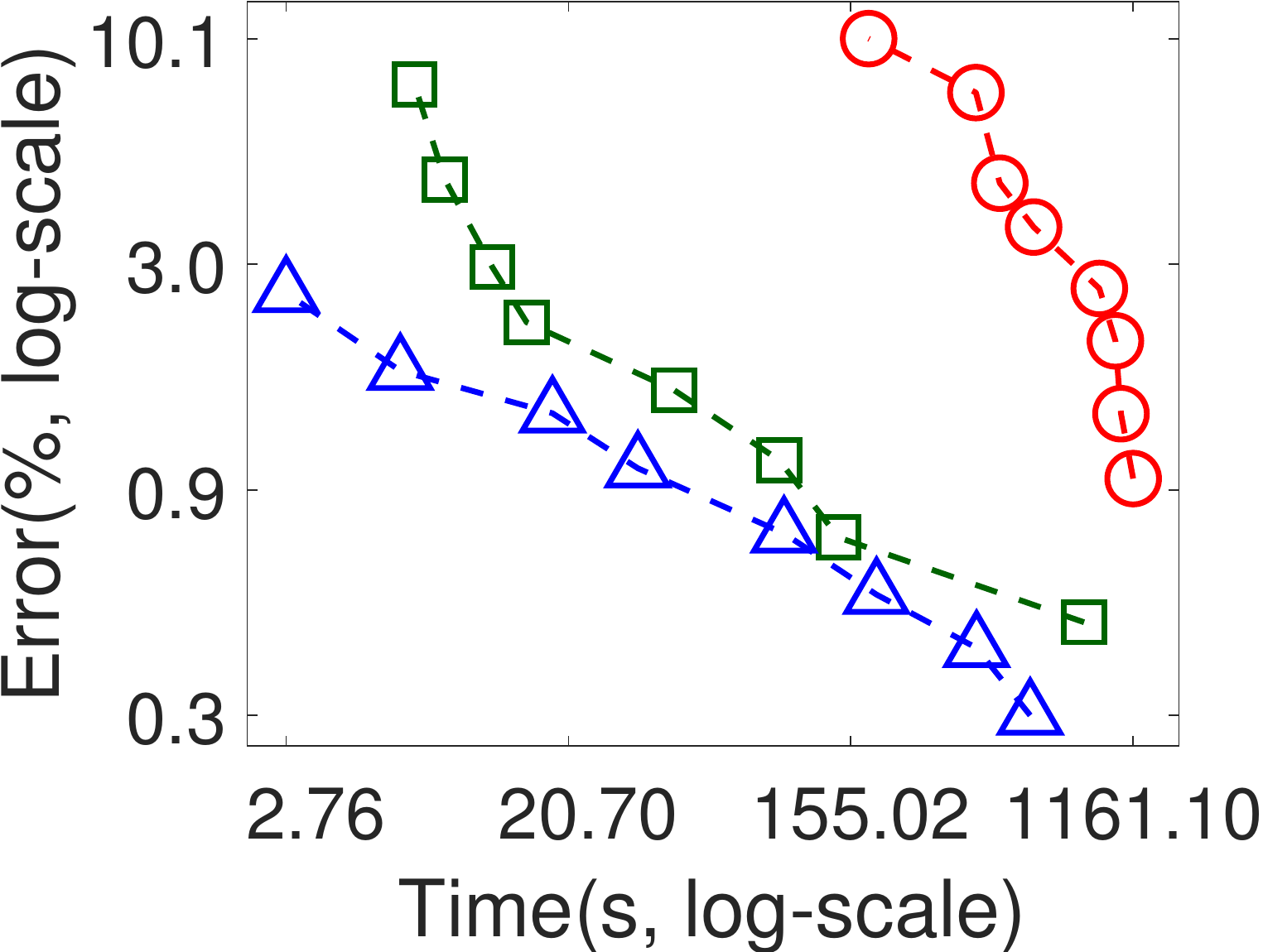}
  }
  \hfill
  \subfigure[Q3 on RC]{
    \label{fig:p:m3:rc}
    \includegraphics[height=0.76in]{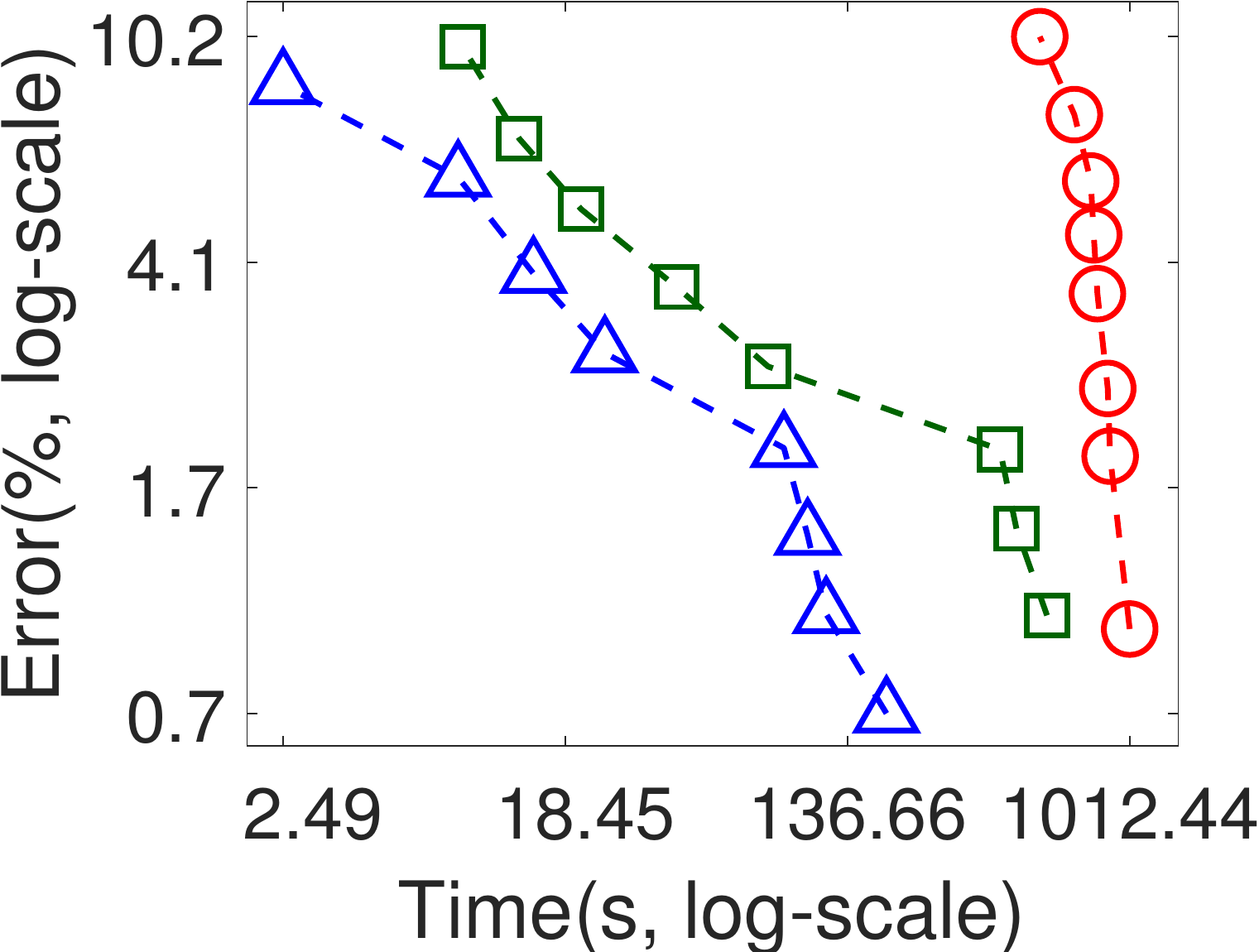}
  }
  \hfill
  \subfigure[Q4 on RC]{
    \label{fig:p:m4:rc}
    \includegraphics[height=0.76in]{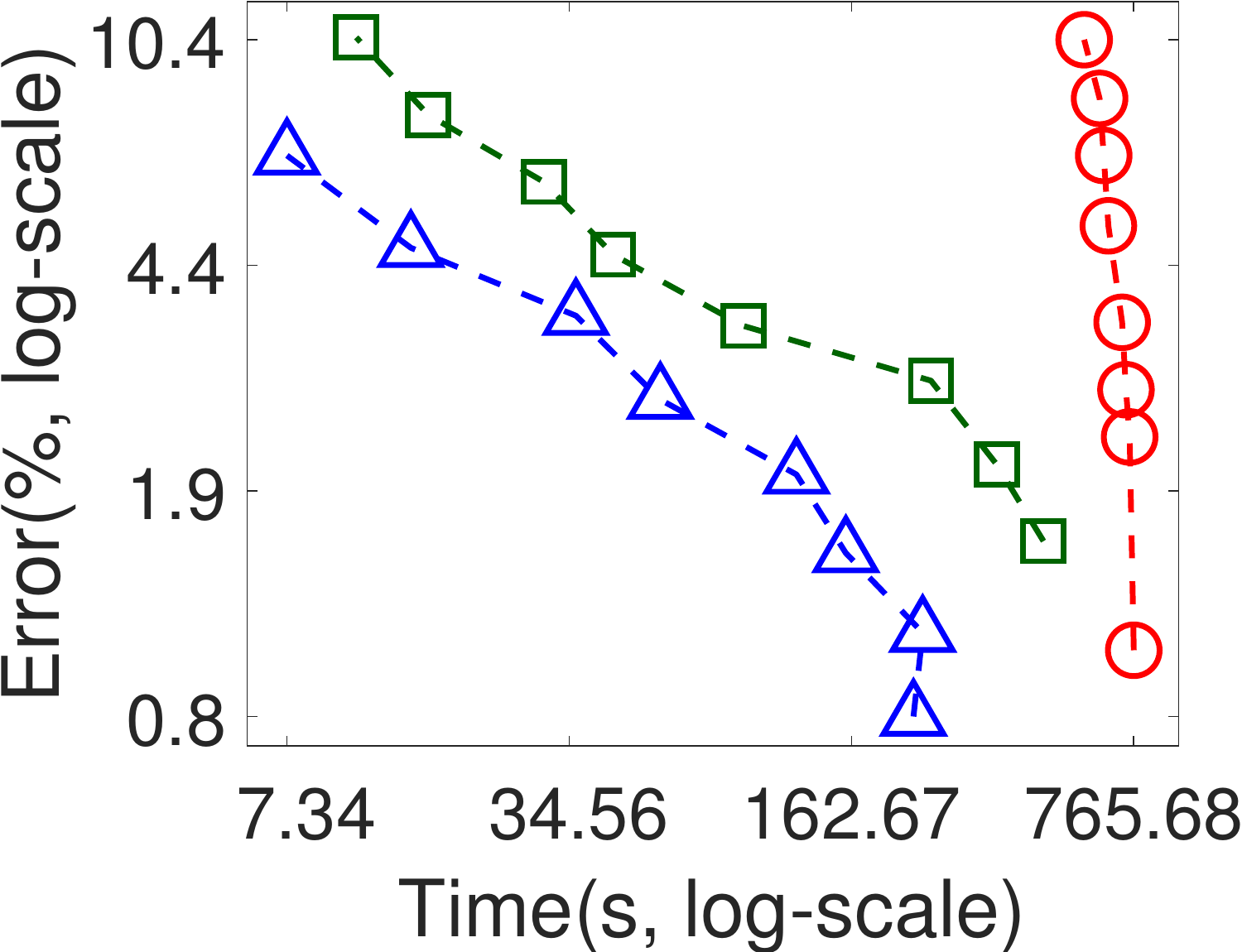}
  }
  \hfill
  \subfigure[Q5 on RC]{
    \label{fig:p:m5:rc}
    \includegraphics[height=0.76in]{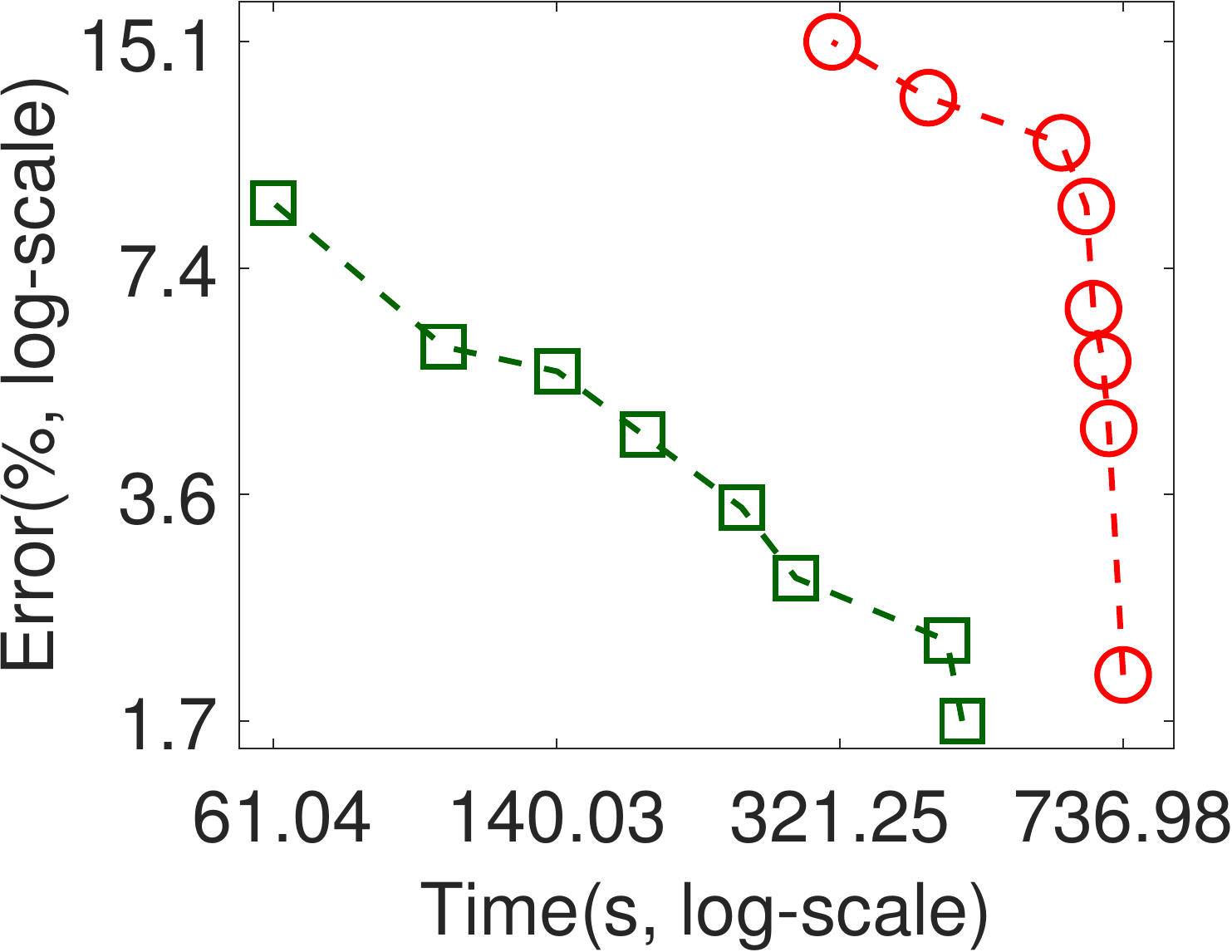}
  }
  \caption{Relative error ($\%$) vs.~running time (in seconds) with varying sampling probability $p$.}
  \Description{P}
  \label{fig:p}
\end{figure}

\begin{figure}[t]
  \centering
  \subfigure[Q3 on BC with varying time span $\delta$]{
    \label{subfig:delta}
    \includegraphics[height=0.9in]{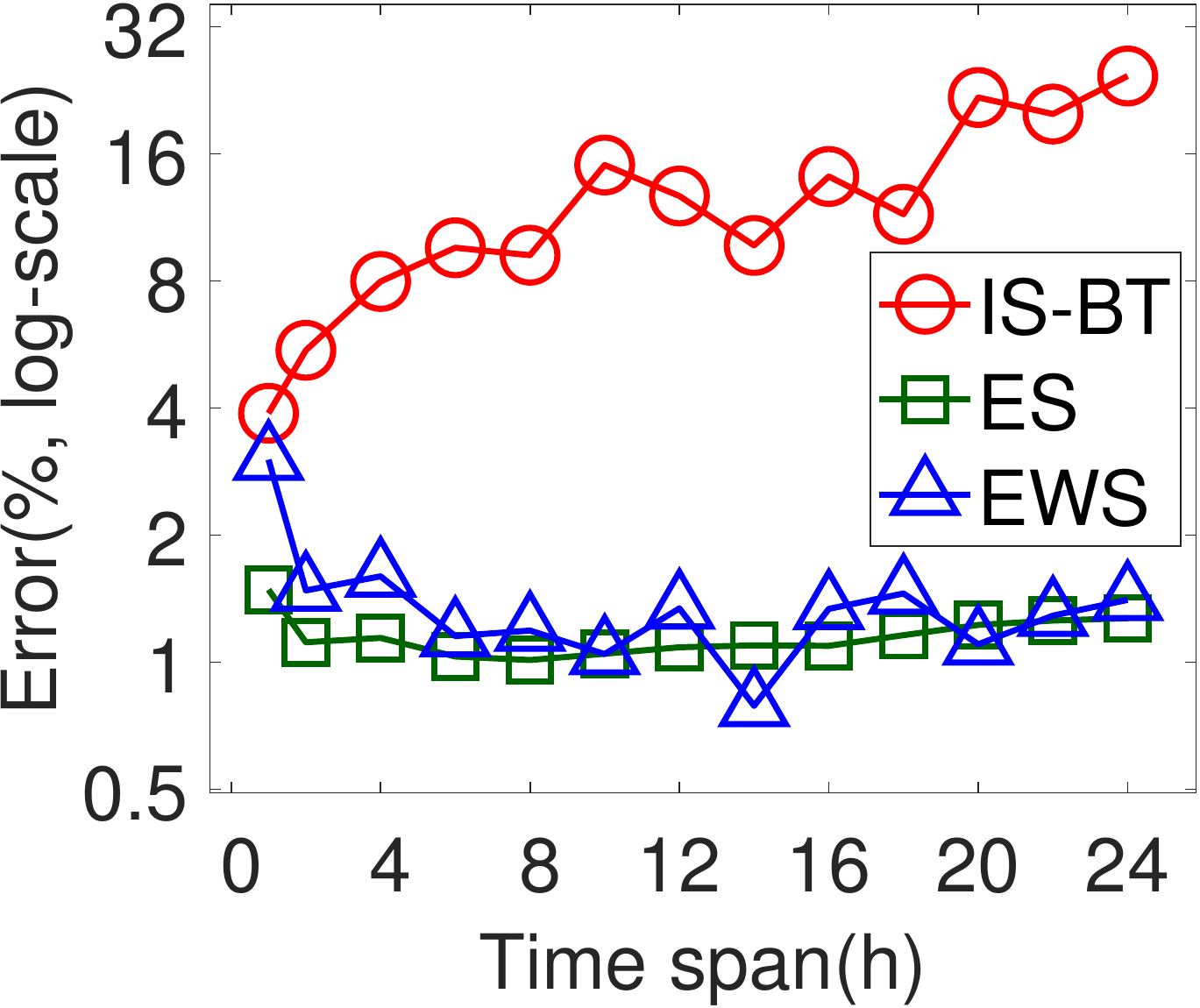}
    \hspace{0.5em}
    \includegraphics[height=0.9in]{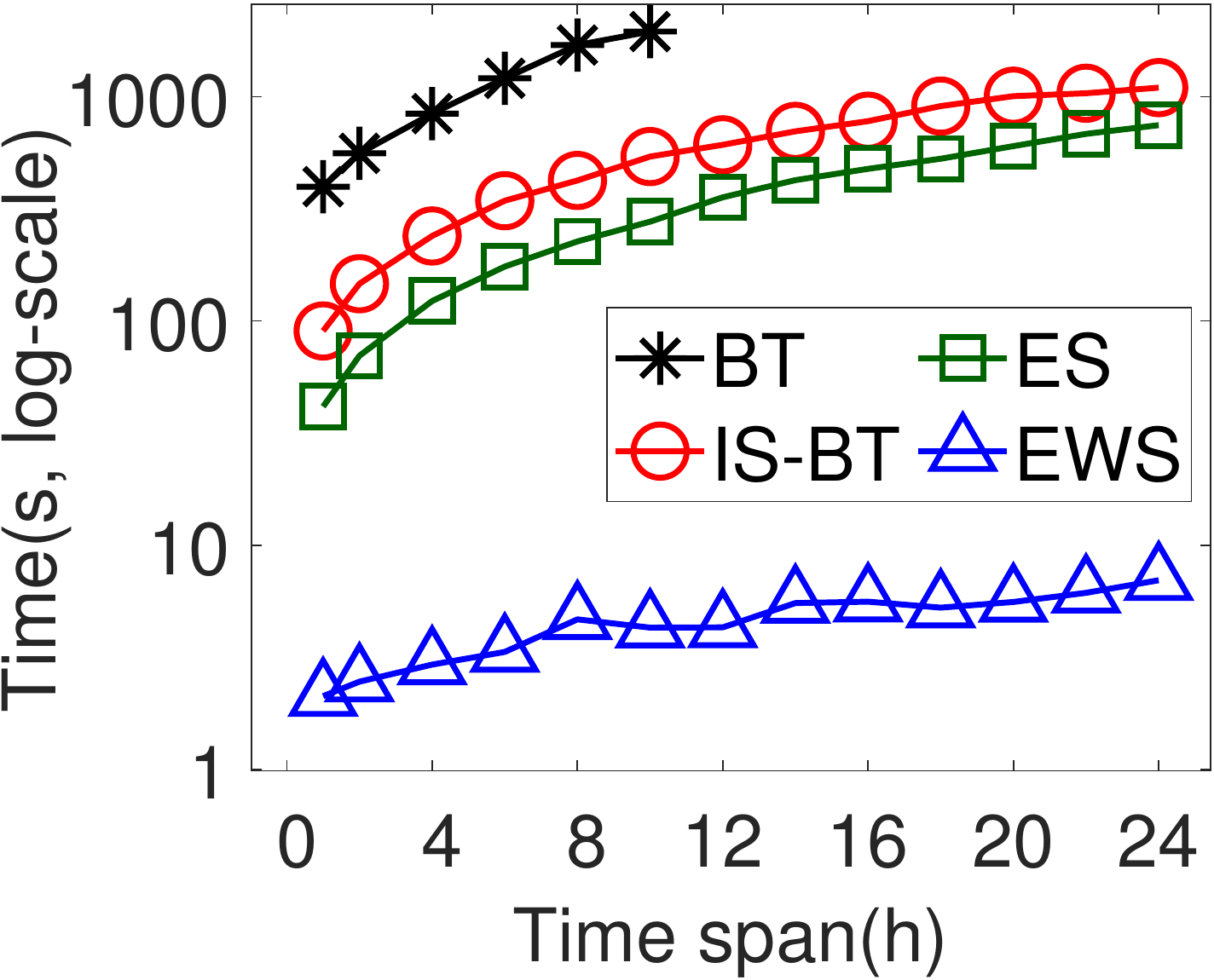}
  }
  \hspace{1em}
  \subfigure[Q2 on RC with varying number of edges $m$]{
    \label{subfig:size}
    \includegraphics[height=0.9in]{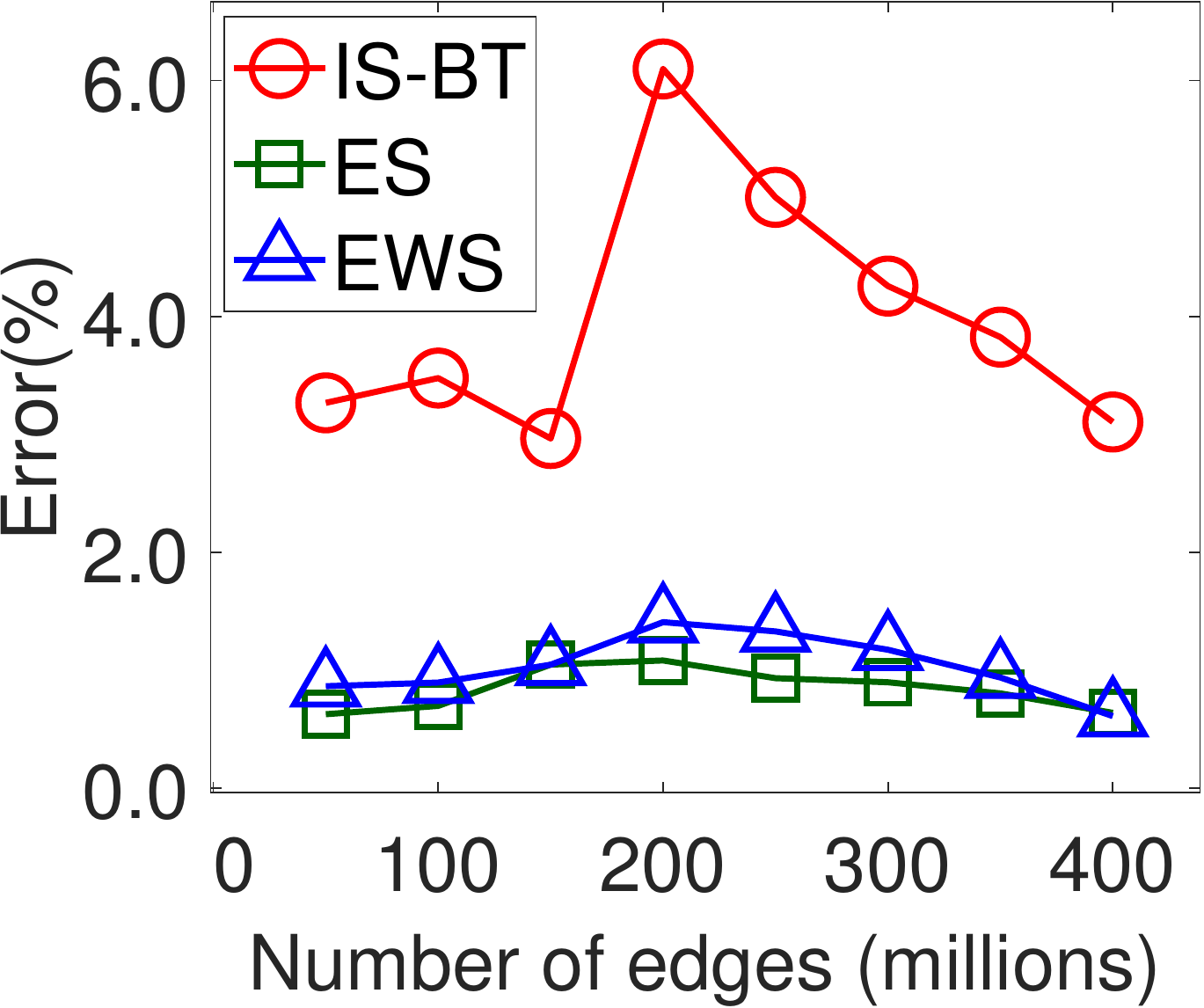}
    \hspace{0.5em}
    \includegraphics[height=0.9in]{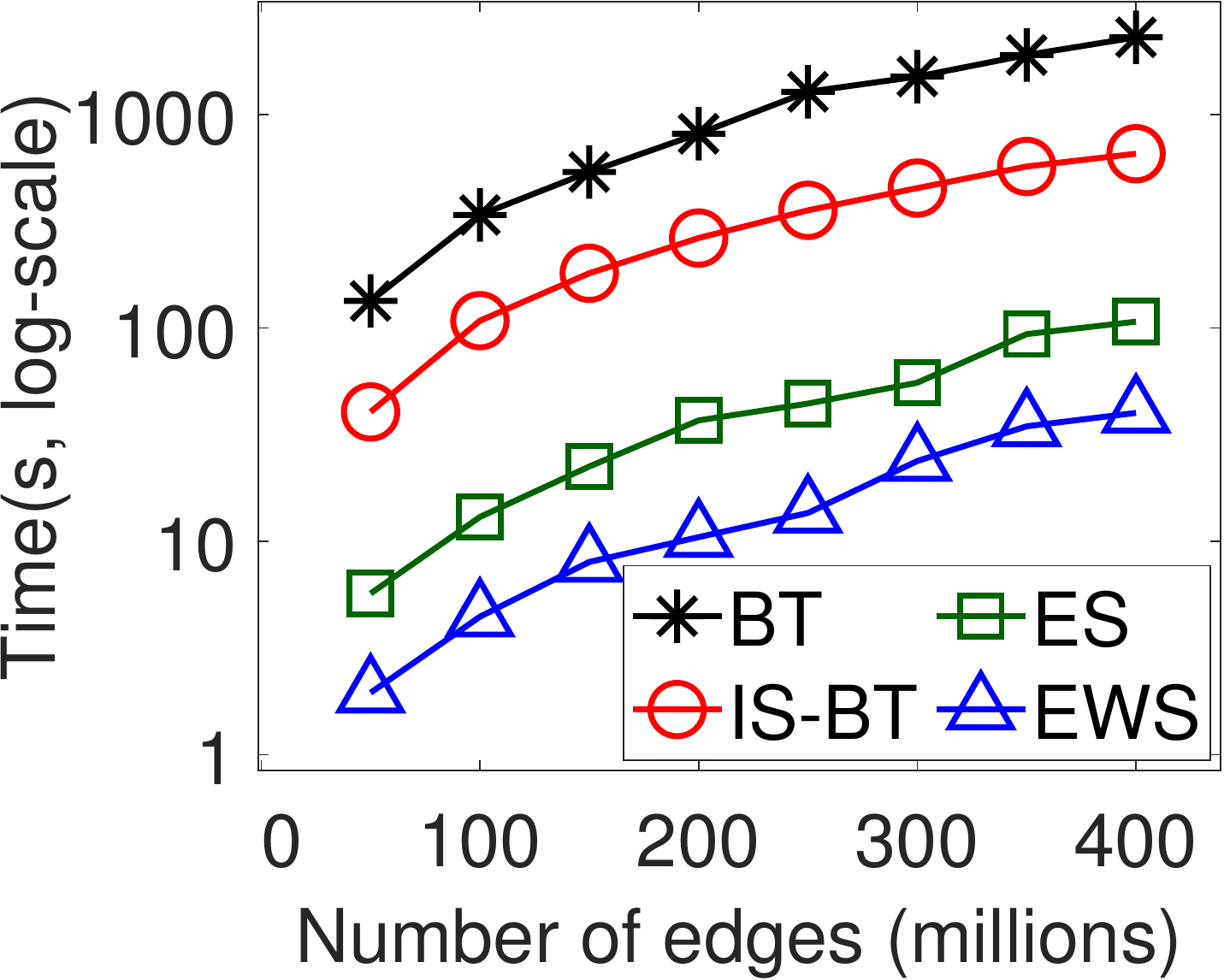}
  }
  \caption{Scalability of different offline algorithms with varying time span $\delta$ and number of temporal edges $m$.}
  \Description{scalability}
  \label{fig:scalability}
\end{figure}

\begin{figure}[t]
  \centering
  \includegraphics[height=0.15in]{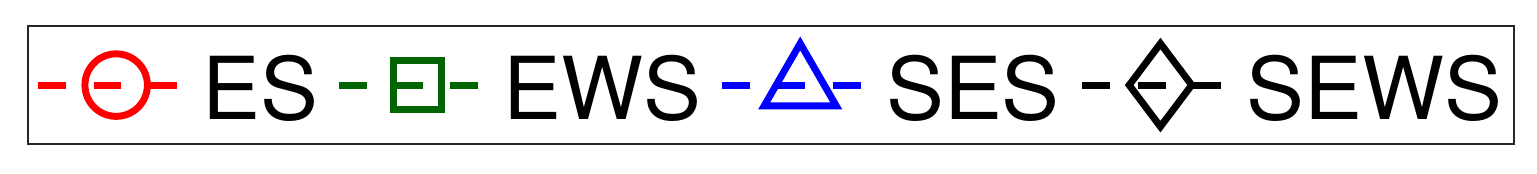}
  \\
  \subfigure[Q1 on AU]{
    \label{fig:r:m1:au}
    \includegraphics[height=0.8in]{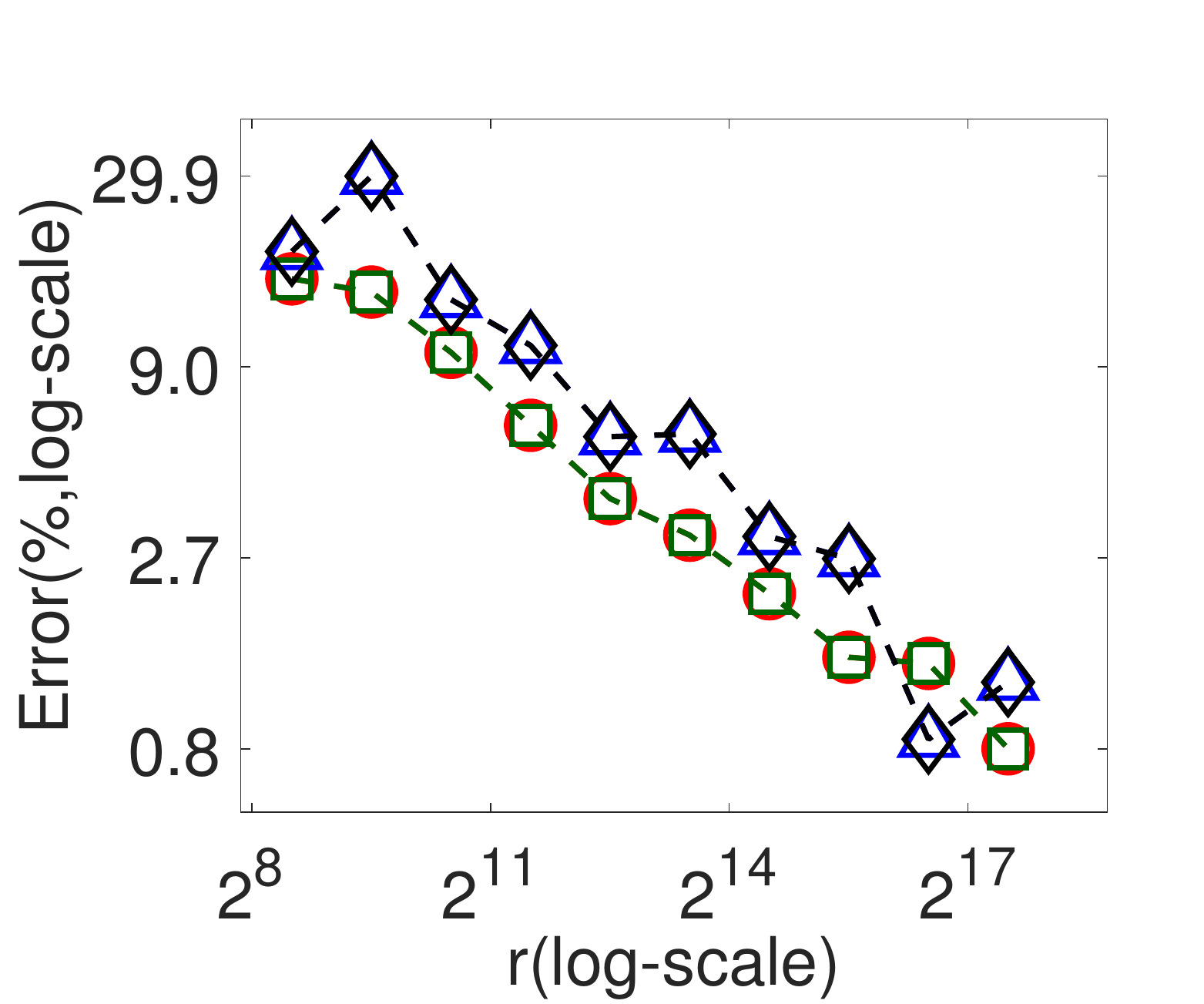}
  }
  \hfill
  \subfigure[Q2 on AU]{
    \label{fig:r:m2:au}
    \includegraphics[height=0.8in]{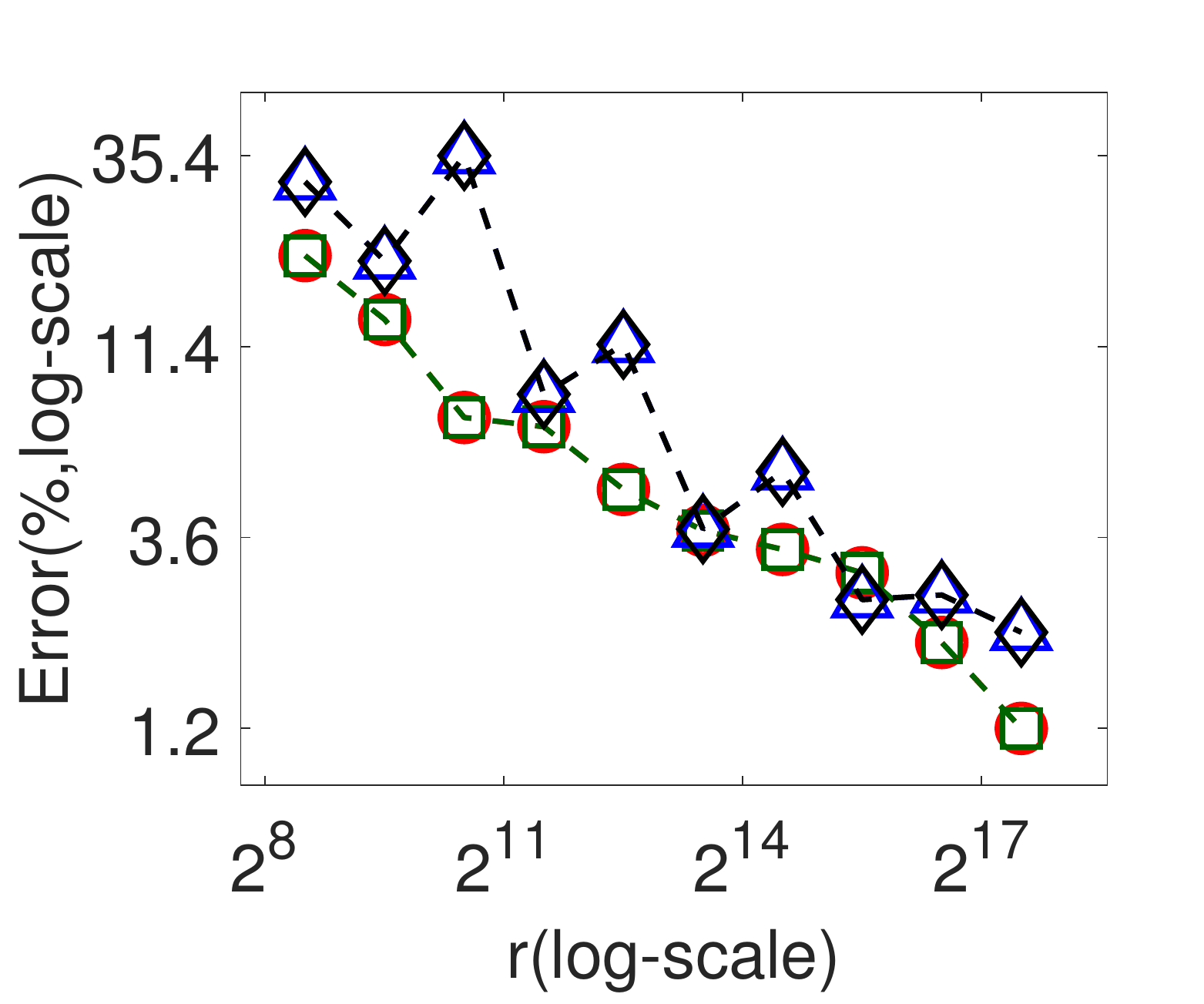}
  }
  \hfill
  \subfigure[Q3 on AU]{
    \label{fig:r:m3:au}
    \includegraphics[height=0.8in]{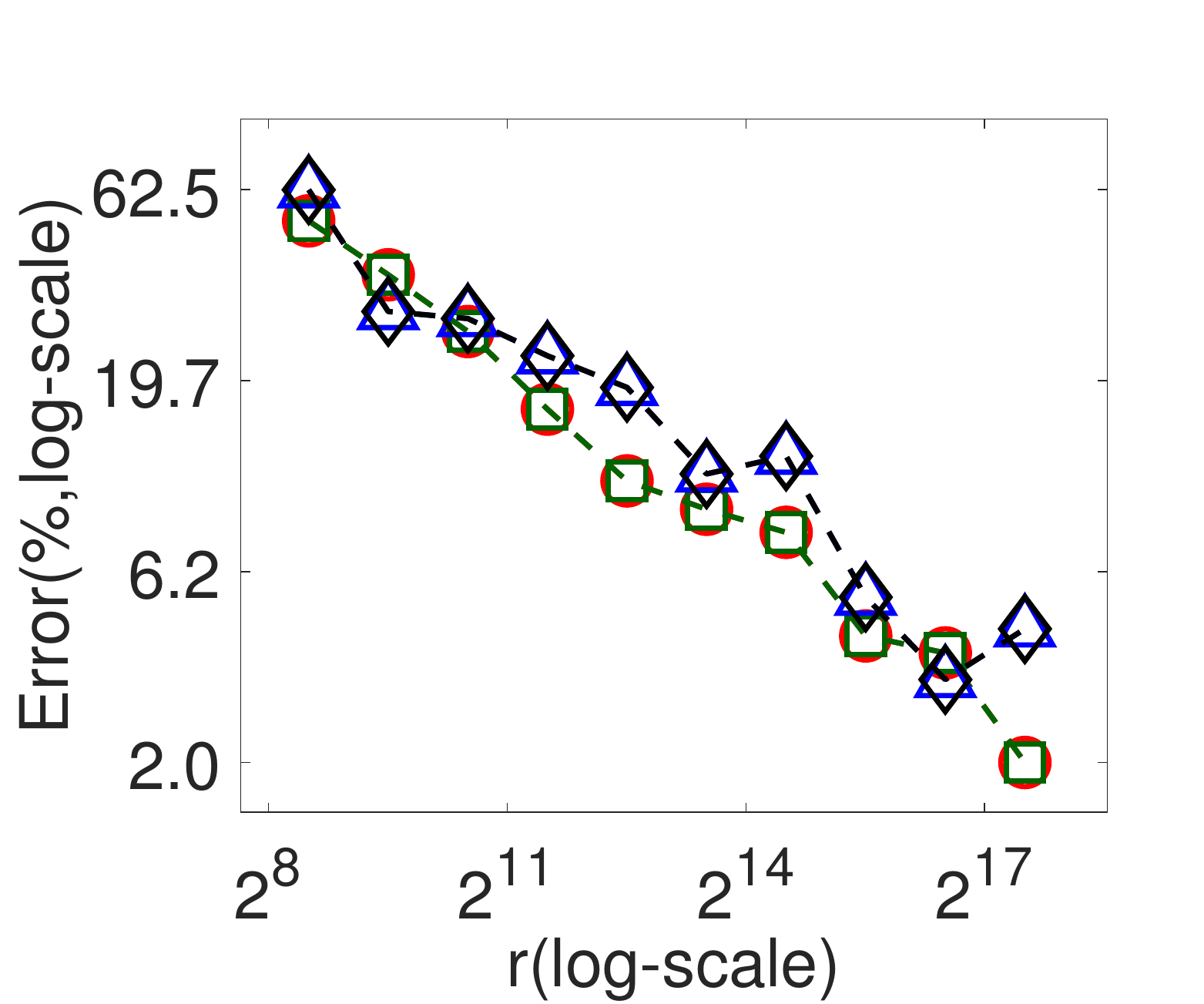}
  }
  \hfill
  \subfigure[Q4 on AU]{
    \label{fig:r:m4:au}
    \includegraphics[height=0.8in]{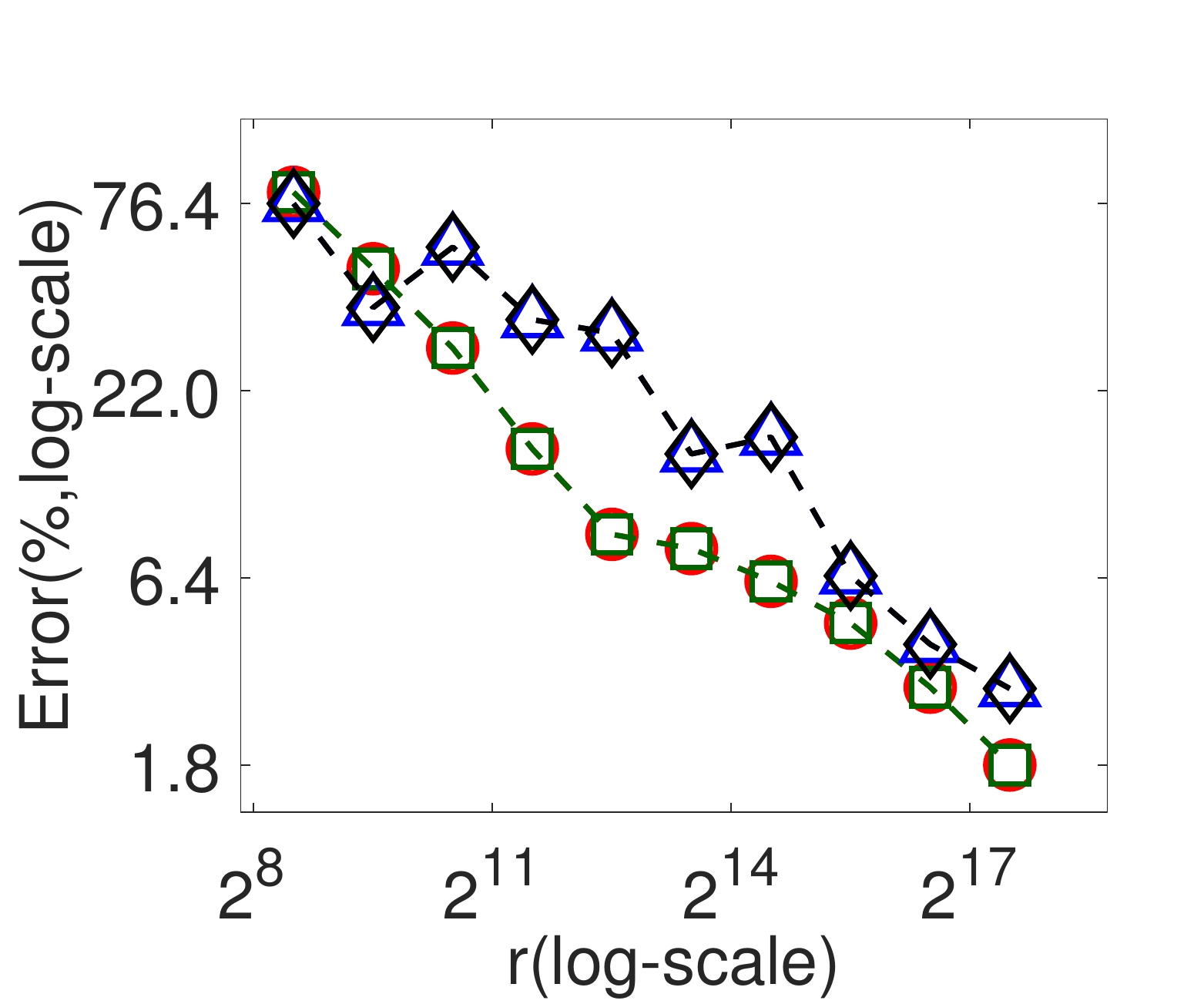}
  }
  \hfill
  \subfigure[Q5 on AU]{
    \label{fig:r:m5:au}
    \includegraphics[height=0.8in]{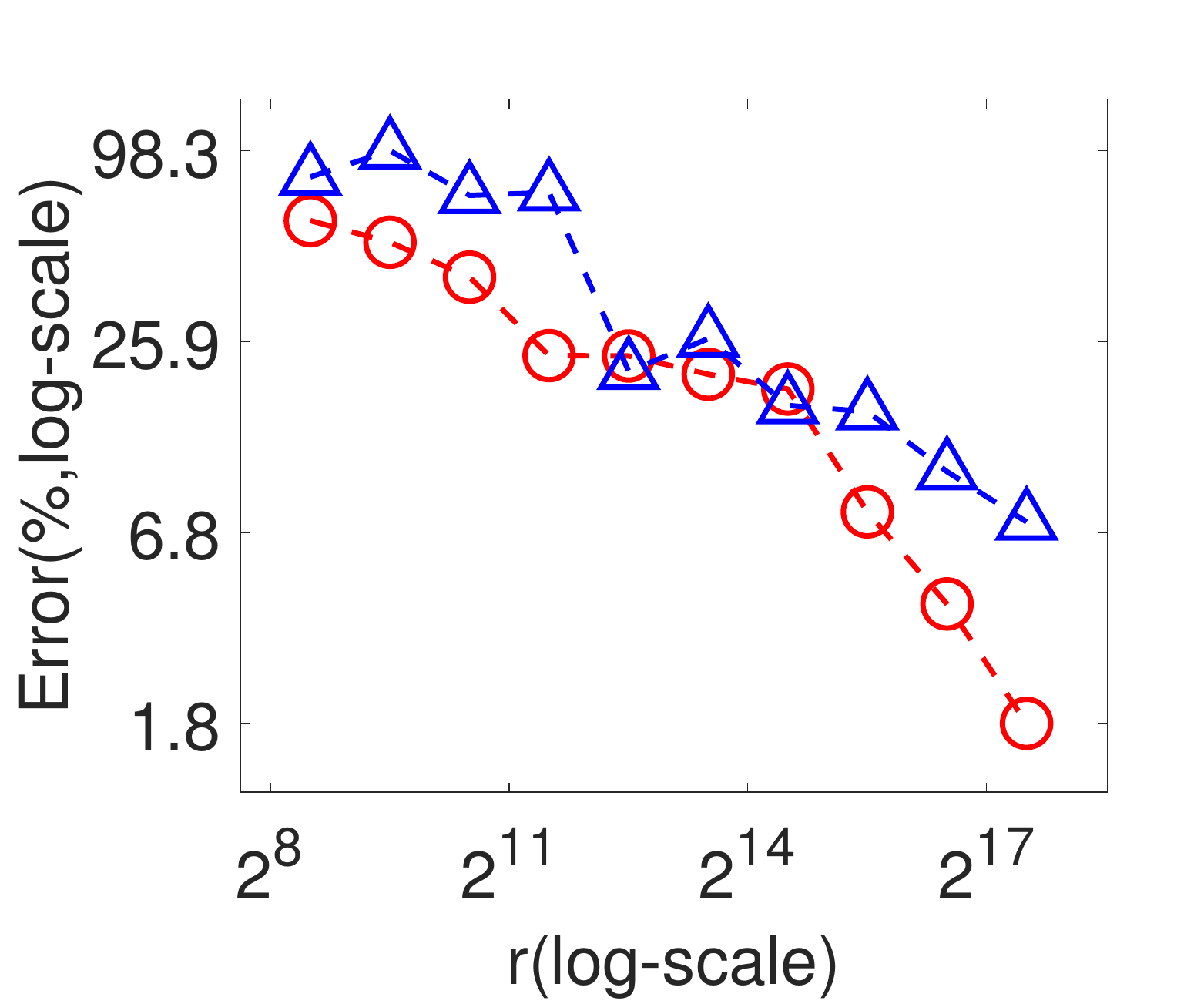}
  }
  \subfigure[Q1 on SU]{
    \label{fig:r:m1:su}
    \includegraphics[height=0.8in]{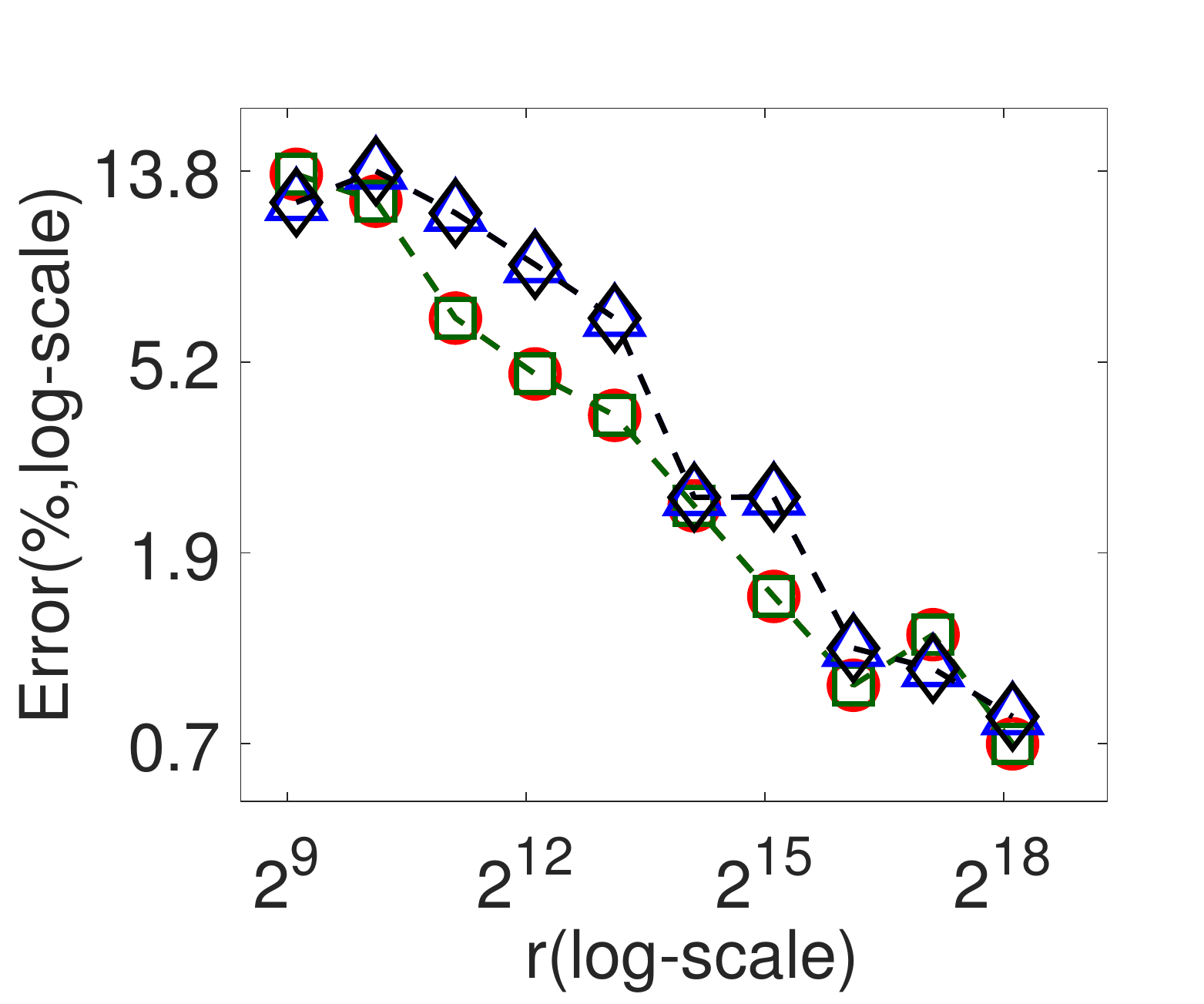}
  }
  \hfill
  \subfigure[Q2 on SU]{
    \label{fig:r:m2:su}
    \includegraphics[height=0.8in]{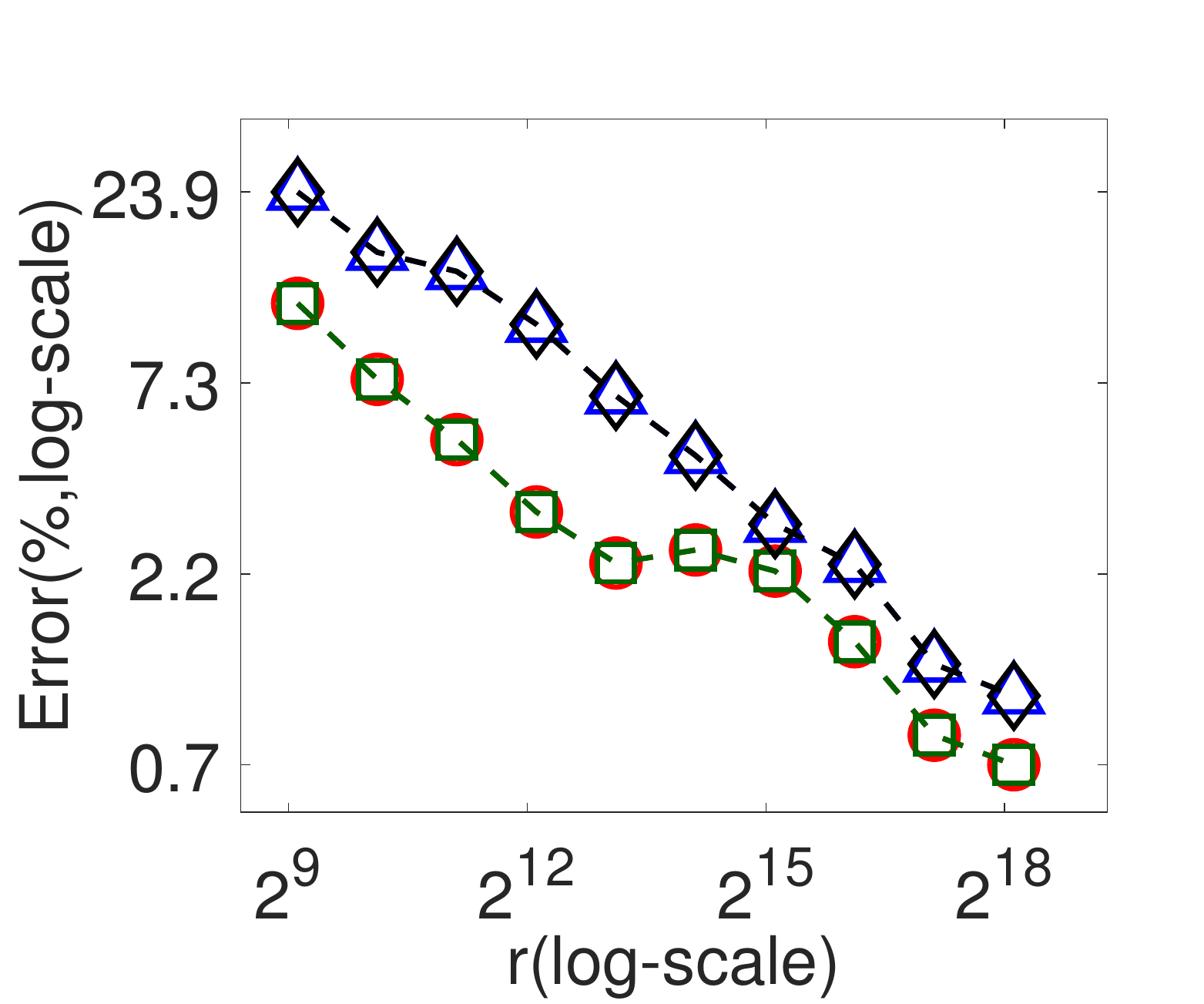}
  }
  \hfill
  \subfigure[Q3 on SU]{
    \label{fig:r:m3:su}
    \includegraphics[height=0.8in]{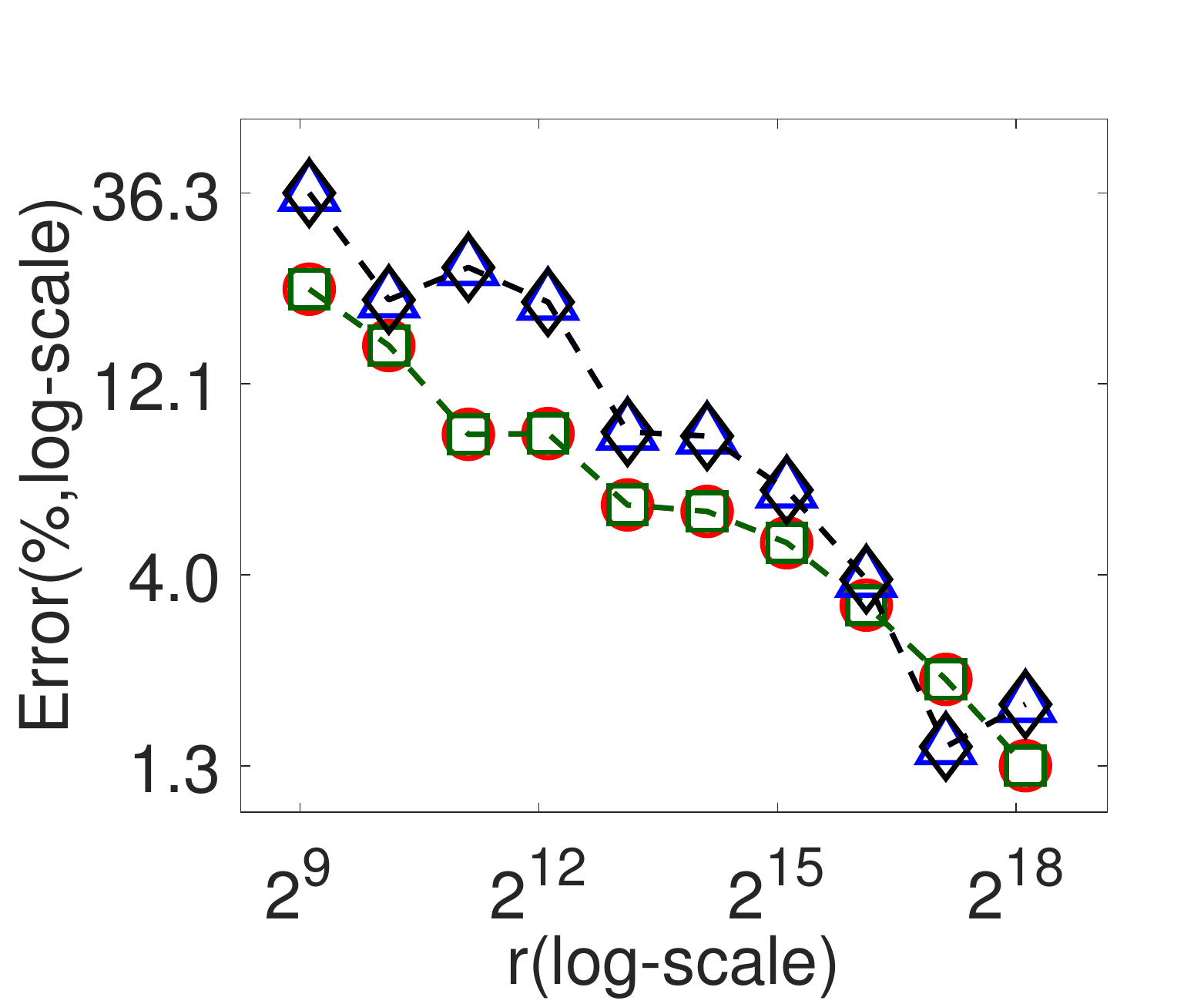}
  }
  \hfill
  \subfigure[Q4 on SU]{
    \label{fig:r:m4:su}
    \includegraphics[height=0.8in]{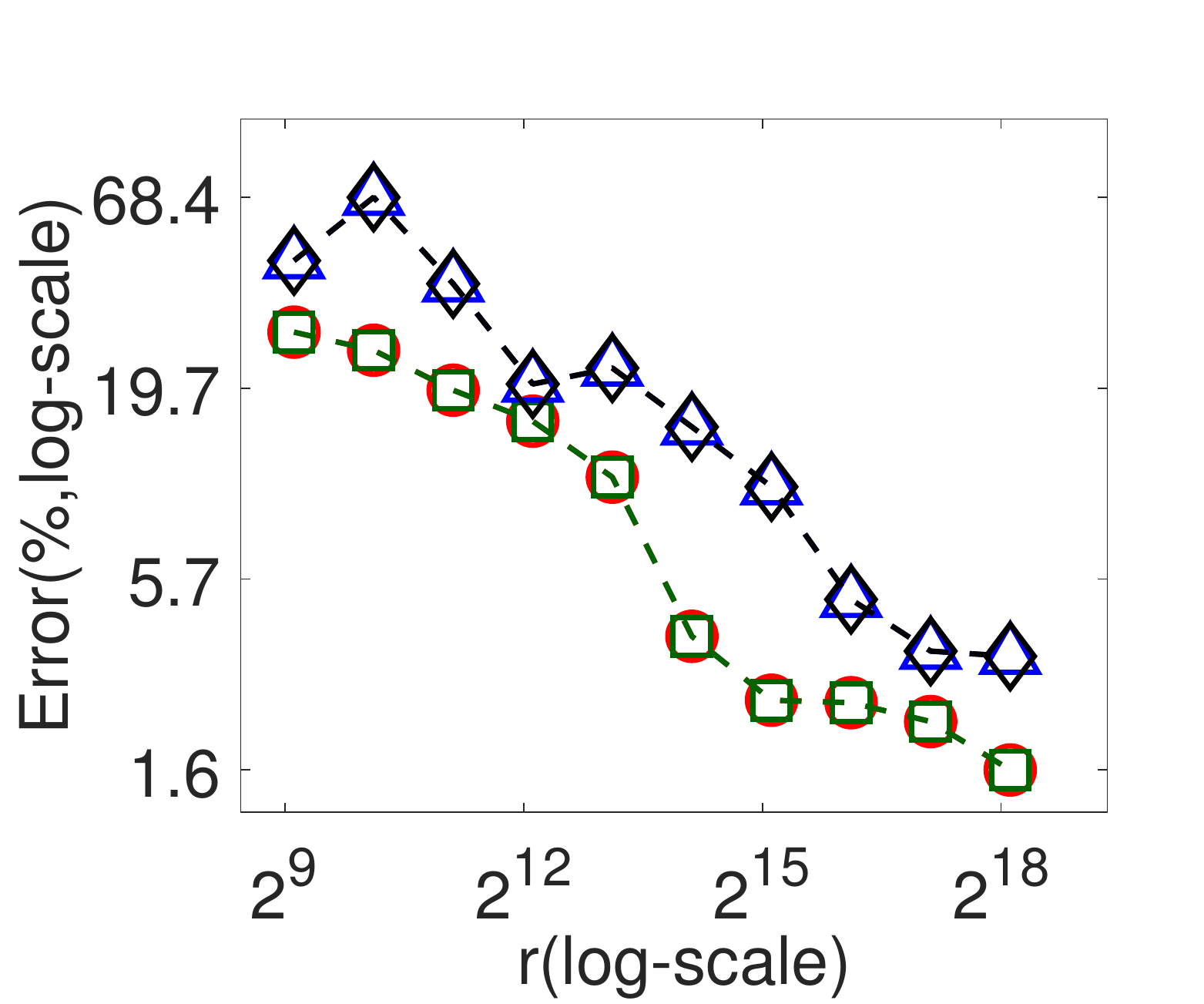}
  }
  \hfill
  \subfigure[Q5 on SU]{
    \label{fig:r:m5:su}
    \includegraphics[height=0.8in]{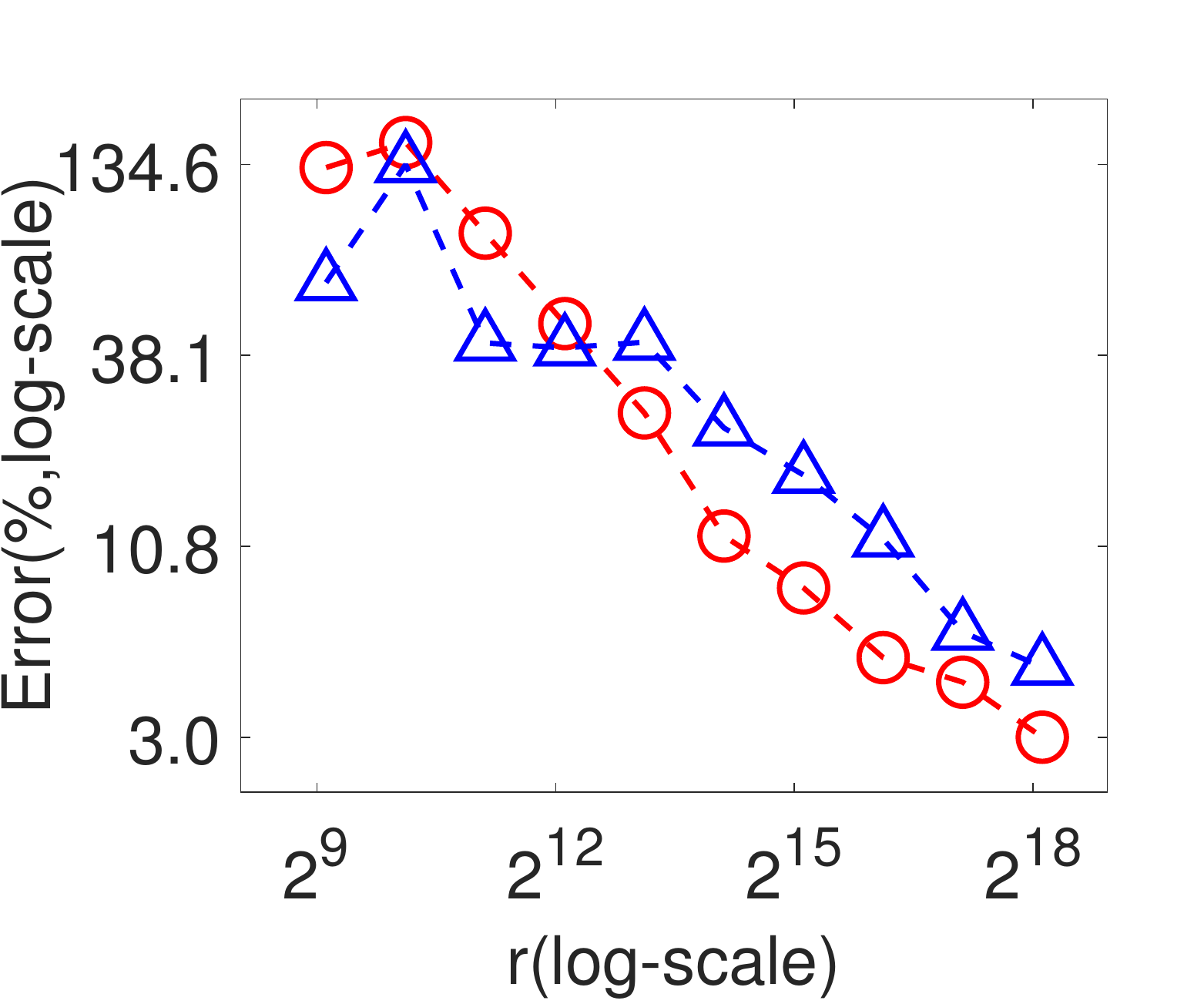}
  }
  \subfigure[Q1 on SX]{
    \label{fig:r:m1:sx}
    \includegraphics[height=0.8in]{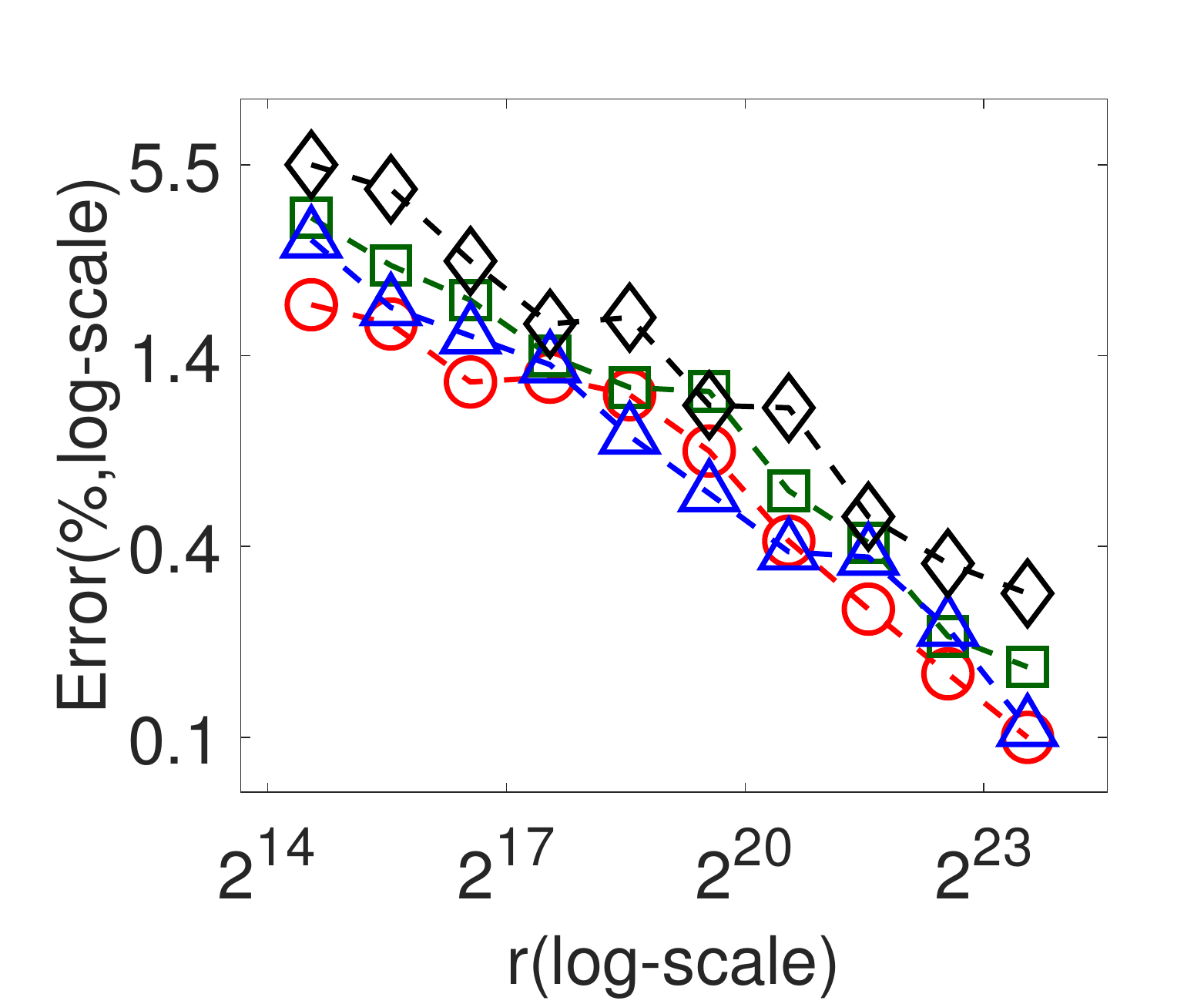}
  }
  \hfill
  \subfigure[Q2 on SX]{
    \label{fig:r:m2:sx}
    \includegraphics[height=0.8in]{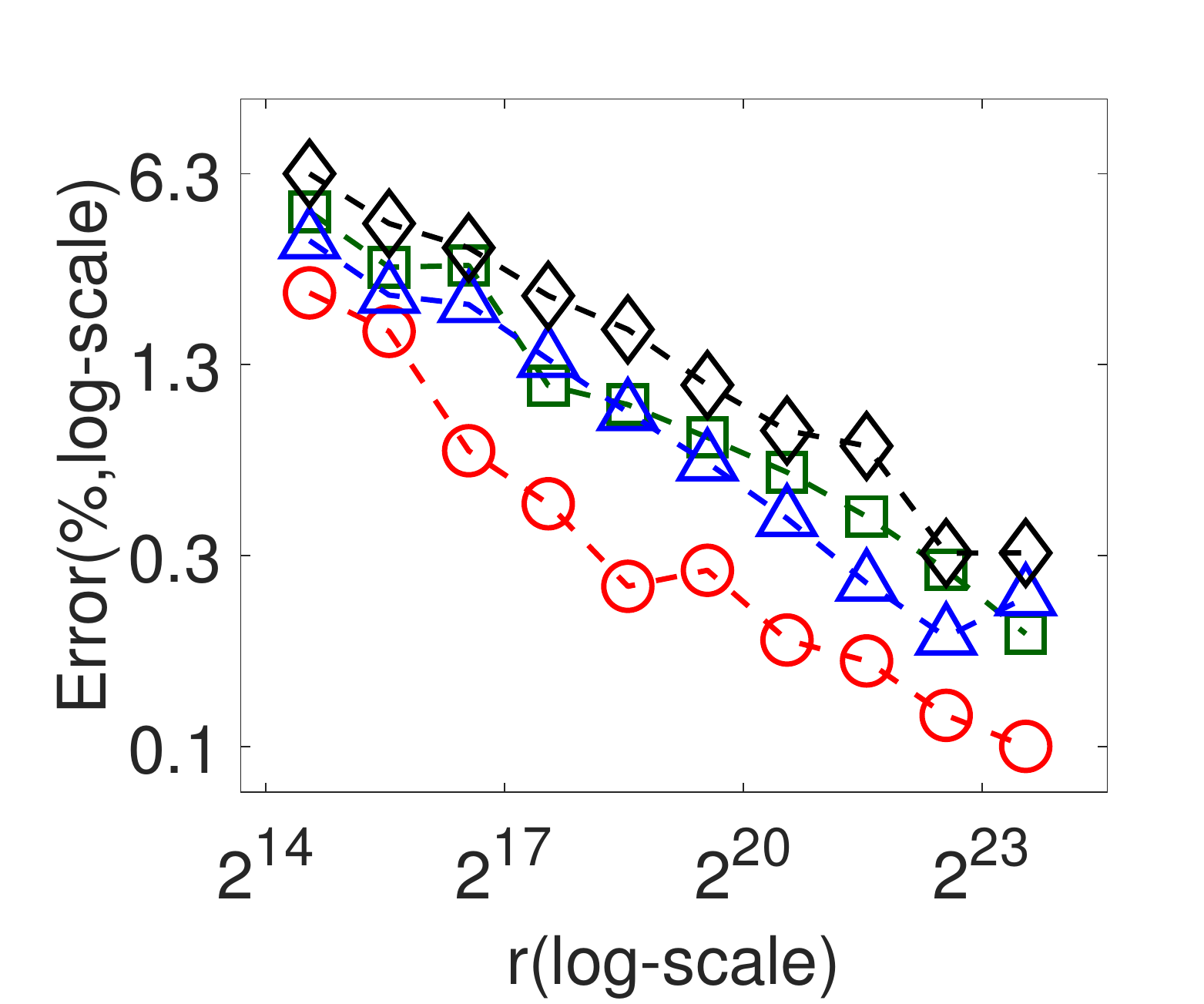}
  }
  \hfill
  \subfigure[Q3 on SX]{
    \label{fig:r:m3:sx}
    \includegraphics[height=0.8in]{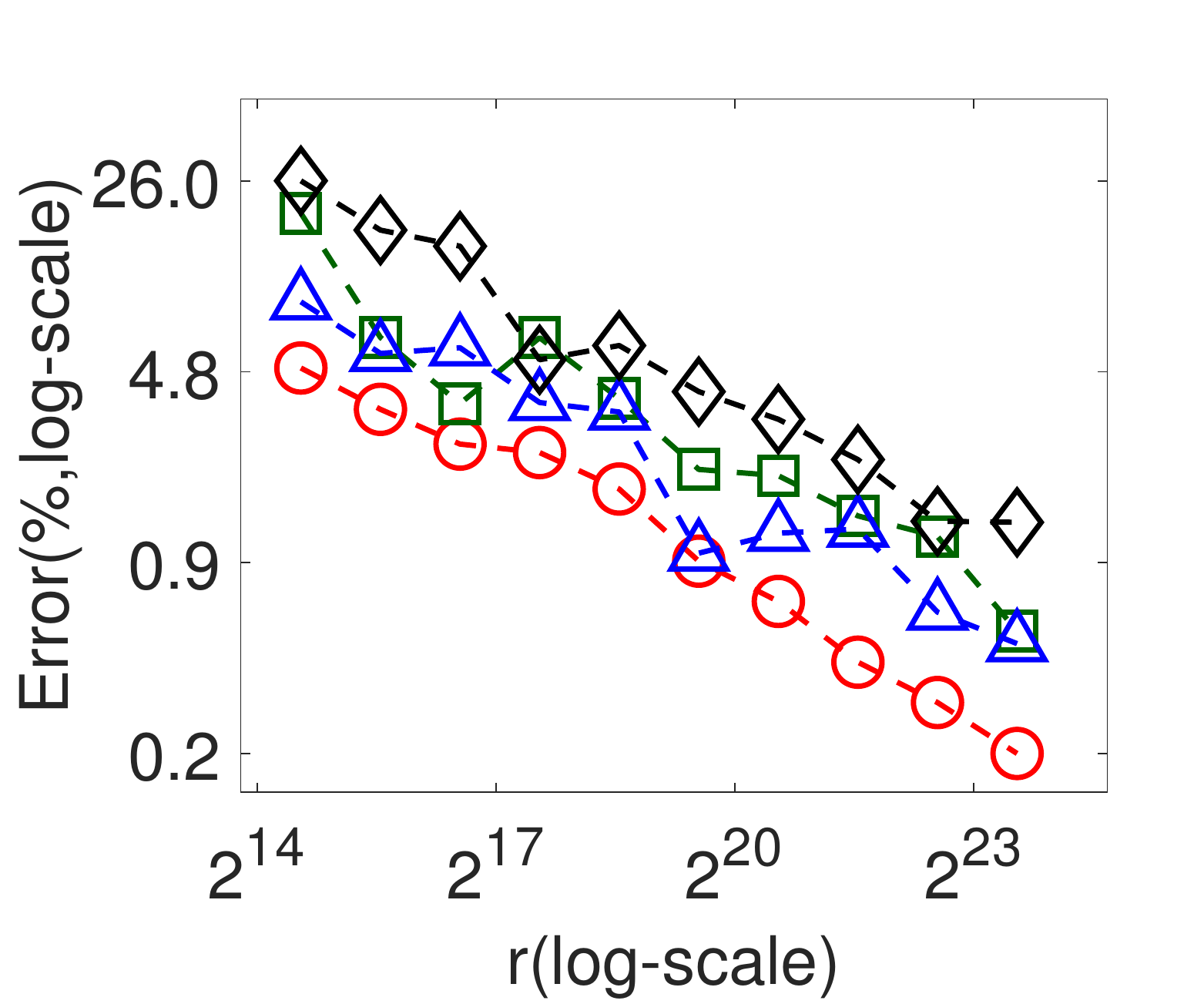}
  }
  \hfill
  \subfigure[Q4 on SX]{
    \label{fig:r:m4:sx}
    \includegraphics[height=0.8in]{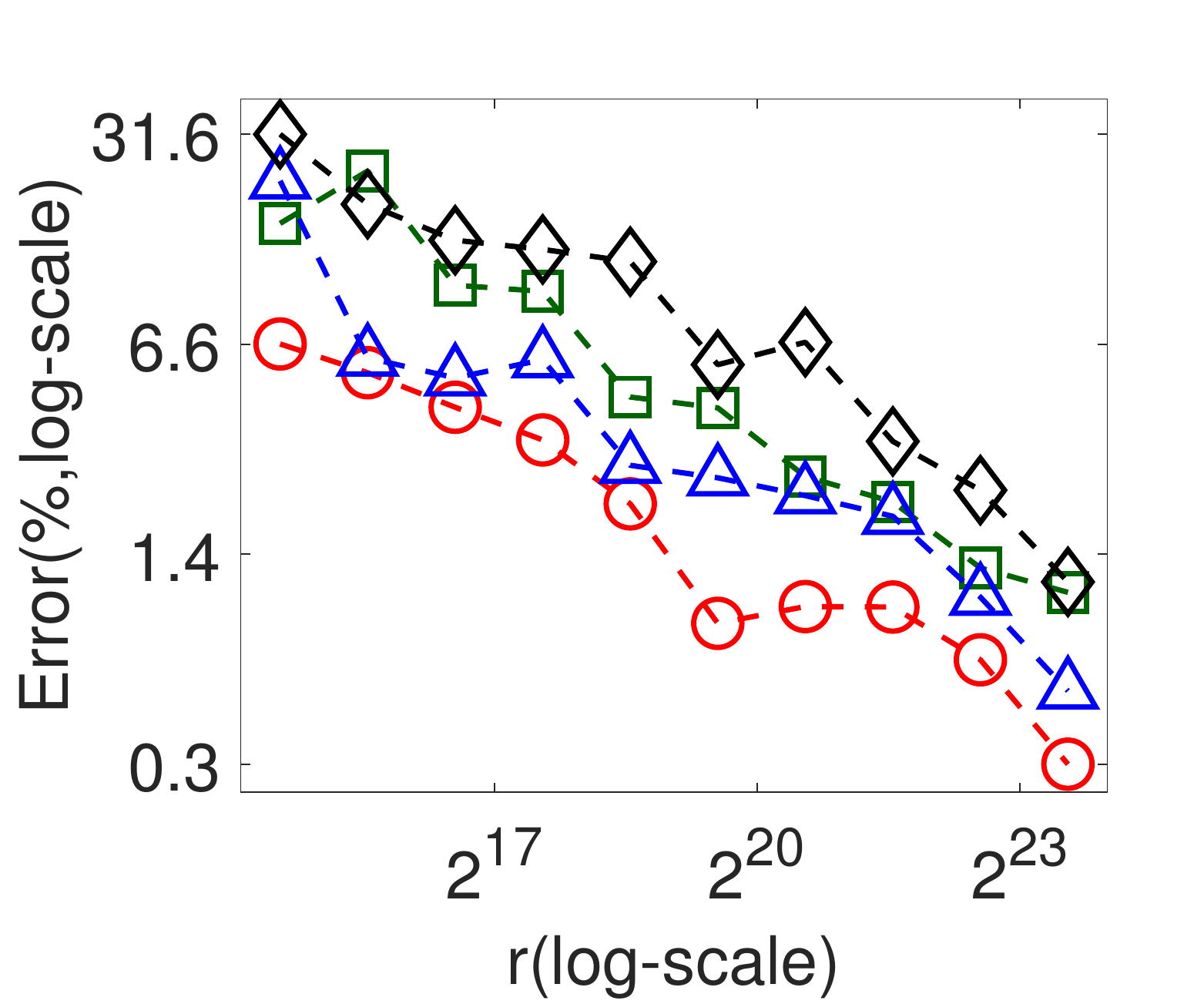}
  }
  \hfill
  \subfigure[Q5 on SX]{
    \label{fig:r:m5:sx}
    \includegraphics[height=0.8in]{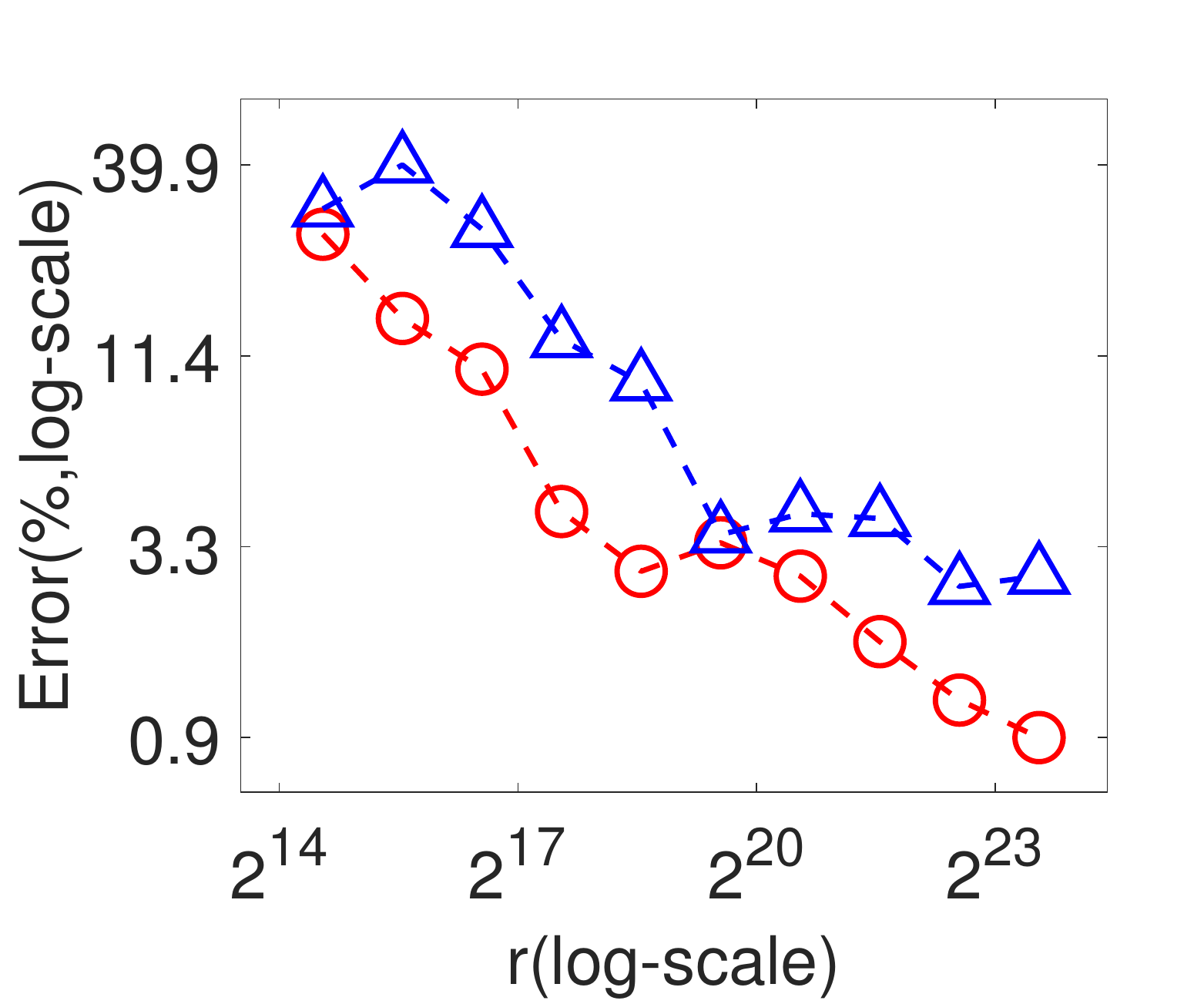}
  }
  \subfigure[Q1 on BC]{
    \label{fig:r:m1:bt}
    \includegraphics[height=0.8in]{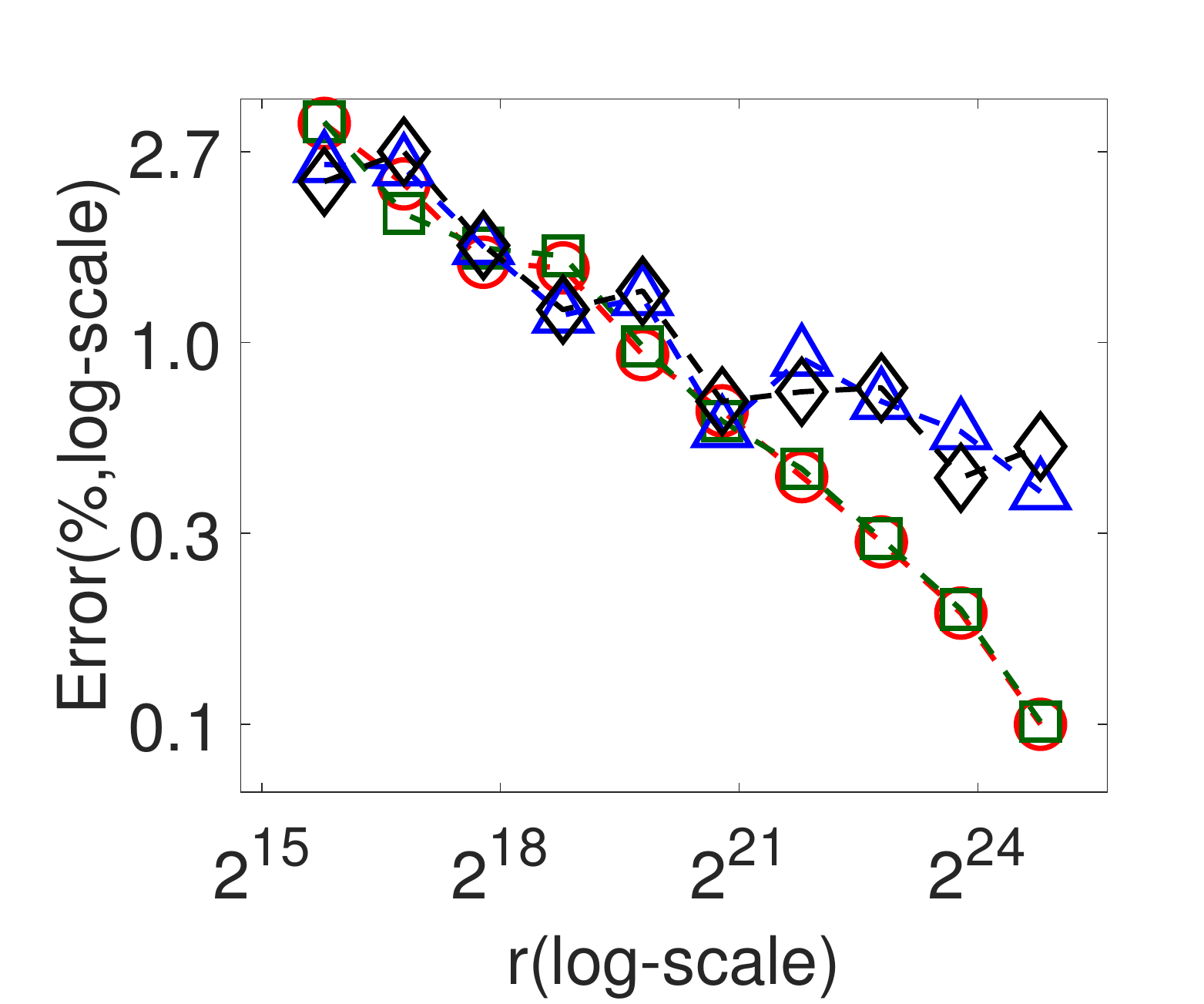}
  }
  \hfill
  \subfigure[Q2 on BC]{
    \label{fig:r:m2:bt}
    \includegraphics[height=0.8in]{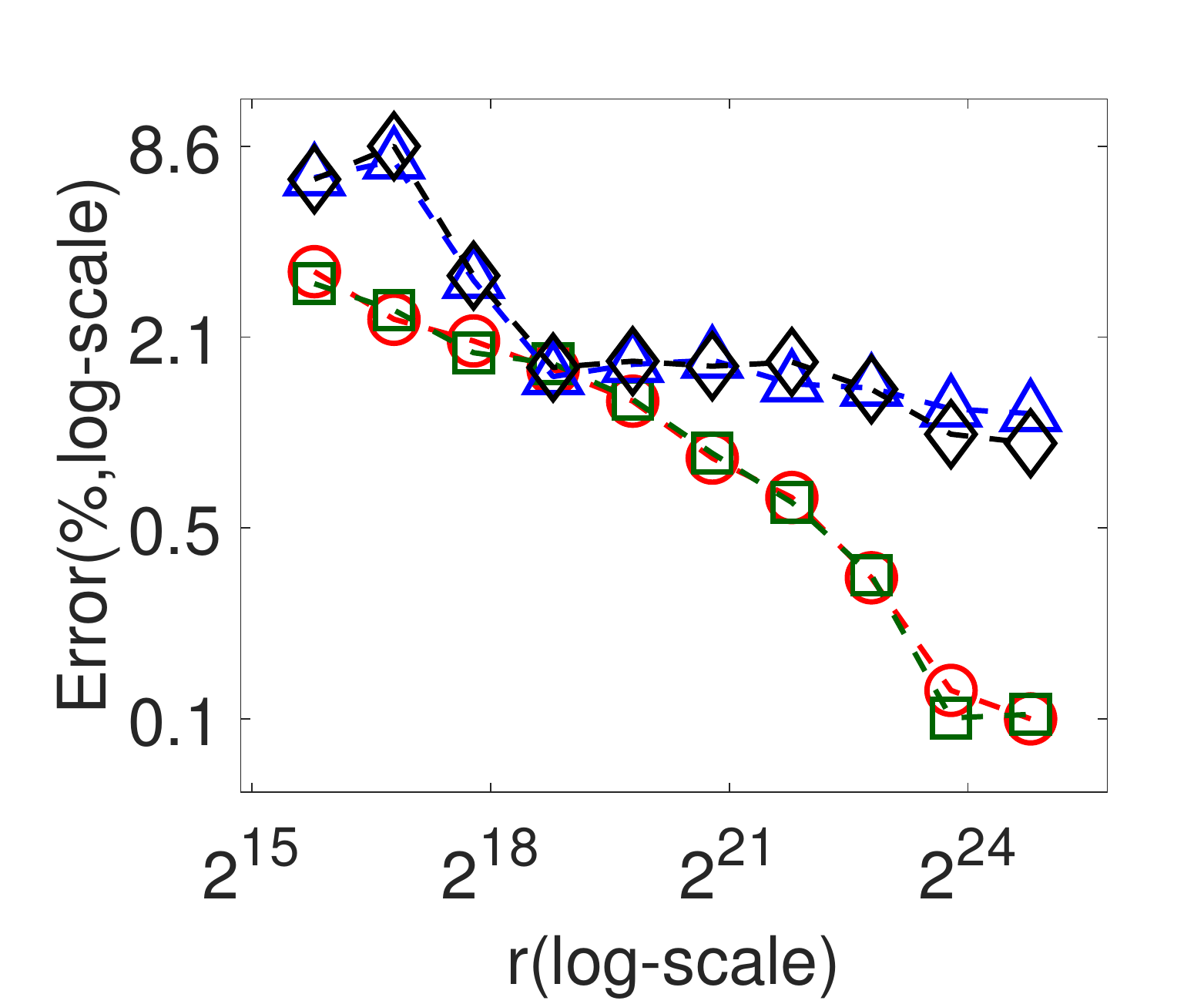}
  }
  \hfill
  \subfigure[Q3 on BC]{
    \label{fig:r:m3:bt}
    \includegraphics[height=0.8in]{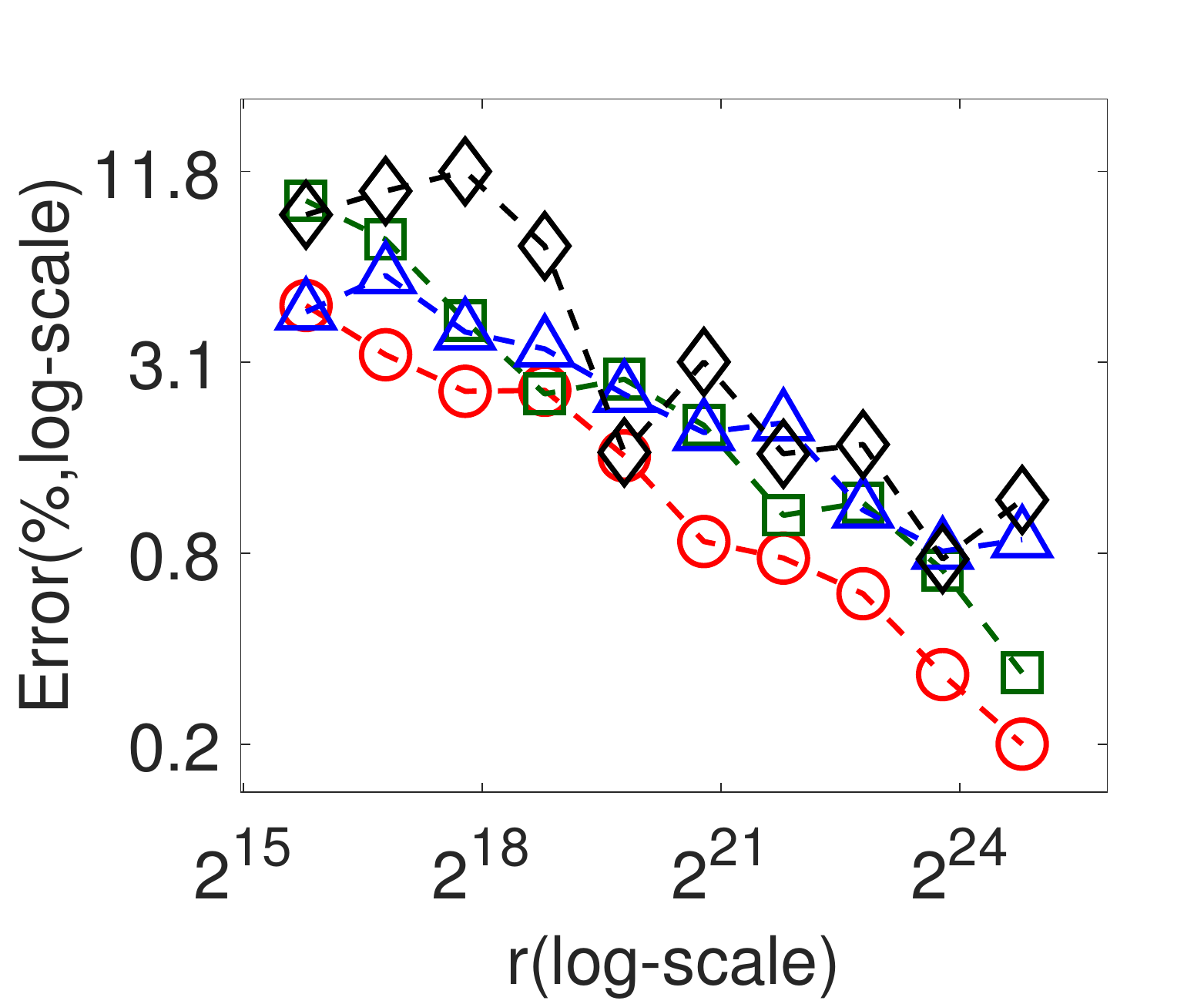}
  }
  \hfill
  \subfigure[Q4 on BC]{
    \label{fig:r:m4:bt}
    \includegraphics[height=0.8in]{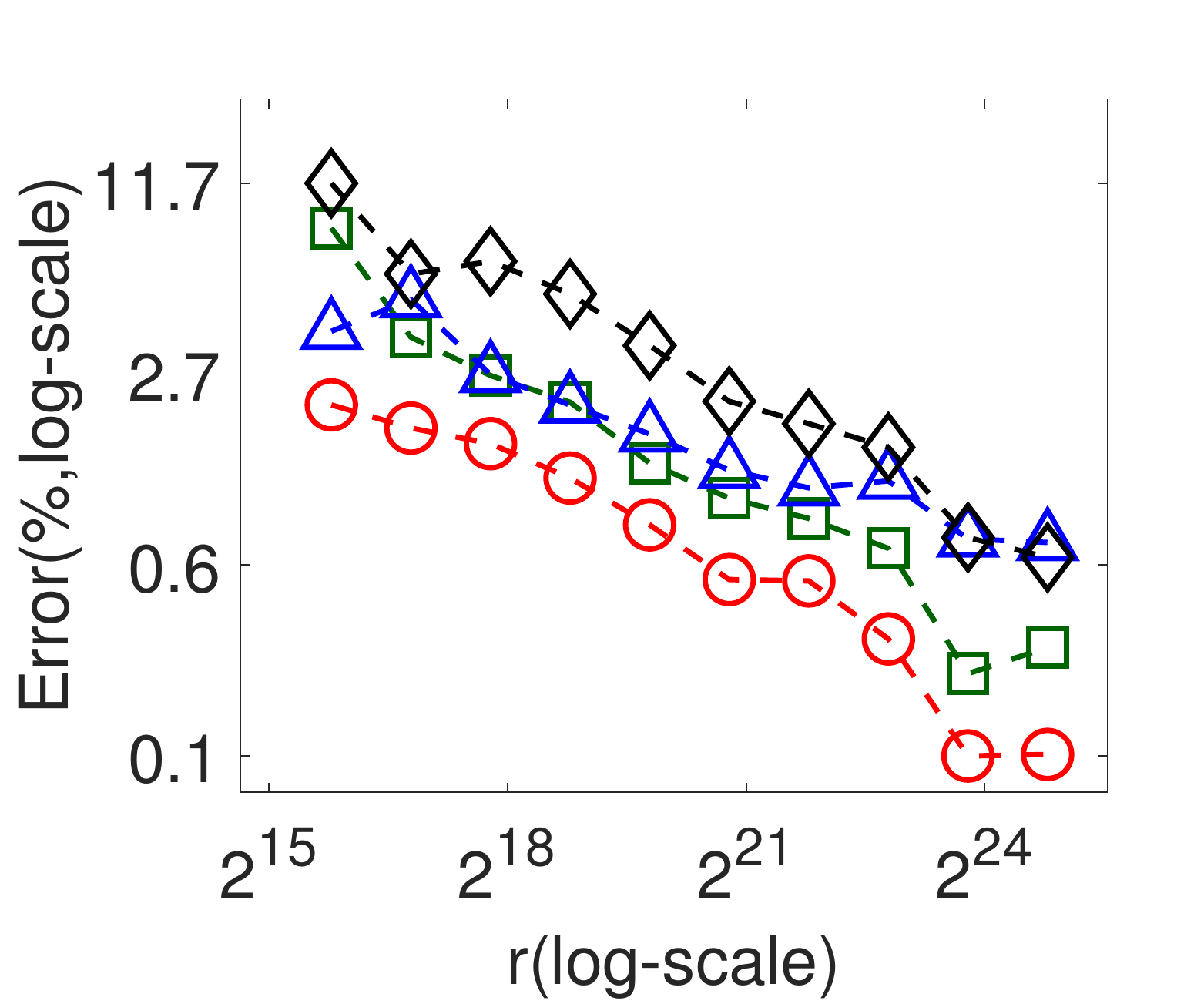}
  }
  \hfill
  \subfigure[Q5 on BC]{
    \label{fig:r:m5:bt}
    \includegraphics[height=0.8in]{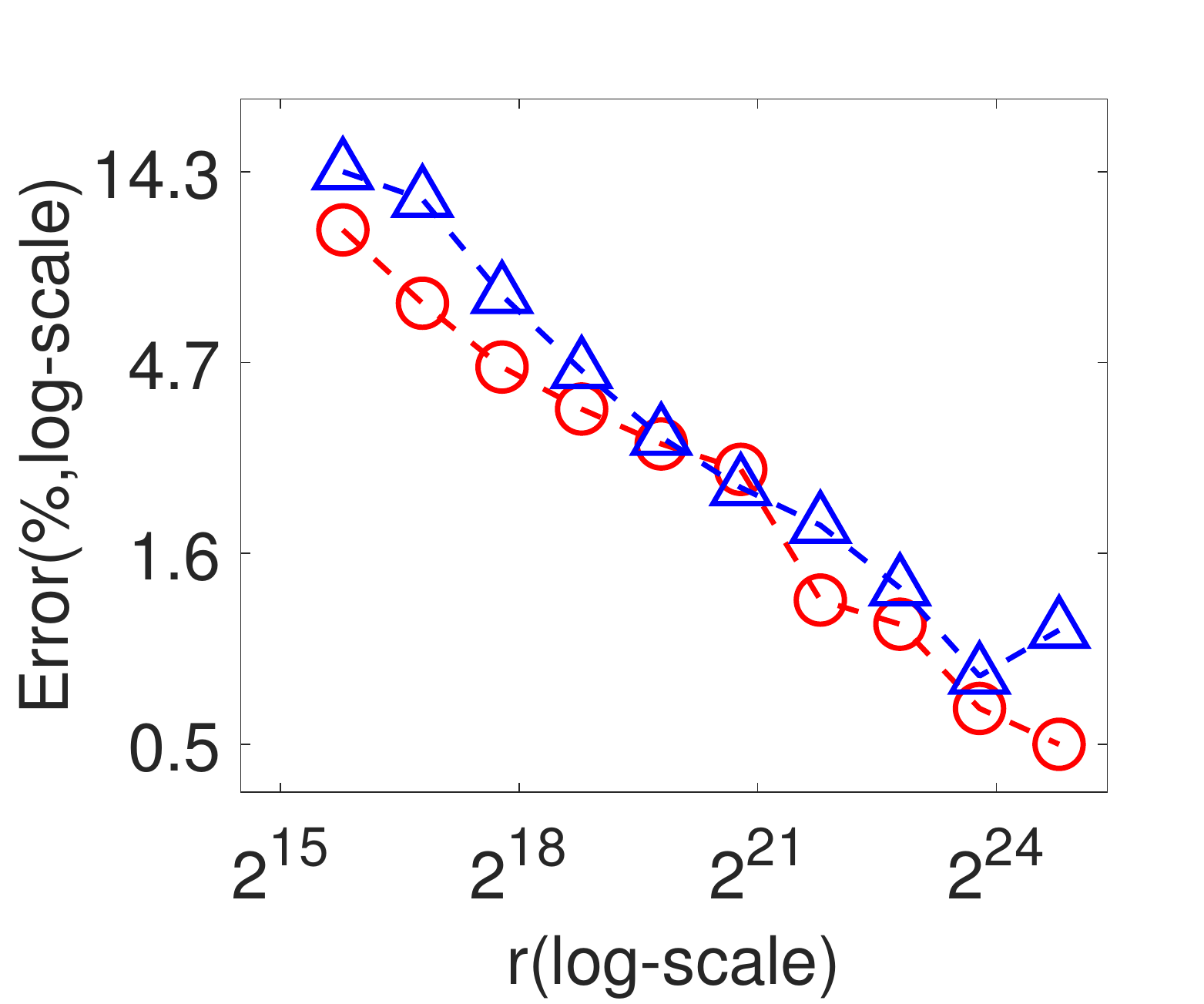}
  }
  \subfigure[Q1 on RC]{
    \label{fig:r:m1:rc}
    \includegraphics[height=0.8in]{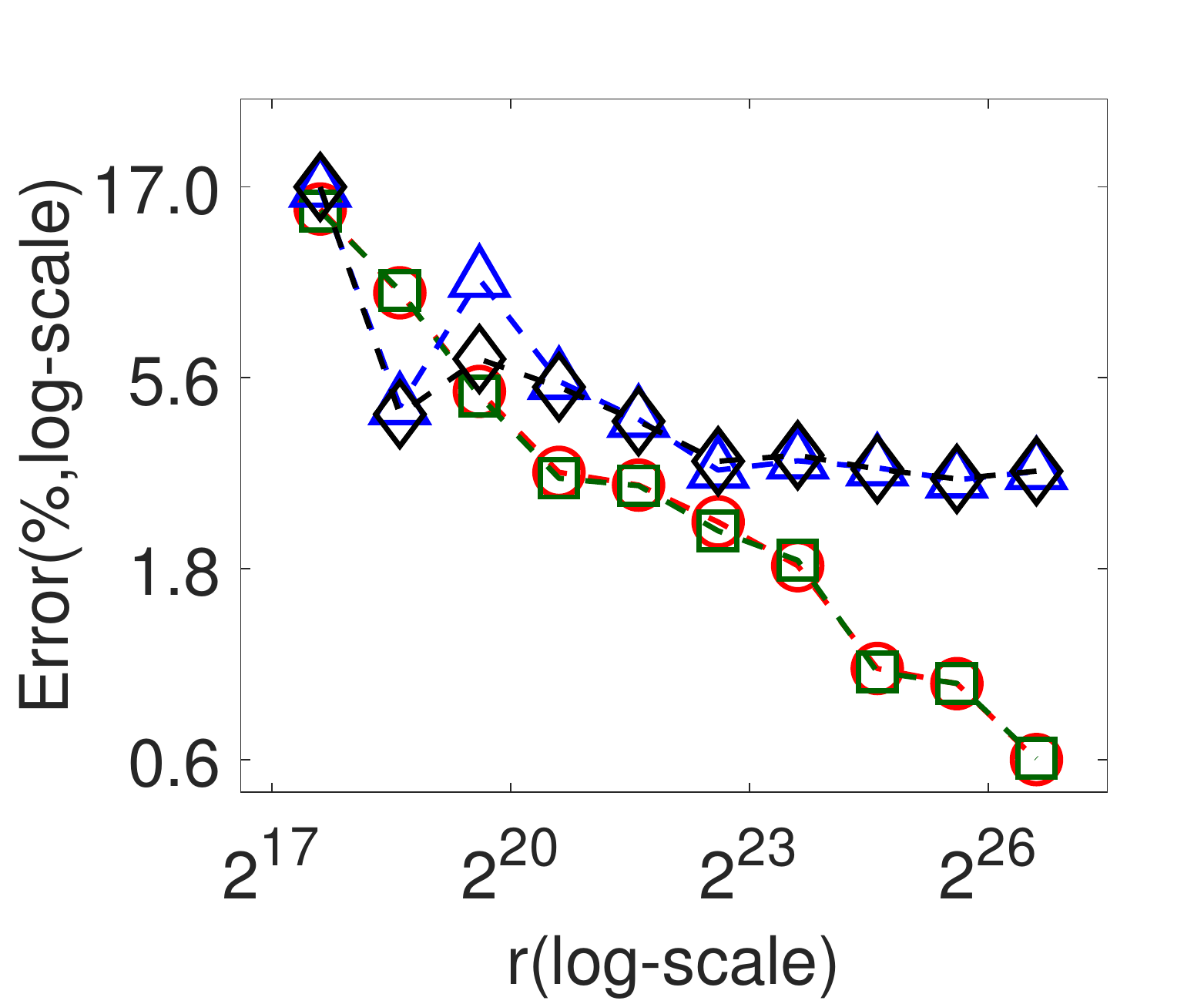}
  }
  \hfill
  \subfigure[Q2 on RC]{
    \label{fig:r:m2:rc}
    \includegraphics[height=0.8in]{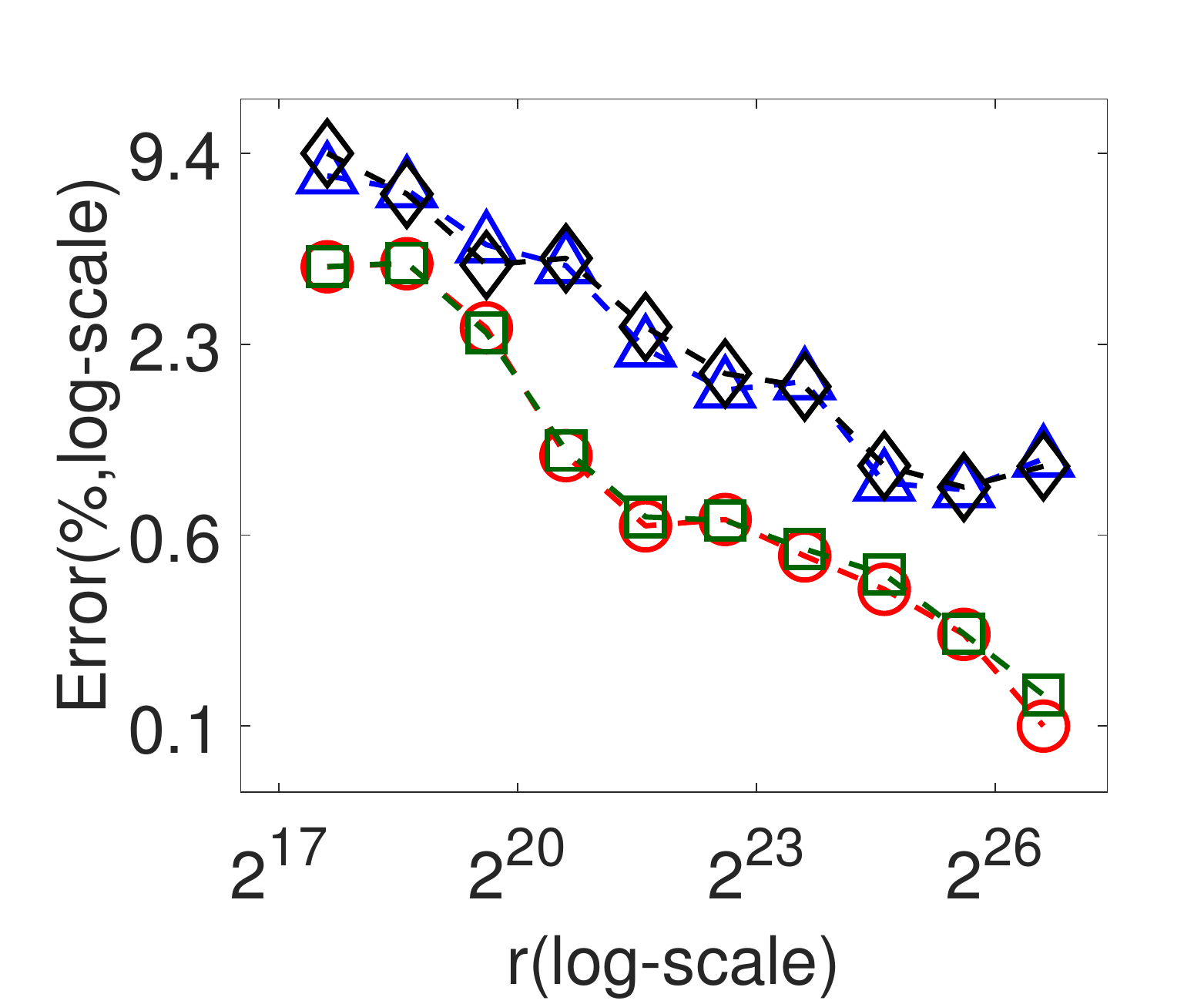}
  }
  \hfill
  \subfigure[Q3 on RC]{
    \label{fig:r:m3:rc}
    \includegraphics[height=0.8in]{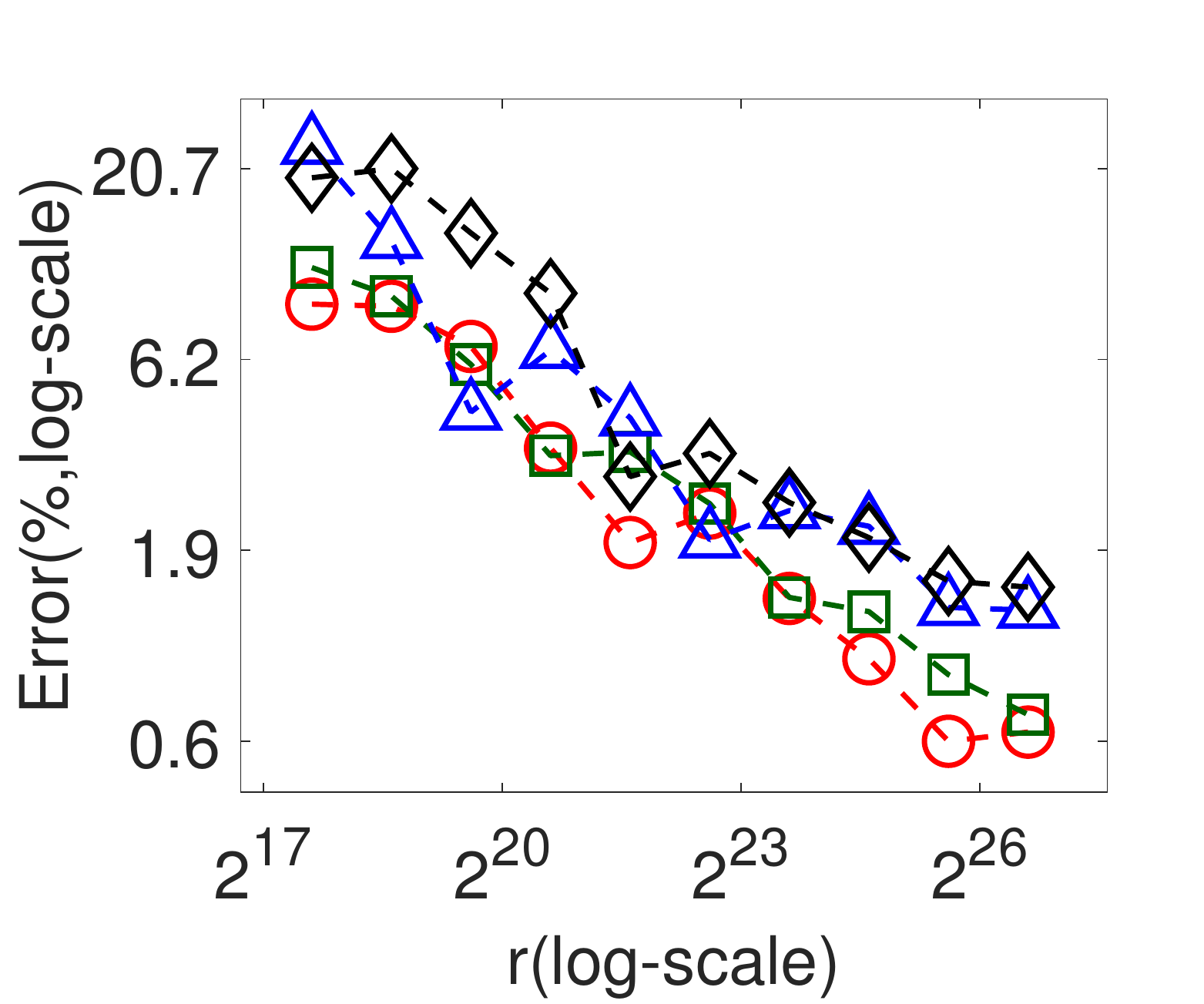}
  }
  \hfill
  \subfigure[Q4 on RC]{
    \label{fig:r:m4:rc}
    \includegraphics[height=0.8in]{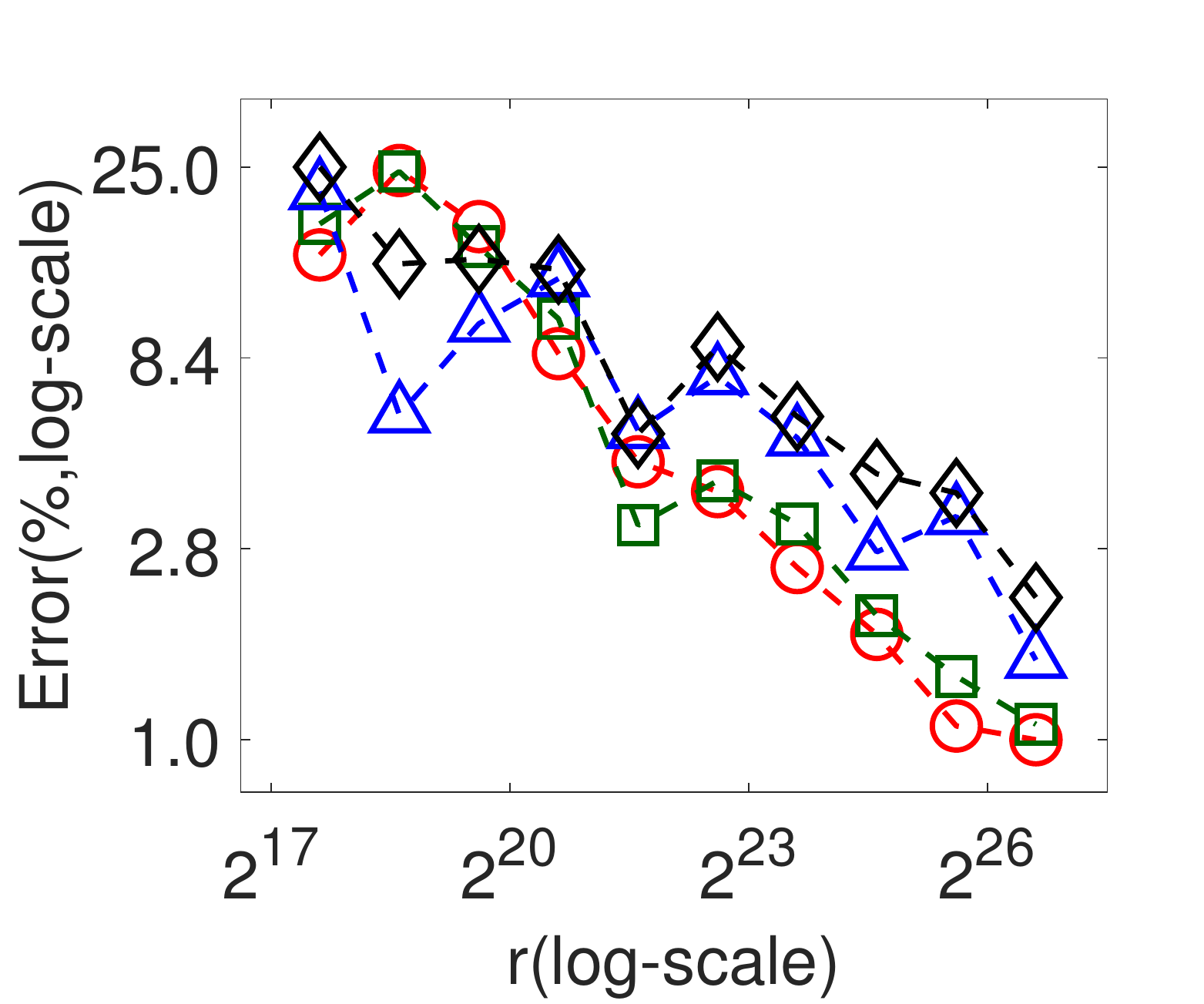}
  }
  \hfill
  \subfigure[Q5 on RC]{
    \label{fig:r:m5:rc}
    \includegraphics[height=0.8in]{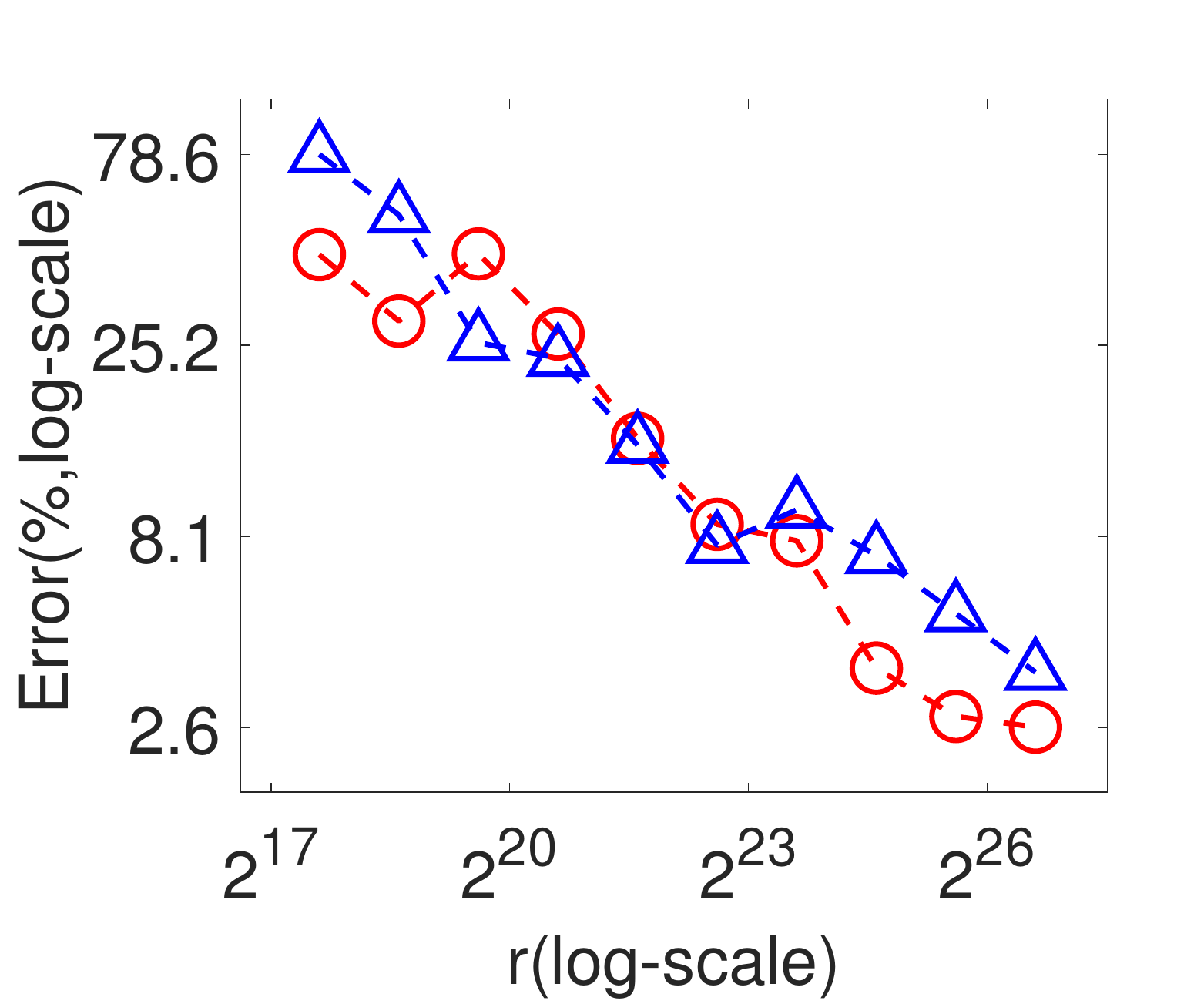}
  }
  \caption{Relative error (\%) with varying sample size $r$ in the streaming setting.}
  \label{fig:r:error}
  \Description{Error-r}
\end{figure}

\begin{figure}[t]
  \centering
  \includegraphics[height=0.15in]{figures/r/legend1.pdf}
  \\
  \subfigure[Q1 on AU]{
    \label{fig:rt:m1:au}
    \includegraphics[height=0.8in]{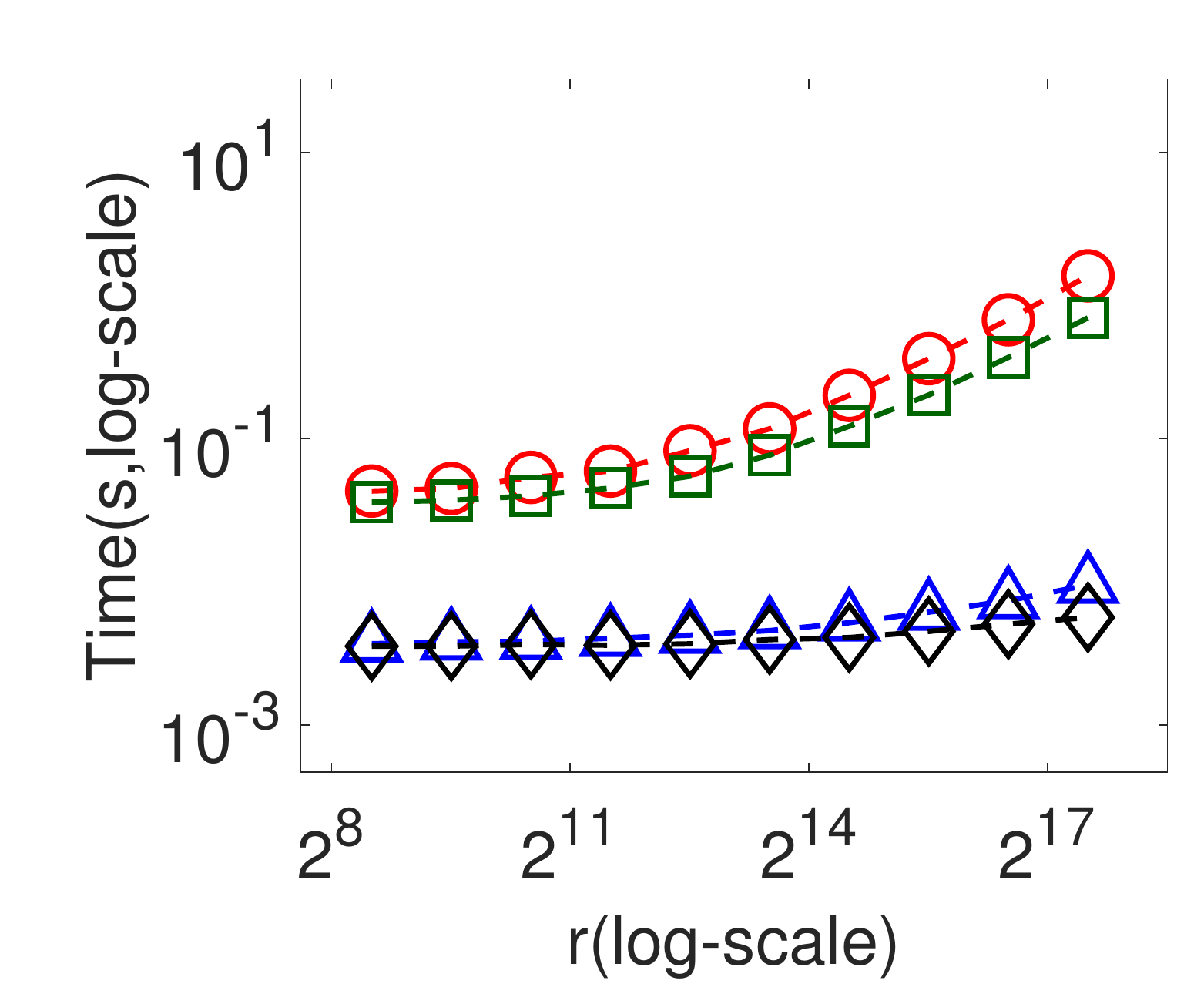}
  }
  \hfill
  \subfigure[Q2 on AU]{
    \label{fig:rt:m2:au}
    \includegraphics[height=0.8in]{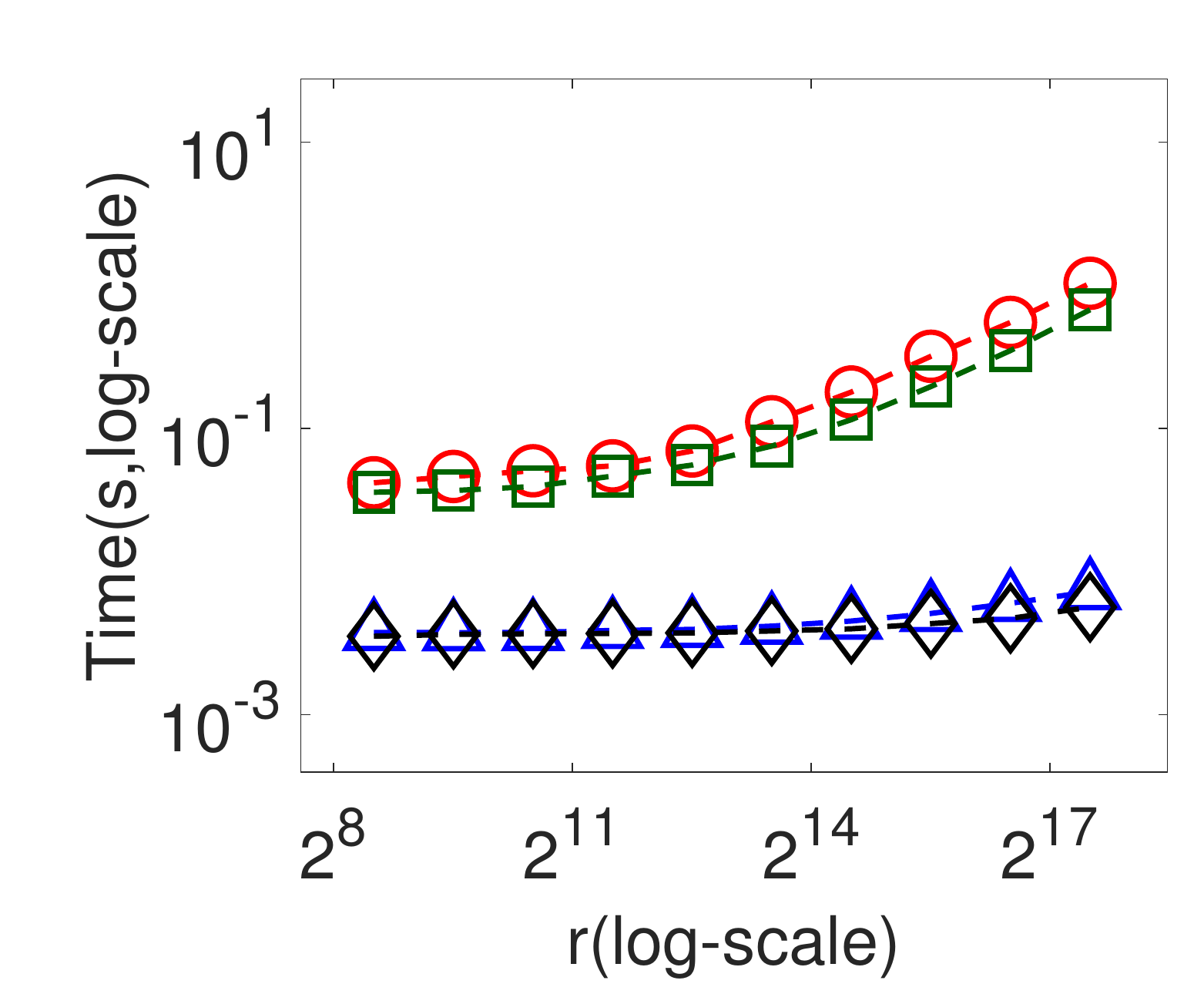}
  }
  \hfill
  \subfigure[Q3 on AU]{
    \label{fig:rt:m3:au}
    \includegraphics[height=0.8in]{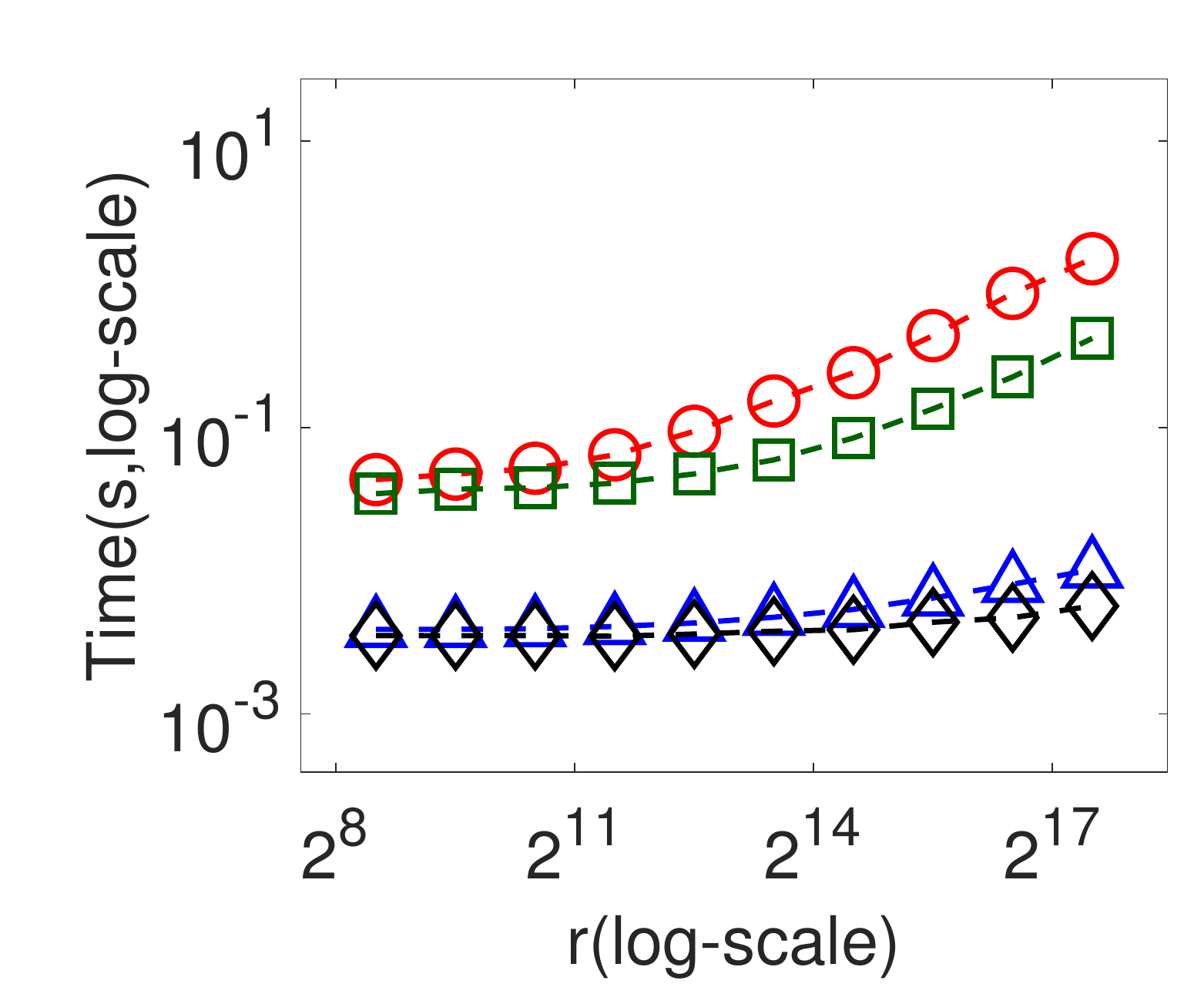}
  }
  \hfill
  \subfigure[Q4 on AU]{
    \label{fig:rt:m4:au}
    \includegraphics[height=0.8in]{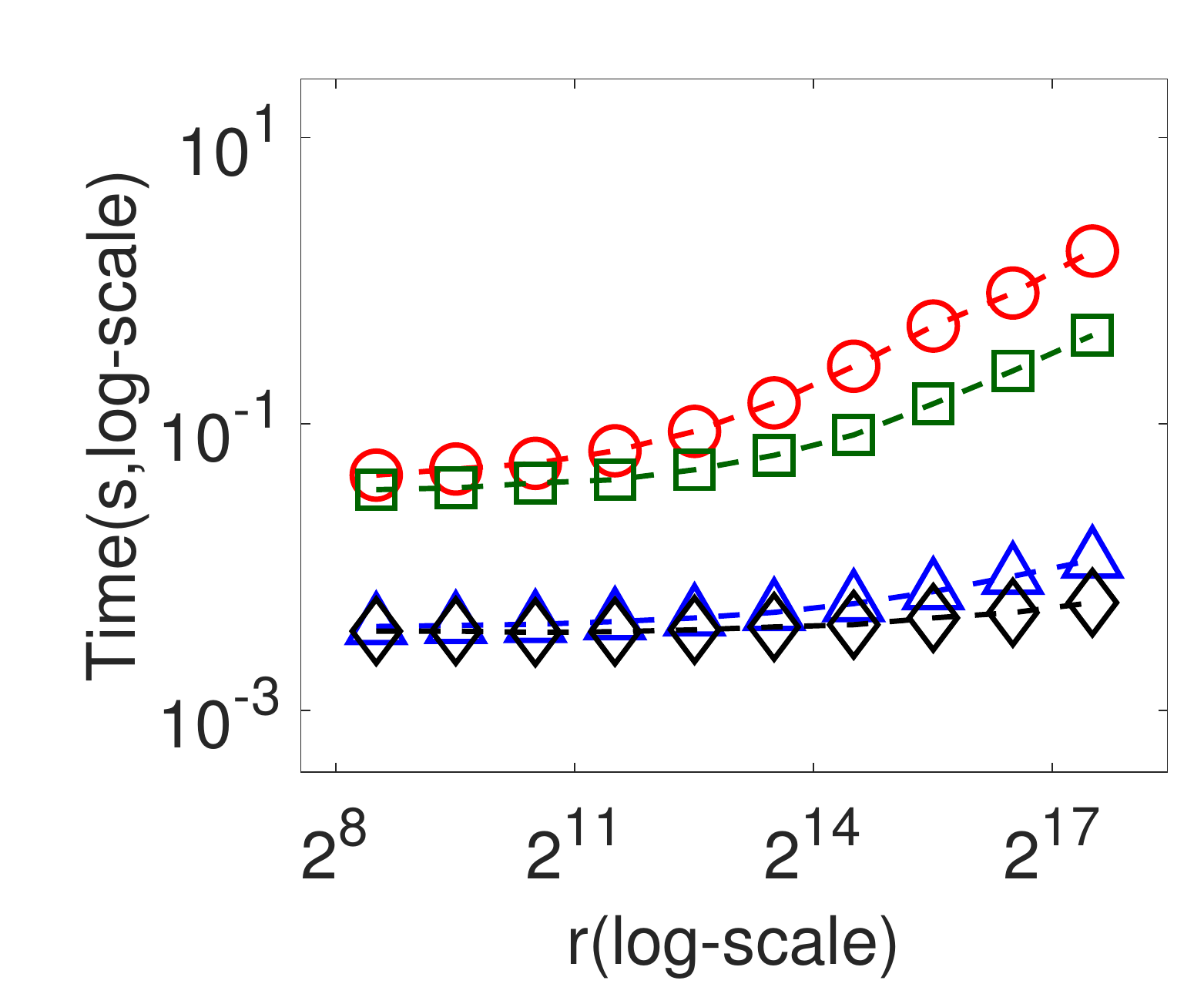}
  }
  \hfill
  \subfigure[Q5 on AU]{
    \label{fig:rt:m5:au}
    \includegraphics[height=0.8in]{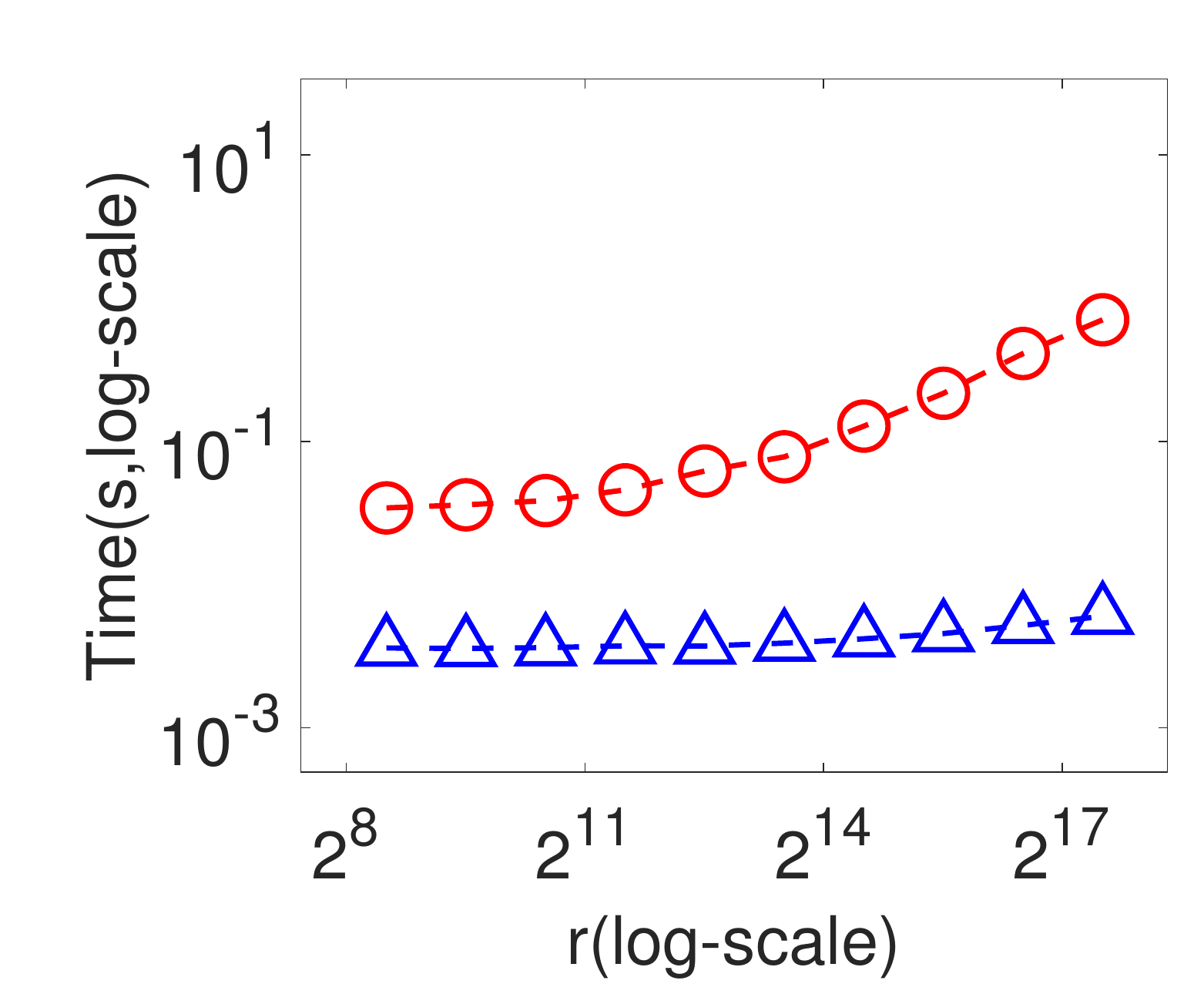}
  }
  \subfigure[Q1 on SU]{
    \label{fig:rt:m1:su}
    \includegraphics[height=0.8in]{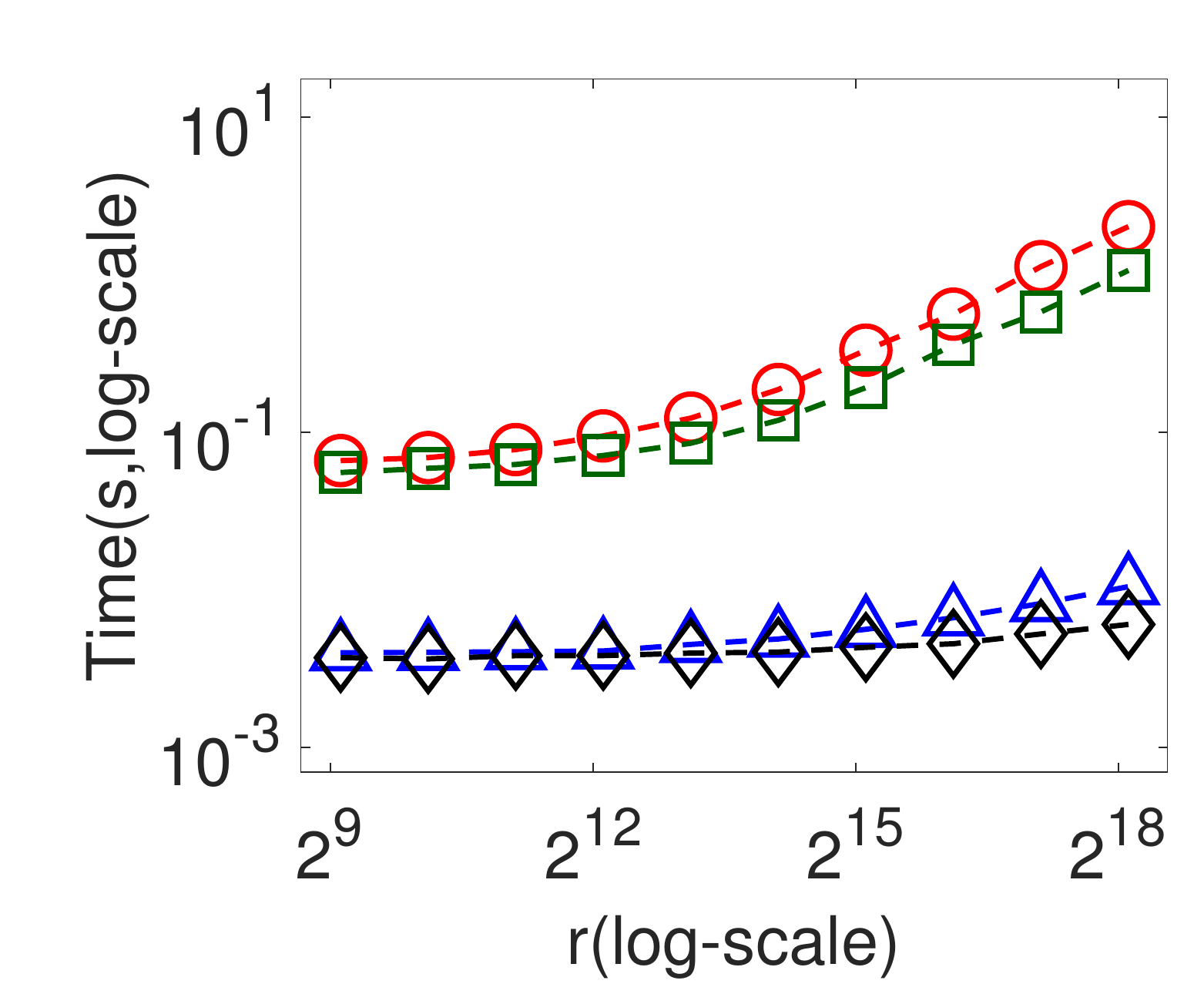}
  }
  \hfill
  \subfigure[Q2 on SU]{
    \label{fig:rt:m2:su}
    \includegraphics[height=0.8in]{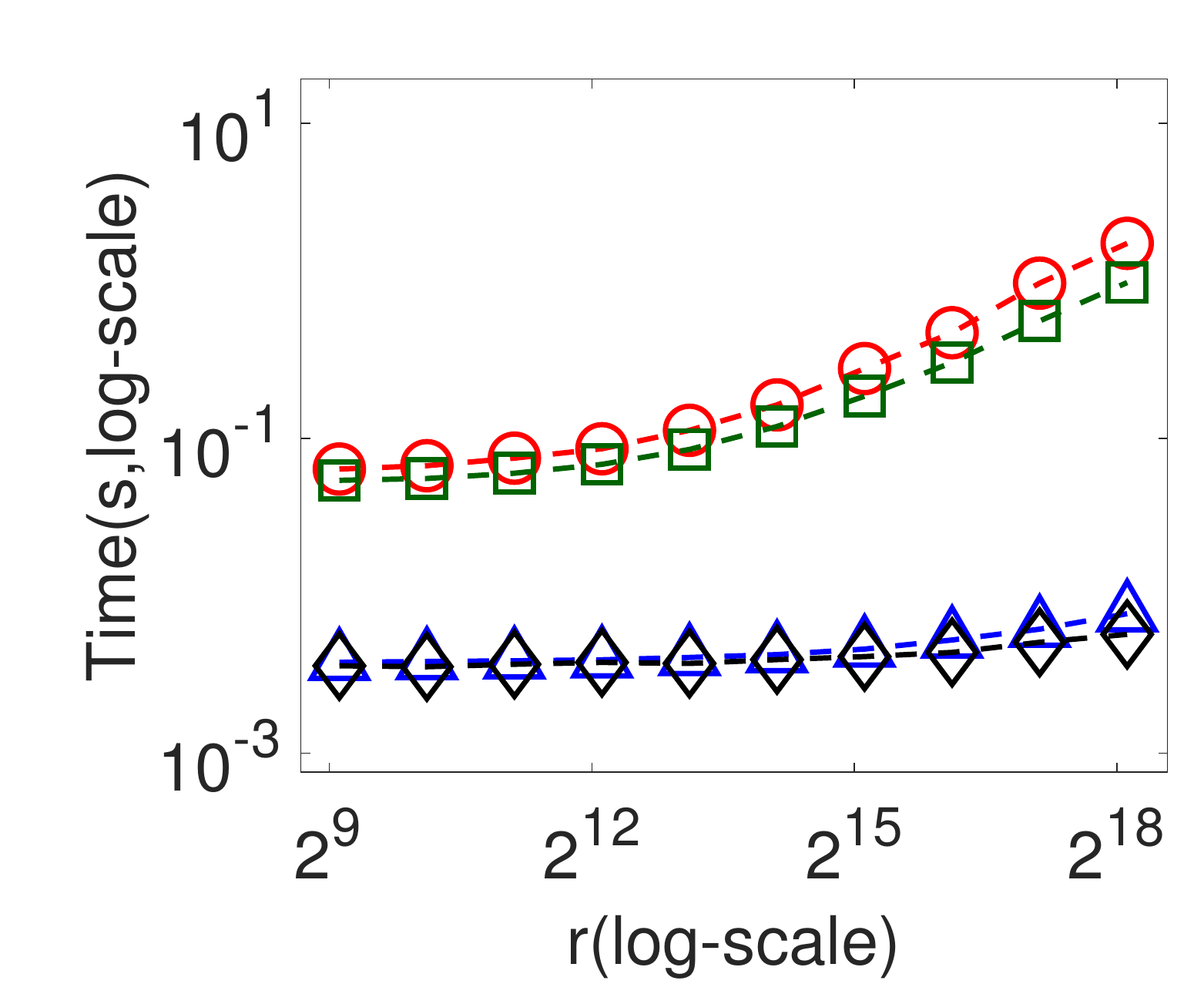}
  }
  \hfill
  \subfigure[Q3 on SU]{
    \label{fig:rt:m3:su}
    \includegraphics[height=0.8in]{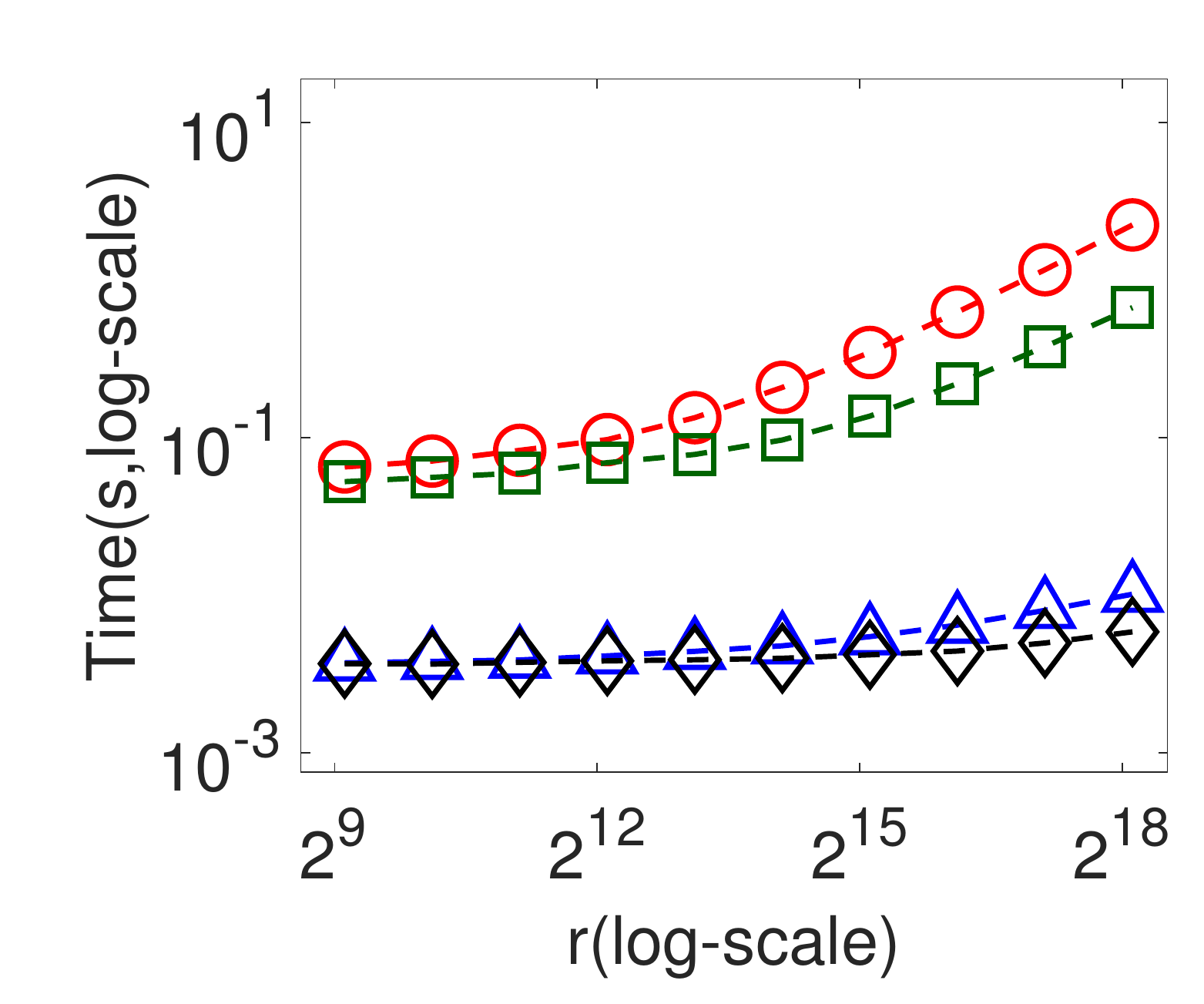}
  }
  \hfill
  \subfigure[Q4 on SU]{
    \label{fig:rt:m4:su}
    \includegraphics[height=0.8in]{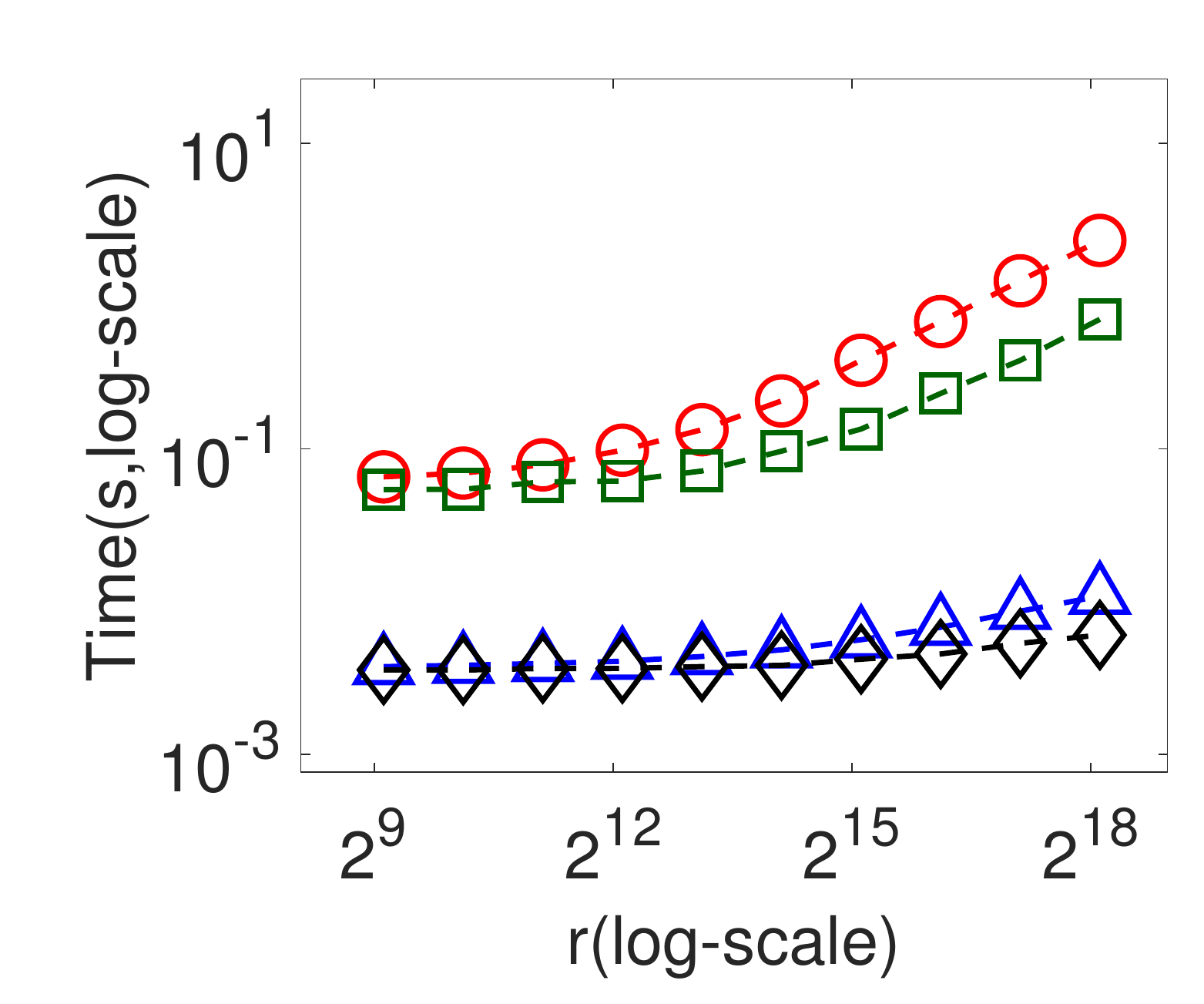}
  }
  \hfill
  \subfigure[Q5 on SU]{
    \label{fig:rt:m5:su}
    \includegraphics[height=0.8in]{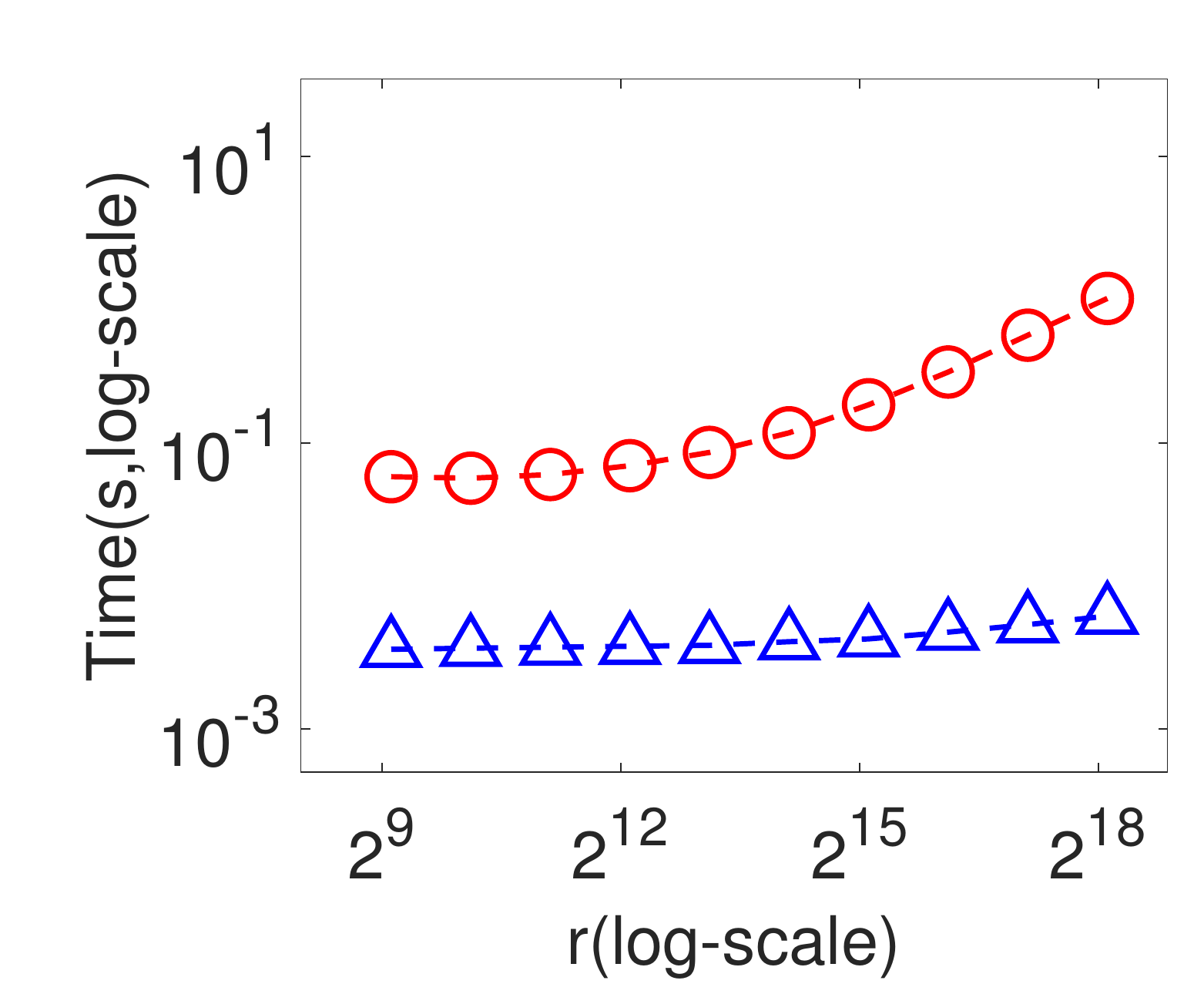}
  }
  \subfigure[Q1 on SX]{
    \label{fig:rt:m1:sx}
    \includegraphics[height=0.8in]{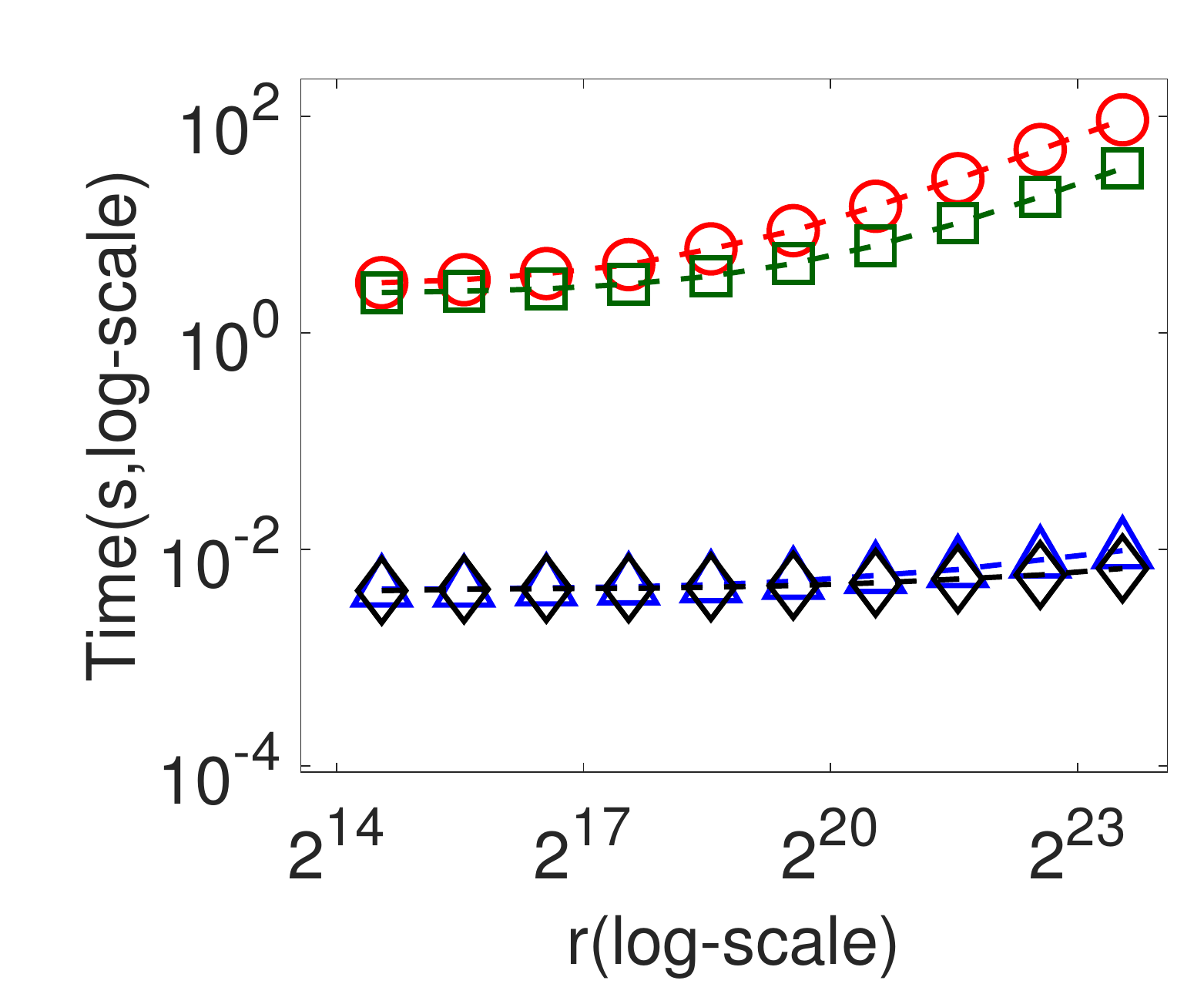}
  }
  \hfill
  \subfigure[Q2 on SX]{
    \label{fig:rt:m2:sx}
    \includegraphics[height=0.8in]{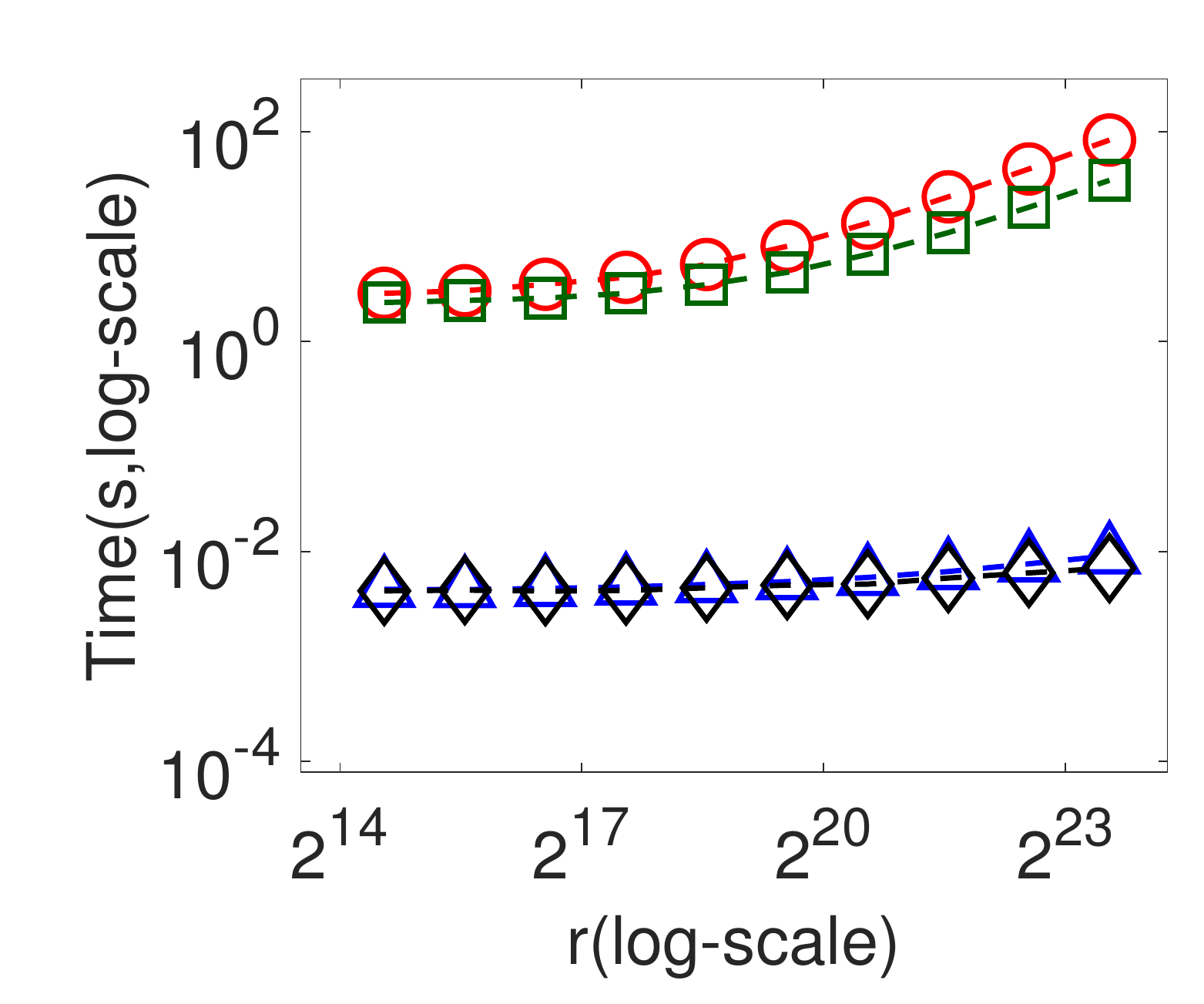}
  }
  \hfill
  \subfigure[Q3 on SX]{
    \label{fig:rt:m3:sx}
    \includegraphics[height=0.8in]{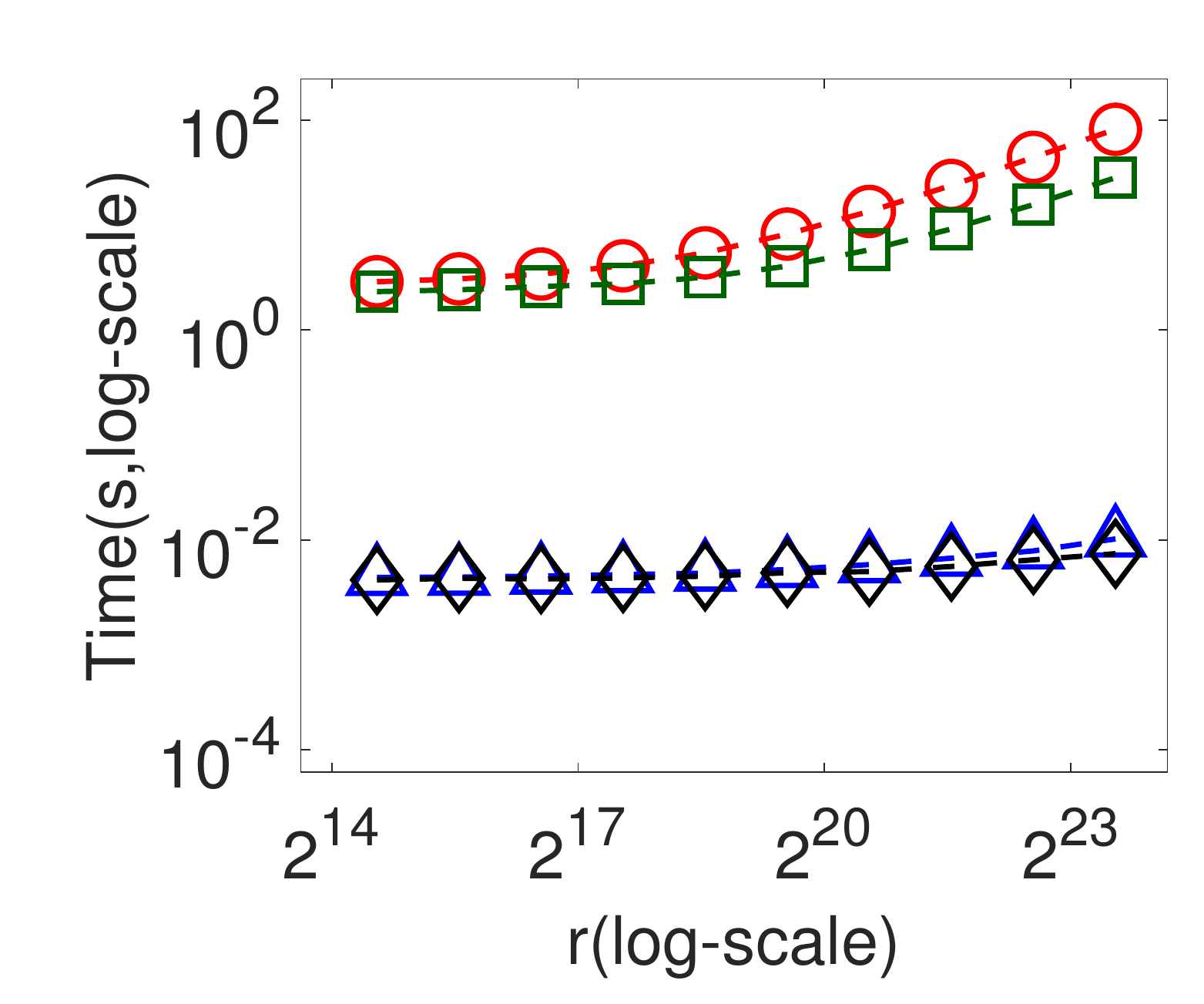}
  }
  \hfill
  \subfigure[Q4 on SX]{
    \label{fig:rt:m4:sx}
    \includegraphics[height=0.8in]{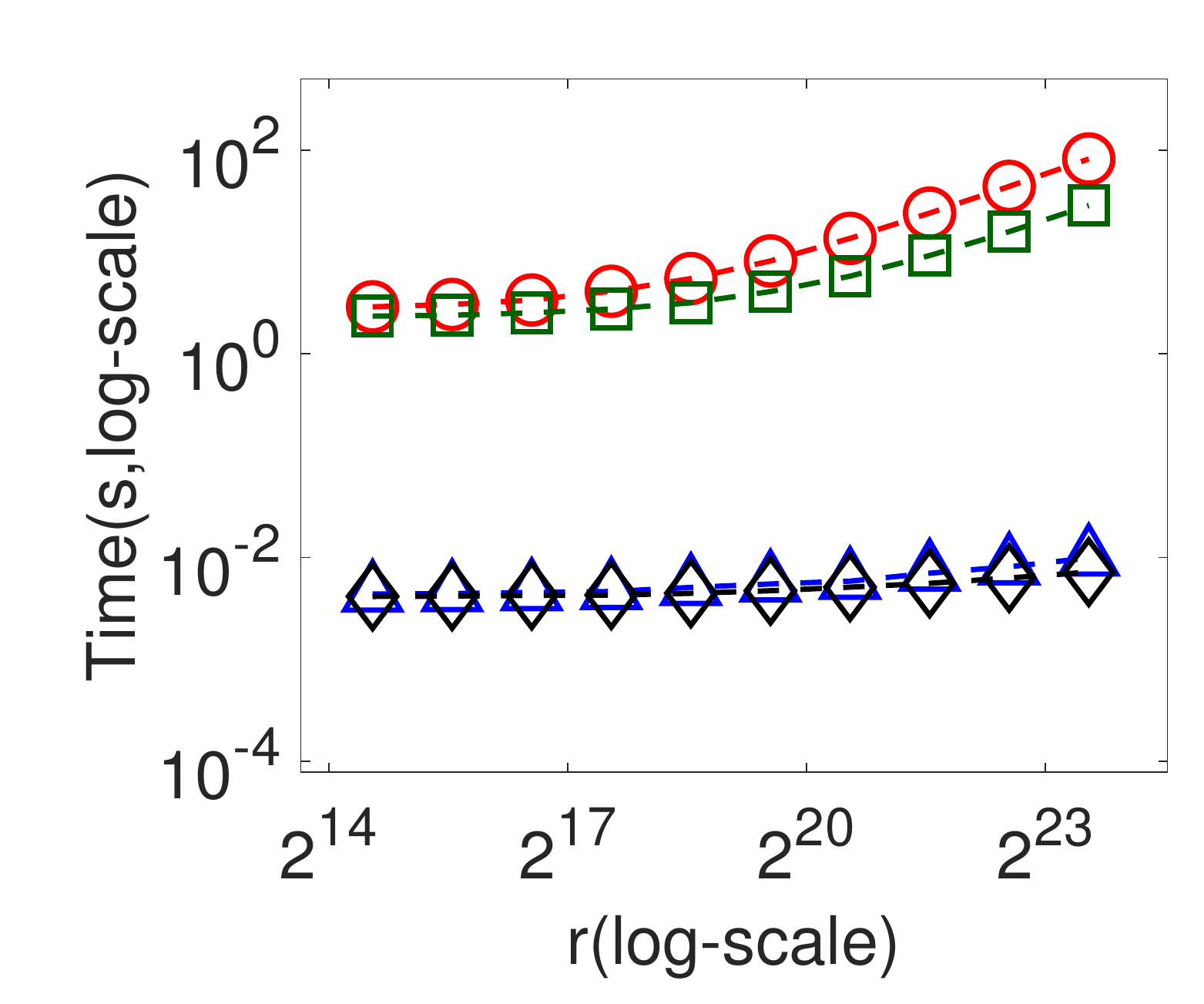}
  }
  \hfill
  \subfigure[Q5 on SX]{
    \label{fig:rt:m5:sx}
    \includegraphics[height=0.8in]{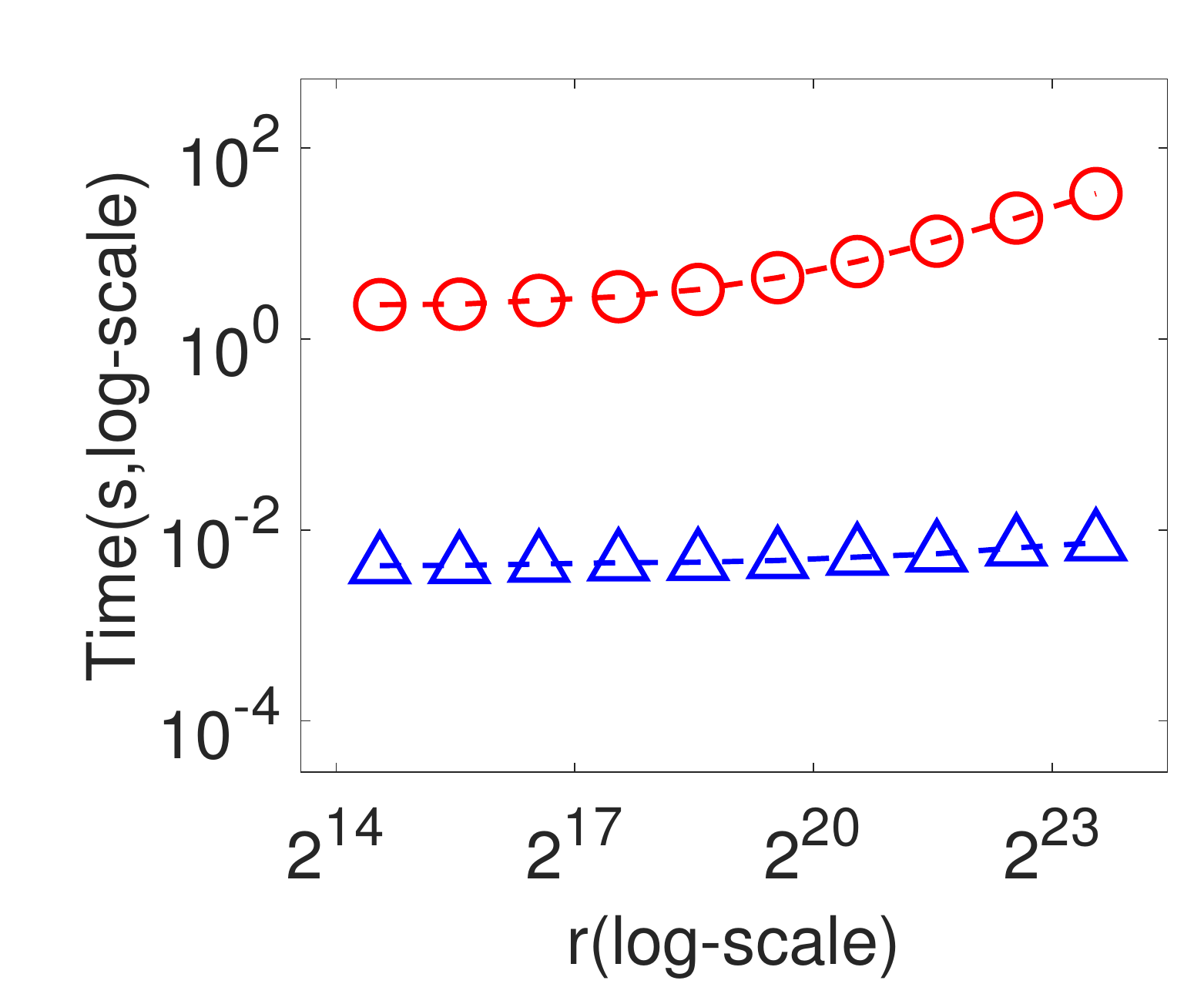}
  }
  \subfigure[Q1 on BC]{
    \label{fig:rt:m1:bt}
    \includegraphics[height=0.8in]{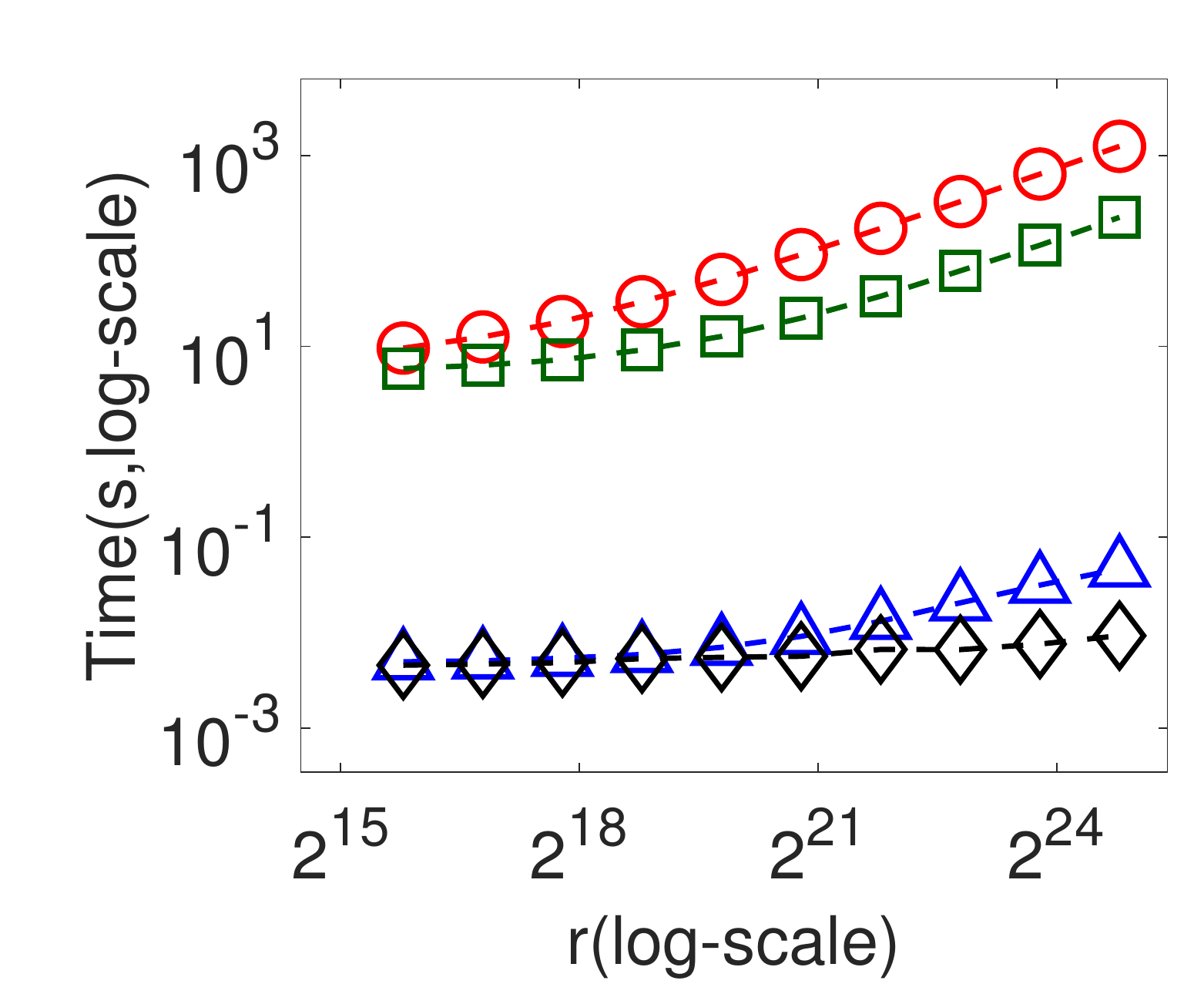}
  }
  \hfill
  \subfigure[Q2 on BC]{
    \label{fig:rt:m2:bt}
    \includegraphics[height=0.8in]{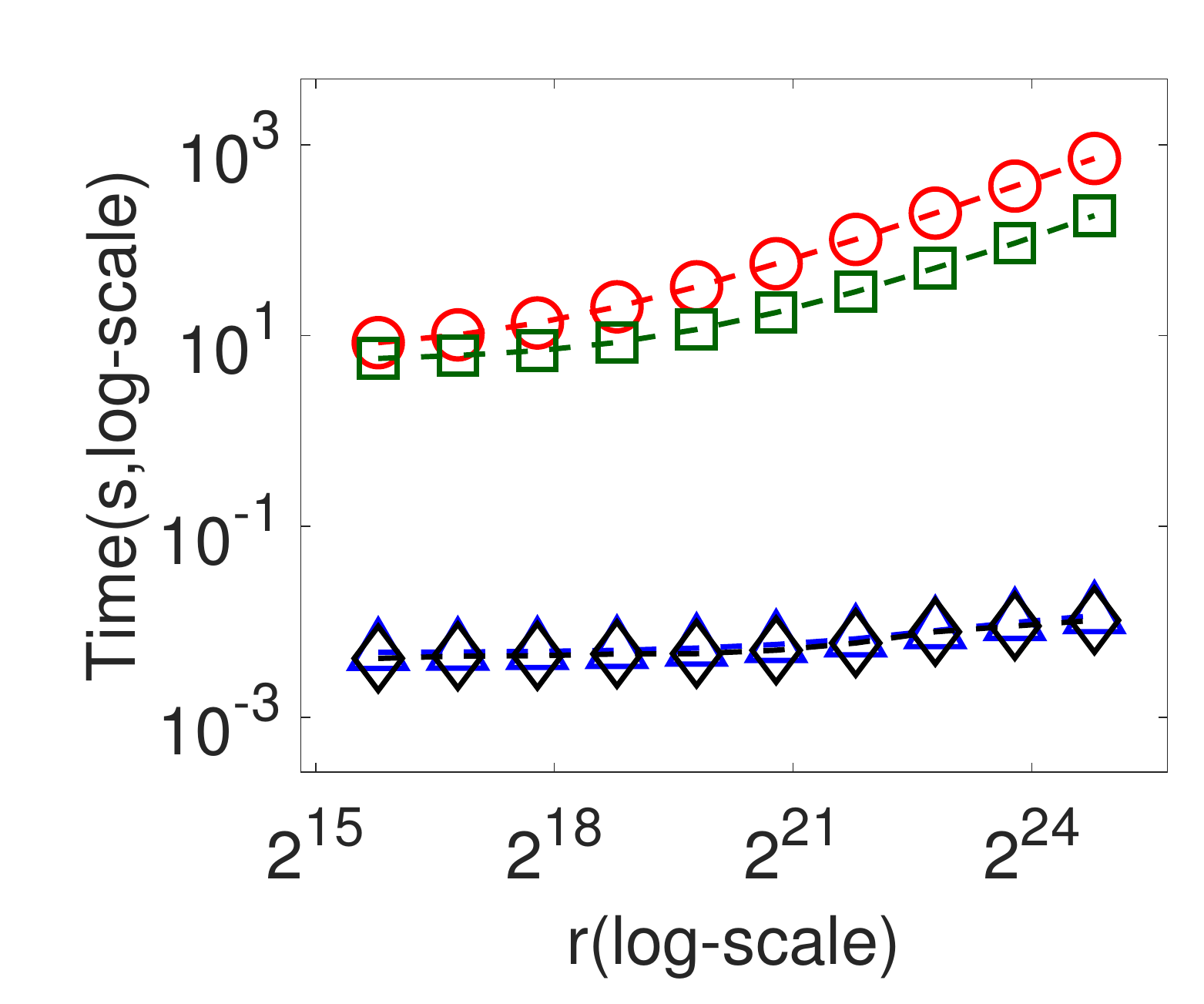}
  }
  \hfill
  \subfigure[Q3 on BC]{
    \label{fig:rt:m3:bt}
    \includegraphics[height=0.8in]{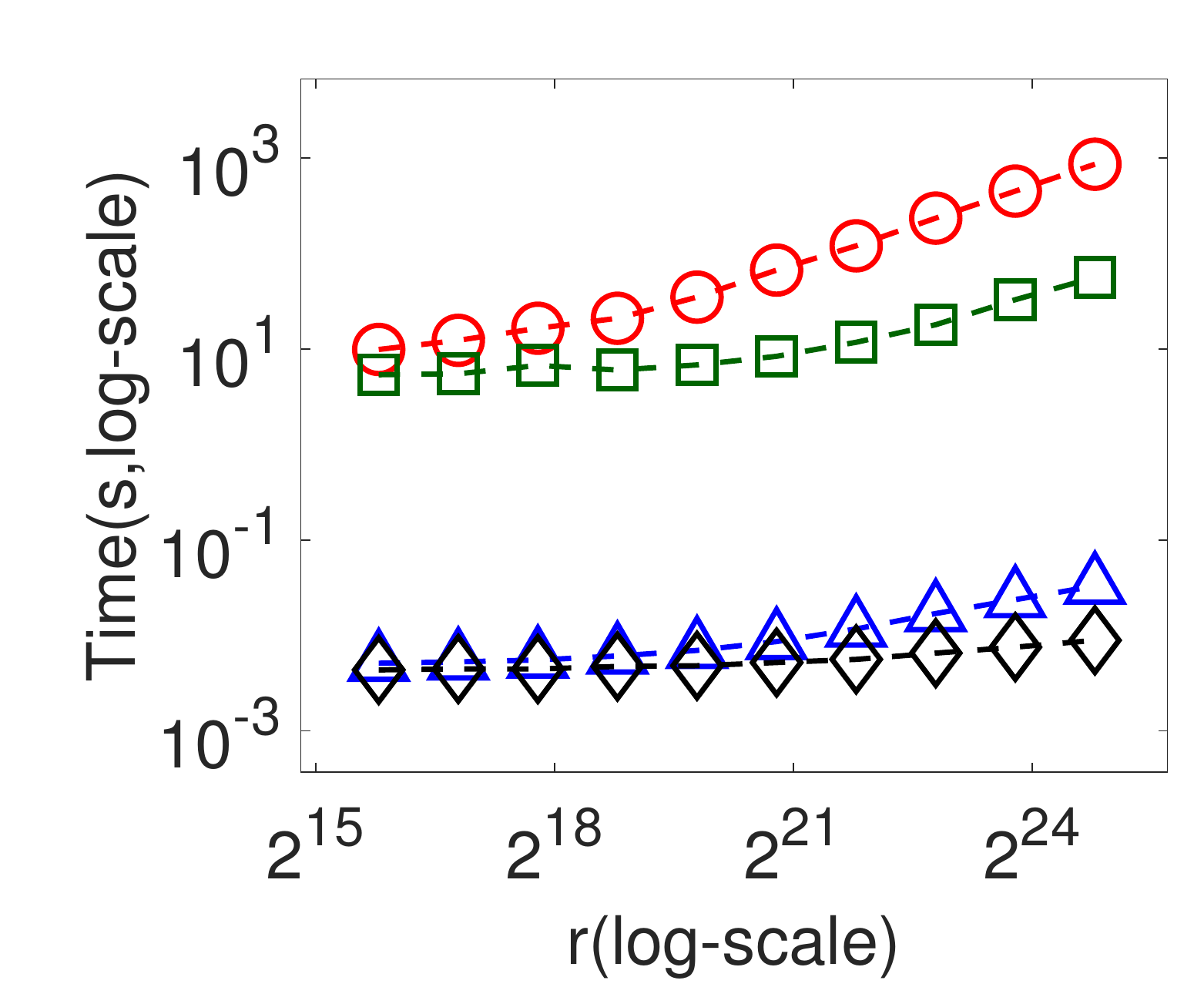}
  }
  \hfill
  \subfigure[Q4 on BC]{
    \label{fig:rt:m4:bt}
    \includegraphics[height=0.8in]{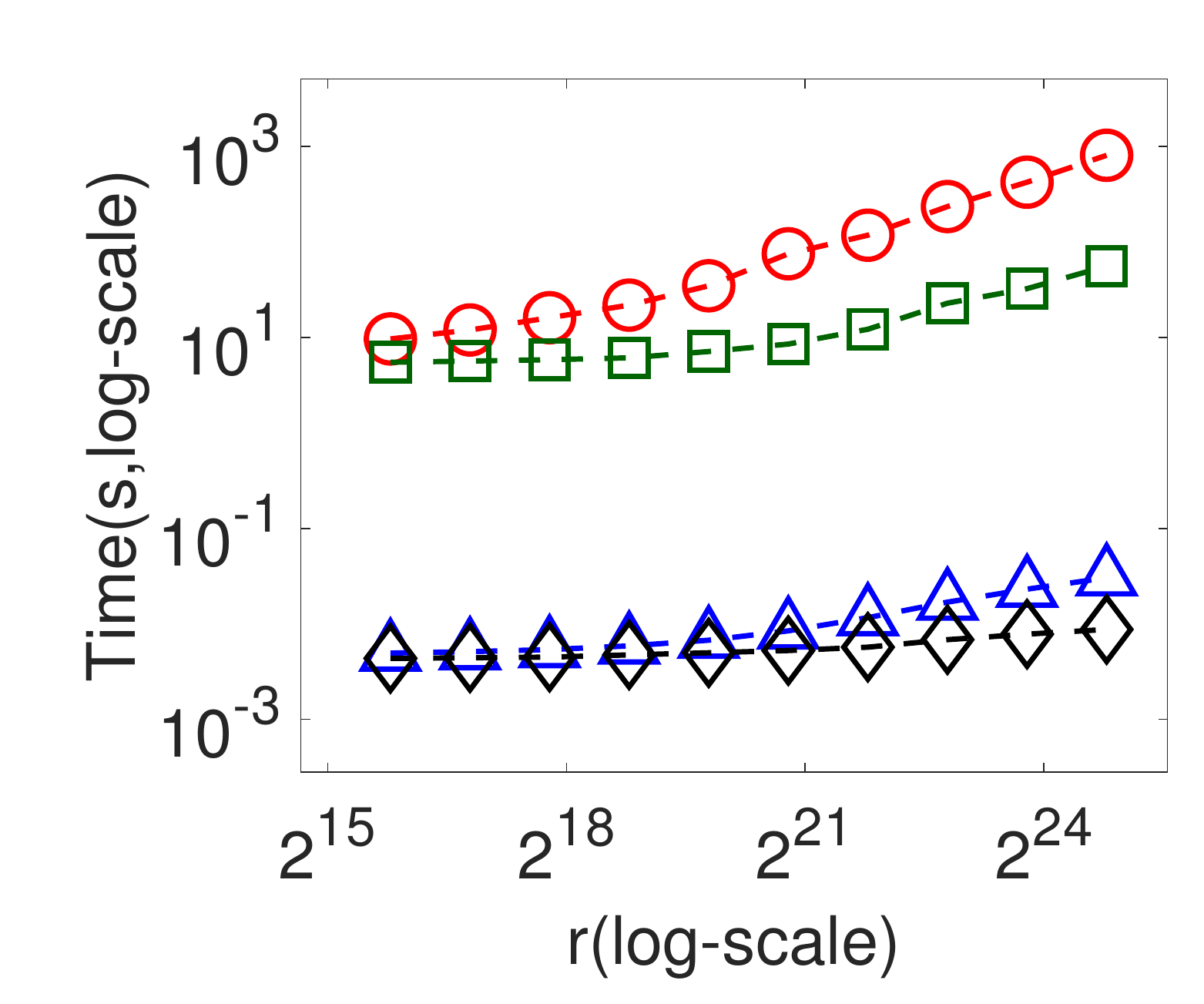}
  }
  \hfill
  \subfigure[Q5 on BC]{
    \label{fig:rt:m5:bt}
    \includegraphics[height=0.8in]{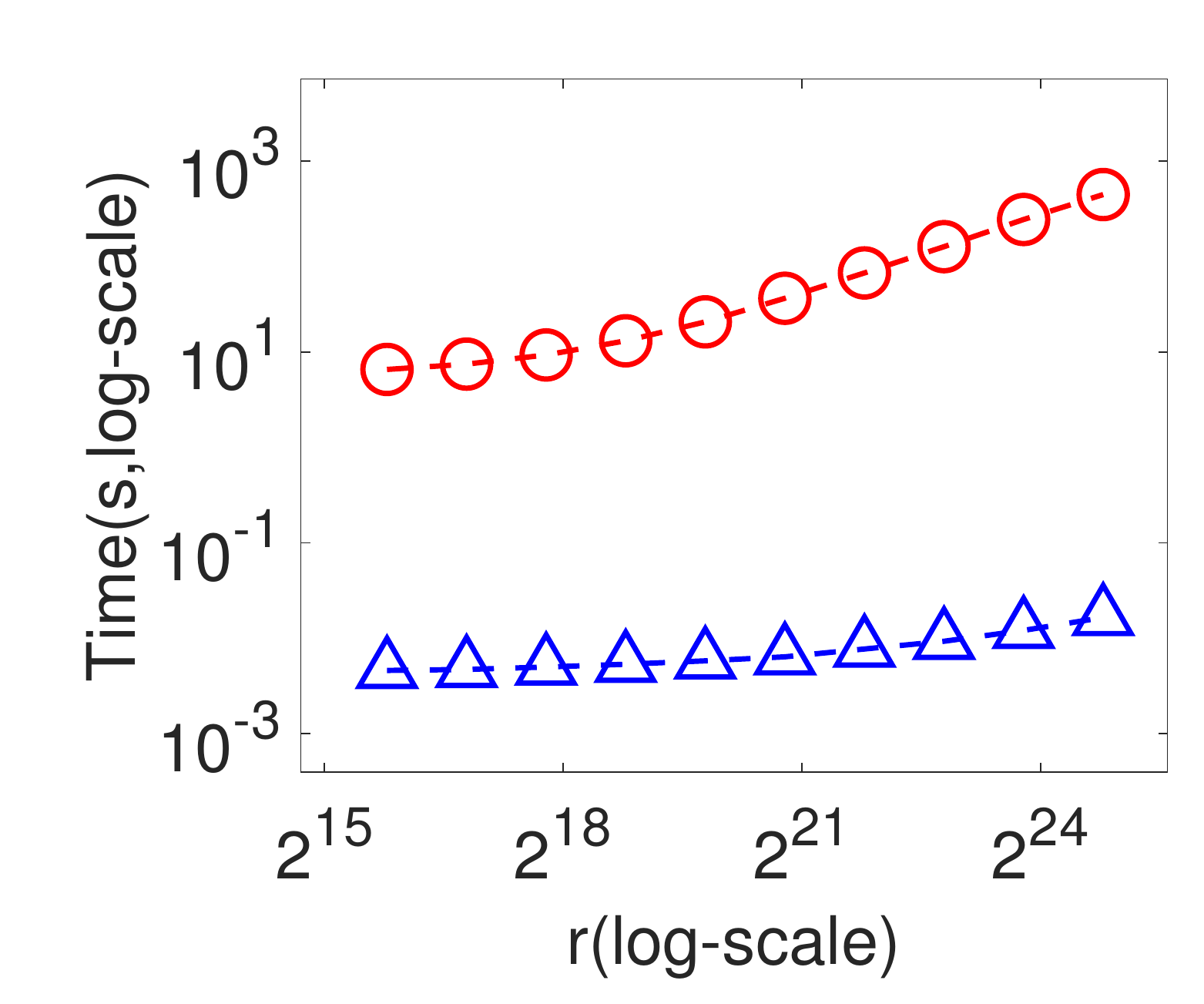}
  }
  \subfigure[Q1 on RC]{
    \label{fig:rt:m1:rc}
    \includegraphics[height=0.8in]{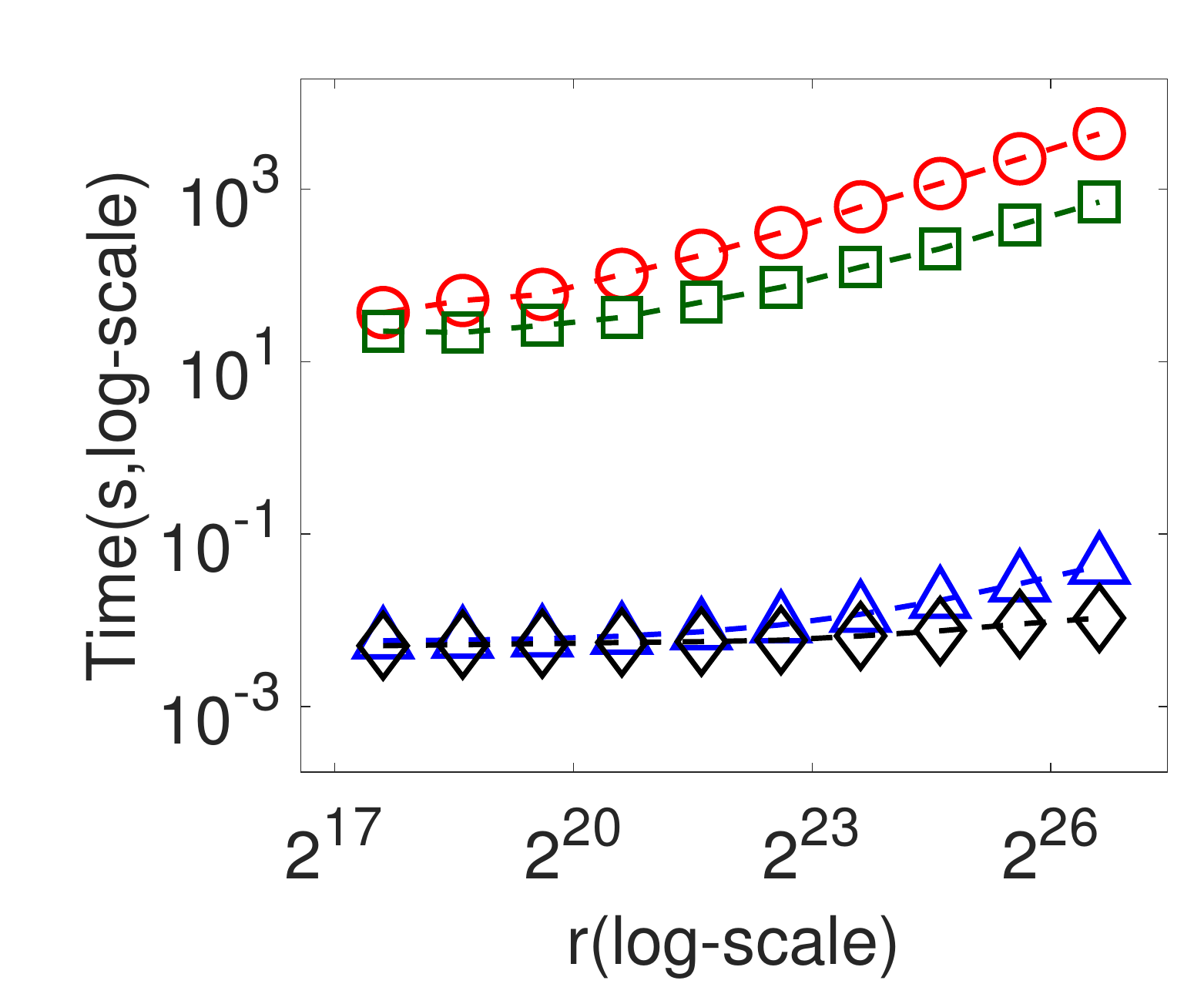}
  }
  \hfill
  \subfigure[Q2 on RC]{
    \label{fig:rt:m2:rc}
    \includegraphics[height=0.8in]{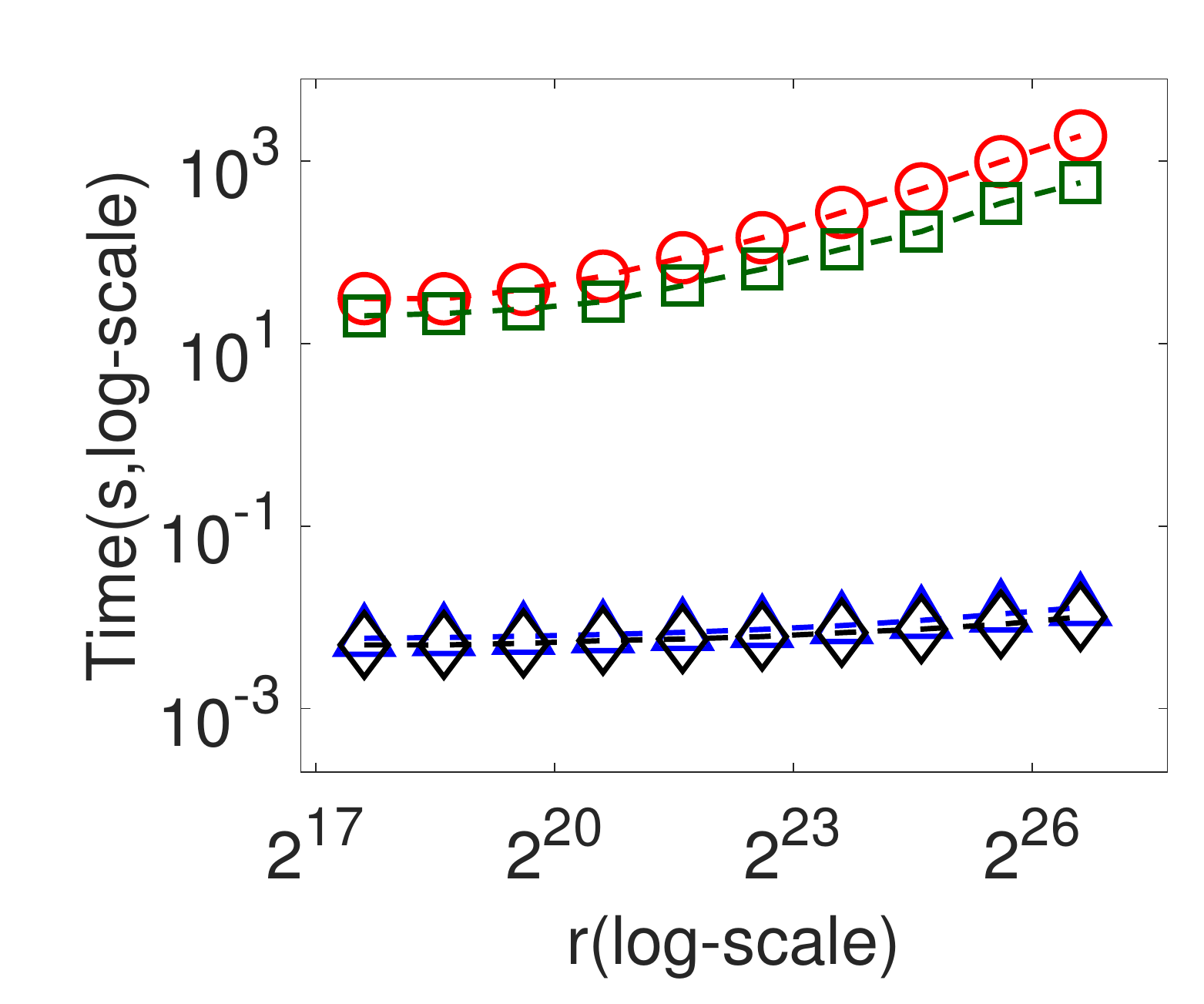}
  }
  \hfill
  \subfigure[Q3 on RC]{
    \label{fig:rt:m3:rc}
    \includegraphics[height=0.8in]{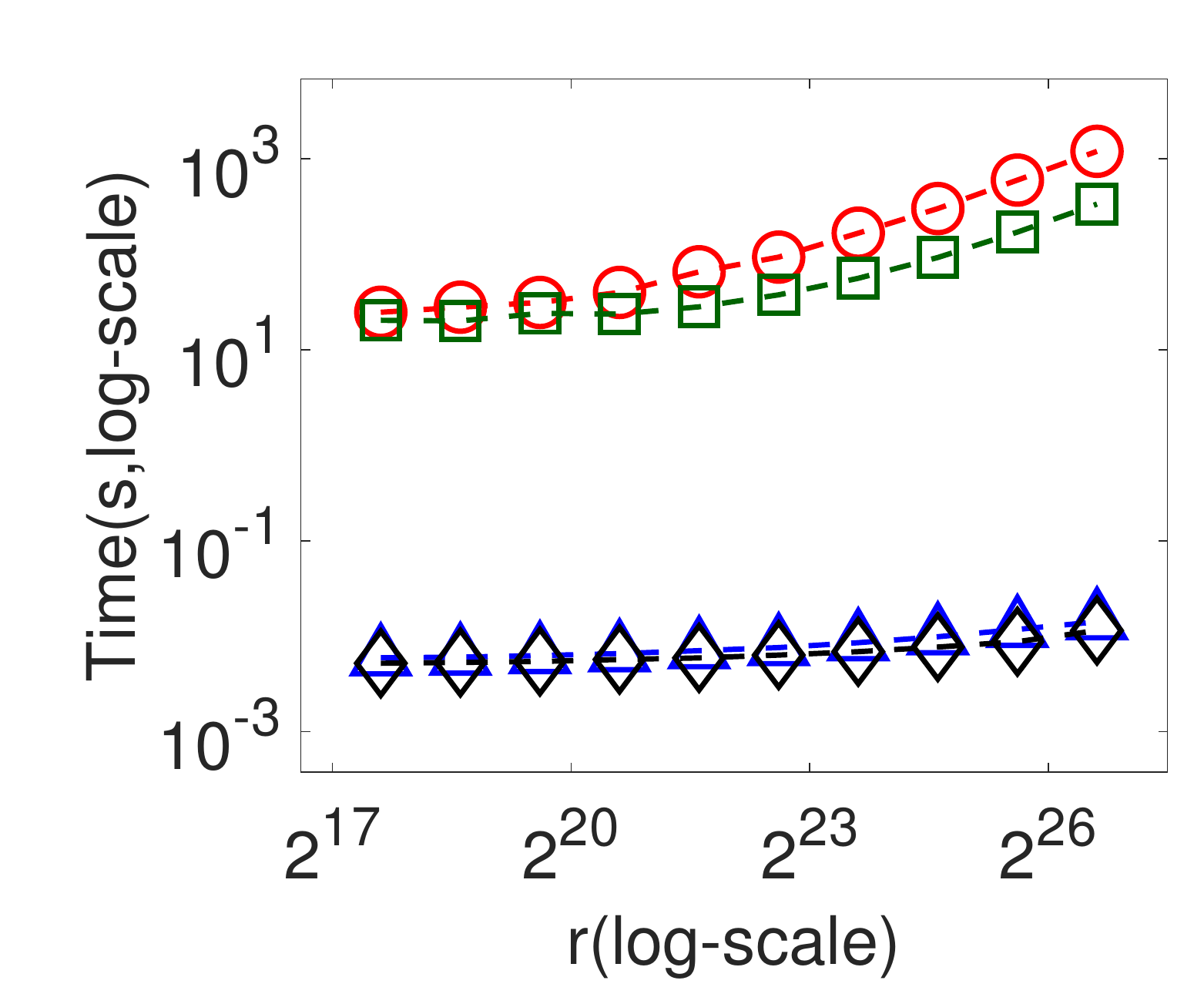}
  }
  \hfill
  \subfigure[Q4 on RC]{
    \label{fig:rt:m4:rc}
    \includegraphics[height=0.8in]{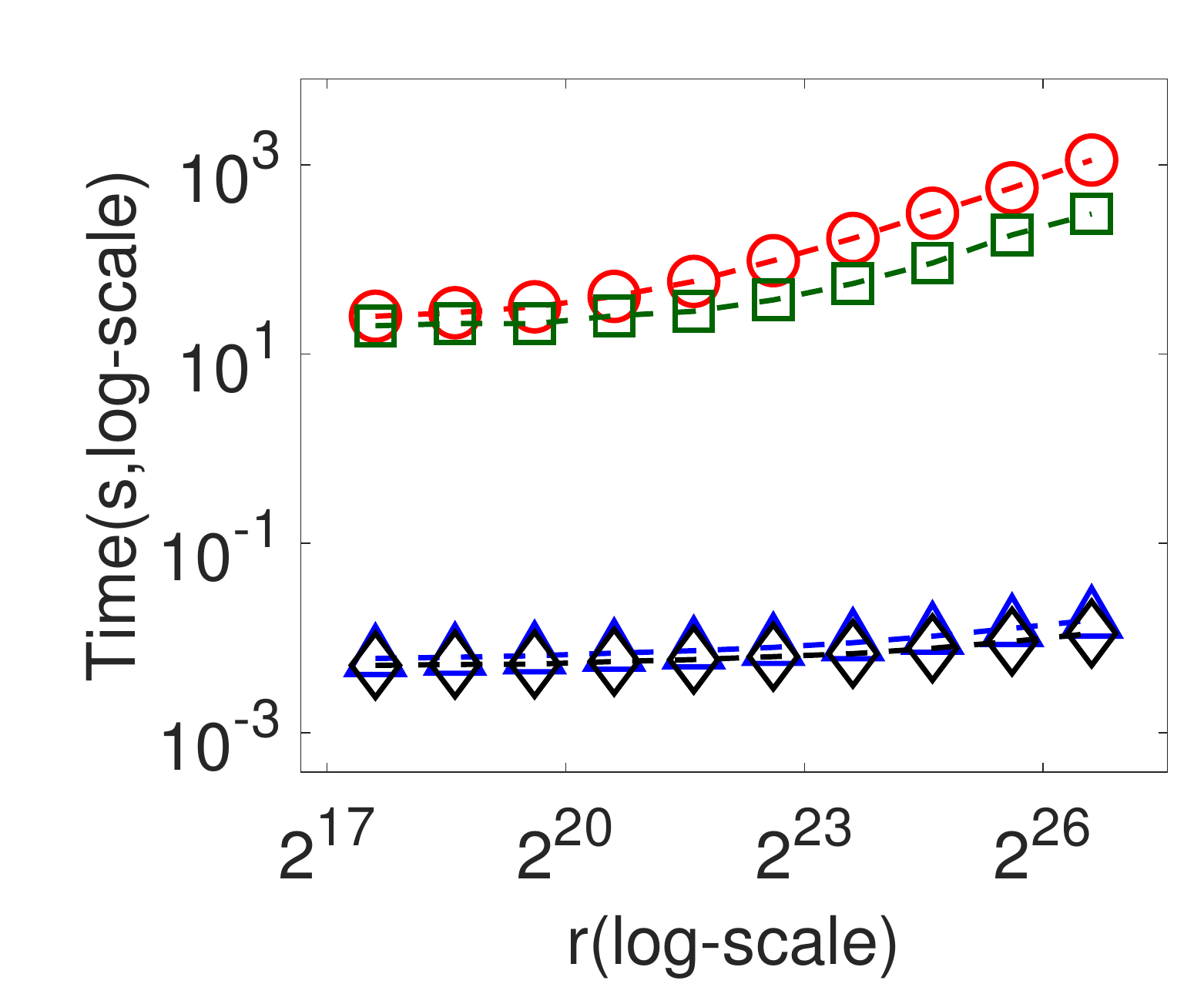}
  }
  \hfill
  \subfigure[Q5 on RC]{
    \label{fig:rt:m5:rc}
    \includegraphics[height=0.8in]{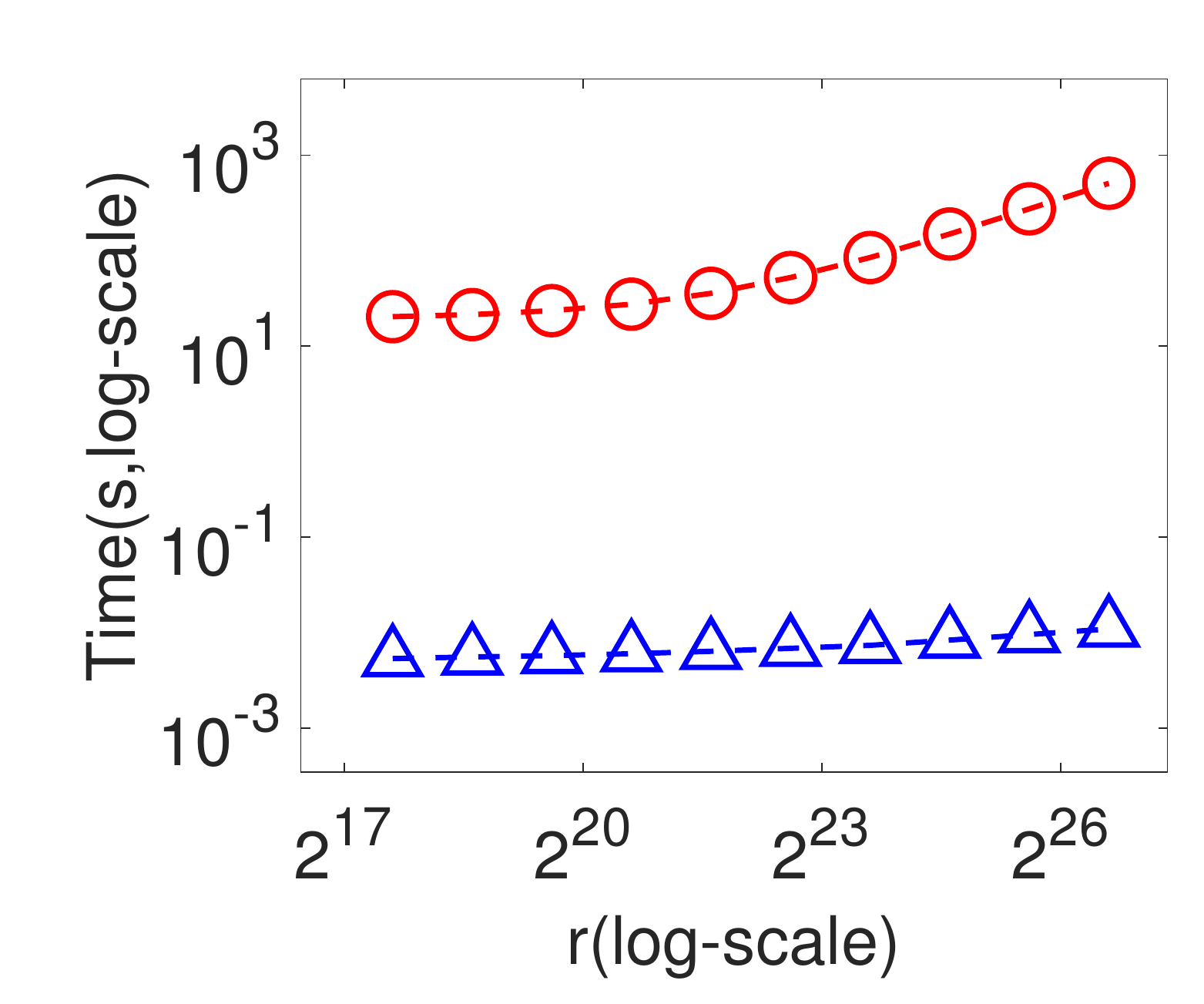}
  }
  \caption{Average update time (in seconds) per 1,000 edges with varying sample size $r$ in the streaming setting.}
  \label{fig:r:time}
  \Description{Time-r}
\end{figure}

\begin{figure}[t]
  \centering
  \subfigure[Q1 on BC]{
    \label{fig:T:m1:bt}
    \includegraphics[height=0.9in]{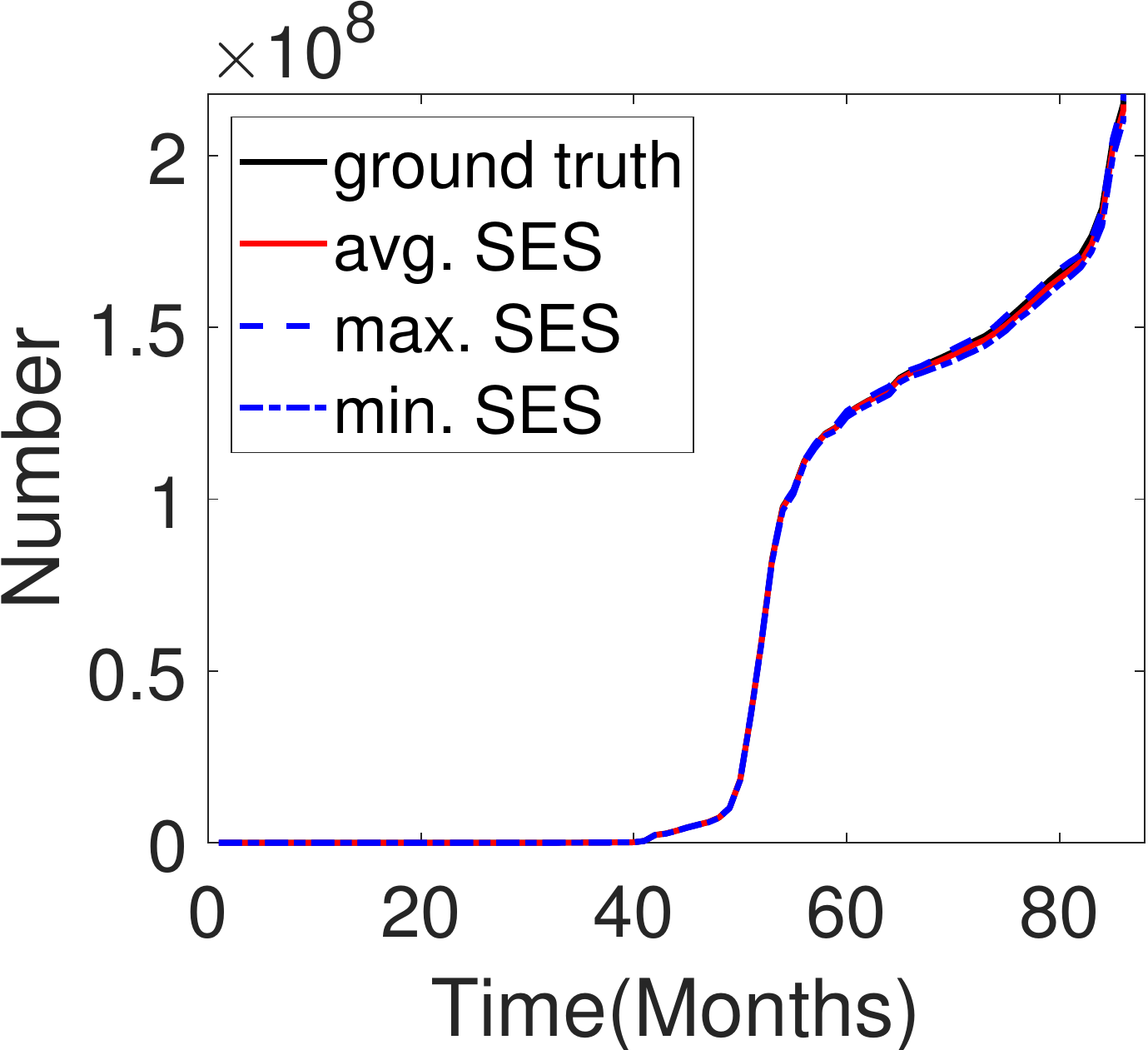}
    \hspace{0.5em}
    \includegraphics[height=0.9in]{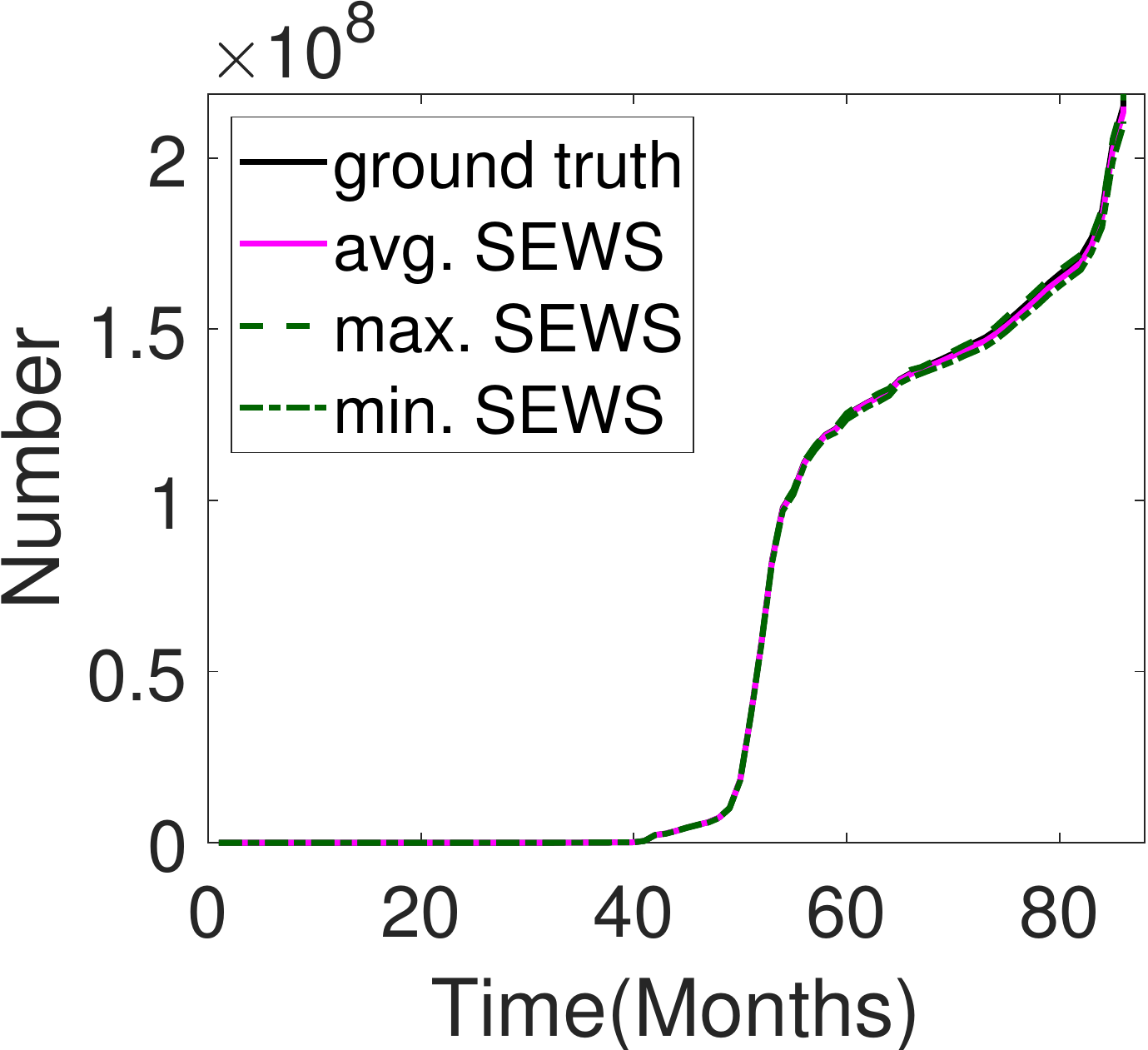}
  }
  \hspace{1em}
  \subfigure[Q2 on BC]{
    \label{fig:T:m2:bt}
    \includegraphics[height=0.9in]{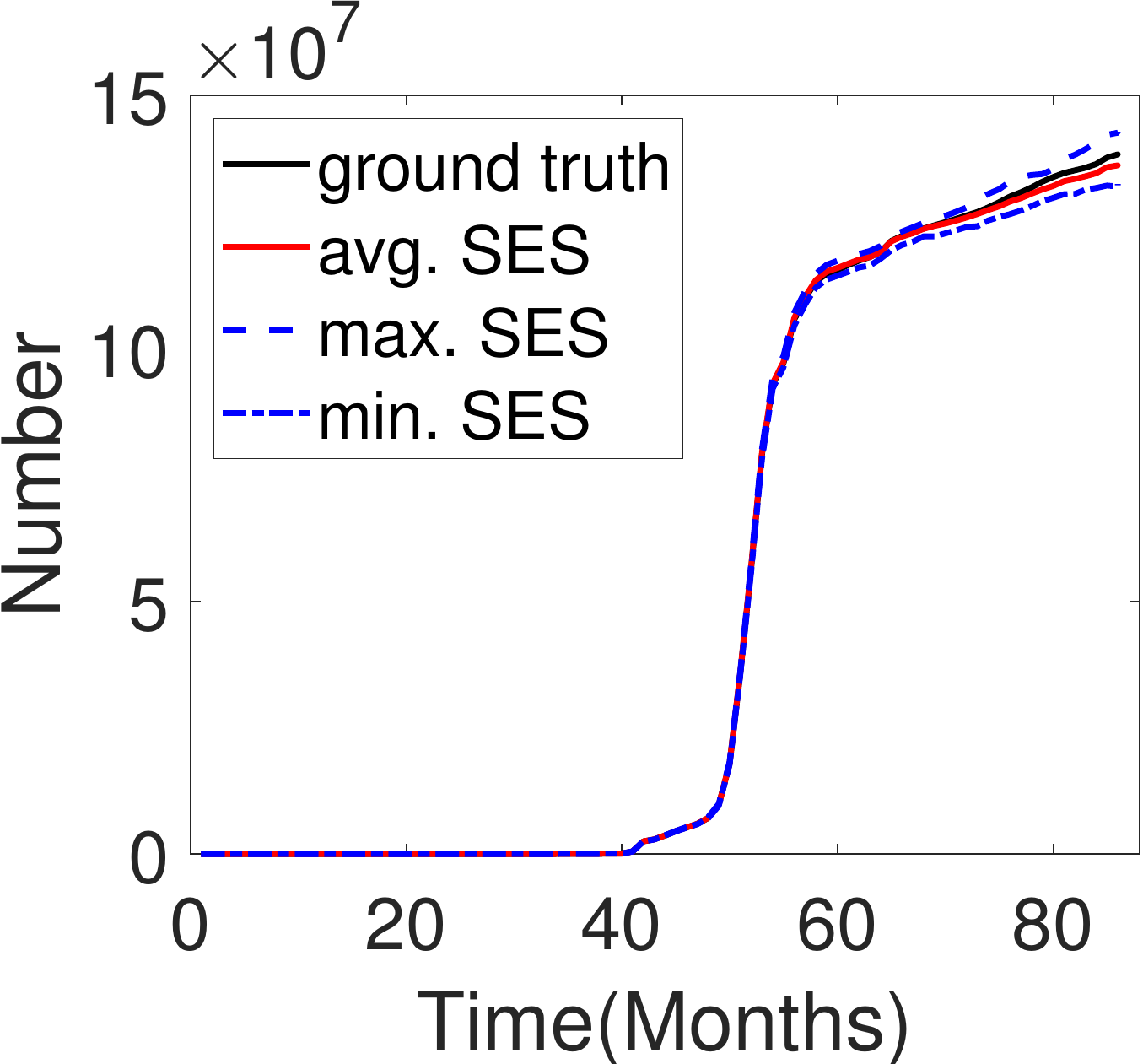}
    \hspace{0.5em}
    \includegraphics[height=0.9in]{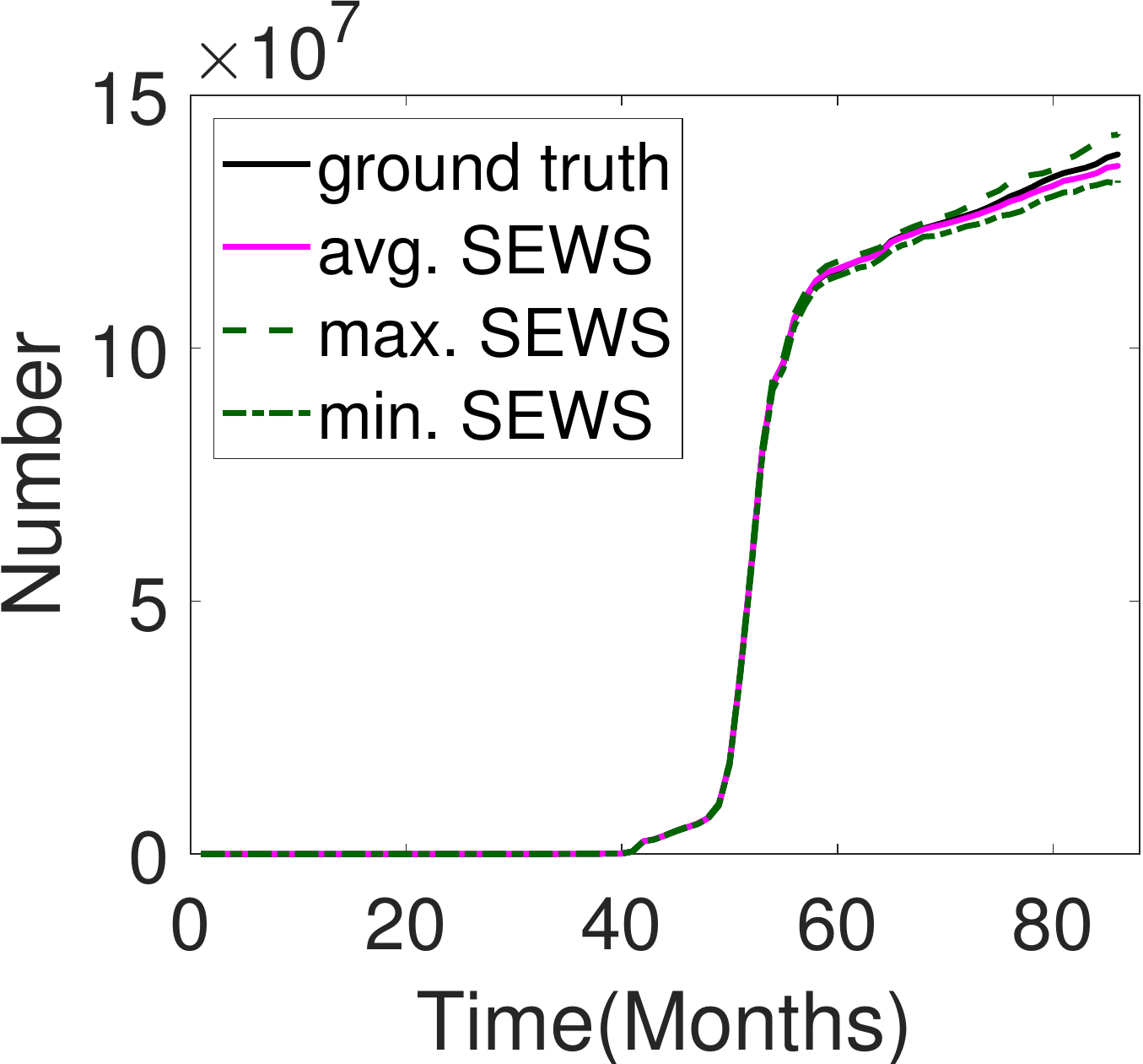}
  }
 \\
  \subfigure[Q3 on BC]{
    \label{fig:T:m3:bt}
    \includegraphics[height=0.9in]{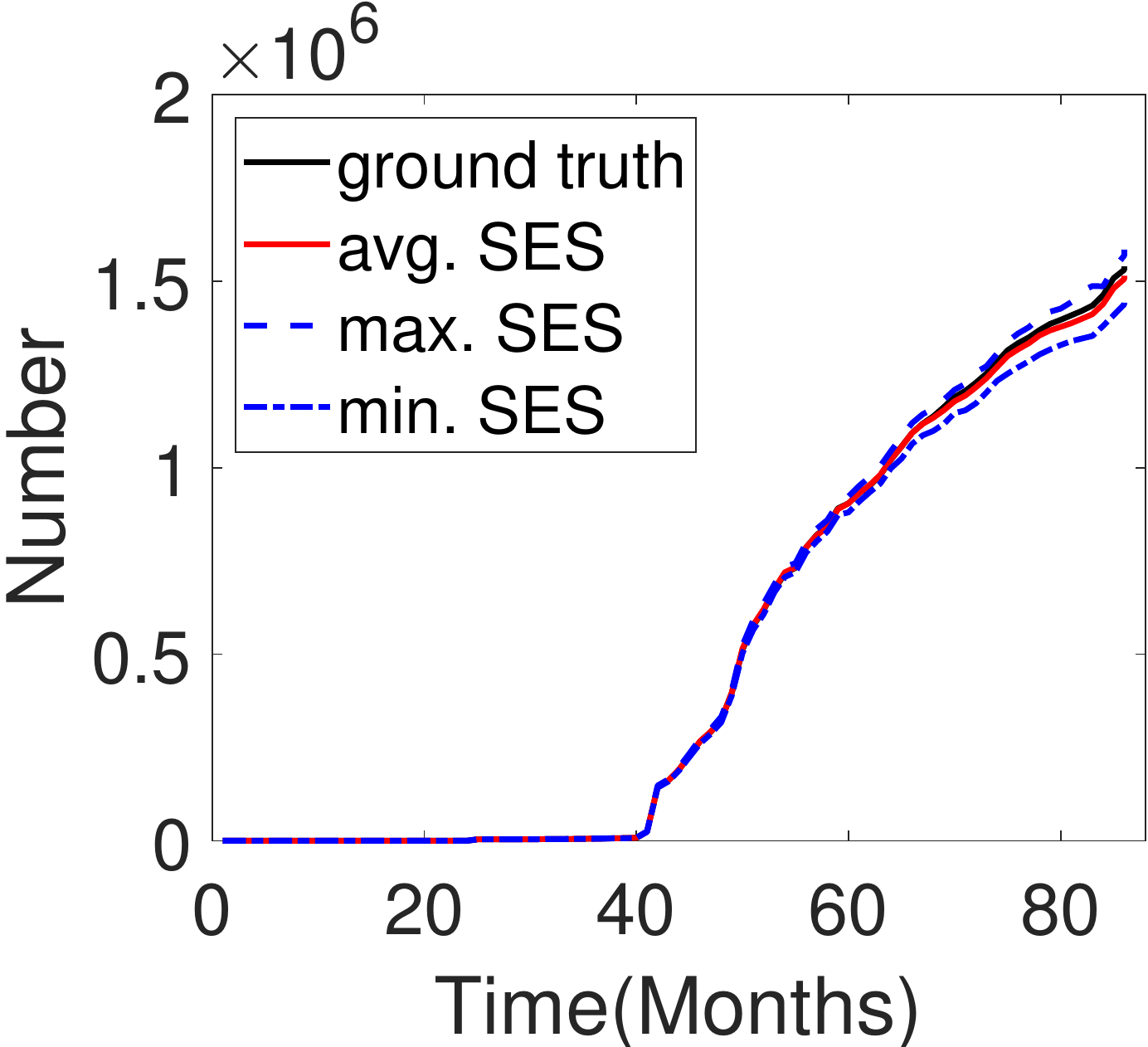}
    \hspace{0.5em}
    \includegraphics[height=0.9in]{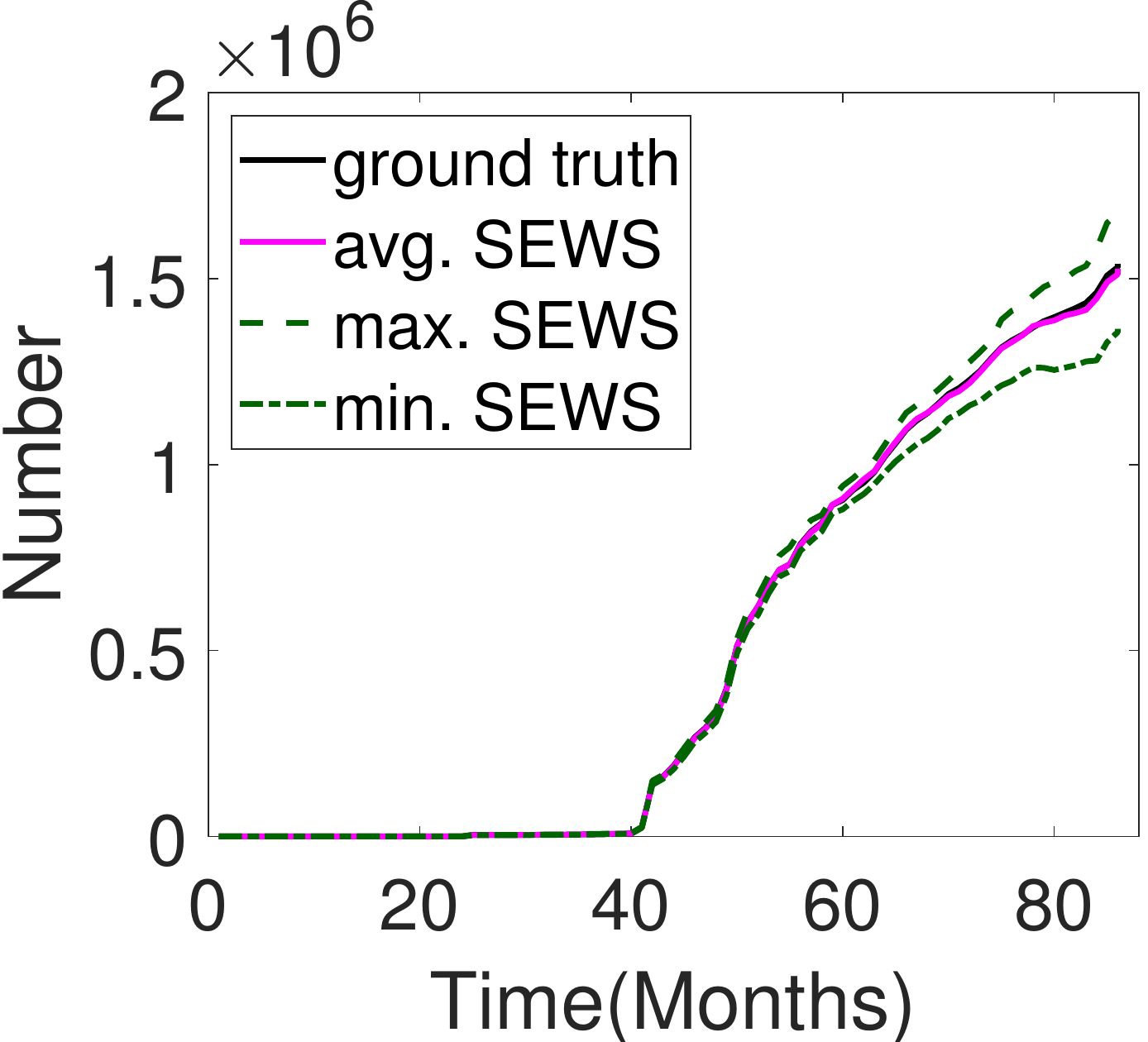}
  }
  \hspace{1em}
  \subfigure[Q4 on BC]{
    \label{fig:T:m4:bt}
    \includegraphics[height=0.9in]{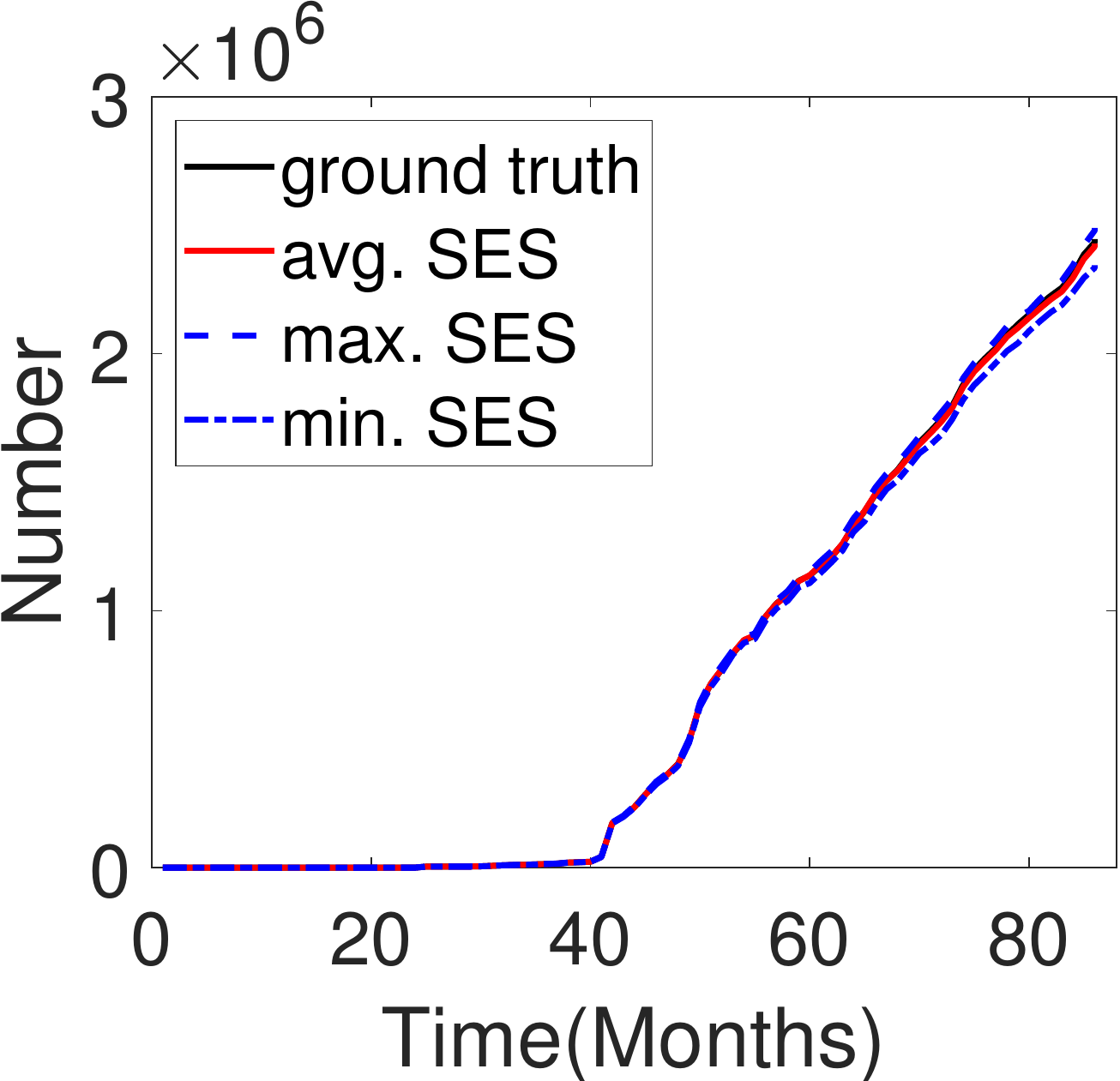}
    \hspace{0.5em}
    \includegraphics[height=0.9in]{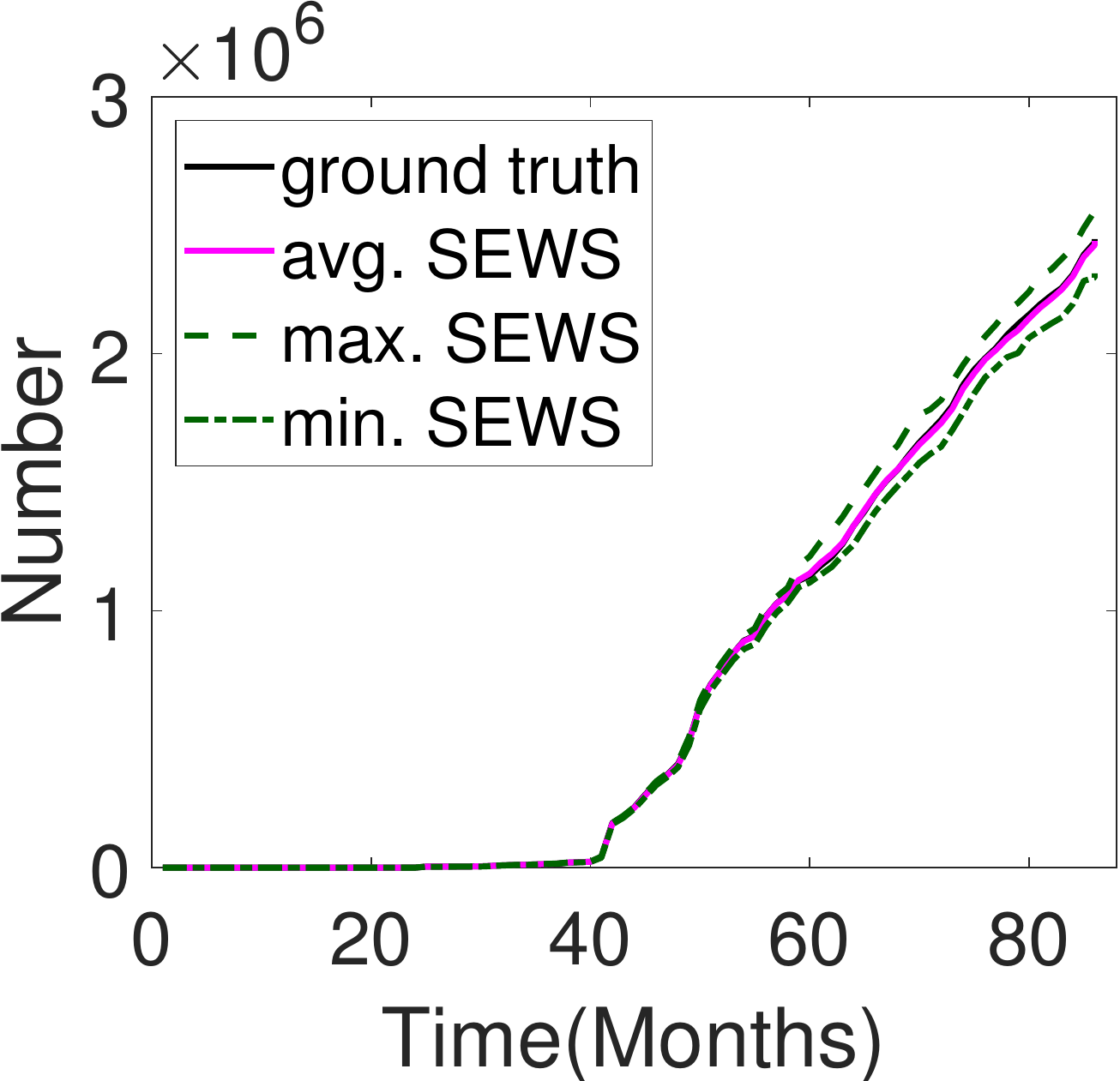}
  }
  \\
  \subfigure[Q1 on RC]{
    \label{fig:T:m1:rc}
    \includegraphics[height=0.9in]{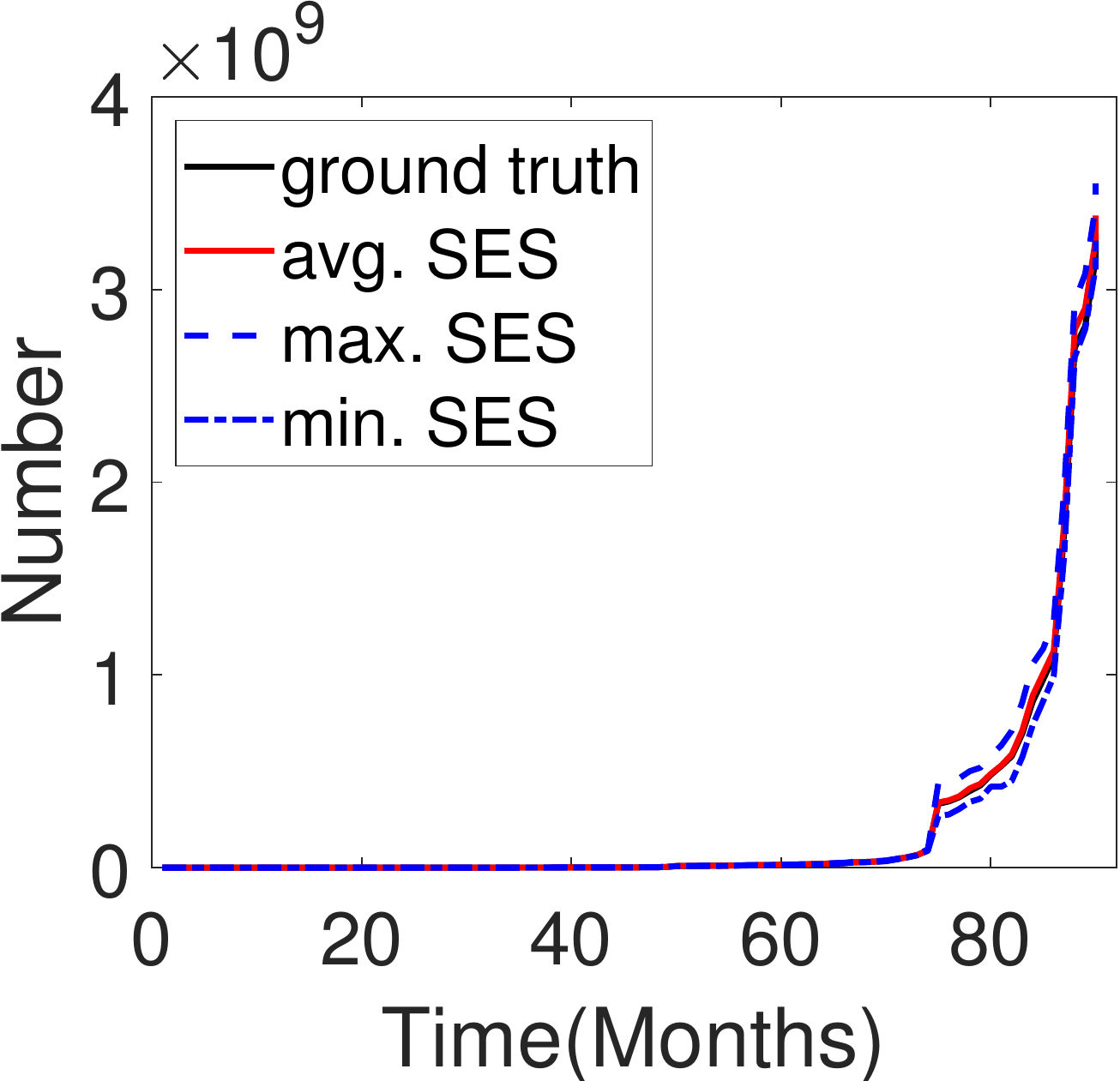}
    \hspace{0.5em}
    \includegraphics[height=0.9in]{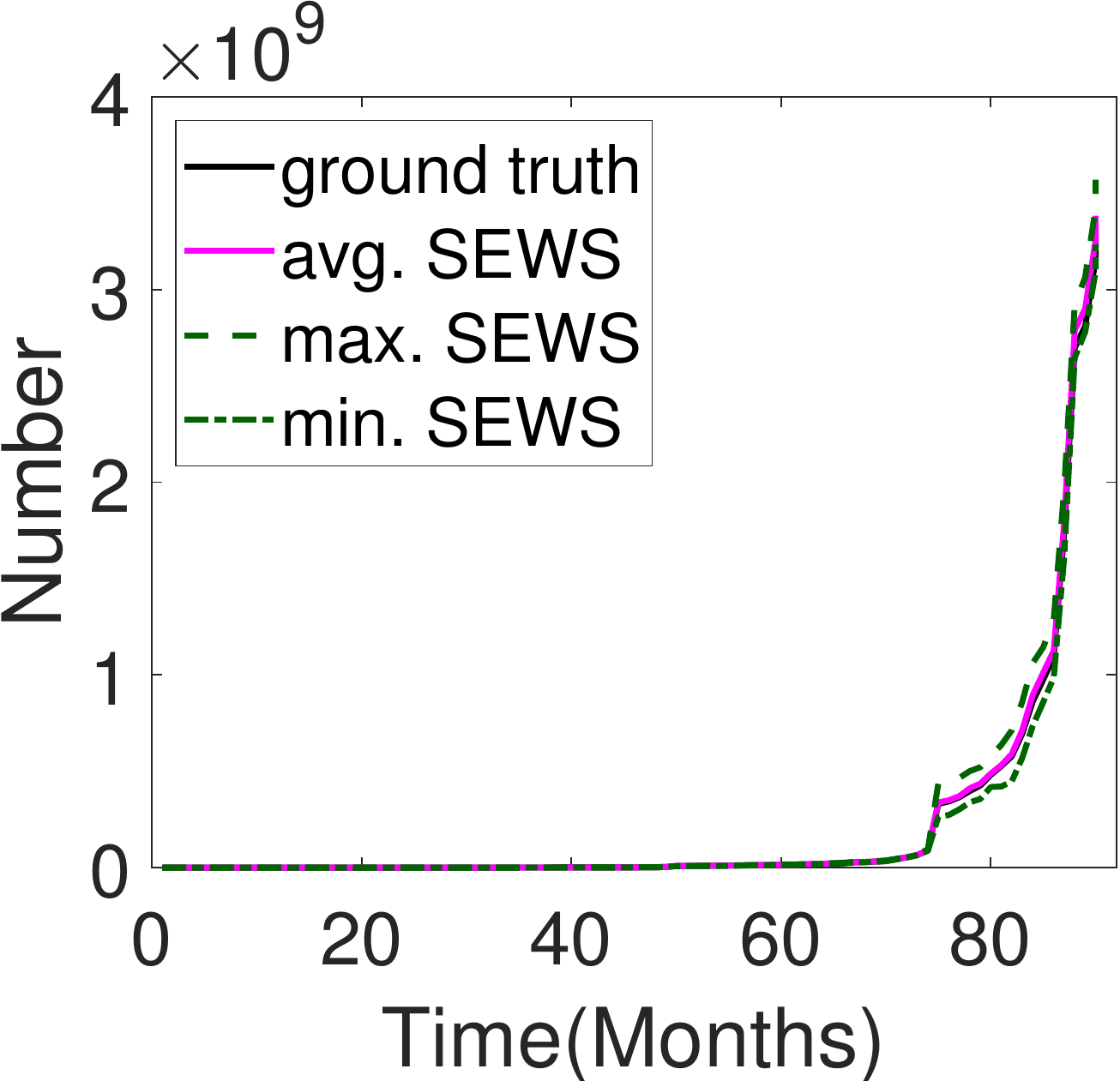}
  }
  \hspace{1em}
  \subfigure[Q2 on RC]{
	\label{fig:T:m2:rc}
	\includegraphics[height=0.9in]{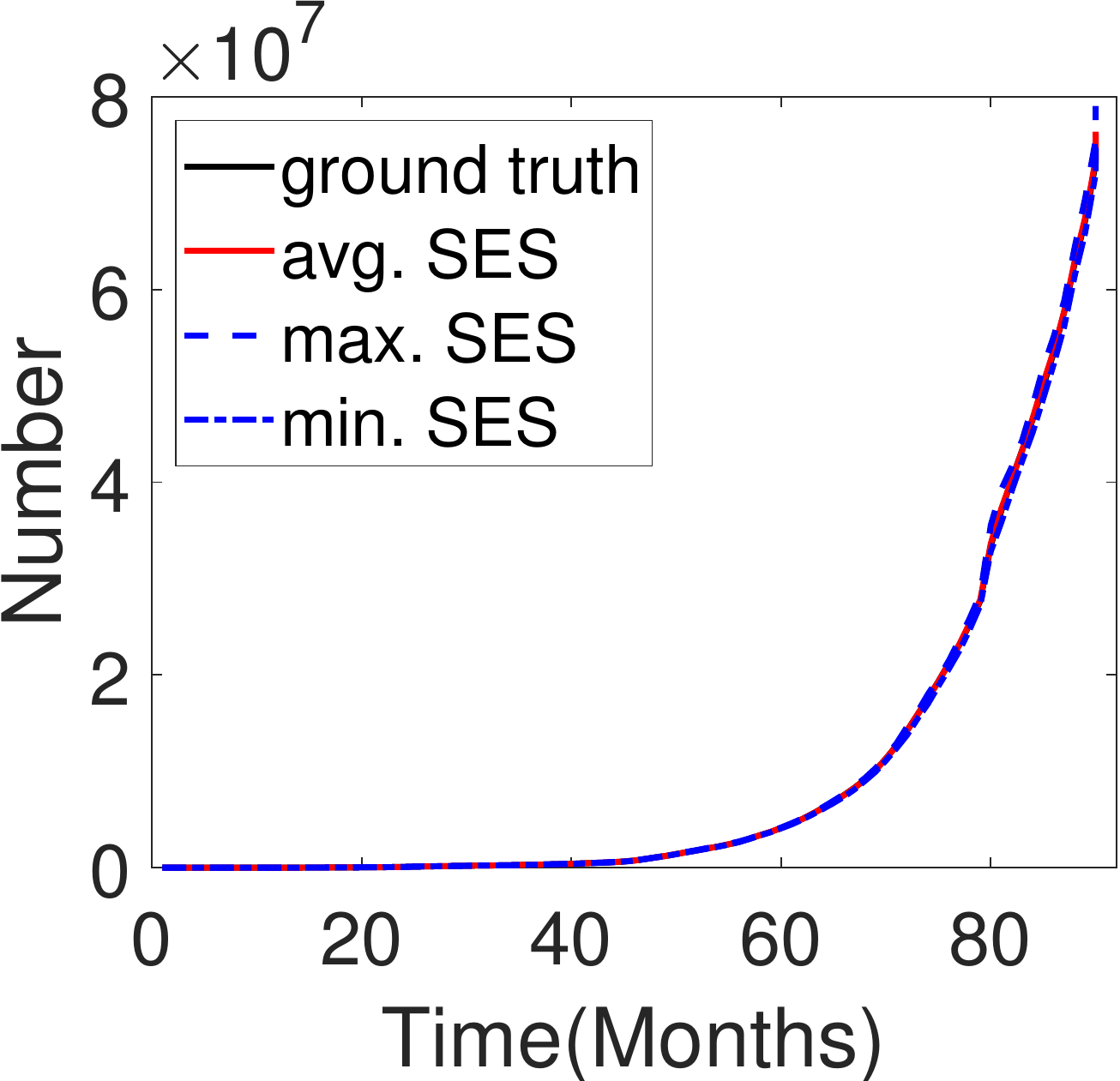}
  \hspace{0.5em}
	\includegraphics[height=0.9in]{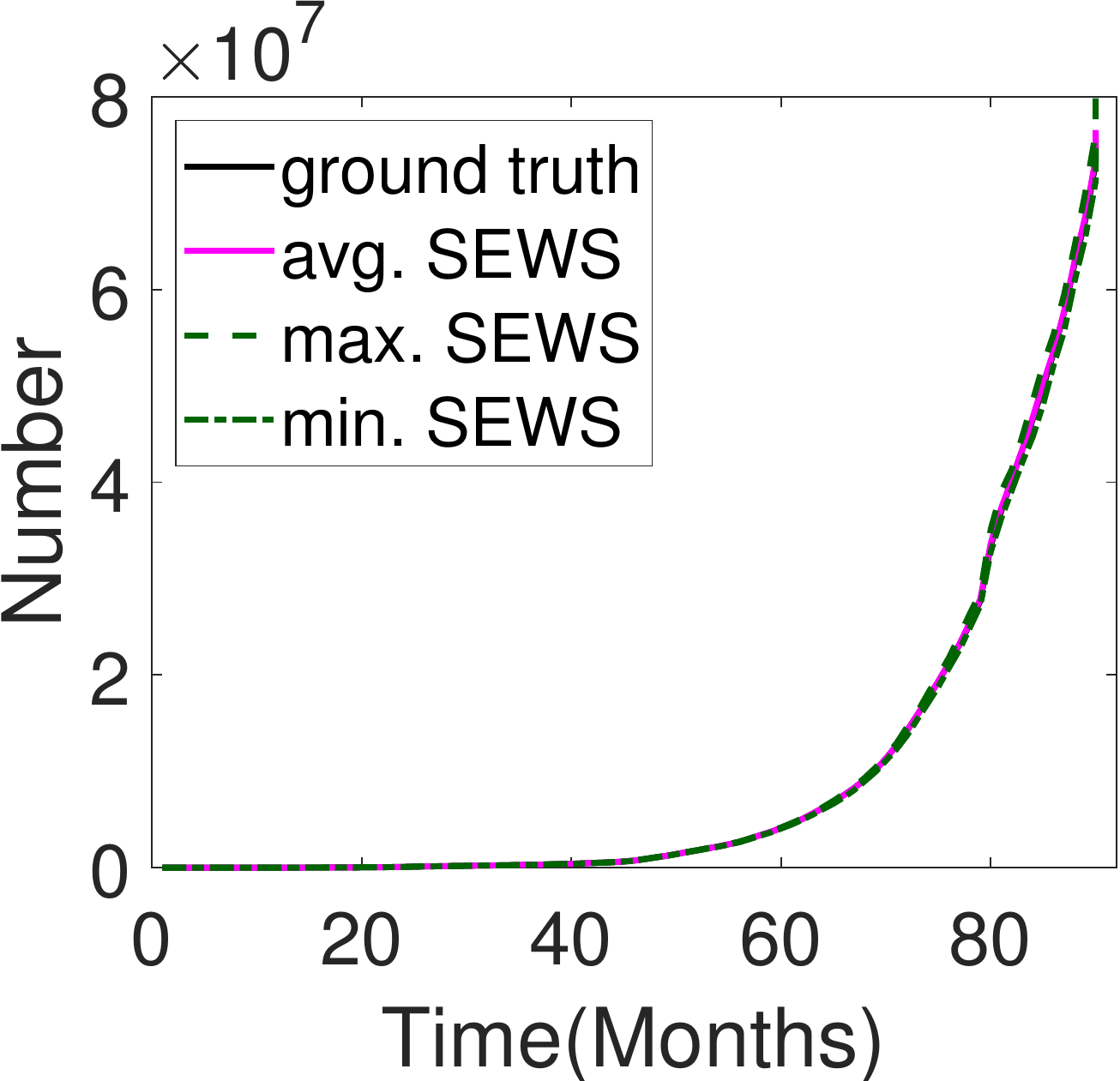}
  }
  \\
  \subfigure[Q3 on RC]{
    \label{fig:T:m3:rc}
    \includegraphics[height=0.9in]{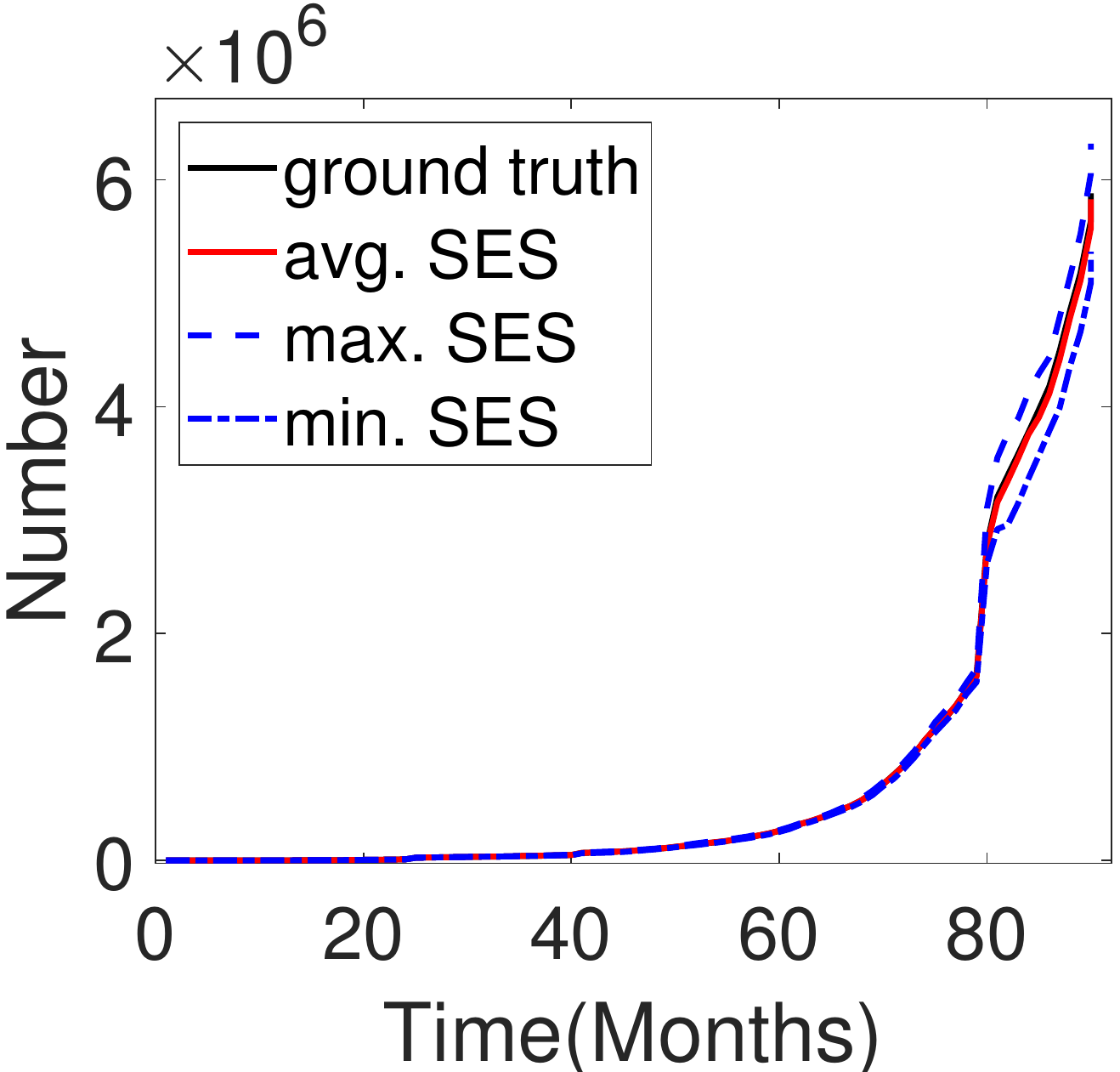}
    \hspace{0.5em}
    \includegraphics[height=0.9in]{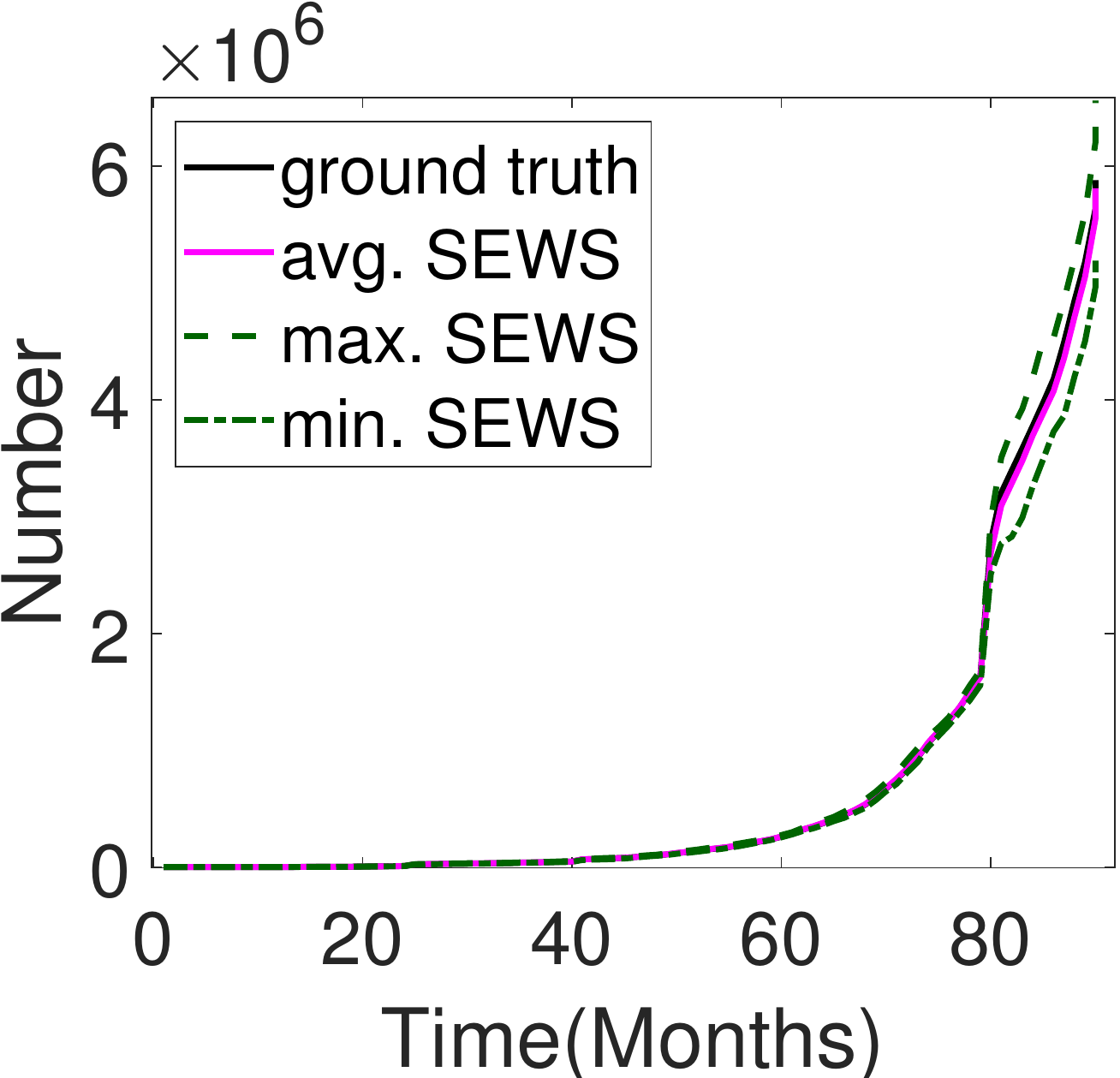}
  }
  \hspace{1em}
   \subfigure[Q4 on RC]{
  	\label{fig:T:m4:rc}
  	\includegraphics[height=0.9in]{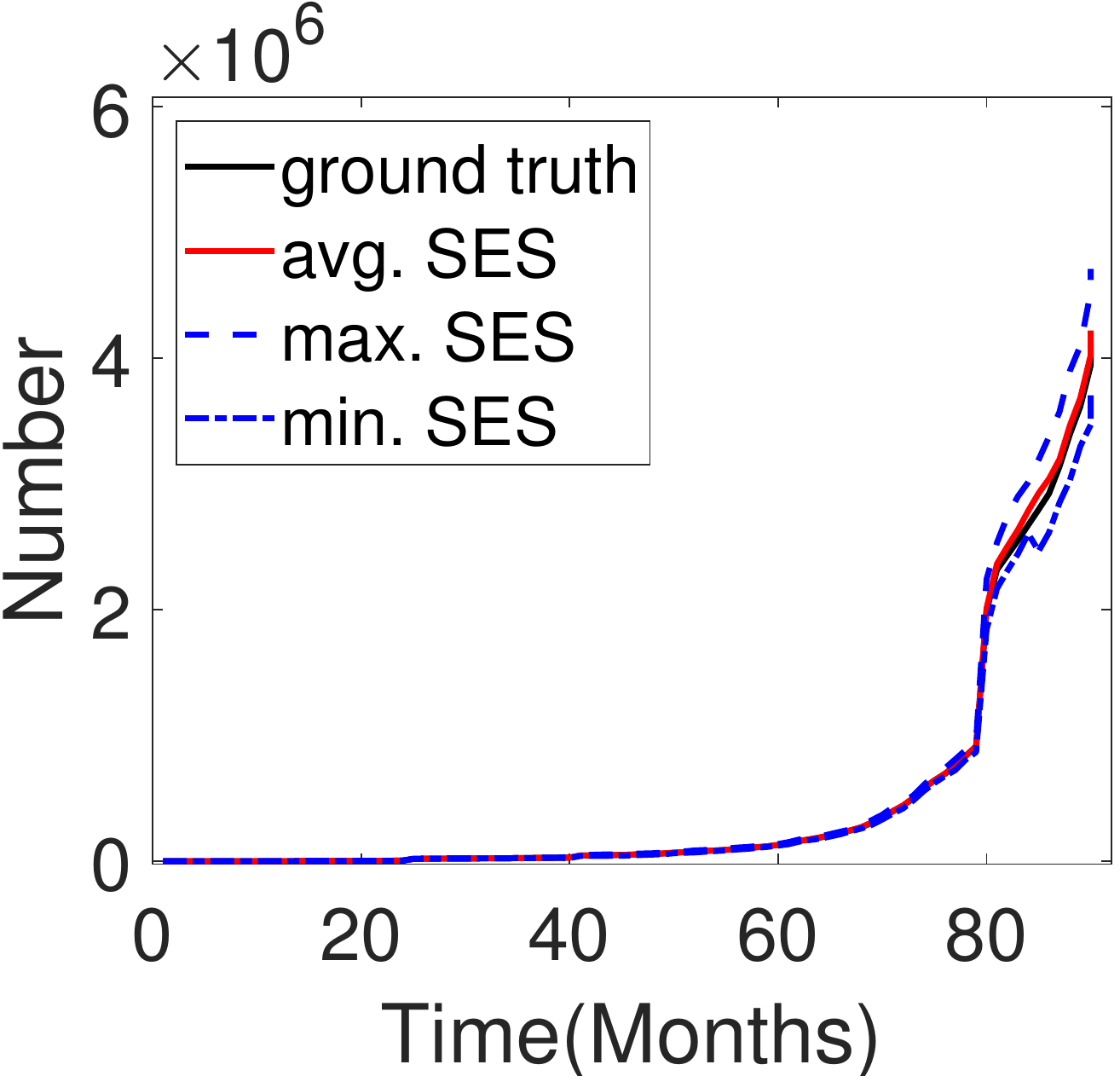}
    \hspace{0.5em}
  	\includegraphics[height=0.9in]{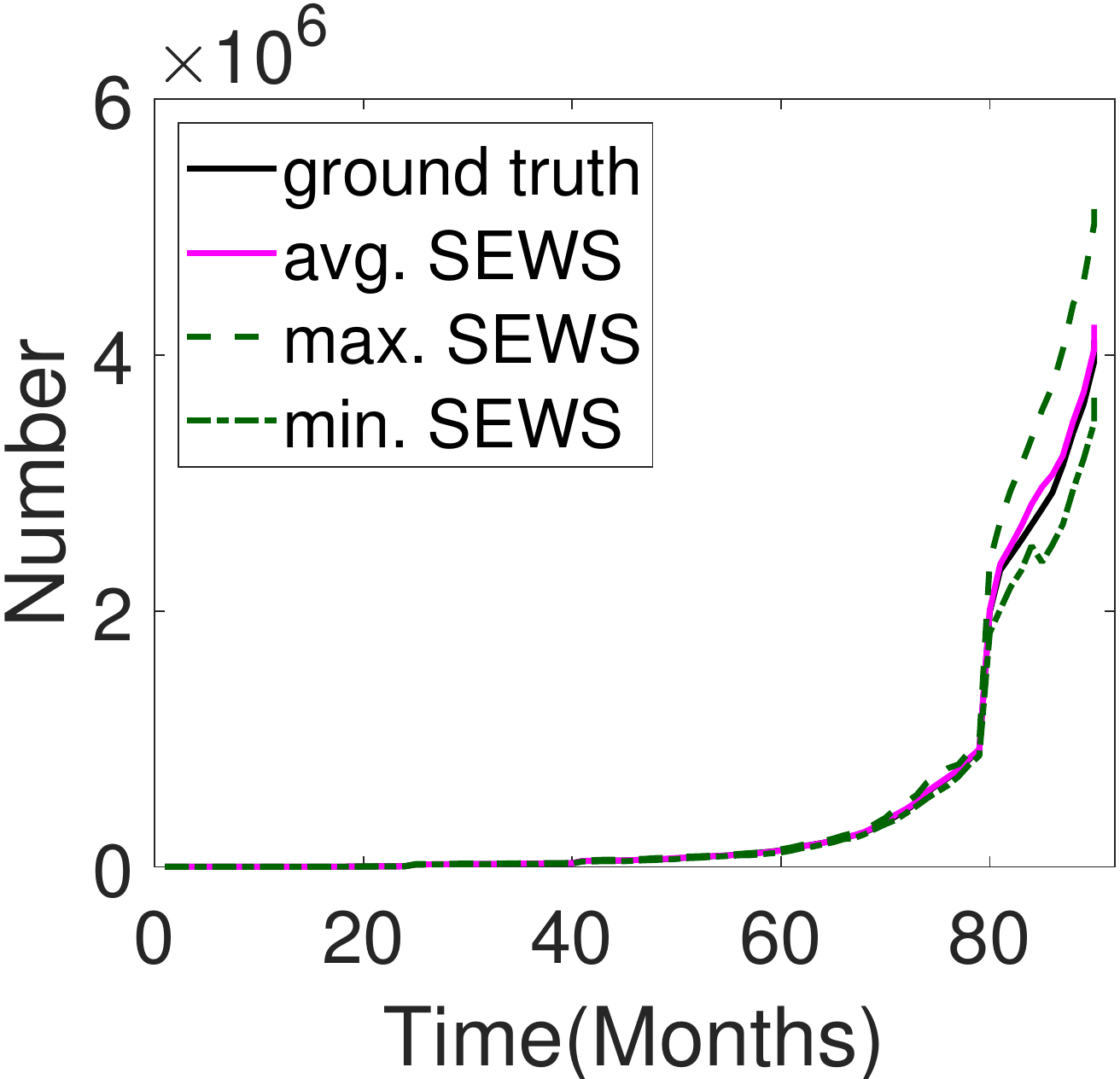}
  }
  \caption{Estimation for the number of instances over time (in months).}
  \label{fig:T:number}
  \Description{Time1}
\end{figure}

\begin{figure}
  \centering
  \includegraphics[height=0.9in]{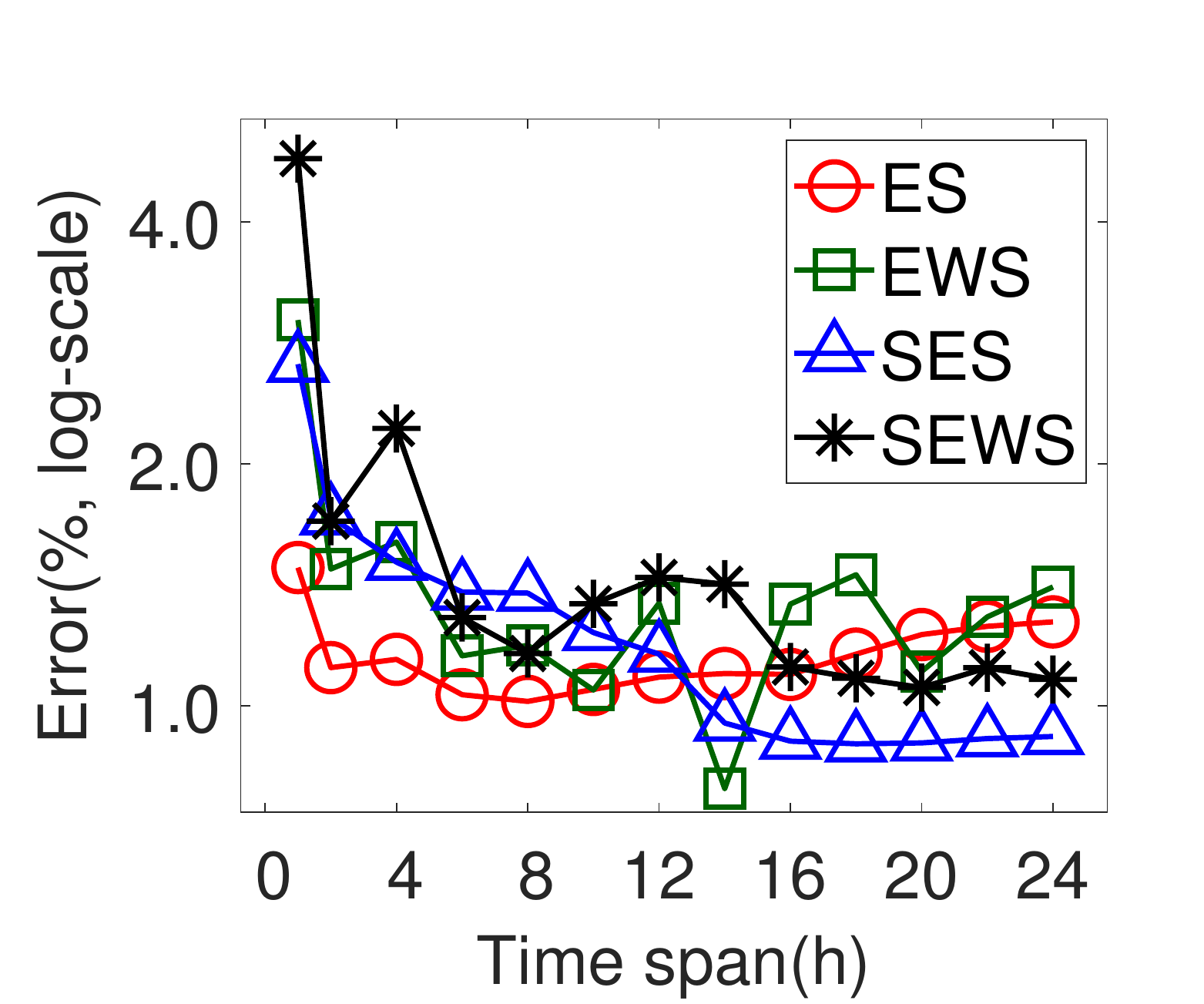}
  \hspace{1em}
  \includegraphics[height=0.9in]{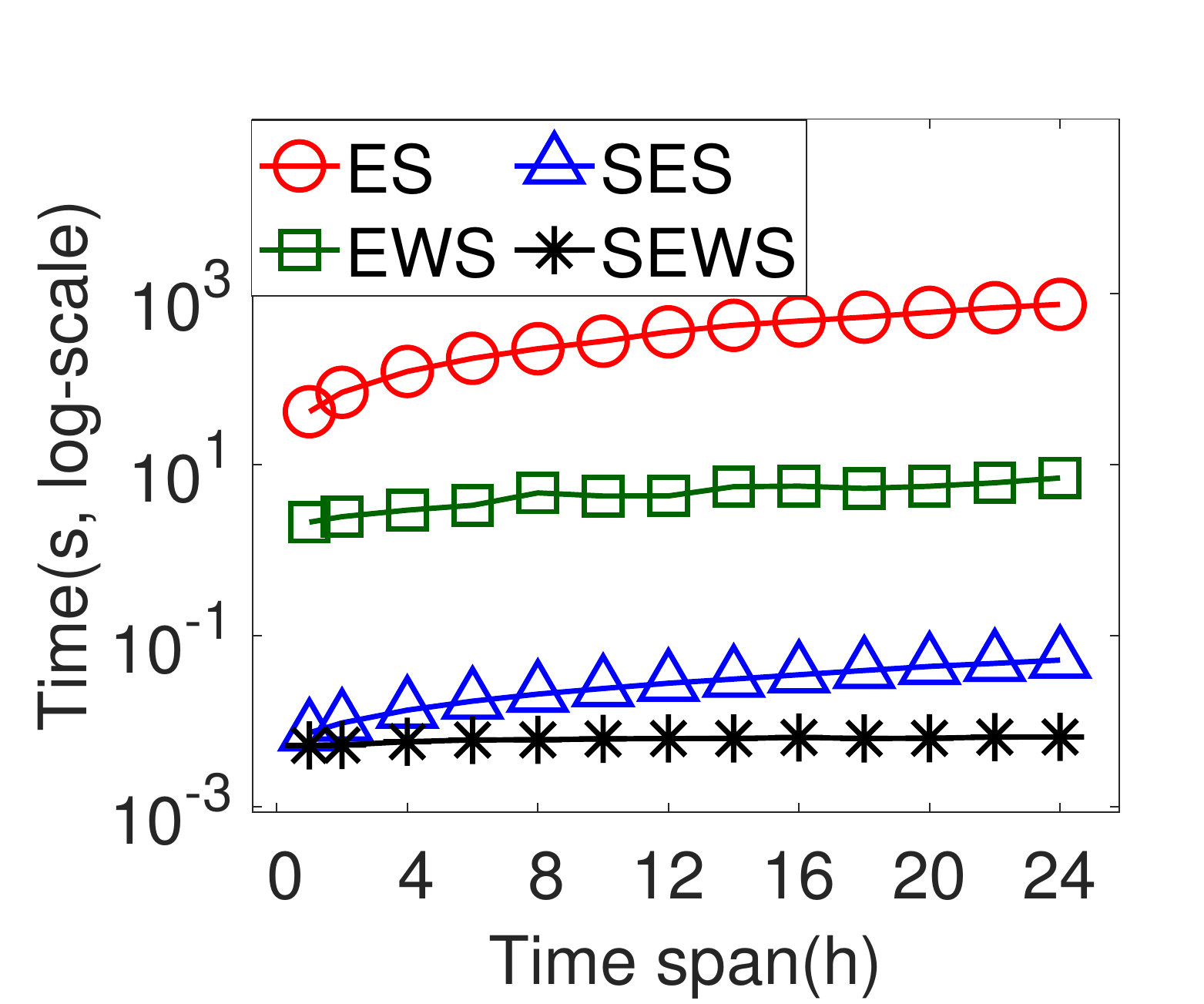}
  \caption{Relative error (\%) and average update time (in seconds) per 1,000 edges with varying time span $\delta$ in the streaming setting for query motif $Q3$ on the BC dataset.}
  \Description{s-scalability}
  \label{fig:s-delta}
\end{figure}

\textbf{Accuracy vs. Efficiency:}
Fig.~\ref{fig:p} demonstrates the trade-offs between \emph{relative error} and \emph{running time} of three sampling algorithms, namely IS-BT, ES, and EWS. For IS-BT, we fix the interval length to $30\delta$ and vary the interval sampling probability from $0.01$ to $1$. For ES and EWS, we vary the edge sampling probability $p$ from $0.0001$ to $0.25$. First of all, ES and EWS consistently achieve better trade-offs between accuracy and efficiency than IS-BT in almost all experiments. Specifically, ES and EWS can run up to $60$x and $330$x faster than IS-BT when the relative errors are at the same level. Meanwhile, in the same elapsed time, ES and EWS are up to $10.4$x and $16.5$x more accurate than IS-BT, respectively. Furthermore, EWS can outperform ES in all datasets except SO because of lower computational overhead. But on the SO dataset, since the distribution of motif instances is highly skewed among edges and thus the \emph{temporal wedge sampling} leads to large errors in estimation, the performance of EWS degrades significantly and is close to or even worse than that of ES. Nevertheless, the effectiveness of \emph{temporal wedge sampling} for EWS can still be confirmed
by the results on the BC and RC datasets.

\textbf{Scalability:}
We evaluate the scalability of different offline algorithms with varying the time span $\delta$ and dataset size $m$. In both experiments, we use the same parameter settings as used for the same motif on the same dataset in Table~\ref{tab:results}. We first test the effect of $\delta$ for $Q3$ on the BC dataset by varying $\delta$ from $1$ hour to $24$ hours. As shown in Fig.~\ref{subfig:delta}, the running time of all the algorithms we compare increases near linearly w.r.t.~$\delta$. BT runs out of memory when $\delta$ is longer than $10$ hours. The relative errors of ES and EWS keep steady for different values of $\delta$ but the accuracy of IS-BT degrades seriously when $\delta$ increases. This is led by the increase in cross-interval instances and the skewness of instances among intervals. Meanwhile, ES and EWS run up to $2.2$x and $180$x faster than IS-BT, respectively, while always having smaller errors. The results for $Q2$ on the RC dataset with varying $m$ are presented in Fig.~\ref{subfig:size}. Here, we vary $m$ from $50$M to near $400$M by extracting the first $m$ temporal edges of the RC dataset. The running time of all the algorithms we compare grows near linearly w.r.t.~$m$. The fluctuations of relative errors of IS-BT explicate that it is sensitive to the skewness of instances among intervals. ES and EWS always significantly outperform IS-BT for different values of $m$: they run much faster, have smaller relative errors, and provide more stable estimates than IS-BT.

\subsection{Experimental Results in Temporal Graph Streams}\label{subsec:exp:results1}

\textbf{Effect of Sample Size $r$:}
In Fig.~\ref{fig:r:error} and~\ref{fig:r:time}, we show the effect of the sample size $r$ on the performance of our algorithms, i.e., ES, EWS, SES, and SEWS, in the streaming setting. Since BT and IS-BT are both offline algorithms and have been shown to be much less efficient that ES and EWS in Section~\ref{subsec:exp:results}, we do not compare with them in the streaming setting anymore. In these experiments, we vary the edge sampling probability $p$ from $0.0001$ to $0.25$ for ES and EWS and the sample size $r$ of reservoir sampling from  $0.0001m$ to $0.25m$ for SES and SEWS, where $m$ is the number of temporal edges in the dataset.

The relative estimation errors of all algorithms with varying $r$ are shown in Fig.~\ref{fig:r:error}. Generally, the estimations of all four algorithms become more accurate when the sample size $r$ (or sampling probability $p$) increases. Moreover, the relative errors of SES and SEWS are slightly larger than those of ES and EWS for the same $r$. This is because SES and SEWS only consider the instances where a sampled edge is mapped to the last edge of the query motif for processing streaming data, whereas ES and EWS take all instances containing each sampled edge for higher accuracy. Nevertheless, SES and SEWS still provide estimates with at most $5\%$ relative errors  for query motifs $Q1$ to $Q4$ on all datasets when $r \geq 0.05m$. For query motif $Q5$, SES also achieves less than $5\%$ relative errors in estimations on all datasets when $r = 0.25m$.

The average time of each algorithm to update $1,000$ edges in the stream with varying $r$ are presented in Fig.~\ref{fig:r:time}. Generally, the update time of all algorithms increases with $r$. SES and SEWS always achieve one to three orders of magnitude speedups over ES and EWS for all query motifs across all datasets. Specially, the advantages of SES and SEWS in terms of efficiency become even greater for larger values of $r$. These results are mainly attributed to the fact that SES and SEWS update the counts incrementally over streams but ES and EWS need to recompute the counts from scratch when a new edge arrives. Finally, SEWS runs faster than SES because of wedge sampling, especially for larger $r$. Meanwhile, the estimation errors of SEWS are comparable to those of SES, as shown in Fig.~\ref{fig:r:error}.

\textbf{Accuracy of SES and SEWS over Time:}
The estimates of the number of instances returned by SES and SEWS over time are shown in Fig.~\ref{fig:T:number}. We illustrate the average, minimum, and maximum of the estimated counts over 10 runs by month on two large datasets BC and RC when $r=0.01m$. First of all, we observe that the number of instances of each motif on both datasets exhibits a sharp rise at some time (e.g., the $49$th and $83$rd months in Fig.~\ref{fig:T:m1:bt}) mostly because the edge distribution is highly skewed. But we see that different query motifs have different temporal distributions over time even on the same dataset (e.g., $Q1$ and $Q2$ on BC). Moreover, the estimation errors of both algorithms increase over time because a fixed sample size $r$ is used. When $r \geq m_t$, all observed edges are sampled and the estimation errors are exactly $0$ for SES and very close to $0$ for SEWS. When $r < m_t$, the estimation errors increase over time because the sampling rates become lower. Then, the average number of instances over $10$ estimations of SES and SEWS is almost indistinguishable from the ground truth, no matter how rapidly the instances are generated. In addition, SES shows higher accuracy and better stability in estimation than SEWS. We can see that the maximum and minimum of its estimates are always closer to the ground truth than those of SEWS. This is because SES counts exactly local motif instances of each sampled edge, whereas SEWS estimates local counts using wedge sampling.

\textbf{Effect of Time Span $\delta$:}
The performance of different algorithms with varying $\delta$ from $1$h to $24$h in the streaming setting for query motif $Q3$ on the BC dataset is shown in Fig.~\ref{fig:s-delta}. With the increase of $\delta$, the estimation errors generally decrease and the average update time per 1,000 edges increases. Nevertheless, SEWS runs the fastest for different $\delta$ among all four algorithms. And its update time is nearly stable when $\delta$ increases. Meanwhile, SEWS still has comparable relative errors with other algorithms. These results confirm that SEWS has the best scalability w.r.t.~the time span $\delta$.

\section{Conclusion}\label{sec:conclusion}

In this paper, we studied the problem of approximately counting a temporal motif in a temporal graph via random sampling. We first proposed a generic Edge Sampling (ES) algorithm to estimate the number of any $k$-vertex $l$-edge temporal motif in a temporal graph. Then, we improved the ES algorithm by combining edge sampling with wedge sampling and devised the EWS algorithm for counting $3$-vertex $3$-edge temporal motifs. Furthermore, we extended the ES and EWS algorithms to the SES and SEWS algorithms, respectively, for processing temporal graph streams using a reservoir sampling-based framework. We provided comprehensive theoretical analyses on the unbiasedness, variances, and complexities of our proposed algorithms. Extensive experiments on several real-world temporal graphs demonstrated the accuracy, efficiency, and scalability of our proposed algorithms. Specifically, ES and EWS ran up to $10.3$x and $48.5$x faster than the state-of-the-art sampling method while having lower estimation errors in the offline setting. In addition, SES and SEWS further achieved up to three orders of magnitude speedups over ES and EWS with comparable estimation errors in the streaming setting.

\bibliographystyle{ACM-Reference-Format}
\bibliography{references}

\end{document}